\documentclass[acmsmall, screen, nonacm]{acmart}

\usepackage{xfrac}
\usepackage{flushend}
\usepackage{multicol}
\usepackage{enumitem}
\usepackage{tikz}
\usepackage{bm}
\usepackage{wrapfig}
\usepackage{worldflags}
\usepackage{multirow}
\AtBeginDocument{%
  \providecommand\BibTeX{{%
    \normalfont B\kern-0.5em{\scshape i\kern-0.25em b}\kern-0.8em\TeX}}}

\usepackage{algorithm}
\usepackage[noend]{algpseudocode}

\usepackage{appendix}
\usepackage{lipsum}
\usepackage{geometry}
\usepackage{amsmath}
\usepackage{subcaption}
\usepackage{mathtools}
\usepackage{amsthm}
\usepackage{xspace}
\usepackage{nicefrac}
\usepackage{hyperref}
\usepackage{pifont}

\setcopyright{rightsretained}
\acmJournal{POMACS}

\usepackage[]{algpseudocode}

\newcommand{\sref}[2]{\hyperref[#2]{#1 \ref{#2}}}

\makeatletter
\newcommand{\algmargin}{\the\ALG@thistlm}

\makeatother
\newlength{\whilewidth}
\algdef{SE}[parWHILE]{parWhile}{EndparWhile}[1]
{\parbox[t]{\dimexpr\linewidth-\algmargin}{%
		\hangindent\whilewidth\strut\algorithmicwhile\ #1\ \algorithmicdo\strut}}{\algorithmicend\ \algorithmicwhile}%
\algnewcommand{\parState}[1]{\State%
	\parbox[t]{\dimexpr\linewidth-\algmargin}{\strut #1\strut}}
\makeatother

\usepackage{etoolbox}
\makeatletter
\patchcmd{\hyper@makecurrent}{%
    \ifx\Hy@param\Hy@chapterstring
        \let\Hy@param\Hy@chapapp
    \fi
}{%
    \iftoggle{inappendix}{%
        \@checkappendixparam{chapter}%
        \@checkappendixparam{section}%
        \@checkappendixparam{subsection}%
        \@checkappendixparam{subsubsection}%
        \@checkappendixparam{paragraph}%
        \@checkappendixparam{subparagraph}%
    }{}%
}{}{\errmessage{failed to patch}}

\newcommand*{\@checkappendixparam}[1]{%
    \def\@checkappendixparamtmp{#1}%
    \ifx\Hy@param\@checkappendixparamtmp
        \let\Hy@param\Hy@appendixstring
    \fi
}
\makeatletter

\newtoggle{inappendix}
\togglefalse{inappendix}

\apptocmd{\appendix}{\toggletrue{inappendix}}{}{\errmessage{failed to patch}}
\apptocmd{\subappendices}{\toggletrue{inappendix}}{}{\errmessage{failed to patch}}

\DeclareMathOperator*{\argmin}{arg\,min}

\newcommand{\OSDM}{\texttt{OSDM}\xspace}
\newcommand{\OSDMS}{\texttt{OSDM-S}\xspace}
\newcommand{\OSDMT}{\texttt{OSDM-T}\xspace}
\newcommand{\PAAD}{\texttt{PAAD}\xspace}
\newcommand{\PALD}{\texttt{PALD}\xspace}
\newcommand{\PALDS}{\texttt{PALD-S}\xspace}
\newcommand{\PALDC}{\texttt{PALD-C}\xspace}
\newcommand{\ALG}{\texttt{ALG}\xspace}
\newcommand{\OPT}{\texttt{OPT}\xspace}
\newcommand{\ADV}{\texttt{ADV}\xspace}
\newcommand{\OCS}{\texttt{OCS}\xspace}
\newcommand{\RORO}{\texttt{RORO}\xspace}
\newcommand{\instance}{\texttt{instance}\xspace}

\colorlet{blue}{black}

\begin{document}

\title{Online Smoothed Demand Management}

\author{Adam Lechowicz}
\affiliation{%
  \institution{University of Massachusetts Amherst}
  \country{USA}
}
\email{alechowicz@cs.umass.edu}

\author{Nicolas Christianson}
\affiliation{%
  \institution{Stanford University}
  \country{USA}
}
\email{christianson@stanford.edu}

\author{Mohammad Hajiesmaili}
\affiliation{%
  \institution{University of Massachusetts Amherst}
  \country{USA}
}
\email{hajiesmaili@cs.umass.edu}

\author{Adam Wierman}
\affiliation{%
  \institution{California Institute of Technology}
  \country{USA}
}
\email{adamw@caltech.edu}

\author{Prashant Shenoy}
\affiliation{%
  \institution{University of Massachusetts Amherst}
  \country{USA}
}
\email{shenoy@cs.umass.edu}
\renewcommand{\shortauthors}{Lechowicz et al.}

\begin{abstract}

We introduce and study a class of online problems called online smoothed demand management $(\texttt{OSDM})$, motivated by paradigm shifts in grid integration and energy storage for large energy consumers such as data centers.  In $\texttt{OSDM}$, an operator makes two decisions at each time step: an amount of energy to be purchased, and an amount of energy to be delivered (i.e., used for computation).  The difference between these decisions charges (or discharges) the operator's energy storage (e.g., a battery).  Two types of demand arrive online: base demand, which must be covered at the current time, and flexible demand, which can be satisfied at any time steps before a demand-specific deadline $\Delta_t$.  The operator's goal is to minimize a cost (subject to the constraints above) that combines a cost of purchasing energy, a cost for delivering energy (if applicable), and smoothness penalties on the purchasing and delivery rates to discourage fluctuations and encourage ``grid healthy'' decisions.  
$\texttt{OSDM}$ generalizes several problems in the online algorithms literature while being the first to fully model applications of interest.  We propose a competitive algorithm for $\texttt{OSDM}$ called $\texttt{PAAD}$ (partitioned accounting \& aggregated decisions) and show it achieves the optimal competitive ratio.  To overcome the pessimism typical of worst-case analysis, we also propose a novel learning framework for $\texttt{OSDM}$ that provides guarantees on the worst-case competitive ratio (i.e., to provide robustness against nonstationarity) while allowing end-to-end differentiable learning of the best algorithm on historical instances of the problem.  We evaluate our algorithms in a case study of a grid-integrated data center with battery storage, showing that $\texttt{PAAD}$ effectively solves the problem and end-to-end learning achieves substantial performance improvements compared to $\texttt{PAAD}$.

\end{abstract}

\maketitle

\section{Introduction}
\label{sec:intro}
The proliferation of renewable energy sources and grid-scale energy storage, the growth of power-hungry AI data centers, and the implementation of dynamic pricing schemes like time-of-use tariffs are fundamentally reshaping the landscape of how energy and compute infrastructures interact. These changes are prompting a paradigm shift in how systems interact with power grids, shifting from a model of passive consumption to one of \textit{active, intelligent demand management}. 
The core challenge in demand management from an energy perspective is to better align the consumption of energy (e.g., from the grid) with external signals such as time-of-use prices, while meeting certain performance (e.g,. service-level agreement (SLA)) objectives on the demand side.
Concurrently, there is a need to consider the ``smoothness'' of these decisions: recent reports suggest that data centers' sudden large fluctuations in energy use are impacting power quality in U.S. grids~\cite{Bloomberg:24}, and sudden changes in load are known to cause cascading failures and even blackouts on the grid, such as the 2025 Iberian Peninsula blackout~\cite{Pons:2025}. 
This prompts an additional objective of \textit{smooth decisions} (e.g., gradual ramping up and down).

Many techniques have been developed for this class of problems, with solution designs that are usually tailored to specific applications such as energy storage management~\cite{Mamun:16:multi-objective-DR-battery} and HVAC control~\cite{Kou:21:HVAC}, among others.  These techniques often rely on machine learning somewhere in the pipeline, e.g., to provide predictions for a model predictive controller~\cite{Serale:18:MPC-HVAC, Morstyn:17:MPC-battery}, inform a day-ahead provisioning problem~\cite{Yang:25:CarbonAwareHVAC}, or to directly make decisions (e.g., using a reinforcement learning agent)~\cite{Azuatalam:20}.  
While these problem-specific solutions show promise, they are seldom deployed in practice.  One reason is that environments such as electricity markets are inherently nonstationary, making it difficult to develop heuristics that deal well with distribution shift~\cite{Levin:22:Texas21} or come with rigorous performance guarantees.
Another reason is that these approaches are typically developed with a single end application in mind, making generalization to other settings difficult. 

On the other side of the spectrum, there is a rich body of work in online algorithms~\cite{Yang:20,Lechowicz:24,ElYaniv:01, Yang:24, Chen:22:InBev} on related problems that are sufficiently abstract to model many different applications, and have been studied under frameworks such as competitive analysis, which provides rigorous performance guarantees.  
Despite this, however, there remains a gap between the existing models studied in the online algorithms literature and the complexity of real-world demand management problems---as of yet, no models exist that fully capture the moving pieces necessary to study the type of problem that arises in, e.g., a grid-integrated data center with local energy storage.

Informed by these challenges, we introduce and study a new class of online optimization problems called \textbf{online smoothed demand management (\OSDM)}.  \OSDM models the challenge faced by an operator who satisfies a stream of incoming demand by purchasing energy from a market with fluctuating prices (e.g., an electricity grid).  We model several features motivated by practice, namely: (i) a mix of inflexible (base) and deferrable (flexible) demands, each with their own unique SLA constraints; (ii) opportunities to arbitrage using local energy storage (e.g., a battery), which may come at a cost; and (iii) a few smoothing cost models for ``unsmooth'' behavior, penalizing large changes in both energy purchasing (ramping) and demand-side delivery.  These costs capture some desired behavior for any algorithmic scheme from both the energy provider (e.g., grid) perspective, as well as the consumer's perspective.  In studying \OSDM, we answer the following question:
\begin{center}
    \textit{Is it possible to design online algorithms for \OSDM that manage the challenges of flexibility, uncertain demand, and penalties for unsmooth decisions while providing meaningful competitive guarantees?}
\end{center}

\noindent Competitive (i.e., worst-case) guarantees are often overly pessimistic in practice, since real-world problem instances are not adversarial.  In recent years, the area of \textit{learning-augmented algorithms}~\cite{Lykouris:18, Purohit:18, Mitzenmacher:22:ALPS} has expanded the online algorithms literature with techniques that can take additional information as input (e.g., predictions about the instance) to improve average-case performance while retaining worst-case guarantees for robustness against nonstationarities.  While this paradigm has successfully been applied to problems that share similarities with \OSDM~\cite{Lechowicz:24,SunLee:21}, traditional learning-augmented designs make strong assumptions about the availability and learnability of \textit{suitable predictions}, which face challenges in \OSDM due to the multiple kinds and nonstationary sources of uncertainty (e.g., demand arrivals, deadlines, and market prices).
Thus, we also ask:
\begin{center}
    \textit{{\color{blue} Can we design data-driven algorithms for \OSDM that directly learn from data (e.g., historical instances) to significantly improve performance while maintaining worst-case guarantees?}}
\end{center}

\subsection{Contributions}

We obtain positive answers to both of the above questions, making the following contributions:
In \autoref{sec:problem}, we formally introduce the \OSDM problem, which captures several realistic features of demand management problems that arise in practice.  To the best of our knowledge, \OSDM is the first online optimization problem that simultaneously captures a mix of flexible and inflexible demands, inventory and/or energy storage dynamics, separate costs for procurement and delivery, and penalties that encourage ``smooth'' behavior.

In \autoref{sec:alg}, we provide the first competitive algorithm for \OSDM, which we call \PAAD (Partitioned Accounting \& Aggregated Decisions, see \autoref{alg:paad}).
\PAAD introduces a driver abstraction that splits \OSDM into subproblems that can be solved independently and aggregated to form a globally competitive solution (\autoref{thm:osdm-upper-bound-1}). 
\PAAD reflects, to the best of our knowledge, the most complex extension of the online pseudo-cost minimization algorithm design paradigm~\cite{SunZeynali:20}; we employ multiple kinds and instances of \textit{threshold functions} to decompose and accommodate the significant heterogeneity of the \OSDM setting.
\PAAD achieves a worst-case competitive ratio that is sublinear in $\nicefrac{p_{\max}}{p_{\min}}$ (the ratio of the procurement price bounds) and linear in other parameters of the problem setting such as the switching and delivery cost magnitudes---in \autoref{sec:fundamental-limits}, we give lower bounds that imply \PAAD achieves the best-possible competitive ratio for \OSDM under our assumptions.

To improve average-case performance and reduce the pessimism of our approach, in \autoref{sec:learn} we present a learning framework for \OSDM called \PALD (Partitioned Accounting \& Learned Decisions).  Compared to existing work on learning-augmented algorithms, the \OSDM setting brings unique challenges that make previously-considered techniques for integrating predictions impractical or ineffective.  To address this challenge, \PALD proposes a novel end-to-end differentiable learning framework that directly predicts the best algorithm parameters (e.g., threshold functions) that minimize \PALD's cost on historical instances, {\color{blue} while ensuring that the learned algorithm satisfies a \textit{robustness certificate} (see \autoref{thm:pald-robustness}) that guarantees a bounded worst-case competitive ratio under adversarial instances, serving as a safeguard against nonstationarity.} This framework will likely be of independent interest as a more flexible and expressive methodology to learn high-performing and robust algorithms across problem domains.

Finally, in \autoref{sec:exp}, we present a case study of \OSDM in a simulation of a grid-integrated data center with co-located battery storage, using real electricity price data and production Alibaba traces~\cite{Alibaba:18}.  We show that \PAAD solves the \OSDM problem effectively in practice, and that \PALD's end-to-end learning capability significantly improves performance compared to \PAAD using data-driven insights.

\subsection{Related Work}

Our study of \OSDM is informed by and builds upon several lines of related work in online algorithms, learning-augmented algorithms, and relevant application areas.  We briefly summarize these below.

\smallskip
\noindent \textbf{Applications of Demand Management and Related Ideas. \ }
\OSDM-like problems arise across multiple domains---%
below, we review some heuristic and control-theoretic approaches that have been proposed to solve these.  Given forecasts (e.g., prices), a common approach is to solve an offline “day-ahead’’ optimization problem, as applied in battery arbitrage~\cite{MohsenianRad:16:day-ahead-battery, Zhang:24:day-ahead-battery, Faraji:20:day-ahead-prosumer}, microgrids~\cite{Silva:20:day-ahead-microgrid, Tandukar:18:day-ahead-VPP}, HVAC~\cite{Yang:25:CarbonAwareHVAC, Zeng:23:day-ahead-AC}, and sustainable computing~\cite{radovanovic2022carbon, Li:17:battery-DC}.
In online settings, a frequently-applied technique is model predictive control (MPC), which uses forecasts to solve a finite-horizon optimization problem at each time step. MPC has been applied to battery management~\cite{Blonsky:20:MPC-bat-management, chehade:2024:MPC-bat-management, Sage:24:RL-battery-storage, Morstyn:17:MPC-battery},  building systems~\cite{Afram:15:MPC-HVAC, Serale:18:MPC-HVAC, Blum:22:MPC-HVAC, Mariano:21:BMS}, microgrids~\cite{MorenoCastro:23:microgrids, Lopez:21:MPC-hybridpp}, and data centers~\cite{Wang:23:hierarchical, Lazic:18:MPC-BMS, Mirhoseininejad:21:DC}.
Select works consider approaches that do not require predictions such as Lyapunov-based control~\cite{Urgaonkar:11:lyapunov-power-cost-management, Guo:13:Lyapunov-battery-DC, Yu:15:Lyapunov-battery-DC, Stai:21:BatteryStorage, Sun:23:Lyapunov-battery-DC}, ADMM for distributed storage~\cite{RigoMariani:22:ADMM}, feedback-based demand response~\cite{Mamun:16:multi-objective-DR-battery}, and stochastic optimization~\cite{Nasiriani:17:stochastic-peak-shaving, Bedi:18}.
These approaches are generally developed with a single application in mind, limiting their generalizability, and do not typically come with theoretical guarantees on performance, particularly in cases where predictions or stochastic models are used.

\smallskip
\noindent \textbf{Related Online Algorithms Problems. \ }
In the literature %
there are a few online problems that share structure with \OSDM.
Of these, perhaps the closest related problem is online linear inventory management (\texttt{OLIM})~\cite{Yang:20}, which models the problem faced by an operator (e.g., a data center) with an inventory (e.g., battery) that must meet a demand at each time step, while minimizing their procurement cost subject to time-varying prices.  Compared to \texttt{OLIM}, \OSDM additionally models a delivery cost to serve demand (i.e., separate from procurement), smoothing penalties, and a mixture of base (inelastic) and flexible demand.
A recent model called online conversion with switching costs (\texttt{OCS})~\cite{Lechowicz:24} has modeled the problem faced by an operator procuring energy (e.g., an EV charger) that must meet a single flexible demand before a deadline (e.g., charging an EV by the end of the day), paying a switching cost when decisions change.
\OSDM generalizes this by including inflexible (i.e., base) demand constraints, delivery costs, and multiple flexible units with heterogeneous deadlines.
Other formulations address energy storage without time-varying procurement prices~\cite{Kim:16, Mo:21} or microgrid scheduling with variable prices but no storage~\cite{Lu:13, Hajiesmaili:16, Zhang:2018:microgrid}.

Each component dynamic of \OSDM has been studied in isolation: procurement and flexible demand via online knapsack and search~\cite{Zhou:08,ElYaniv:01,Fung:21}, per-step demand via online matching~\cite{Mehta:07, Huang:24}, and smoothness penalities via metrical task systems~\cite{Borodin:92, FriedmanLinial:93}
However, the combination of these in \OSDM introduces new technical challenges in the design of robust and efficient algorithms that preclude the direct application of existing results (see, e.g., \autoref{sec:warmup}).  
Online resource allocation problems~\cite{jaillet2011online, Devanur:19, Wang:17:ORA, Jiao:17:ORA, Klimm:19, Zhang:17:ORA, Balseiro:23} are also related, though they usually assume a fixed (non-replenishable) resource. Some variants consider inventories~\cite{Larsen:10, Cheung:21, Chen:22:InBev, He:22, HIHAT:23, Yang:24}, but typically treat either demand or inventory evolution as exogenous, whereas in \OSDM both are decisions.

\section{Problem Formulation} 
\label{sec:problem}

In this section, we introduce the online smoothed demand management (\OSDM) problem, provide motivating examples as context for our modeling, and state assumptions we place on the problem.

\subsection{Online Smoothed Demand Management (\OSDM)} \label{sec:OSDM_statement}

\noindent \textbf{Problem Statement. \ }
{\color{blue} Consider an operator who makes decisions over a discrete time horizon $t \in [T]$.} At each time step $t$, several inputs are revealed online.  The first of these is a \textit{market price} for energy, denoted by $p_t \geq 0$.  Furthermore, two types of demand arrive: a base (inflexible) demand $b_t \ge 0$ that must be satisfied immediately, and a flexible demand $f_t \ge 0$ that can be deferred until a specified deadline $\Delta_t > t$.  Further, the operator can control the state of a local energy storage unit (e.g., a battery) with a maximum capacity of $S \geq 0$. We denote the state of charge at time step $t$ as $s_{t}$, where the initial storage state is empty (i.e., $s_0 = 0$).\footnote{The assumption that the initial storage state is $s_0 = 0$ is without loss of generality---our results hold when $s_0 > 0$.}

At each time step $t$, the operator must make two decisions without knowledge of the future.  First, they must decide how much energy to purchase, denoted by $x_t \ge 0$.  Second, they must decide how much energy to deliver to satisfy demand, denoted by $z_t \ge 0$.  The delivery decision can use a combination of newly purchased energy and energy drawn from storage.
The operator's goal is to satisfy the above constraints while minimizing a cost consisting of four components: the cost of purchasing energy, penalties for an unsmooth purchasing rate (e.g., a switching or tracking cost), a cost for delivering energy (which may depend on the storage state and the market price), and a penalty for an unsmooth delivery rate (we consider a switching cost). The offline (i.e., full time horizon) version of \OSDM is formalized as follows:
\begin{align}
    [\OSDM] \min_{\{x_t, z_t\}_{t\in [T]}} & \sum\nolimits_{t=1}^{T} \Big [ \underbrace{p_t x_t}_{\text{purchase cost}}  + \underbrace{\mathcal{D}(z_t, s_{t-1}, p_t)}_{\text{delivery cost}} \Big ] + \sum\nolimits_{t=1}^{T+1} \Big[ \underbrace{ \mathcal{S}(x_t, x_{t-1}) + \delta \vert z_t - z_{t-1} \vert}_{\text{smoothing penalty}} \Big] \label{eq:osdm-objective},\\ 
    \text{ s.t. } \quad & {\color{blue} s_t = s_{t-1} + x_t - z_t, \text{  with  } 0 \le s_t \le S, \nonumber} \\
    &b_t \le z_t \le b_t + \sum\nolimits_{\tau=1}^{t} f_\tau - \sum\nolimits_{\tau: \Delta_\tau < t} f_\tau, \nonumber \\
    & {\color{blue} \sum\nolimits_{t = \tau_1}^{\tau_2} z_t - b_t \geq \sum\nolimits_{t : t \geq \tau_1 \text{ and } \Delta_t \leq \tau_2} f_t \quad \quad\forall (\tau_1, \tau_2) : \tau_1 \leq \tau_2 \text{ and } \tau_1, \tau_2 \in [T].}
\end{align}
In the objective \eqref{eq:osdm-objective}, $\delta \geq 0$ is a known coefficient for the delivery switching penalty, $\mathcal{S}(\cdot)$ is a function that implements the desired smoothness penalty on purchasing decisions, and $\mathcal{D}(\cdot)$ is a function that captures the delivery cost. For both of the latter terms, we consider several models inspired by real-world applications---see \autoref{sec:smoothness-penalty-def} and \autoref{sec:delivery-cost-def} for a discussion of each, respectively.  We assume decisions outside of the time horizon are zero: $x_0 = z_0 = 0$ and $x_{T+1} = z_{T+1} = 0$.

As outlined above, the problem is subject to several constraints for all $t \in [T]$.  In order, they are as follows:  First, the storage state is governed by the inventory dynamics $s_t = s_{t-1} + x_t - z_t$, and the state must remain within physical limits, i.e., $0 \le s_t \le S$.  \OSDM models a ``discharged'' storage state as $s_t = 0$, but this can capture the more realistic case where $s_t = 0$ corresponds to a minimum charge level for the storage.
Next, the amount of energy purchased, i.e., $x_t$, must be sufficient to cover the desired delivery decision, potentially drawing from storage, but cannot overfill the storage. These constraints are captured by $z_t - s_{t-1} \le x_t \le z_t + (S - s_{t-1})$.  
{\color{blue}
The delivery decision at time $t$ ($z_t$) must \emph{cover} the base demand and is allowed to satisfy any flexible demand that has arrived but is not past due---this is captured by $b_t \le z_t \le b_t + \sum_{\tau=1}^{t} f_\tau - \sum_{\tau: \Delta_\tau < t} f_\tau$. Finally, the cumulative delivery for any subinterval of $[T]$ must satisfy base demand and any flexible demands whose lifetime lies within this interval: this is captured by $\sum\nolimits_{t = \tau_1}^{\tau_2} z_t - b_t \geq \sum\nolimits_{t : t \geq \tau_1 \text{ and } \Delta_t \leq \tau_2} f_t \; \forall (\tau_1, \tau_2) : \tau_1 \leq \tau_2 \text{ and } \tau_1, \tau_2 \in [T]$.
We summarize the key pieces of \OSDM in \autoref{fig:osdm-diagram}.
We remark that if flexible demands are not considered, the two decision variables in the \OSDM formulation can be simplified to a single decision variable (i.e., $\{ x_t - z_t \}_{t\in [T]}$)---this simplified version of \OSDM is of independent interest, and we refer to it as ``\OSDM with only base demand'' throughout the paper (see e.g., \sref{Corollary}{cor:osdm-lower-bound-base-demand} for tighter results in this regime).
}

While \eqref{eq:osdm-objective} captures the offline version of the problem, our goal is to design an online algorithm that chooses $\{x_t, z_t\}$ at each time step $t$ \textit{without knowledge of the future online inputs} of $\{ p_\tau, b_\tau, f_\tau, \Delta_\tau \}_{\tau > t}$.  To do so, we follow the literature on online algorithms and competitive analysis~\cite{Manasse:88, Borodin:92} and evaluate performance via the \textit{competitive ratio}~\cite{Manasse:88, Borodin:92}, defined as follows:

\begin{definition}[Competitive ratio] \label{dfn:comp-ratio}
{\it
Let $\OPT(\mathcal{I})$ denote the optimal offline cost for an arbitrary \OSDM input instance $\mathcal{I} \in \Omega$ (where $\Omega$ is the set of all feasible inputs), and let $\ALG(\mathcal{I})$ denote the cost of an online algorithm $\ALG$ over the same instance.
$\ALG$ is said to be $\alpha$-competitive if for all $\mathcal{I}$, $\ALG(\mathcal{I}) \leq \alpha \OPT(\mathcal{I}) + C$ holds, where $\alpha \geq 1$ is the \textit{competitive ratio}, and $C \geq 0$ is an additive constant.}
\end{definition}
\noindent Note that a \textit{smaller} competitive ratio implies that the online algorithm is guaranteed to be \textit{closer} to the offline optimal solution. The competitive ratio is a \textit{worst-case} performance metric for online algorithms; as such, bounds on this metric ensure robustness against nonstationarity in the underlying environment, which is crucial for the applications motivating \OSDM. For instance, electricity prices are known to be volatile and non-stationary~\cite{Qu:24}, and demand profiles may be non-stationary in certain applications.
However, algorithms optimized for competitive guarantees are often overly conservative in practice---to address this limitation in \OSDM, we build on the learning-augmented algorithms literature~\cite{Lykouris:18, Purohit:18,Mitzenmacher:22:ALPS} to propose data-driven algorithms that can leverage historical data to significantly improve average-case performance while still maintaining worst-case competitive guarantees against non-stationary inputs.  We discuss this further in \autoref{sec:learn}.

\begin{figure*}[t]
    \centering
    \vspace{-1em}
    \includegraphics[width=\textwidth]{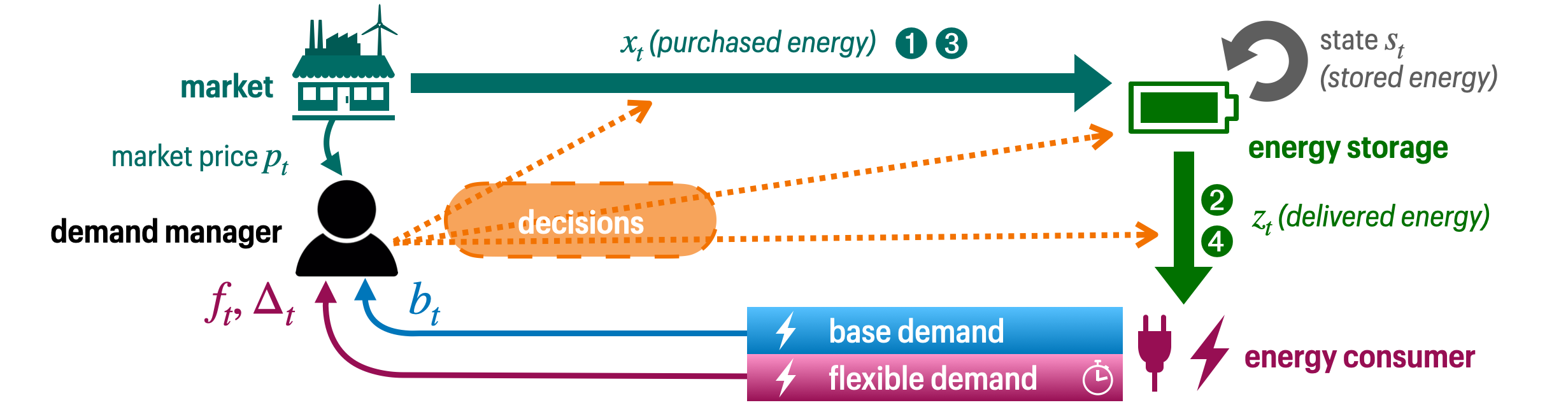} \vspace{-1em}
    \caption{Diagram of the \OSDM problem at a single (discrete) time step $t$. The demand manager specifies a purchasing decision $x_t$, a delivery decision $z_t$, and (implicitly) a storage state $s_t$.  As a function of these decisions, they pay \ding{202} a purchasing cost $p_t x_t$, \ding{204} a smoothness penalty $\mathcal{S}(x_t, x_{t-1})$ that discourages fluctuations in the purchasing rate, \ding{203} a delivery cost of serving demand $\mathcal{D}(z_t, s_{t-1}, p_t)$, and \ding{205} a switching penalty $\delta \vert z_t - z_{t-1} \vert$ that discourages fluctuations in the delivery rate. 
    }
    \label{fig:osdm-diagram}
\end{figure*}

\subsection{Motivating Applications} \label{sec:examples}

We present several \textit{demand management} applications that can be modeled by \OSDM, motivated by emerging \textit{grid integration} problems. Owing to its generality, \OSDM applies broadly across energy and operations domains. We focus here on the high-level mapping from problems to their \OSDM instantiations, deferring a longer discussion to \autoref{apx:examples}.

\smallskip
\noindent \textbf{Grid-integrated Data Center with Storage. \ } Consider a grid-connected data center with local energy storage.
At each time step $t$, the operator decides how much electricity to purchase ($x_t$) from the grid at price $p_t$, and how much energy to use to meet demand ($z_t$); the difference is stored in or drawn from the battery with state $s_t$.
Base demands $b_t$ represent e.g., interactive workloads that must be satisfied immediately, while flexible demands $f_t$ capture delay-tolerant batch jobs with deadlines $\Delta_t > t$.
The operator minimizes their cost of electricity subject to a smoothing penalty ($\mathcal{S}(\cdot)$) that penalizes them for large fluctuations, as sudden shifts in load cause instability and even blackouts on the power grid~\cite{Kwon:25}.
Additional penalties model the cost of discharging the battery ($\mathcal{D}(\cdot)$) to account for degradation, promote keeping the battery sufficiently charged to accommodate the risk of an outage, and smooth the rate of internal consumption ($\delta |z_t - z_{t-1}|$) to reduce wear-and-tear on components.
Unlike prior demand-management formulations~\cite{Yang:20}, \OSDM jointly captures both flexible and inflexible workloads and the need for smooth grid- and demand-side decisions.  This application is the focus of our case study in \autoref{sec:exp}.

\smallskip
\noindent \textbf{Thermal Energy Demand Management. \ } An operator manages a thermal energy system (for example, district heating) with a local storage tank~\cite{NREL:ThermalStorage}.
At each time step $t$, the operator decides how much thermal energy to purchase ($x_t$) at price $p_t$ and how much to deliver to meet demand ($z_t$); the difference is stored in or drawn from the tank with state $s_t$.
Base demands $b_t$ represent immediate thermal needs, while flexible demands $f_t$ can be deferred until a deadline $\Delta_t>t$.
{\color{blue}
The operator minimizes the total cost to procure energy $p_tx_t$ and deliver energy (e.g., using heat exchangers) $\mathcal{D}(\cdot)$, all subject to a smoothing penalty that models power grid stability needs (e.g., if electricity is used to generate thermal energy), since demand spikes can cause grid failure~\cite{FT:25:europe}. }
An optional penalty on the delivery side $\delta |z_t-z_{t-1}|$ models wear-and-tear on thermal components.

\smallskip
\noindent \textbf{Just-in-Time Manufacturing with Material Inventory. \ } 
A factory converts raw materials into finished products using on-site inventory.
At each step $t$, it decides how much material to purchase ($x_t$) at price $p_t$ and how much product to produce/deliver ($z_t$), maintaining warehouse state $s_t$.
Base orders $b_t$ must be filled immediately, while flexible orders $f_t$ can be deferred until deadline $\Delta_t>t$.
{\color{blue}
The objective is to minimize a purchase cost and a fixed per-unit processing cost $p_tx_t + \mathcal{D}(\cdot)$, while smoothing the purchasing rate ($\mathcal{S}(\cdot)$) and production rate ($\delta |z_t-z_{t-1}|$).
}

\smallskip
\noindent \textbf{Flow Battery Storage Management. \ }
A large electricity consumer operates a grid-connected flow battery to shift energy usage~\cite{Breeze:19}.
At each time $t$, the operator purchases electricity ($x_t$) at price $p_t$ and delivers electricity to meet demand ($z_t$), with battery charge state $s_t$.
{\color{blue}
The goal is to minimize the total cost $p_tx_t$ subject to a smoothing penalty informed by grid considerations as above, and accounting for system losses and auxiliary energy needs through $\mathcal{D}(\cdot)$.
}

\subsection{Modeling Smoothness Penalty} \label{sec:smoothness-penalty-def}

In \OSDM, the smoothness penalty $\mathcal{S}(\cdot)$ in \autoref{eq:osdm-objective} captures an objective to make ``smooth'' purchasing decisions over time. This is often an important goal (see, e.g., the applications in \autoref{sec:examples}); for instance, industrial consumers like data centers can have electricity consumption patterns with large fluctuations that destabilize the grid and reduce power quality~\cite{Bloomberg:24}.
Motivated by these considerations, we consider two natural models for the smoothness penalty $\mathcal{S}(\cdot)$.
First, we consider a \textit{switching cost} model (with corresponding problem referred to as \OSDMS, hereinafter), where the penalty charges the operator proportionally to the \textit{change} in purchasing decisions:
\vspace{-0.2em}
\begin{definition}[Switching cost for purchasing decisions] \label{def:osdm-switching}
{\it
    In \eqref{eq:osdm-objective}, let $\mathcal{S}(x_t, x_{t-1}) \coloneqq \gamma |x_t - x_{t-1}|$, where $\gamma \geq 0$ is a known coefficient. }
\end{definition}
\vspace{-0.2em}

\noindent Second, we consider a \textit{tracking cost} model (\OSDMT, hereinafter), where the goal is not just to smooth the purchasing rate, but to follow an externally-provided target signal $\{ a_t \}_{t \in [T]}$.  
This is motivated by recent interest in programs where grid operators provide \textit{load shaping} signals to large electricity consumers to encourage consumption patterns that aid grid stability~\cite{Jiang:14:LoadShaping,Aryandoust:17:LoadShaping}.
This target signal may be known in advance (e.g., day-ahead) or revealed online (e.g., real-time), and the tracking cost penalizes the operator based on the deviation from this target signal. 
\vspace{-0.2em}
\begin{definition}[Tracking cost for purchasing decisions] \label{def:osdm-tracking}
{\it
    In \eqref{eq:osdm-objective}, let $\mathcal{S}(x_t, x_{t-1}) \coloneqq \eta |x_t - a_t|$,
    where $\eta \geq 0$ is a tracking cost coefficient known a priori, and $a_t \geq 0$ is a target signal revealed at time $t$.}
\end{definition}
\vspace{-0.2em}
{\color{blue}
We remark that the \OSDM objective in \autoref{eq:osdm-objective} models the smoothness penalty for the delivery cost (i.e., decisions $\{z_t\}_{t \in [T]}$) as a fixed switching cost (with coefficient $\delta$).  In principle, this penalty could also be modeled as a tracking cost, and we believe that analogous results are attainable in such a model.  We focus on the case where specifically the purchasing decisions are subject to a tracking cost for brevity---in a motivating application such as electricity purchasing, the grid operator provides suitable signals (e.g., reserve capacity signals~\cite{Emerald:25}) to realize this model (on the purchasing side) in practice.
}

\subsection{Modeling Delivery Cost} \label{sec:delivery-cost-def}

The delivery cost $\mathcal{D}(\cdot)$ in \autoref{eq:osdm-objective} represents the additional expense incurred when delivering energy to the consumer. While often negligible ($\mathcal{D}(\cdot)=0$), it can vary significantly across applications.
In a thermal energy system, for instance, efficiency improves when the reservoir (i.e., storage) is fuller or when the temperature differential is small, implying a delivery cost that \textit{decreases} with the storage state and depends on energy price due to pumping or heat-exchange losses~\cite{Ruhnau:19}.
In a just-in-time manufacturing context, $\mathcal{D}(\cdot)$ may reflect a fixed processing cost independent of storage or price.
On the other hand, in a flow battery storage system, using stored energy requires some additional input energy to drive pumps and circulate liquid~\cite{Breeze:19}, yielding a cost that \textit{increases} with discharge volume and depends on the contemporaneous energy price.

{\color{blue}
To capture this diversity of cases across multiple applications, in what follows, we define a class of monotone and price-dependent delivery costs. We believe this form of delivery cost to be a particularly challenging case that, if analyzed, offers corollary insights for other delivery cost models.  It is defined as follows:
}

\begin{definition}[Monotone affine price-dependent delivery cost] \label{def:osdm-delivery-pd-monotone}
{\it
{\color{blue}
    Let $c \geq 0$ and $\varepsilon \geq 0$ be delivery cost coefficients known a priori.  We say that the delivery cost $\mathcal{D}(\cdot)$ in \eqref{eq:osdm-objective} is a \textbf{monotone affine price-dependent delivery cost} if it satisfies one of the following:} {\color{blue}
    \begin{enumerate}[label=(\alph*)]
        \item \textit{(Increasing case)} Let $\mathcal{D}(z_t, s_{t-1}, p_t) \coloneqq (\frac{c s_{t-1}}{S} + \varepsilon) p_t z_t$---
        this delivery cost increases linearly from $\varepsilon p_t z_t$ to $(c + \varepsilon) p_t z_t$ as the storage state increases from $s_{t-1} = 0$ to $s_{t-1} = S$.
        \item \textit{(Decreasing case)} Let $\mathcal{D}(z_t, s_{t-1}, p_t) \coloneqq (c - \frac{c s_{t-1}}{S} + \varepsilon) p_t z_t$---
        this delivery cost decreases linearly from $(c + \varepsilon) p_t z_t$ to $\varepsilon p_t z_t$ as the storage state increases from $s_{t-1} = 0$ to $s_{t-1} = S$.
    \end{enumerate}}
}
\end{definition}

{\color{blue}
\noindent Several alternative delivery cost models are ``covered by'' \sref{Def.}{def:osdm-delivery-pd-monotone}---for instance, setting $c = 0$ recovers a price-dependent delivery cost that is \textit{constant} as a function of the storage state.  Similarly, a reasonable alternative model might be a delivery cost that \textit{is} state-dependent but doesn’t depend on the price $p_t$---such a model is captured as a simpler sub-case of \sref{Def.}{def:osdm-delivery-pd-monotone}, and our algorithm (see \autoref{sec:alg}) would directly capture any such simpler cases.
In the rest of the paper, we primarily give results for the \OSDMS and \OSDMT problems using the \textit{decreasing} monotone delivery cost case defined in \sref{Def.}{def:osdm-delivery-pd-monotone}(b), because it is the most challenging case to analyze. \autoref{thm:osdm-decreasing-worse-increasing} shows that given the same values of $c$ and $\varepsilon$, the best achievable competitive ratios for the increasing case in \sref{Def.}{def:osdm-delivery-pd-monotone}(a) imply it is a strictly easier problem than the decreasing case.  Thus, our results primarily focus on this ``hard case,'' holding for the general monotone definition in \sref{Def.}{def:osdm-delivery-pd-monotone} and any simpler subcases of it.}
We note that setting the delivery cost according to \sref{Def.}{def:osdm-delivery-pd-monotone}(b) with $c > 0$ results in a non-convex (bilinear) objective in \eqref{eq:osdm-objective} due to the dependence on $s_{t-1}$.\footnote{The online problem is unaffected by non-convexity because the state $s_{t-1}$ is a fixed value at time $t$ (not a decision variable).}

\subsection{Assumptions and Notation}\label{sec:assumptions}

In this section, we summarize our main assumptions and additional notation for \OSDM.
First, we assume that the market prices are bounded: $p_t \in [p_{\min}, p_{\max}] \ \forall t \in [T]$, where $0 < p_{\min} \le p_{\max}$. $p_{\min}$ and $p_{\max}$ are known to the operator a priori.  This is a standard assumption in the literature on online search and related problems such as online knapsack~\cite{ElYaniv:01,Zhou:08,Larsen:10}.  It is known that without such an assumption, no online algorithm can achieve a bounded competitive ratio~\cite{ElYaniv:01,Lorenz:08}.

We also assume that it is always feasible to satisfy demand as it becomes due, i.e., if necessary, it is always possible to set $x_t = b_t + \sum_{\tau: \Delta_\tau = t} f_\tau$ and $z_t = b_t + \sum_{\tau: \Delta_\tau = t} f_\tau$ for any $t \in [T]$. This is a natural assumption that ensures the problem is well-posed---without such an assumption, an online algorithm can be backed into a corner where feasibility becomes impossible.  While we do not consider \textit{rate constraints} in our setting of \OSDM (i.e., constraints on the amount of energy that can be purchased or delivered at any time), this gives an implicit condition relevant to such constraints\textemdash see discussion of future work in \autoref{sec:conc} for more on this point.

We assume that the switching cost coefficients are not ``too large'' relative to the prices.  Formally, $\gamma$ and $\delta$ satisfy $\gamma + \delta \leq \frac{p_{\max} - p_{\min}}{2}$.  %
Note that if $\gamma + \delta > \frac{p_{\max} - p_{\min}}{2}$, then competitive decision-making is simple: an algorithm should choose decisions that minimize the total switching cost (i.e., by spreading decisions evenly over time).  Thus, we focus on the ``interesting regime'' where an online algorithm must balance the purchasing and switching costs.  %

{\color{blue}
Using the notation of \sref{Def.}{def:osdm-delivery-pd-monotone}, we assume the delivery cost coefficients $c \geq 0$ and $\varepsilon \geq 0$ satisfy $0 \leq c + \varepsilon \leq 1$.  This ensures that the delivery cost is not ``too large'' relative to the purchasing cost, and is realistic in the target applications of \OSDM, where the delivery cost (e.g., the added cost to serve demand using stored energy) is a fraction of the original cost to procure said energy.
}

In \OSDMT, we make two additional assumptions related to the target signal and the tracking cost.  First, we assume that the cumulative target is at most the demand, i.e., $\sum_{t=1}^{T} a_t \leq \sum_{t=1}^{T} (b_t + f_t)$.
This ensures that $\sum_{t \in [T]} a_t$ is reasonable relative to the total demand over the horizon.  Without this assumption, there are unrealistic \OSDM instances with, e.g., a small amount of demand and a consistently large target signal that force any algorithm to pay a large penalty for failing to have enough demand to meet the target.
Finally, we assume that the tracking cost coefficient $\eta$ satisfies $\eta \leq \frac{p_{\max} - p_{\min}}{2}$---this ensures that the tracking cost coefficient is not ``too large'' relative to the range of market prices.  If this assumption is violated (i.e., if $\eta > \frac{p_{\max} - p_{\min}}{2}$), then it is always preferable (from a cumulative purchasing cost perspective) to set $x_t = a_t$ wherever possible.  Thus, we focus on the ``interesting regime'' where an online algorithm balances the tracking cost and the purchasing cost, allowing for slight deviations from the target signal when they are advantageous.

Throughout the paper, we use the shorthand notation $\Phi(w, z)$ and $\Psi(w, z)$ to denote the \textit{definite integrals} of functions $\int_w^z \phi(u) \ du$ and $\int_w^z \psi(u) \ du$, respectively, for any $0 \leq w \leq z$.  We also use the shorthand $P \coloneqq (1+c+\varepsilon)p_{\max}$, $\kappa \coloneqq \gamma + \delta$, and $\omega \coloneqq \tfrac{(1+c+\varepsilon)}{(1+\varepsilon)}$.

\section{Warmup and Preliminary Results}
\label{sec:warmup}

In this section, we start by discussing similarities between \OSDM and related problems in the literature.  We review one such problem and an existing algorithm for it, showing that an extension of this algorithm to \OSDM fails to achieve
a good
competitive ratio, even in a simplified setting---this warmup result illustrates the complexity of the problem and motivates observations that we use in the design of our main algorithm, \PAAD, described in \autoref{sec:alg}.

\smallskip
\noindent \textbf{Relation to Online Conversion with Switching Costs. \ } 
A much-simplified version of the \OSDM problem defined in \autoref{sec:problem} has been studied under the name of \textit{online conversion with switching costs} (\OCS)~\cite{Lechowicz:24}.
This problem is a special case of \OSDMS where there is a single unit of flexible demand (whose size is known a priori) that must be satisfied by the end of the time horizon (i.e., $f_0 = 1$ and $\Delta_0 = T$).
\OCS also does not consider base demand (i.e., $b_t = 0$ for all $t$), there is no storage (i.e., $S=0$), and there is no delivery cost (i.e., $c = \varepsilon = \delta = 0$).
The goal is to minimize the total cost of purchasing an asset subject to time-varying prices, while also minimizing a switching cost incurred from changing the purchase amount over time.
For the \OCS problem, an algorithmic framework called ``ramp-on, ramp-off'' (\RORO) achieves the best possible competitive ratio~\cite{Lechowicz:24}.
The key idea underlying \RORO is to maintain a \textit{single} threshold function and completion state (i.e., ``how much demand has been satisfied'') that guides the algorithm's decision-making; this threshold is designed to capture the marginal trade-off between purchasing at the current price and waiting for better prices that may not arrive (see \sref{Appendix}{apx:roro-details} for more details on the algorithm and setting).  %

A key challenge preventing the direct application of \RORO to the \OSDM setting is uncertainty in how much demand will arrive over time, since \RORO is designed to handle an a priori fixed total demand. A natural approach to %
deal with such uncertainty is a 
\textit{doubling strategy} that progressively increases the total demand that \RORO is specified to satisfy.  We describe this extension below.

\smallskip
\noindent \textbf{Doubling Strategy for \RORO. \ }
To apply the \RORO framework in the setting of \OSDM, our key idea is to maintain a guess $\smash{\hat{d}}$ of the total flexible demand that will arrive over time, and create instances of \RORO as if there is a single unit of flexible demand of size $\smash{\hat{d}}$ that must be satisfied by the end of the total time horizon $T$ (known a priori).  Initially, we set $\smash{\hat{d}} = 2^0 = 1$ and $j = 0$.  Whenever a new flexible demand $f_t$ arrives, we first attempt to ``assign it'' to the most-recently initialized instance of \RORO by checking whether adding $f_t$ to its total assigned demand would violate the guessed bound $\smash{\hat{d}}$. If not, then we assign $f_t$ to this current \RORO instance and allow it to continue; otherwise, 
we double the guess $\smash{\hat{d}} \gets 2^{j+1}$, increment $j$, and create a new instance of \RORO that is responsible for satisfying all flexible demands arriving from time $t$ onward (including $f_t$). Prior instances of \RORO operate until all of the flexible demands assigned to them are satisfied (or until their deadlines have passed), and the global purchasing decision is the sum of the decisions of all active instances of \RORO.
We formalize this doubling extension of \RORO in the Appendix (see \autoref{alg:roro-doubling}). 
To analyze this extension of \RORO, we restrict to a subset of \OSDMS instances with only flexible demand (i.e., $b_t = 0$ for all $t$), no delivery cost (i.e., $c = \varepsilon = \delta = 0$), no storage ($S=0$), and we assume that the maximum demand unit is upper-bounded by some $\sigma$ (i.e., $\max_{t \in [T]} f_t \le \sigma$), where $\sigma < 1$.  Then, we can characterize the competitive ratio of the doubling extension of \RORO as follows:

\begin{theorem}\label{thm:roro-doubling-upper-bound}
    For the subset of \OSDMS instances described above, the ``doubling extension'' of \RORO is $\zeta$-competitive, where $\alpha_{\texttt{RORO}}$ is the optimal competitive ratio for \OCS (see \autoref{eq:alpha_roro}) %
    and $\zeta$ is at least:
    \[
    \zeta \geq \min \Bigg\{ \frac{\alpha_{\texttt{RORO}}}{1-\sigma} + \frac{\alpha_{\texttt{RORO}}}{1-\sigma} \cdot \frac{\big( p_{\max} - 2\gamma - \frac{p_{\max}}{\alpha_{\texttt{RORO}}} \big) \exp\big( \frac{\sigma}{\alpha_{\texttt{RORO}}} \big) - \frac{\sigma p_{\max} (1-\sigma)}{\alpha_{\texttt{RORO}}}}{p_{\min}}, \frac{p_{\max} + 2\gamma}{p_{\min}} \Bigg\}.
    \]
\end{theorem}

\noindent We give the full proof of \autoref{thm:roro-doubling-upper-bound} in \sref{Appendix}{apx:roro-doubling-proof}.  
This result shows that the competitive ratio of the doubling \RORO algorithm can be \emph{arbitrarily bad}---in particular, the first entry in the $\min$ expression above grows to infinity as $\sigma$ approaches one, implying that the competitive ratio of the doubling \RORO algorithm is not better than the trivial $O\big( \tfrac{p_{\max}}{p_{\min}} \big)$ bound achievable by any feasible \OSDM algorithm. 
{\color{blue}
Intuitively, a simple extension of \RORO cannot account for the demand that arrives in a coarse-grained manner: it fails to resolve the ``gap'' between the demand the algorithm anticipates and the actual demand that arrives, and this gap leads to its suboptimality.
This limitation motivates the modular and adaptive design of our main algorithm \PAAD, which we present next.
}

\section{A Unified Competitive Algorithm}
\label{sec:alg}

In this section, we present a unified algorithm called \textbf{P}artitioned \textbf{A}ccounting \& \textbf{A}ggregated \textbf{D}ecisions (\PAAD) to solve \OSDM.
Motivated by the result in \autoref{sec:warmup} that demonstrates the necessity of \textit{accounting} for \OSDM, \PAAD introduces abstractions called \textit{drivers} that perform a double duty of accounting for distinct demand units while making decisions that can be aggregated into
global purchasing and delivery actions.
We formally introduce the \PAAD algorithm in \autoref{sec:paad-algorithm} and analyze its competitive ratio for \OSDMS and \OSDMT in \autoref{sec:paad-analysis}.

\subsection{The \PAAD Algorithm} \label{sec:paad-algorithm}

The core idea of \PAAD is to handle each newly-arrived unit of demand with an abstract \textit{driver}, which produces decisions which are aggregated into a globally feasible purchasing and delivery action at each time.
We start by detailing the driver types used in \PAAD before describing the algorithm itself.

\noindent \textbf{Drivers. \ } \PAAD uses two types of drivers to manage the purchasing and delivery decisions: (i) \textit{base demand} drivers (denoted by set $\mathcal{B}$) that manage units of base demand, and (ii) \textit{flexible demand} drivers (denoted by set $\mathcal{F}$) that manage units of flexible demand.  Each driver is assigned a \textit{size} based on the demand it is responsible for, and makes local decisions based on its own state.  %

At initialization ($t=0$), \PAAD creates an initial \textit{manager} driver that charges the storage when the market price is sufficiently low---this driver is instantiated as a base demand driver with index $0$ ($\mathcal{B} \gets 0$) and size $d^{(0)} \coloneqq S$ (the storage capacity).  
At each time step $t \in [T]$, upon the arrival of $p_t, b_t,$ and $f_t$ (with deadline $\Delta_t$), drivers are created as follows.  If $b_t > 0$, then a new base demand driver with even index $2t$ is created and added to the set of base drivers $\mathcal{B}$, with size $d^{(2t)} \coloneqq b_t$.  If $f_t > 0$, then a new flexible demand driver with odd index $2t+1$ is created and added to the set of flexible drivers $\mathcal{F}$, with size $d^{(2t+1)} \coloneqq f_t$ and deadline slack $\Delta_t$.  

Each base demand driver $i \in \mathcal{B}$ makes local decisions governed by a \textit{purchasing threshold function} $\smash{\phi_b^{(i)}(w) : w \in [0,d^{(i)}]}$, and updates a state variable $\smash{w_b^{(i)}}$ that indicates how much has been purchased so far.  $\smash{\phi_b^{(i)}}$ is non-increasing in $w$, and captures the marginal trade-off between making progress towards satisfying the driver's total demand and waiting for (potentially) better prices.
\begin{definition}[Base demand threshold function] \label{def:base-threshold}
{\it
    Given a base driver with size denoted by $d$, the threshold function $\phi_b$ for \OSDMS is defined as follows, where $\alpha$ is the competitive ratio (see \sref{Thm.}{thm:osdm-upper-bound-1}):
    \[
    \phi_b(w) = p_{\max} + 2 \gamma + p_{\min} c + \Big( \frac{p_{\max}(1+c+\varepsilon) + 2\kappa}\alpha - \Big( p_{\max}(1+\varepsilon) + p_{\min}c + \frac{2\kappa}{T} \Big) \Big) \exp\Big( \frac{w}{\alpha d} \Big) : w \in [0,d],
    \]
}
\end{definition} 

\noindent Each flexible demand driver $l \in \mathcal{F}$ maintains a \textit{purchasing threshold function} $\smash{\phi_f^{(l)}(w) : w \in [0, d^{(i)}]}$ and a \textit{delivery threshold function} $\smash{\psi^{(l)}(v): v \in [0, d^{(i)}]}$ (defined below), with state variables $\smash{w_f^{(l)}}$ and $\smash{v_f^{(l)}}$ denoting the amount of their demand already purchased and delivered, respectively.  %
\begin{definition}[Flexible demand threshold functions] \label{def:flex-thresholds}
{\it
    Given a flexible driver with size denoted by $d$, the purchasing threshold function $\phi_f$ for \OSDMS is defined as follows, where $\omega = \tfrac{(1+c+\varepsilon)}{(1+\varepsilon)}$ and $\alpha' = \nicefrac{\alpha}{\omega}$:
    \[
    \phi_f(w) = p_{\max} + p_{\min} c + 2 \gamma + \Big( \frac{p_{\max}+ 2\gamma}{\alpha'} - \Big( p_{\max} + p_{\min}c + \frac{2\gamma}{T}\omega \Big) \Big) \exp\Big( \frac{w}{\alpha' d} \Big) : w \in [0,d],
    \]
    and the delivery threshold function $\psi$ for \OSDMS is defined as:
    \[
    \psi(v) = p_{\max} (c + \varepsilon) + 2\delta + \Big( \frac{p_{\max}(c+\varepsilon) + 2\delta}{\alpha'} - \Big( p_{\max}(c+\varepsilon) + \frac{2\delta}{T} \omega \Big) \Big) \exp\left( \frac{v}{\alpha' d} \right) : v \in [0,d].
    \]
    }
\end{definition}

\begin{figure*}[b]
    \vspace{-1em}
    \centering
    \includegraphics[width=0.95\textwidth]{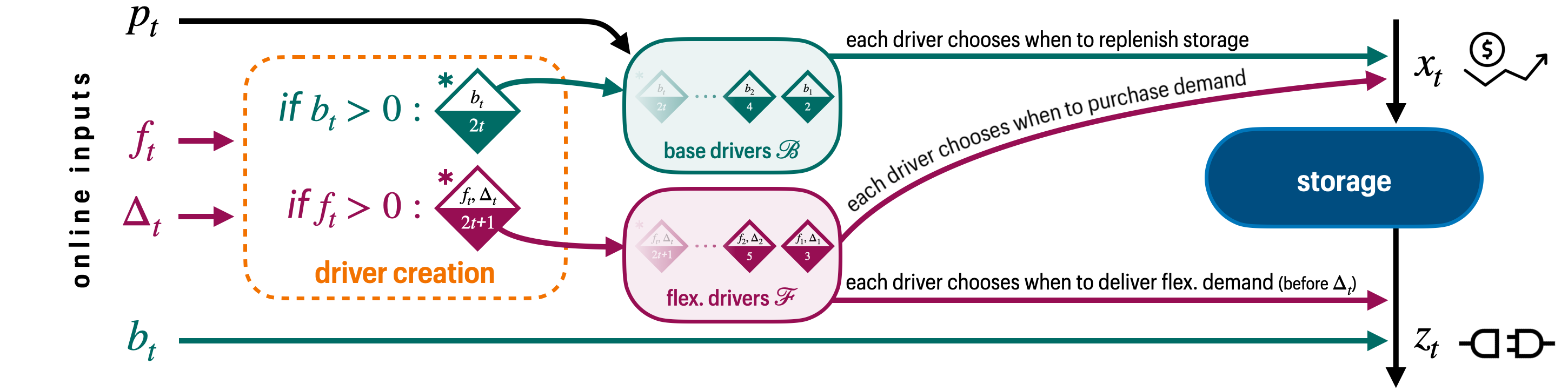} \vspace{-1em}
    \caption{{\color{blue}Diagram of the intuition behind the \PAAD algorithm.  Each driver accounts for a single unit of base or flexible demand, making local decisions based on its own state and role that are aggregated into a single global decision. Purchasing decisions $x_t$ govern how much to purchase at the current market price, and delivery decisions $z_t$ dictate how much demand is satisfied at the current time (subject to feasibility constraints).}}
    \label{fig:paad-diagram}
    \vspace{-1em}
\end{figure*}

\noindent \textbf{Algorithm. \ } We now describe the \PAAD algorithm itself.  The main pseudocode is provided in \autoref{alg:paad}, and the pseudo-cost minimization subroutine for individual drivers is detailed in \autoref{alg:pcm-subroutine}. %
The key idea guiding \PAAD's behavior is to have each driver make local decisions based on its own state and threshold functions, while also carefully accounting for switching costs across multiple drivers, allowing the algorithm to amortize global switching costs.  %
All local decisions are \textit{aggregated} into a single global decision: $\smash{x_{t} \leftarrow \sum_{i \in \mathcal{B} \cup \mathcal{F}} x_{t}^{(i)},  z_{t} \leftarrow b_t + \sum_{i \in \mathcal{F}} z_{t}^{(i)}}$.  \PAAD ensures feasibility by checking for binding constraints (i.e., flexible demands nearing their deadlines, Line 4 in \autoref{alg:pcm-subroutine}) and ensuring that sufficient energy is purchased to meet the delivery decision (Line 21). %
\PAAD amortizes global switching costs by communicating \textit{pseudo-decisions} to each driver, which are adjusted based on the global purchasing decision from the previous time step.  This allows each driver to recognize when increasing its purchase amount can help to amortize the total switching cost across multiple drivers, reducing the pessimism inherent in a purely local approach.   

For each driver $i \in \mathcal{B} \cup \mathcal{F}$, this is formalized as the definition of 
$\smash{\hat{x}^{(i)}_{t-1}}$ on Line 2 in \autoref{alg:pcm-subroutine}, where $q(x_{t-1})$ is the difference between the sum of all driver decisions and the global decision at time $t-1$. 
Intuitively, $q(x_{t-1})$ captures any ``extra'' purchasing from the previous time step that was not attributable to any currently active driver (e.g., due to base demand or inactive drivers). 
{\color{blue}
Note that Line 3 (\autoref{alg:pcm-subroutine}) includes a $\gamma \vert x \vert$ term---this term accounts for the worst-case future cost of ``switching off'' whenever $x > 0$ (e.g., when the sequence ends at time $T$, decisions $x_{T+1} = z_{T+1} = 0$).  By considering the pseudo-decision $\smash{\hat{x}_{t-1}^{(i)}}$ instead of the actual previous decision $\smash{x_{t-1}^{(i)}}$ when making the next decision, \PAAD amortizes this ``future cost'' across multiple drivers, resulting in smoother global decisions.} The same is applied to the delivery decisions of flexible drivers.
Note that \PAAD may generate as many as $\Omega(T)$ drivers on a given problem instance. The necessity of these many drivers is supported by our warmup result in \autoref{sec:warmup}, which shows that a granular, per-demand accounting is likely necessary to achieve the best possible competitive ratio when future demands are unknown. However, \PAAD's driver refresh (Line 23) ensures that, in practice, there will likely be far fewer than $\Omega(T)$ drivers active at any given time. In particular, a driver is only active for a finite period---e.g., until flexible drivers' deadlines pass, or until base drivers complete their purchasing. 
We discuss this point further and give runtime measurements in \sref{Appendix}{apx:runtime}.

\noindent \begin{minipage}[t]{\linewidth}
\null 
\vspace{-1em}
\begin{algorithm}[H]
\caption{Pseudo-cost Minimization Subroutine for \PAAD Drivers (\texttt{GetDecision})}
\label{alg:pcm-subroutine}
{\small
\begin{algorithmic}[1]
\State \textbf{input:} driver sets $\mathcal{B}$ and $\mathcal{F}$, driver index $i$, excess $q(x_{t-1})$, price $p_t$, switching coeff. $\gamma$, threshold $\Phi$
\State \textbf{define} \textit{pseudo-decision} $\hat{x}^{(i)}_{t-1} \gets x^{(i)}_{t-1} + \frac{ q(x_{t-1}) \cdot d^{(i)} }{\sum_{i\in \mathcal{B} \cup \mathcal{F}} d^{(i)} }$
\State \textbf{solve} $x^{(i)}_t \gets \text{argmin}_{x \in [0, d^{(i)} - w_{t-1}^{(i)}]} \ p_t x + \gamma \vert x - \hat{x}^{(i)}_{t-1} \vert + \gamma \vert x \vert - \Phi^{(i)}(w^{(i)}_{t-1}, w^{(i)}_{t-1}+x)$
\If {$i \in \mathcal{F}$ \textbf{and} $\Delta^{(i)} = t$} $x^{(i)}_t \gets d^{(i)} - w_{t-1}^{(i)}$ \Comment{flexible demand must be satisfied \textit{now} (at the deadline)}
\EndIf
\State \textbf{output:} driver's decision for time step $t$: $x^{(i)}_t$
\end{algorithmic}
}
\end{algorithm}
\end{minipage}
\noindent \begin{minipage}[t]{\linewidth}
\null 
\vspace{-0.5em}
\begin{algorithm}[H]
\caption{\textbf{P}artitioned \textbf{A}ccounting \& \textbf{A}ggregated \textbf{D}ecisions (\PAAD) Algorithm for \OSDMS}
\label{alg:paad}
{\small
\begin{algorithmic}[1]
\State \textbf{input:} Storage capacity $S$, switching cost coefficients $\gamma, \delta$, delivery cost $\mathcal{D}(\cdot)$, threshold functions $\Phi_b, \Phi_f, \Psi$
\State \textbf{initialize:} Storage state $s_0 \gets 0$, decisions $x_0 \gets 0$, $z_0 \gets 0$, driver sets $\mathcal{B} \gets \emptyset$, $\mathcal{F} \gets \emptyset$
\For{$t = 1$ to $T$: \textbf{observe} $p_t, b_t, f_t, \Delta_t$, }
    \If{$s_{t-1} = 0$ or $b_{t} > S$} \Comment{refresh base drivers and storage manager driver}
        \State Create new storage manager driver with index $0$ and size $d^{(0)} = S$; $\mathcal{B} \gets \{0\}$ 
    \EndIf
    \If{$0 < b_t < S$}
        \State Create base driver with index $2t$ and size $d^{(2t)} = b_t$; $\mathcal{B} \gets \mathcal{B} \cup \{2t\}$
    \EndIf
    \If{$0 < f_t$}
        \State Create flexible driver with index $2t+1$, size $d^{(2t+1)} = f_t$ and deadline $\Delta^{(2t+1)} = \Delta_t$; $\mathcal{F} \gets \mathcal{F} \cup \{2t+1\}$
    \EndIf
    \State \textbf{compute} $q(x_{t-1}) \gets x_{t-1} - \sum_{i \in \mathcal{B} \cup \mathcal{F}} x^{(i)}_{t-1}$, \quad $q(z_{t-1}) \gets z_{t-1} - \sum_{i \in \mathcal{F}} z^{(i)}_{t-1}$ 
    \State \Comment{captures ``extra'' purchase/delivery due to prior demand or inactive drivers}
    \For{each active flexible driver $i \in \mathcal{F}$} 
        \State $z^{(i)}_t \gets \texttt{GetDecision}(\mathcal{B}, \ \mathcal{F}, \ i, \ q(z_{t-1}), \ \mathcal{D}(z_t, s_{t-1}, p_t), \ \delta, \ \Psi^{(i)})$ 
        \State \textbf{update} delivery state: $v^{(i)}_t = v^{(i)}_{t-1} + z^{(i)}_t$
    \EndFor
    \State \textbf{aggregate} delivery decisions $z_t \gets b_t + \sum_{i \in \mathcal{F}} z_t^{(i)}$ and \textbf{compute} buying $\texttt{cap}_t \gets z_t + (S - s_{t-1})$
    \For{each active driver $i \in \mathcal{B} \cup \mathcal{F}$ \textbf{while} $\texttt{cap}_t > 0$}
        \State $x^{(i)}_t \gets \texttt{GetDecision}(\mathcal{B}, \ \mathcal{F}, \ i, \ q(x_{t-1}), \ p_t, \ \gamma, \ \Phi^{(i)})$ 
        \State $x^{(i)}_t \gets \min\{x^{(i)}_t, \texttt{cap}_t \}$ \Comment{ note $\Phi^{(i)}$ indicates $\Phi_b$ or $\Phi_f$ depending on driver type }
        \State \textbf{update} buying $\texttt{cap}_t \gets \texttt{cap}_t - x^{(i)}_t$ and purchasing state $w^{(i)}_t = w^{(i)}_{t-1} + x^{(i)}_t$
    \EndFor
    \State \textbf{aggregate} global decisions: $x_t \gets \sum_{i \in \mathcal{B} \cup \mathcal{F}} x_t^{(i)}$ 
    \State $x_t \gets \max\{ x_t, z_t - s_{t-1} \}$ \Comment{ensure sufficient energy to deliver}
    \State \textbf{update} storage state: $s_t = s_{t-1} + x_t - z_t$
    \State \textbf{refresh} drivers: remove any driver $i$ with $w^{(i)}_t = d^{(i)}$ or $\Delta^{(i)} = t$ from $\mathcal{B}$ or $\mathcal{F}$, respectively
\EndFor
\State \textbf{output:} Purchase decisions $\{x_t\}_{t=1}^T$, delivery decisions $\{z_t\}_{t=1}^T$
\end{algorithmic}
}
\end{algorithm}
\vspace{0.5em}
\end{minipage}

\subsection{Competitive Analysis} \label{sec:paad-analysis}

We now analyze the \PAAD algorithm and show that it achieves the best possible competitive ratio for \OSDM.  We start by stating our main result.

\begin{theorem}
\label{thm:osdm-upper-bound-1}
\PAAD is $\alpha$-competitive for \OSDMS, where $\alpha$ is given by:
{\color{blue}
\begin{equation}
\alpha = \frac{\omega}{\Bigg[ W\Bigg(-\frac{\left( P - (1+\varepsilon) p_{\min} \right) \exp\big( \frac{- \nicefrac{\omega 2\kappa}{T} - p_{\min} c - P}{P + 2\kappa} \big) }{P + 2\kappa}\Bigg)+\frac{P + p_{\min} c - \nicefrac{\omega 2\kappa}{T}}{P + 2\kappa} \Bigg] }, \label{eq:alpha1}
\end{equation}}
where $W(\cdot)$ is the Lambert W function~\cite{Corless:96LambertW}, $P = (1+c+\varepsilon)p_{\max}$, $\kappa = \gamma + \delta$ and $\omega = \frac{(1+c+\varepsilon)}{(1+\varepsilon)}$.
\end{theorem}

\noindent This expression for $\alpha$ is complex, but a leading order approximation yields asymptotics that are sublinear in the ratio $\nicefrac{p_{\max}}{p_{\min}}$ (specifically, $\alpha \approx O\left(\sqrt{\nicefrac{p_{\max}}{p_{\min}}}\right)$), and linear in both the switching cost coefficients $\gamma, \delta$ and the delivery cost coefficient $c$. We plot $\alpha$ as a function of these parameters in \sref{Appendix}{apx:paad-proof}, and prove that it is the optimal competitive ratio in \autoref{sec:fundamental-limits}.

{\color{blue}
We remark on several regimes of interest for $\alpha$ where \OSDMS essentially ``recovers'' the optimal competitive ratio of other online problems.  First, in a case with zero delivery cost (and zero delivery switching cost), letting $T \to \infty$, $\omega = 1$, $\kappa = \beta$, $c = \epsilon = 0$, and $P = U-2\beta$,\footnote{{\color{blue}In contrast to \OSDM, the original \OCS paper~\cite{Lechowicz:24} assumed that any necessary compulsory trading at the end of the time horizon could be accomplished without incurring additional switching cost, captured by the $-2\beta$ in this term.}} $\alpha$ exactly reduces to the optimal competitive ratio for online conversion with switching costs, where $\beta$ is the switching coefficient~\cite{Lechowicz:24}. Furthermore, letting $T \to \infty$, $\omega \to 1$, $\kappa = 0$, $c = \epsilon = 0$, and $P \to U$, $\alpha$ reduces to the optimal competitive ratio for the minimization version of one-way trading~\cite{Lorenz:08}.
}

\begin{corollary}
\label{cor:osdmt-upper-bound-2}
\PAAD is $\alpha_{\texttt{T}}$-competitive for \OSDMT, for $\alpha_{\texttt{T}}$ given by:
{\color{blue}
\begin{equation}
\alpha_{\texttt{T}} = \frac{\omega}{\Bigg[ W\Bigg(-\frac{\left( P - (1+\varepsilon) p_{\min} \right) \exp\big( \frac{- \nicefrac{\omega 2\delta}{T} - p_{\min} c - P}{P + 2(\eta+\delta)}\big) }{P + 2(\eta + \delta)}\Bigg)+\frac{P + p_{\min} c -  \nicefrac{\omega 2\delta}{T}}{P + 2(\eta+\delta)} \Bigg] }, \label{eq:alphaT}
\end{equation}}
This result uses a redefinition of the thresholds and pseudo-cost minimization problem (see \sref{Def.}{dfn:new-thresholds-and-pseudo-cost-osdmt}).
\end{corollary}

\subsection{Proof Overview}

We sketch the proof of \autoref{thm:osdm-upper-bound-1}, relegating its full proof and that of \sref{Corollary}{cor:osdmt-upper-bound-2} to \autoref{apx:paad-proof}.

\noindent \textbf{Proof Sketch of \autoref{thm:osdm-upper-bound-1}. \ } 
To show this result, we prove lemmas to characterize the cost of \OPT, the cost of \PAAD, and the relationship between the two.  First, note that \PAAD's solution is feasible (i.e., it satisfies all \OSDM constraints, see \sref{Lemma}{lem:osdm-feasibility-1})---this follows by \PAAD's definition.

For an arbitrary \OSDM instance, we then carefully lower bound the cost of \OPT---compared to prior applications of the pseudo-cost minimization framework, the key technical challenge is that it is not sufficient to derive a single global lower bound on \OPT's cost (e.g., as in \OCS~\cite{Lechowicz:24}), since there are many units of demand with differing constraints.  The key idea of our analysis in \sref{Lemma}{lem:osdm-optimal-lower-bound-1} is to partition the instance into \textit{active} and \textit{inactive} periods based on \PAAD's base driver state. This enables us to derive a lower bound on \OPT's cost that is a combination of the ``local'' best prices during each active period and the best price during all inactive periods, which is itself lower bounded.
Next, we derive an upper bound on \PAAD's total cost by leveraging the definition of the pseudo-cost minimization to ``charge'' each driver's decisions to their threshold functions (see \sref{Lemma}{lem:osdm-paad-upper-bound-1}), while accounting for any worst-case purchasing, switching, and delivery costs that arise due to binding constraints for a particular unit of demand.  To relate the cost of \PAAD to that of \OPT, we prove two lemmas that give a relation between each threshold function and its integral:
\begin{lemma}
\label{lem:osdm-threshold-function-relation-1}
By the definition of the threshold function $\phi_b(\cdot)$, the following relation always holds:
$\Phi_b (0, w) + (1\!-\!w) (p_{\max}+2\gamma) + p_{\max}(c+\varepsilon) + 2\delta - c w p_{\min} = \alpha\left[ \phi_b(w) - 2\gamma + \varepsilon p_{\max} + \frac{2\kappa}{T}\right] \ \forall w \in [0,1].
$
\end{lemma}

\begin{lemma}
\label{lem:osdm-threshold-function-relation-2}
By the definitions of the threshold functions $\phi_f(\cdot)$ and $\psi_f(\cdot)$, the following relation always holds:
$
\Phi_f (0, w) + (1\!-\!w) (p_{\max}+2\gamma) - c w p_{\min} + \Psi_f (0, v) + (1\!-\!v) (p_{\max}(c+\varepsilon) + 2\delta) =
\alpha'\Big[ \phi_f(w) + \psi_f(v) - 2\kappa + \frac{2\kappa \omega}{T }\Big] \ \forall w \in [0,1], v \in [0,w].
$
\end{lemma}

\noindent The rest follows by contradiction---we show that if \sref{Lemma}{lem:osdm-optimal-lower-bound-1} and \ref{lem:osdm-paad-upper-bound-1} hold, then if the threshold functions adhere to \sref{Lemmas}{lem:osdm-threshold-function-relation-1} and \ref{lem:osdm-threshold-function-relation-2}, $\frac{\PAAD(\mathcal{I}) - p_{\max} \hat{s}}{\OPT(\mathcal{I})}$ must be at most $\alpha$, completing the proof. \hfill $\square$

\section{Fundamental Limits}
\label{sec:fundamental-limits}

In this section, given the results in \autoref{sec:alg}, we ask whether any algorithm can achieve a better competitive bound for \OSDM than \PAAD.  In \sref{Theorems}{thm:osdm-lower-bound-1} and \ref{thm:osdmt-lower-bound-2} below, we answer this question in the negative, showing that \PAAD's competitive ratio is the best achievable among all deterministic online algorithms for \OSDMS and \OSDMT, respectively. Along the way, we show results to contextualize the relative ``hardness'' of \OSDM in subcases of interest, such as the case of only base demand (see \sref{Corollary}{cor:osdm-lower-bound-base-demand}) or the case of price-dependent increasing delivery cost (recall \sref{Def.}{def:osdm-delivery-pd-monotone}).

\smallskip
\noindent \textbf{Lower Bound for \OSDMS. \ } 
We start by considering the switching cost case of \OSDMS.  \autoref{thm:osdm-upper-bound-1} shows that \PAAD is $\alpha$-competitive (defined in \eqref{eq:alpha1}) for \OSDMS.  We now prove that $\alpha$ is optimal by constructing a difficult set of \OSDMS instances on which no deterministic algorithm can achieve a better competitive ratio; this \textit{lower bounds} the achievable ratio of any deterministic online algorithm for \OSDMS (see \sref{Def.}{dfn:comp-ratio}).
We give a sketch of the result, deferring the full proof to \autoref{apx:lowerbounds}.

\begin{theorem}
\label{thm:osdm-lower-bound-1}
There exists a set of \OSDMS instances for which no deterministic online algorithm \ALG can achieve a competitive ratio better than $\alpha$ (for $\alpha$ defined in \eqref{eq:alpha1}).
\end{theorem}

\noindent \textbf{Proof Sketch of \autoref{thm:osdm-lower-bound-1}. \ }
In \sref{Def.}{def:x-decreasing-instance-switching}, we define a set of \textit{$x$-decreasing instances} for \OSDMS (denoted by $\{\mathcal{I}_{x}\}_{x\in [p_{\min}, p_{\max}]}$) where $x$ is the ``best'' (i.e., lowest) price in the instance.
These instances consist of a single unit of flexible demand at time $t=1$ with a deadline of $\Delta_1 = T$.  An adversarial price sequence presents market prices that generally decrease over time, interrupted by ``spikes'' to the highest price $p_{\max}$ that force an online algorithm \ALG to pay a high switching cost.  This captures a trade-off between being eager or reluctant to accept a certain price $x$---if \ALG purchases too much at prices $> x$, it may miss better prices later on, but if it waits to purchase at prices $< x$, it may be forced to purchase and deliver at a high price $p_{\max}$ (e.g., if the sequence ends after $x$).

Under these special instances, the cost of any deterministic online algorithm \ALG can be described by two (arbitrary) purchasing and delivery functions $h(x), z(x): [p_{\min}, p_{\max}] \rightarrow [0,1]$ that describe the fraction of the demand purchased and delivered at prices $\geq x$.  Since \OPT is easy to describe on these instances (i.e., $\OPT(\mathcal{I}_{x}) = x + \frac{2\gamma}{T} + \varepsilon x + \frac{2\delta}{T})$, we show that the purchasing and delivery functions of an $\alpha^\star$-competitive \ALG must satisfy a necessary condition that relates $h(x), z(x)$ to $\alpha^\star$ for all $x \in [p_{\min}, p_{\max}]$.  By applying Grönwall's Inequality~\cite[Theorem 1, p.356]{Mitrinovic:91}, this gives the following condition that $\alpha^\star$ must satisfy at optimality: $\smash{1 = \frac{\alpha^\star}{\omega} \ln \Big[ \frac{P - (1+\varepsilon) p_{\min}}{P - \frac{\omega [ P + 2 \kappa ]}{\alpha^\star} + \frac{\omega 2\kappa}{T} + p_{\min} c} \Big]}$.  Solving this transcendental for $\alpha^\star$ yields the result, completing the proof. \hfill $\square$

\smallskip
{\color{blue}
In the above result, we have shown that \PAAD achieves the best possible competitive ratio for \OSDMS.  The proof of \autoref{thm:osdm-lower-bound-1} relies on a class of \OSDM instances that consist purely of flexible demand.  In the following corollary, we use a similar class of ``hard instances'' to argue that the competitive ratio \textit{improves} if we focus on the case of \OSDMS instances with \textit{only base demand}---recall from \autoref{sec:OSDM_statement} that this case corresponds to a simplification of \OSDM wherein the two decision variables of \OSDM can be simplified to a single decision variable (i.e., $\{ x_t - z_t \}_{t\in [T]}$).
}%

\begin{corollary}
\label{cor:osdm-lower-bound-base-demand}
There exists a set of \OSDMS instances with only base demand for which no deterministic online algorithm \ALG can achieve a competitive ratio better than $\alpha_{\texttt{B}}$, given by: 
\begin{equation}
\alpha_{\texttt{B}} \geq \Bigg[ W\Bigg(-\frac{\big( (1+c)p_{\max} - p_{\min} \big) \exp \big( - \frac{P +2\delta + \frac{2\gamma}{T}}{P + 2\kappa} \big)}{P + 2\kappa}\Bigg)+\frac{P +2\delta + \frac{2\gamma}{T}}{P + 2\kappa} \Bigg]^{-1} \kern-1em, \label{eq:alphaB}
\end{equation}
where $P = (1+c+\varepsilon)p_{\max}$. Given fixed $p_{\min}, p_{\max}, c, \varepsilon, \gamma$, and $\delta$, we have $\alpha_{\texttt{B}} < \alpha$ (for $\alpha$ defined in \eqref{eq:alpha1}).
\end{corollary}

\noindent We relegate the full proof of \sref{Corollary}{cor:osdm-lower-bound-base-demand} to \autoref{apx:osdm-lower-bound-base-demand}.  The main insight of \sref{Corollary}{cor:osdm-lower-bound-base-demand} is that the additional complexity of flexible demands makes the general \OSDM problem (i.e., allowing both flexible and base demand) strictly harder to solve in the online setting. 

{\color{blue}
We remark that \autoref{thm:osdm-lower-bound-1} and \sref{Corollary}{cor:osdm-lower-bound-base-demand} do not rule out the possibility of randomization improving the achievable competitive ratio.  While several works~\cite{ElYaniv:01, Lorenz:08, Zhou:08} have shown that randomization does not improve the competitive ratio for continuous or fractional online allocation problems such as one-way trading and online knapsack, the smoothness terms in the objective of \OSDM (and related problems such as \OCS) may change this property.
}

{\color{blue}
\smallskip
\noindent \textbf{Delivery Cost Dynamics. \ }
In the above results and throughout the paper, we focus on the case of a price-dependent delivery cost that is \textit{decreasing} as a function of the storage state.
Amongst the class of delivery costs defined in \sref{Def.}{def:osdm-delivery-pd-monotone}, we claimed this is the most challenging case to consider in \autoref{sec:problem}.  We now prove this claim by showing that \OSDM
with a price-dependent \textit{increasing} delivery cost is strictly easier.  As a corollary, any simpler special cases of \sref{Def.}{def:osdm-delivery-pd-monotone} (e.g., a price-dependent delivery cost that is constant in the storage state) are also easier in a competitive sense.
}

\begin{theorem}
\label{thm:osdm-decreasing-worse-increasing}
Amongst the class of monotone delivery costs defined in \sref{Def.}{def:osdm-delivery-pd-monotone}, the best achievable competitive ratios for \OSDM with an \textit{increasing delivery cost} are strictly better than those for \OSDM with a decreasing delivery cost (i.e., the cases considered in \autoref{thm:osdm-lower-bound-1}, \ref{thm:osdmt-lower-bound-2} and \sref{Corollary}{cor:osdm-lower-bound-base-demand}).
\end{theorem}

{\color{blue}
\noindent We relegate the full proof of \autoref{thm:osdm-decreasing-worse-increasing} to \autoref{apx:osdm-decreasing-worse-increasing}, which relies on lower bound ``hard instance'' constructions similar to those used in \autoref{thm:osdm-lower-bound-1} and \sref{Corollary}{cor:osdm-lower-bound-base-demand}.  The main insight is that our general competitive results focus on the most challenging setting of a delivery cost that is price-dependent and decreasing as a function of the storage state\textemdash these results directly extend to the other ``simpler'' cases of delivery cost models.
}

\smallskip
\noindent \textbf{Lower Bound for \OSDMT. \ }
Following the results for \OSDMS, we now consider the tracking cost case of \OSDMT.  Recall that \sref{Corollary}{cor:osdmt-upper-bound-2} shows that \PAAD is $\alpha_{\texttt{T}}$-competitive (defined in \eqref{eq:alphaT}) for \OSDMT.  We now prove this is optimal by constructing 
a difficult set of instances and tracking target
sequences
on which no algorithm can do better than $\alpha_{\texttt{T}}$-competitive, thus showing that \PAAD achieves the best possible competitive ratio for \OSDMT.  We defer the proof to \autoref{apx:osdmt-lower-bound-2}.

\begin{theorem}
\label{thm:osdmt-lower-bound-2}
There exists a set of \OSDMT instances for which no deterministic online algorithm \ALG can achieve a competitive ratio better than $\alpha_{\texttt{T}}$ (for $\alpha_{\texttt{T}}$ defined in \eqref{eq:alphaT}).
\end{theorem}

\section{Learning Data-Driven Decisions}
\label{sec:learn}

{\color{blue}
In \autoref{sec:OSDM_statement}, we mentioned that algorithms designed using competitive analysis are often overly pessimistic, and can sometimes perform poorly compared to simple heuristics used in practice. However, these heuristics---as well as data-driven, machine-learned algorithms---can suffer a lack of robustness to nonstationarity or distribution shift.
In this section, we propose a novel paradigm that directly learns ``the best \OSDM algorithm'' lying within a defined set of algorithms that satisfy a given (desired) worst-case performance guarantee, improving average-case performance while retaining robustness.
To accomplish this, we show that the pseudo-cost design of \PAAD provides a mechanism by which one can theoretically describe the general set of ``$\rho$-robust algorithms'' (for any $\rho > \alpha$, see \sref{Def.}{dfn:robust-threshold-base}).
To realize this learning paradigm in practice, we show a \textit{end-to-end learning methodology} (see \autoref{sec:pald-learning}) that leverages recent techniques from the literature on learning-to-optimize (L2O)~\cite{Chen:21:L20}, such as differentiable optimization layers.  We begin by reviewing some relevant preliminaries.

}

\smallskip
\noindent \textbf{Preliminaries \& Challenges. \ }
{\color{blue}
In the literature on learning-augmented online algorithms, the competitive ratio (\sref{Def.}{dfn:comp-ratio}) is replaced with two metrics: \textit{consistency} and \textit{robustness}~\cite{Lykouris:18, Purohit:18, Mitzenmacher:22:ALPS}.  If \ALG is a learning-augmented online algorithm provided with some advice \ADV (e.g., a prediction about the instance), then \ALG is said to be $c$-\textit{consistent} if it is $c$-competitive when \ADV is \textit{perfectly accurate}, and $\rho$-\textit{robust} if it is $\rho$-competitive, regardless of how bad (e.g., adversarial) \ADV is.

There are two common models for advice. The first is \textit{decision advice}, where \ADV predicts the optimal decisions at each time step---this type of advice is frequently used for problems with multiple time steps and switching costs, such as metrical task systems~\cite{Antoniadis:20MTS} and the \OCS problem discussed in \autoref{sec:warmup}~\cite{Lechowicz:24}
The second is \textit{input prediction advice}, where \ADV predicts a characteristic of the problem instance, such as the number of skiing days in a ski rental problem~\cite{Purohit:18}. }%
This type of advice has been used for related problems such as online knapsack~\cite{Daneshvaramoli:25} and one-way trading~\cite{SunLee:21}, which is itself closely related to \OCS.
Both models face limitations in the context of \OSDM. 
Obtaining decision advice directly from an ML model can be difficult due to the generalization challenges posed by complex problem structure---existing methods instead predict input features (e.g., a price forecast in \OCS), solve the corresponding ``certainty equivalent'' offline problem, and use that solution as advice~\cite{Lechowicz:24}.
{\color{blue}
In \OSDM, however, joint uncertainty in prices and demand arrivals/deadlines makes this approach brittle to prediction error and impractical in practice.
On the other hand, point predictions about instances (e.g., estimates of the best price) are straightforward to obtain, but are too simplistic to provide useful insights for \OSDM. 
}
For instance, even given a perfect prediction of the best (i.e., lowest) price, it is 
not possible to determine optimal
purchasing and delivery schedules subject to the unknown demand arrivals/deadlines and switching/tracking costs.

{\color{blue}
On top of these limitations to existing prediction models,
there is a gap between theory and practice where metrics are concerned. In particular, while robustness is useful as a worst-case ``guardrail'' against, e.g., significantly out-of-distribution inputs, the \textit{consistency} metric is less meaningful---it is only defined for perfect predictions that are unrealistic in practice, and it is still a worst-case metric (i.e., it is defined as the competitive ratio given accurate predictions). While recent work has considered ``smooth degradation'' of algorithm performance to small prediction errors in online search problems \cite{Benomar:25, NEURIPS2024_11c6625b}, it is not clear that these methods can be extended to the substantially more complex setting of \OSDM.
In practice, the actual objective is for an algorithm to perform well \textit{on average}: if an algorithm performs very well in practice but has bad theoretical consistency, that is preferable to an algorithm that has good consistency in theory but poor real-world performance.  This motivates us to propose a model that circumvents the challenges above while still providing the key benefits of ``both worlds''---improved performance in practice with worst-case guarantees.
}

\subsection{The \PALD Framework}
\label{sec:pald}

{\color{blue}
To tackle the above challenges, rather than defining a new prediction model and designing an algorithm around it, we propose a new learning framework (Partitioned Accounting \& Learned Decisions (\PALD)) that 
directly ``learns the best algorithm'' from within a general characterized set of algorithms, where each algorithm in the set satisfies a certain desired robustness (i.e., each algorithm guarantees a certain competitive ratio).
}

{\color{blue}
\smallskip
\noindent \textbf{Motivation. \ } 
Under the pseudo-cost minimization framework, a key step in bounding the competitive ratio of \PAAD (\sref{Thm.}{thm:osdm-upper-bound-1}) uses characteristics of the threshold functions ($\phi_b$ for base drivers, $\phi_f$ and $\psi_f$ for flexible drivers) that govern the algorithm's decisions. 
In particular, the definition of the pseudo-cost minimization problem allows us to relate \PAAD's cost to a lower bound on \OPT via these threshold functions.
Then, for a given \textit{target robustness} $\rho > \alpha$ (where $\alpha$ is the competitive ratio of \PAAD for \OSDMS), we can 
exactly define a \textit{feasible set} of threshold functions, based on the necessary relations for the competitiveness proof, that guarantee $\rho$-robustness for \PALD. 
A similar characterization of ``robust and consistent thresholds'' was derived by~\citet{Benomar:25} for one-max search to analytically investigate smoothness in algorithms with predictions.
In \PALD however, our goal is to directly \textit{learn} the best thresholds within a feasible set that \textit{only} considers robustness, empirically optimizing performance on historical instances while retaining a bounded competitive ratio.
Given learned thresholds $\hat{\phi}_b, \hat{\phi}_f$ and $\hat{\psi}_f$, \PALD's pseudocode is exactly the same as that of \PAAD (see \autoref{alg:paad}) but with $\hat{\phi}_b, \hat{\phi}_f$ and $\hat{\psi}_f$ in place of $\phi_b, \phi_f$ and $\psi_f$  (analytical thresholds).
}

\smallskip
\noindent \textbf{Robustness Certificate. \ }
In what follows, we define feasible sets of \textit{robust threshold functions} that guarantee $\rho$-robustness.  The \PALD framework accommodates any arbitrary learning mechanism as long as one can parameterize threshold functions and ensure they lie within these feasible sets, providing a robustness certificate that we use to prove \PALD's worst-case guarantees.
We focus on the setting of \OSDMS in the main body, deferring definitions, theorems, and proofs for \OSDMT to \autoref{apx:pald-robustness-osdmt}. We start with the definition of the feasible set for base demand drivers:

\begin{definition}[Robust threshold set for base drivers in \OSDMS]
\label{dfn:robust-threshold-base}
{\it
Given a target robustness $\rho > \alpha$ ($\alpha$ as defined in \eqref{eq:alpha1}), a learned threshold function $\hat{\phi}_b$ must lie in the following feasible set:
}
{\small
\begin{align*}
\mathcal{R}_b(\rho) \coloneqq \big\{ & \phi_b: [0,1] \to [p_{\min}, p_{\max}] \ \big| \ \phi_b \text{ is monotone non-increasing},  \phi_b(1) \leq p_{\min} + 2\gamma, \text{ and } \forall w \in [0,1] : \\
& \quad \Phi_b(0,w) + (1-w)(p_{\max} + 2\gamma) + p_{\max}(c+\varepsilon) + 2\delta - c w p_{\min} \leq \rho \left[ \phi_b(w) - 2\gamma + \varepsilon p_{\max} + \nicefrac{2\kappa}{T}\right] \big\}.
\end{align*}
}
\end{definition}

\noindent If a learned threshold function $\hat{\phi}_b$ lies in the set $\mathcal{R}_b(\rho)$, we have the necessary conditions to prove that \PALD is $\rho$-robust (with respect to base demand drivers).  The first two conditions (monotonicity and the inequality at $w\!=\!1$) ensure that $\hat{\phi}_b$ ``sweeps'' the range of possible prices to provide a necessary range of lower bounds on \OPT.  The third condition is the key inequality that relates \PALD's cost to that of the offline optimal via the pseudo-cost paradigm---on the left-hand-side, the integral and additive terms correspond to an upper bound on \PALD's total cost, while the right-hand-side corresponds to a lower bound on \OPT's cost scaled by the target robustness $\rho$.  The same logic extends, with some additional complexity, to the case of flexible demand drivers:

\begin{definition}[Robust threshold set for flexible drivers]
\label{dfn:robust-threshold-flexible}
{\it
Given a target robustness $\rho > \alpha$, learned threshold functions $\hat{\phi}_f$ and $\hat{\psi}_f$ for flexible demand drivers must lie in the following joint feasible set: 
}
{\small
\begin{align*}
\mathcal{R}_f(\rho) \coloneqq \big\{ & (\phi_f, \psi_f): [0,1]^2 \to [p_{\min}, p_{\max}]^2 \ \big| \ \phi_f, \psi_f \text{ are monotone non-increasing}, \ \phi_f(1) \leq p_{\min} + 2\gamma, \\
& \psi_f(1) \leq p_{\min}(c+\varepsilon) + 2\delta,  \text{ and } \forall w \in [0,1], v \in [0,w]: \Phi_f(0,w) + (1\!-\!w)(p_{\max}\!+\!2\gamma) - c w p_{\min} + \\
& \hspace{9.5em} \Psi_f(0,v) + (1\!-\!v)(p_{\max}(c\!+\!\varepsilon) \!+\! 2\delta) \leq \rho \big[ \nicefrac{1}{\omega} \big( \phi_f(w) + \psi_f(v) - 2\kappa \big) + \nicefrac{2\kappa}{T}\big] \big\}.
\end{align*}
}
\end{definition}

\noindent Under \sref{Def.}{dfn:robust-threshold-base} and \sref{Def.}{dfn:robust-threshold-flexible}, we have the following robustness certificate for \PALD when given any learned thresholds $\hat{\phi}_b, \hat{\phi}_f$ and $\hat{\psi}_f$ that lie in the respective feasible sets:
\begin{theorem}
\label{thm:pald-robustness}
Given learned threshold functions that lie in the feasible sets $\hat{\phi}_b \in \mathcal{R}_b(\rho)$ and $(\hat{\phi}_f, \hat{\psi}_f) \in \mathcal{R}_f(\rho)$ for some $\rho > \alpha$ ($\alpha$ defined in \eqref{eq:alpha1}), the \PALD algorithm is $\rho$-robust for \OSDMS.
\end{theorem}

{\color{blue}
\noindent \textbf{Proof Sketch of \autoref{thm:pald-robustness}.} (full proof in~\autoref{apx:pald-robustness}) \ 
Since \PALD is identical to \PAAD save for the learned threshold functions, we inherit feasibility from \sref{Lemma}{lem:osdm-feasibility-1}.  
First, we lower bound the cost of \OPT; as in the proof of \sref{Thm.}{thm:osdm-upper-bound-1}, we partition the time horizon into active and inactive periods. %
Since the partitioning relies on the value of the base driver threshold function and \PALD's is learned, we must use the definition of $\mathcal{R}_b$ to prove that $\OPT$'s cost during inactive periods can be bounded in the same manner. %
Next, we upper bound the cost of \PALD by ``charging'' each driver's decisions to their corresponding threshold functions and accounting for any worst-case purchasing, switching, and delivery costs that arise due to constraints. %
We %
show that the same charging argument from \sref{Thm.}{thm:osdm-upper-bound-1} still holds with learned thresholds.
Finally, we use the definitions of $\mathcal{R}_b(\rho)$ and $\mathcal{R}_f(\rho)$ to relate the cost of \PALD to that of \OPT.  Using a similar contradiction argument as in the proof of \sref{Thm.}{thm:osdm-upper-bound-1}, the last inequalities in the definition of the feasible sets are exactly those needed to show that the ratio $\frac{\PALD(\mathcal{I}) - p_{\max} \hat{s}}{\OPT(\mathcal{I})}$ must be at most $\rho$, completing the proof.  \hfill $\square$
}

{\color{blue}
\smallskip
\noindent \textbf{Remark. \ } 
\PALD is similar to certain existing approaches, but distinct in important ways.
For instance, \citet{Zeynali:21:DataDriven} propose learning the best exponent for an exponential threshold function in the online knapsack problem---our approach is distinguished by its expressiveness in allowing any arbitrary threshold function while retaining worst-case guarantees.
Our concept of ``learning a threshold'' is structurally similar to the notion of learning dual functions in online convex optimization~\cite{Lobos:21}---\PALD is distinguished by the ``robustness certificate,'' which additionally guarantees a certain (adversarial) competitive ratio.
\citet{Li:22:ExpertCalibrated} learn an ML model to solve online convex optimization that is calibrated (although not guaranteed) to be robust with respect to a baseline---in contrast, \PALD allows us to specify the worst-case level performance we are willing to incur, and to learn an algorithm that performs well subject to that constraint with certainty. 

}

\subsection{Learning Methodology}
\label{sec:pald-learning}

{\color{blue}
In this section, we describe the methodology we use in our case study (see \autoref{sec:exp}) 
to instantiate a ``proof-of-concept'' of the \PALD framework in practice.
We parameterize each threshold function as a piecewise-affine function using a ``knot'' representation at a fixed grid of $K >1$ points $\{0, \nicefrac{1}{K-1}, \nicefrac{2}{K-1}, \ldots, 1\}$.  We learn a parameter vector $\mathbf{y} \in \mathbb{R}_{+}^K$ where each element defines the function's value at a knot (between knots, the function is interpolated linearly).  We denote by $\mathbf{y}_b, \mathbf{y}_f, \mathbf{y}_\psi$ the parameter vectors for $\smash{\hat{\phi}_b, \hat{\phi}_f, \hat{\psi}_f}$, respectively.  This makes the robust sets tractable:
}
\begin{lemma}
\label{lem:convex-robust-set}
Letting $\smash{\hat{\phi}_b, \hat{\phi}_f, \hat{\psi}_f}$ be piecewise-affine functions parameterized by $\mathbf{y}_b, \mathbf{y}_f, \mathbf{y}_\psi$ %
as outlined above, the feasible sets defined in \sref{Def.}{dfn:robust-threshold-base} and \ref{dfn:robust-threshold-flexible} are convex sets in $\mathbf{y}_b$ and $(\mathbf{y}_f, \mathbf{y}_\psi)$, respectively. 
\end{lemma}

{\color{blue}
\noindent The full proof is in \autoref{apx:convex-robust-set}---the key idea is that the necessary robustness certificates for \PALD reduce to box constraints, linear inequalities, and affine halfspaces, whose intersection is convex.
}

Recall that the analytical thresholds $\phi_b, \phi_f$ and $\psi_f$ (\sref{Defs.}{def:base-threshold} \& \ref{def:flex-thresholds}) are continuous and monotone non-increasing functions---by definition of the \PAAD algorithm, these lie in the feasible sets $\mathcal{R}_b(\alpha)$ and $\mathcal{R}_f(\alpha)$, respectively.  For $\rho$ increasing away from $\alpha$, the feasible sets $\mathcal{R}_b(\rho)$ and $\mathcal{R}_f(\rho)$ become strictly larger, since the robustness conditions become easier to satisfy.  Since piecewise-affine functions are known to be dense in continuous functions~\cite{Cheney:1998:Density, Pinkus:2005:Density}, the analytical thresholds can be approximated arbitrarily well using piecewise-affine functions with sufficiently many knots.  Thus, for a given robustness $\rho > \alpha$, the feasible sets $\mathcal{R}_b(\rho)$ and $\mathcal{R}_f(\rho)$ are non-empty since they at least contain the analytical thresholds. 
\sref{Lemma}{lem:convex-robust-set} implies that it is efficient to \textit{project} any learned piecewise-affine function into the feasible sets---this makes robustness tractable to enforce during learning, and unlocks several learning techniques which we discuss below.\footnote{ {\color{blue} In principle, \PALD could use substantially different learning techniques than what we consider here---our main contribution is the notion of directly finding ``the best algorithm'' amongst a set of algorithms satisfying a given worst-case guarantee. }}

{\color{blue}
\smallskip
\noindent \textbf{Directly Learning Thresholds. \ }
We implement the \PALD algorithm by using differentiable convex optimization layers (e.g., CVXPYLayers~\cite{CVXPYlayers:19}) to solve the pseudo-cost minimization problem at each time step.
This allows backpropagation through the entire \PALD procedure to minimize empirical competitive ratio on historical instances.
Given a data set of historical instances $\mathcal{D} = \{\mathcal{I}_1, \ldots, \mathcal{I}_N\}$, we can use this differentiable implementation of \PALD to find a single set of parameters $\mathbf{y}_b, \mathbf{y}_f, \mathbf{y}_\psi$ that minimizes the empirical competitive ratio on $\mathcal{D}$, using projected gradient descent (PGD) to ensure that, during training, the learned thresholds lie in the robust sets $\mathcal{R}_b(\rho)$ and $\mathcal{R}_f(\rho)$ for a target $\rho > \alpha$.  We refer to this as \texttt{PALD-S} (``\texttt{S}'' for simple) in our case study.
}

{\color{blue}
\smallskip
\noindent \textbf{End-to-End Contextual Learning. \ }
While directly learning a single set of thresholds improves average-case performance, it cannot leverage instance-specific \textit{context} that is extremely useful in practice.  For instance, our case study in \autoref{sec:exp} uses locational marginal price (LMP) data for different grids; LMPs are known to correlate with factors such as the time of day, the weather, and the season~\cite{Qu:24}.} Using a learning method that is aware of such context can enable better performance in practice by learning thresholds that are tailored to likely operating conditions. 

{\color{blue}
To realize this context-aware approach, we employ a neural network $F_\Theta(\mathbf{x}) \to \mathbf{y}_b, \mathbf{y}_f, \mathbf{y}_\psi$ (with parameters $\Theta$) that takes as input a context vector $\mathbf{x}$ about a given driver and outputs predicted vectors $\mathbf{y}_b, \mathbf{y}_f, \mathbf{y}_\psi$ that define the threshold functions.  The final layer of this neural network is a differentiable projection layer that projects the output parameters $\mathbf{y}_b, \mathbf{y}_f, \mathbf{y}_\psi$ into the feasible sets $\mathcal{R}_b(\rho)$ and $\mathcal{R}_f(\rho)$ to ensure target robustness $\rho > \alpha$.
We construct an unsupervised learning pipeline that simulates the differentiable \PALD algorithm on a historical data set $\mathcal{D}$.  For each instance, each new driver in \PALD makes a forward pass through the neural network $F_\Theta$ to obtain its context-aware thresholds $\smash{\hat{\phi}_b, \hat{\phi}_f}$, and $\smash{\hat{\psi}_f}$.  After observing the resulting performance, we backpropagate through the entire pipeline to learn a better mapping between context and thresholds.  We refer to this technique as \texttt{PALD-C} in our case study, and give a diagram of it in \sref{Appendix}{apx:pald-implementation-details}.
}

\section{Case Study: A Grid-integrated Data Center with Local Energy Storage}
\label{sec:exp}
We conclude with a case study application of \OSDM and our algorithms for the motivating application of demand management in a grid-integrated data center with local energy storage.

\smallskip
\noindent \textbf{Experimental Setup. \ } We simulate a data center with co-located energy storage that purchases electricity from a real-time market and serves a mix of base and flexible demand (e.g., interactive and batch jobs, respectively).

\noindent \textit{$\triangleright$ Electricity Price Data. \ } We use a year of locational marginal price (LMP) data at a 15-minute granularity in four U.S. grid regions: California ISO (\texttt{CAISO}), Electric Reliability Council of Texas (\texttt{ERCOT}), PJM Interconnection (\texttt{PJM}), and ISO New England (\texttt{ISONE})~\cite{GridStatus}.  For \texttt{PALD-Contextual}, we also use day-ahead LMP forecasts at an hourly granularity as contextual features.  We truncate any negative or small prices to \$$1$ per MWh to avoid a divide-by-zero in the competitive ratio, and cap prices at the 99.9$^\text{th}$ percentile to remove outliers.  We summarize key statistics in \autoref{tab:characteristics}.

\begin{table}[h]
\centering %
\caption{Statistics for LMP traces.  All traces span the time period 01/01/2024-12/31/2024 at a 15-minute granularity, and LMP data is reported in \$ (USD) per megawatt-hour (MWh). } \vspace{-1em}
\label{tab:characteristics}
\begin{tabular}{|l|l|l|l|l|l|}
\hline
Trace & Min. & Max. & Mean & Std. Dev. & Coeff. of Var.~\cite{GridStatus} \\ \hline
\texttt{CAISO} & 1.00 & 523.85 & 45.63 & 34.91 & 0.765 \\ \hline 
\texttt{ISONE} & 1.00 & 336.27 & 39.57 & 31.44 & 0.794 \\ \hline 
\texttt{PJM} & 1.00 & 325.93 & 31.10 & 26.95 & 0.867 \\ \hline 
\texttt{ERCOT} & 1.00 & 877.48 & 25.02 & 41.30 & 1.651 \\ \hline
\end{tabular}
\end{table}

\noindent \textit{$\triangleright$ Demand Data. \ } To model the data center's demand, we use Alibaba's 2018 production trace~\cite{Alibaba:18}, which provides seven days of task start and end times at a 1-second granularity.  In our experiments, time steps are fixed to 15 minutes (matching the price data), so we \textit{aggregate} the demand: we partition the trace into 15-minute buckets and compute a \textit{weighted number of tasks} for each bucket based on the number of jobs and their active duration (a task that runs for the entire bucket contributes 1, while one that runs for only 5 minutes contributes $\nicefrac{1}{3}$).  We plot this quantity for the full trace in the Appendix, in \autoref{fig:weighted-active-jobs}. 
Based on this, we compute a scaled demand as $\nicefrac{\text{weighted no. of tasks}}{\text{scale factor}}$, where a scale factor is chosen based on the size of the energy storage.  In most of our experiments, we set this factor to reflect a battery capacity that can meet the daily peak demand for a single (15-minute) time step on $\nicefrac{2}{7}$ days in the trace.  In some experiments, we vary this scale factor to study the effect of storage capacity.  Since the trace does not specify task types, we also set a parameter $\texttt{prop\_base} \in [0.5, 1]$ to probabilistically classify tasks as base or flexible: with probability \texttt{prop\_base}, a task is ``base'', and otherwise it is ``flexible.''  We typically set $\texttt{prop\_base}=0.5$, which yields instances that consist of half base and half flexible demand on average.

We simulate 1,200 instances of \OSDM for each parameter configuration.  In addition to the data parameters above, we isolate the effect of parameters including the delivery cost coefficients $c, \varepsilon$, the switching cost coefficients $\gamma, \delta$, time horizon $T$, and the tracking cost coefficient $\eta$ (for \OSDMT).  Unless otherwise specified, we set $T=48$ (i.e., a 12-hour horizon), $\gamma=10$, $\delta=5$, and $\eta=0$ (i.e., no tracking cost).  
In experiments with a tracking cost (i.e., with non-zero $\eta$), we set a tracking target sequence $\{a_t \}_{t\in [T]}$ that evenly distributes the total demand over the time horizon (i.e., $a_t \approx \nicefrac{D}{T}$, where $D$ is the total demand), while randomly choosing between two to four time steps to set $a_t = 0$, simulating load-shedding events.  
We implement a price-dependent decreasing delivery cost as a soft penalty to encourage the decision maker to keep the battery charged, setting $c=0.2$ and $\varepsilon=0.05$ unless otherwise specified.  These parameters quantify the risk tolerance of the data center as it relates to keeping the battery charged for, e.g., backup power.

{\color{blue}
\noindent \textit{$\triangleright$ Algorithms. \ } We implement and test three algorithms, denoted as \PAAD, \PALDS, and \PALDC.  We solve for the offline optimal solution using GurobiPy~\cite{gurobi} to compute the empirical competitive ratio (ECR).
\PAAD is the competitive algorithm described in \autoref{sec:alg}, while \PALDS and \PALDC are instantiations of the learning framework described in \autoref{sec:learn}---we typically set a robustness factor of $\rho = 5\alpha$ (where $\alpha$ is the standard competitive ratio) to ensure robustness while allowing for some benefit from learning, and set $K=10$.  
We implement the \PALD algorithm using CVXPYLayers~\cite{CVXPYlayers:19,CVXPY} to optimize downstream task performance %
via backpropagation. 
\PALDS directly learns one set of threshold functions to minimize overall ECR---to capture seasonality in the traces, we train a set of \PALDS thresholds for each month in a given region's LMP data (i.e., 12 sets of thresholds total), using 100 random training instances for each month.
In contrast, \PALDC is a neural network-based approach that uses context features to predict threshold functions for each driver---we train it on 100 random training instances using all 12 months in a given region's LMP data, using the time, month, tracking target (if $\eta > 0$), and statistics about day-ahead LMP forecasts (i.e., $\min$, $\max$, avg., std. dev.), as features.
We defer more implementation details of \PALDS \& \PALDC to~\autoref{apx:pald-implementation-details}.
}

\begin{figure*}[t]
    \centering
    \vspace{-1em}
    \minipage{\textwidth}
    \begin{center}
    \includegraphics[width=0.35\linewidth]{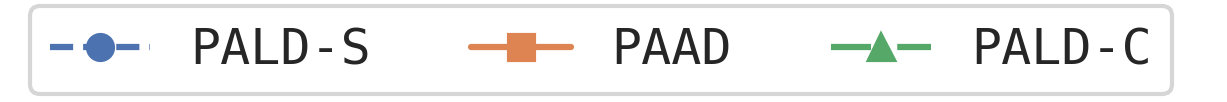}\vspace{0.1em}
    \end{center}
    \endminipage\hfill\\
        
    \minipage{0.24\textwidth}
    \includegraphics[width=\linewidth]{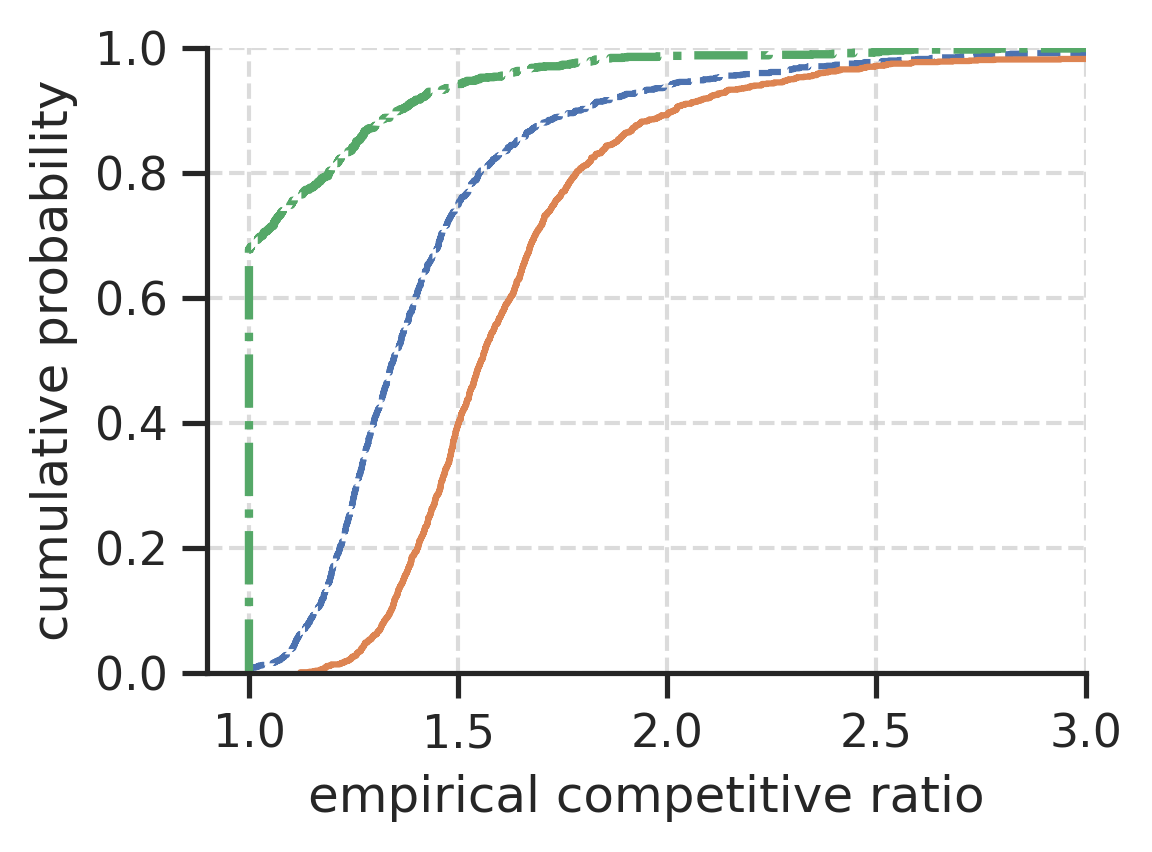}\vspace{-1em}
    \caption{ CDFs of ECR for all algorithms in default \texttt{CAISO} experiments.}\label{fig:cdf}
    \endminipage 
    \hfill
    \minipage{0.73\textwidth}
    \null
    \vspace{-2em}
    \begin{subfigure}{0.32\textwidth}
        \centering
        \includegraphics[width=\textwidth]{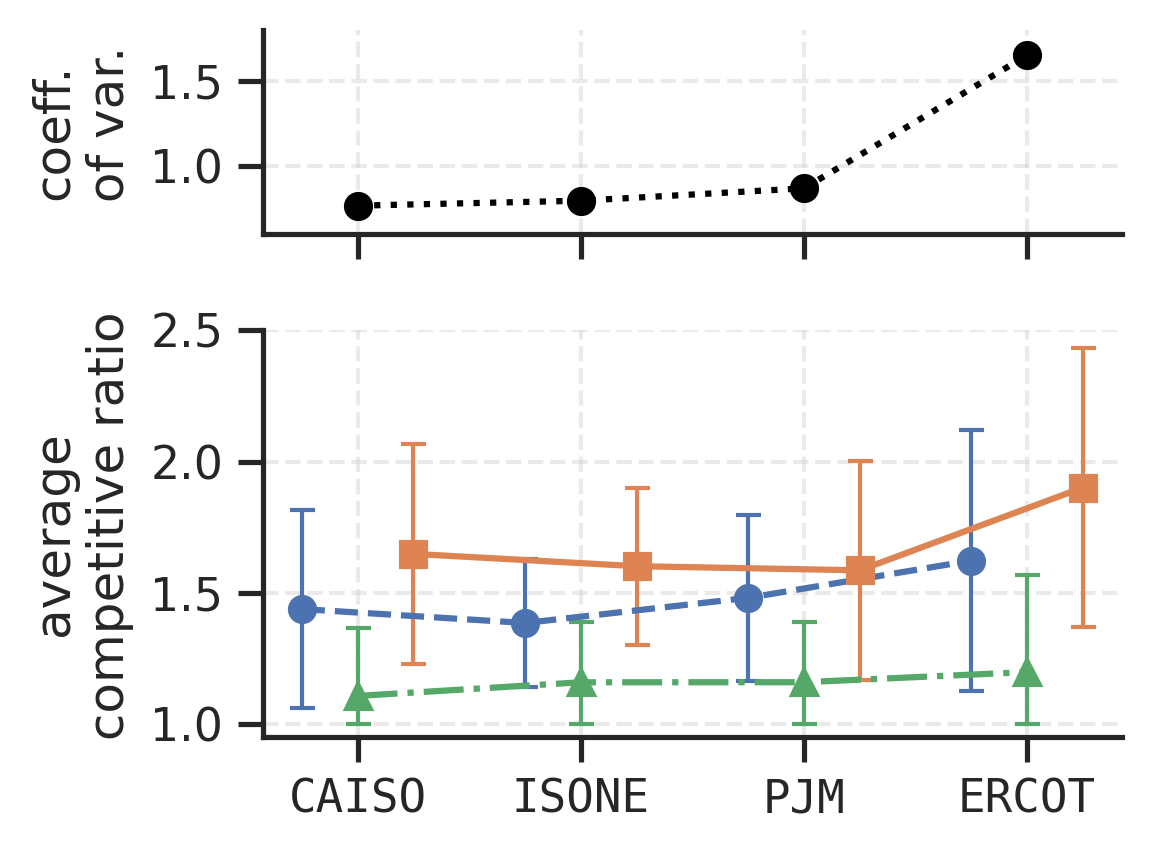}\vspace{-0.5em}
        \caption{Grid regions}
        \label{fig:regions}
    \end{subfigure}
    \hfill
    \begin{subfigure}{0.32\textwidth}
        \centering
        \includegraphics[width=\textwidth]{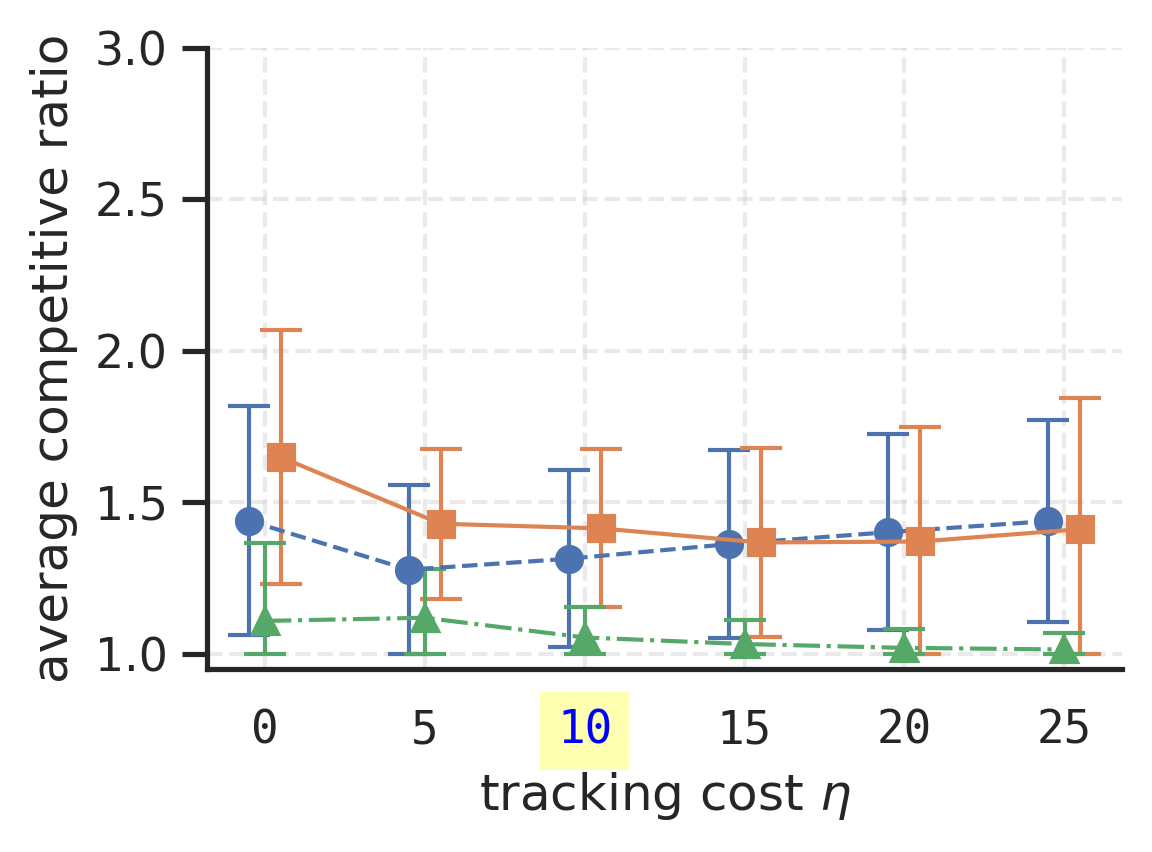}\vspace{-0.5em}
        \caption{Tracking cost $\eta$}
        \label{fig:etas}
    \end{subfigure}
    \begin{subfigure}{0.32\textwidth}
        \centering
        \includegraphics[width=\textwidth]{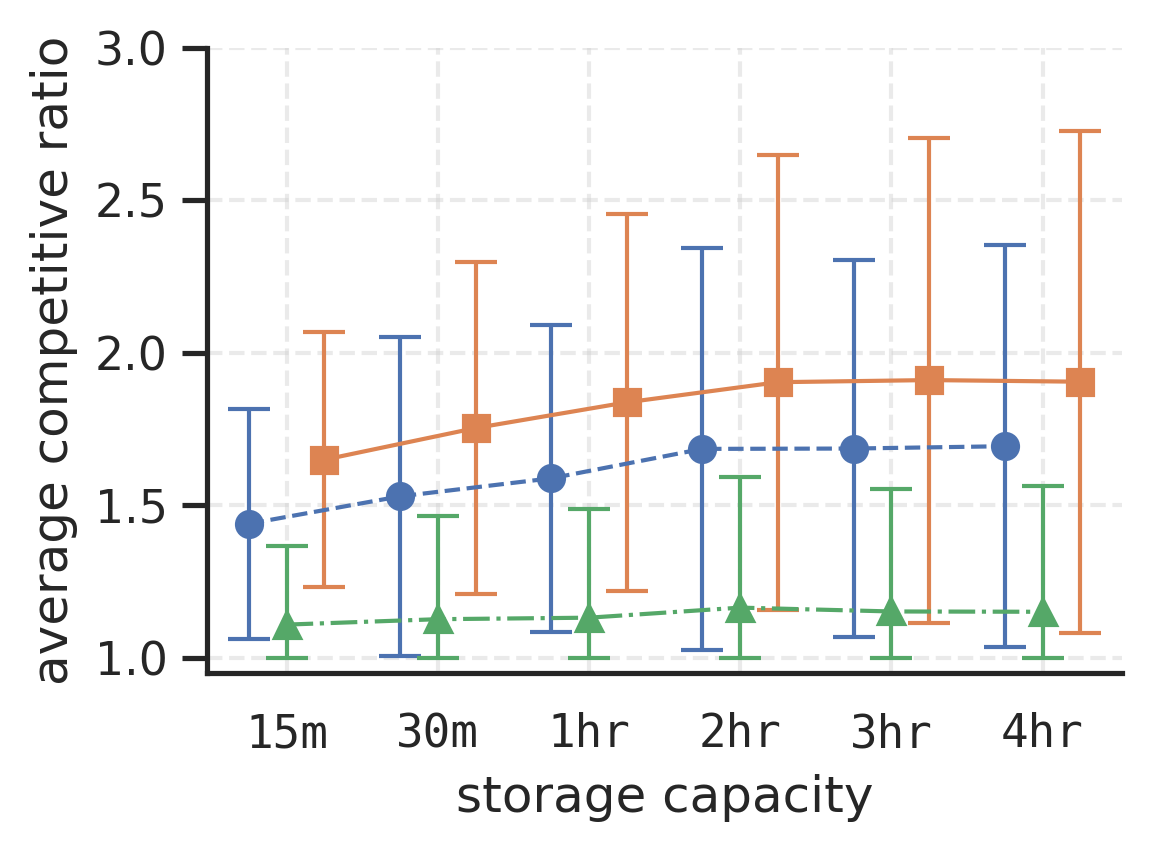}\vspace{-0.5em}
        \caption{Storage size}
        \label{fig:scale-factors}
    \end{subfigure}
    \vspace{-1em}
    \caption{Average ECR for all algorithms in varying regions, with a tracking cost, and varying storage size.  In (a), top subfigure is coeff. of var. for each region.}
    \label{fig:varied1}
    \vspace{-2em}
    \endminipage
\end{figure*}

\smallskip
\noindent \textbf{Experimental Results. \ } 
We highlight key experiments in the main body, referring to \autoref{apx:exp-results} for an investigation of algorithm runtimes and the effect of parameters $\delta$ and $\rho$. 
A summary is given in \autoref{fig:cdf}, which plots a cumulative distribution function (CDF) of the ECR for all tested algorithms in the \texttt{CAISO} region, with the default parameter settings. 
Our competitive algorithm \PAAD is $1.65$-competitive on average and $\leq 2.3$-competitive in 95\% of cases.  By learning a single set of thresholds, \PALDS improves on \PAAD by $12.7\%$ on average and $8.7\%$ at the 95$^\text{th}$ percentile.  Finally, by leveraging contextual features about instances, \PALDC significantly improves on \PAAD by $32.8\%$ on average and $31.61\%$ at the 95$^\text{th}$ percentile.  Remarkably, \PALDC's performance is nearly optimal (i.e., ECR $\approx$ 1) in roughly $\nicefrac{1}{2}$ of instances, highlighting the value of a context-aware approach.
We now consider the impact of grid region, tracking cost, and storage capacity on performance, training \PALD once for each region, for the tracking cost $\eta = 10$, and for each capacity.

\noindent \textit{$\triangleright$ Grid Regions. \ } To capture the effect of different electricity price distributions and characteristics, \autoref{fig:varied1}(a) shows the results of the main experiment with the same parameters in three other regions, namely \texttt{ERCOT}, \texttt{PJM}, and \texttt{ISONE}.  Both \PALDS and \PALDC are retrained on a test of training instances generated from each region.
We expect to see some variance in performance across these regions, since the global price bounds ($p_{\min}$ and $p_{\max}$) are not homogeneous---interestingly, we find that the performance of algorithms roughly corresponds to the variability of prices in each region, with \texttt{ERCOT} (highest coefficient of variation) showing the most challenging case.

\noindent \textit{$\triangleright$ Tracking Case. \ } In \autoref{fig:varied1}(b), we show algorithm performance in a setting with a tracking cost (i.e., instead of a switching cost, so $\gamma=0$), varying $\eta$ between $0$ and $25$.
Both \PALDS and \PALDC are retrained with $\eta=10$. %
We find that the performance of \PAAD and \PALDC slightly improves as $\eta$ grows---\PALDS continues to improve on \PAAD in the small $\eta$ regime, but for large (out-of-distribution) $\eta > 15$, \PALDS's performance degrades because the single set of thresholds \PALDS learns becomes overly conservative (i.e., waits too long to purchase) when $\eta$ grows.

\noindent \textit{$\triangleright$ Storage Capacities. \ } 
Recall that the default setting is where the storage is sized to meet the peak demand on two of the seven days in the Alibaba trace for a single 15-minute time slot---in \autoref{fig:varied1}(c), we plot the effect of larger storage sizes that meet up to $4$ hours of the same peak demand on two of the seven days (i.e., a $16\times$ increase in storage size).  As the storage size increases and \OPT gains flexibility for shifting, the performance of all algorithms slightly degrades, although \PALDC consistently outperforms the other algorithms across all storage cases.

\noindent Next, we plot the effect of parameters that are \textit{out-of-distribution} for \PALD.  Throughout the plots, we \textit{highlight} the instance that is in-distribution (i.e., parameters match training set).

\begin{figure*}[t]
    \centering
    \vspace{-1em}
    \begin{subfigure}{\textwidth}
        \centering
        \includegraphics[width=0.35\textwidth]{figs/legend.png}
        \label{fig:top}
    \end{subfigure}
        
    \begin{subfigure}{0.24\textwidth}
        \centering
        \includegraphics[width=\textwidth]{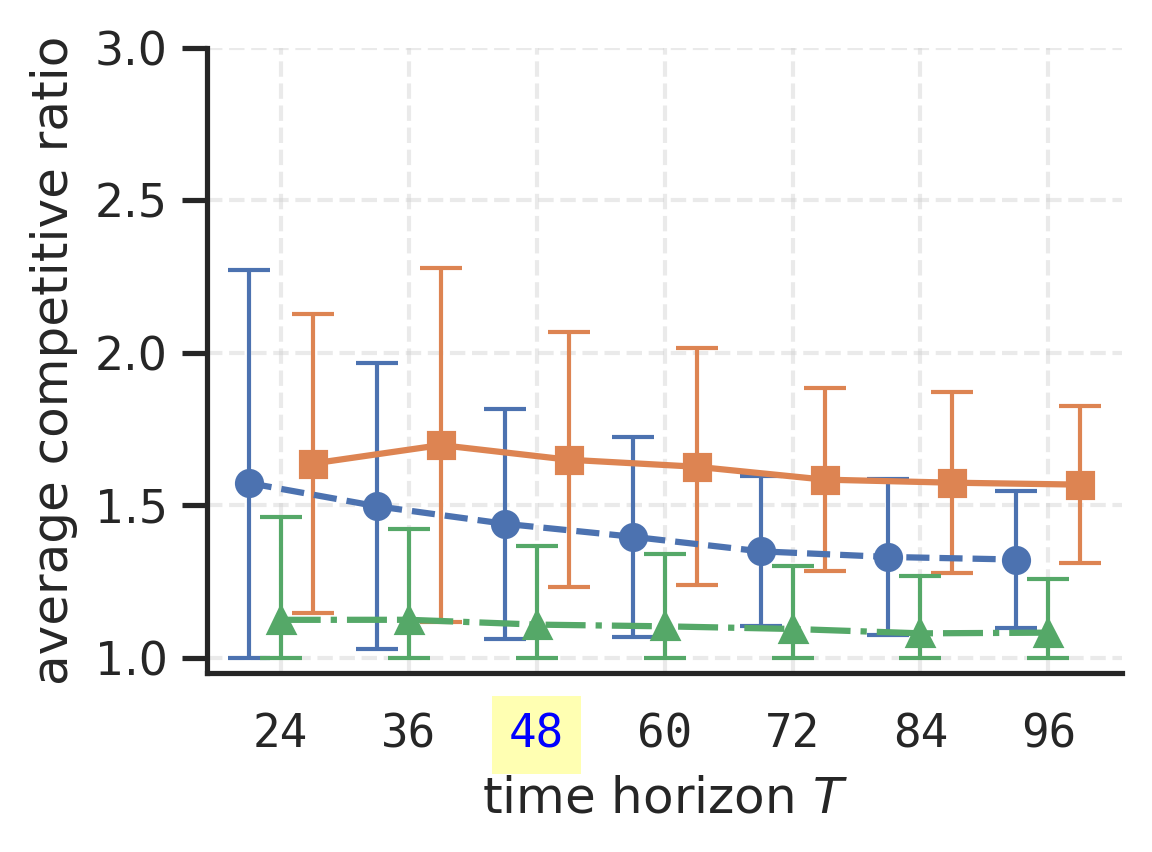}\vspace{-0.5em}
        \caption{Time horizon $T$}
        \label{fig:horizons}
    \end{subfigure}
    \hfill
    \begin{subfigure}{0.24\textwidth}
        \centering
        \includegraphics[width=\textwidth]{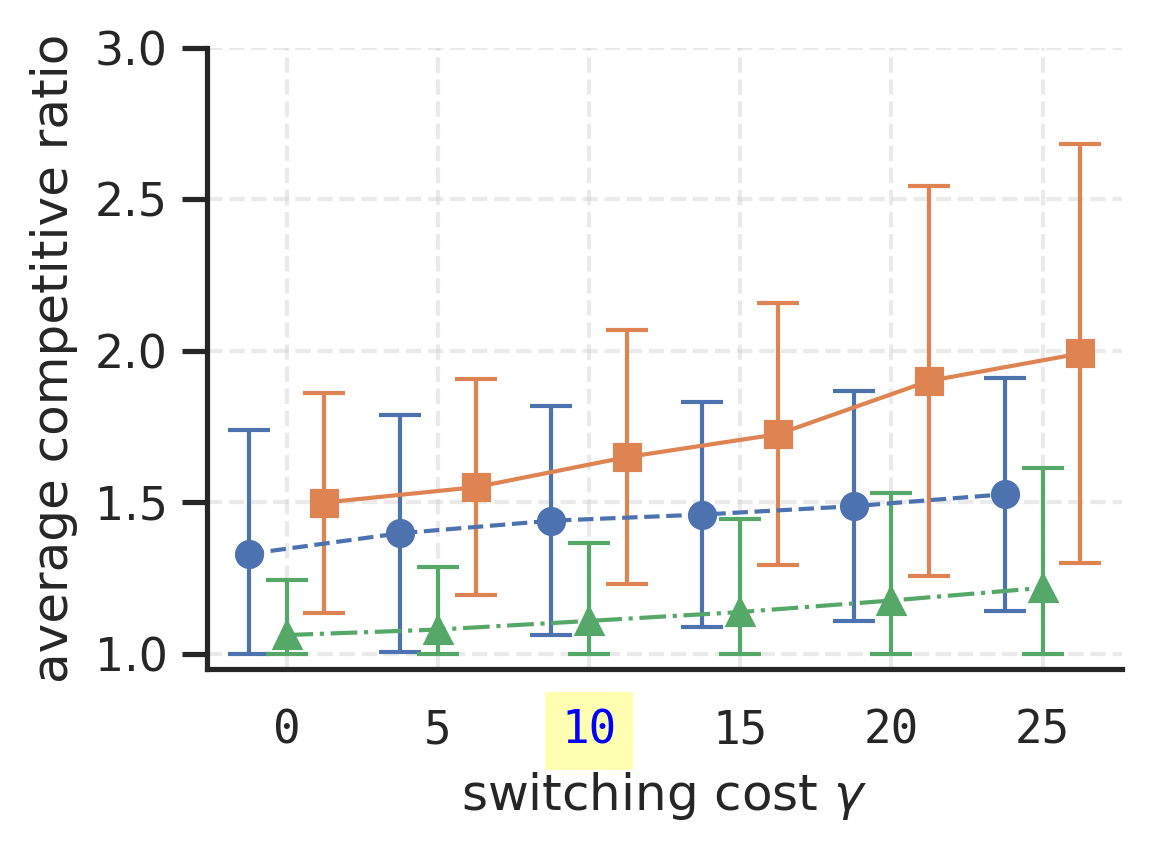}\vspace{-0.5em}
        \caption{Switching cost $\gamma$}
        \label{fig:gammas}
    \end{subfigure}
    \hfill
    \begin{subfigure}{0.24\textwidth}
        \centering
        \includegraphics[width=\textwidth]{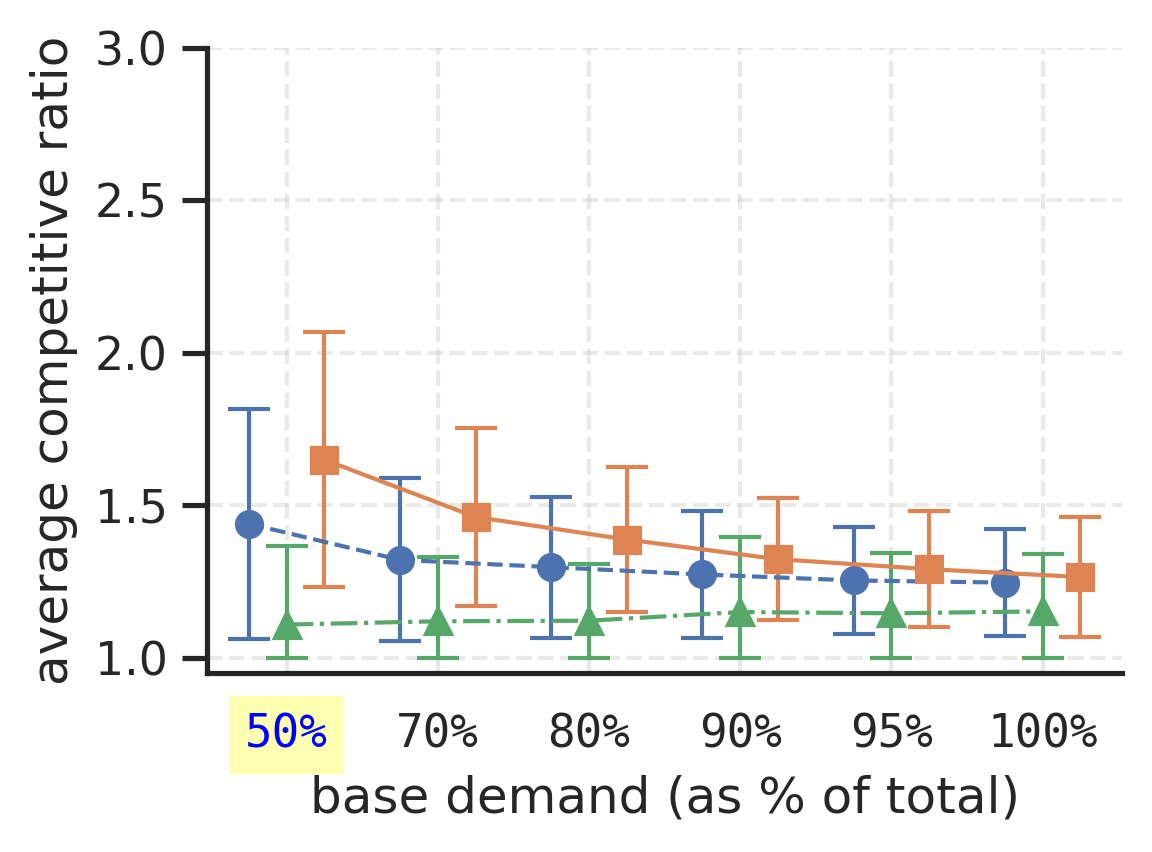}\vspace{-0.5em}
        \caption{Demand mix}
        \label{fig:prop-base}
    \end{subfigure}
    \begin{subfigure}{.24\textwidth}
        \centering
        \includegraphics[width=\textwidth]{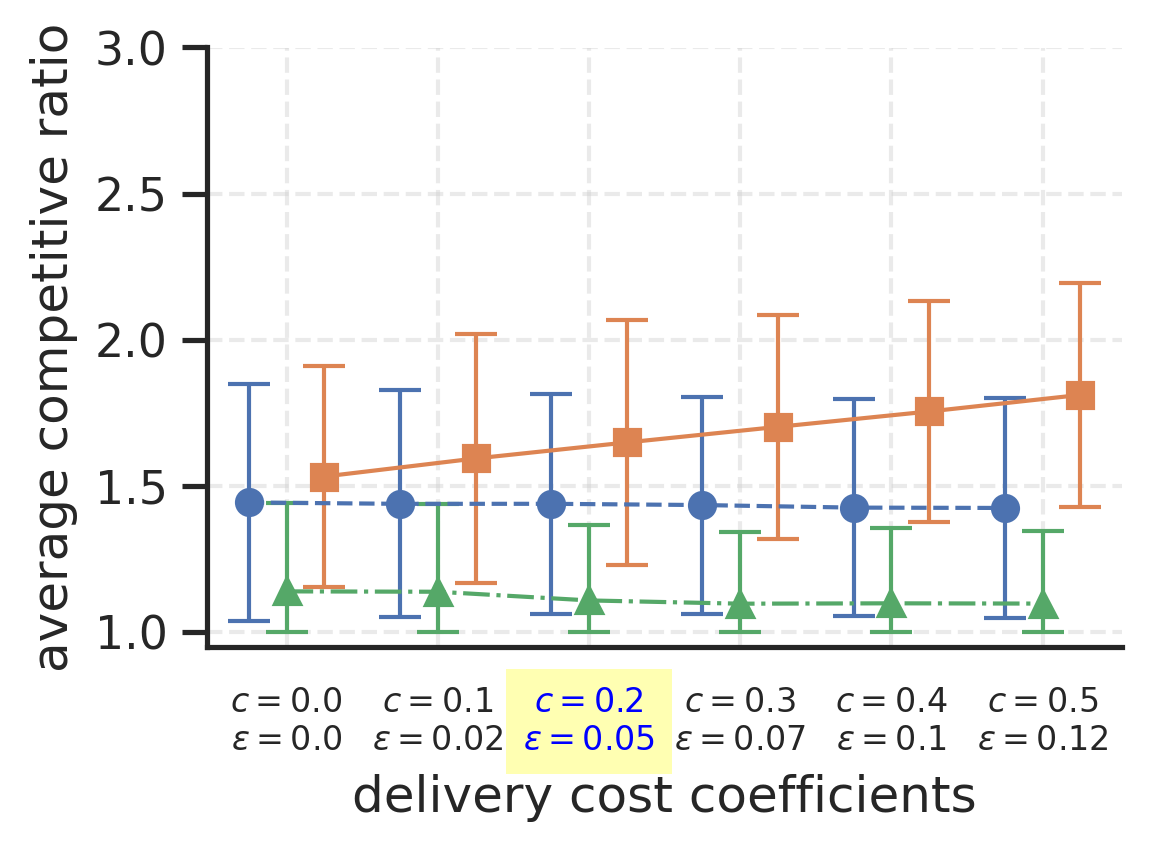}\vspace{-0.5em}
        \caption{Delivery cost ($c$ and $\varepsilon$)}
        \label{fig:deliveries}
    \end{subfigure}
    \vspace{-1em}
    \caption{Average ECR for all algorithms in \texttt{CAISO} with varying parameters as specified.}
    \label{fig:varied2}
    \vspace{-1.5em} %
\end{figure*}

\noindent \textit{$\triangleright$ Time Horizons. \ } \autoref{fig:varied2}(a) shows the effect of varying the time horizon $T$ between $24$ and $96$ in \texttt{CAISO}.  We find that the performance of all algorithms improves slightly with longer horizons, which we attribute to the increased flexibility in using storage over longer time periods.  Notably, \PALDC consistently outperforms the other algorithms across all horizon lengths.

\noindent \textit{$\triangleright$ Switching Costs. \ } \autoref{fig:varied2}(b) shows the effect of varying the switching cost coefficient $\gamma$ (with fixed $\delta = 5$) between $0$ and $25$ in \texttt{CAISO}.  As predicted by the theoretical results, the performance of all algorithms degrades slightly as the switching cost coefficient $\gamma$ increases---in this respect, \PALDS and \PALDC are less sensitive to increases in switching costs than \PAAD.  

\noindent \textit{$\triangleright$ Varying Demands. \ } 
We estimate the effect of varying demand distribution on the performance of all algorithms in the \texttt{CAISO} region.  Recall that the default setting splits demand evenly between base and flexible demand ($\texttt{prop\_base} = 0.5$)---in \autoref{fig:varied2}(c), we plot the effect of varying \texttt{prop\_base} between $0.5$ and $1.0$ in the \texttt{CAISO} region.  We find that the performance of \PAAD and \PALDS generally improves as the proportion of base demand increases (deviating further from the training instances), while \PALDC's performance slightly degrades, although it still outperforms the other algorithms.

\noindent \textit{$\triangleright$ Delivery Costs. \ } In \autoref{fig:varied2}(d), we show the effect of varying the delivery cost parameters $c$ and $\varepsilon$ (see \sref{Def.}{def:osdm-delivery-pd-monotone}) in the \texttt{CAISO} region.  As $c$ and $\varepsilon$ increase, the delivery cost has a more significant effect on the resulting cost. As predicted by the theoretical results, \PAAD's performance degrades slightly as a function of $c$ and $\varepsilon$, while the performance of \PALDS and \PALDC remains relatively constant.

\section{Conclusion}
\label{sec:conc}

Motivated by paradigm shifts promoting active and intelligent demand management models (e.g., demand response), we introduce and study online smoothed demand management (\OSDM)---the first online problem to jointly capture key practical features such as mixed inflexible and flexible demand, local energy storage, and costs penalizing unsmooth behavior.
We present \PAAD, an algorithm achieving the optimal competitive ratio for \OSDM, and \PALD, a novel differentiable framework that learns the best algorithm from data while preserving worst-case guarantees. We evaluate both in a case study of a grid-integrated data center with co-located storage.

{\color{blue}
Several directions remain for future work. One natural extension is to incorporate \textit{rate constraints}, where purchase and/or delivery rates are capped per time step as in some related models~\cite{Lechowicz:24, Yang:20}.  While the monotone state- and price-dependent delivery cost we consider (\sref{Def.}{def:osdm-delivery-pd-monotone}) is sufficiently general to capture some motivating applications, the delivery cost could be generalized---for instance, one might consider a cost that depends on a second adversarial price sequence (separate from $p_t$). 
More broadly, while our \PALD learning framework is developed for \OSDM, its approach of learning ``the best algorithm'' out of a set of algorithms that satisfy a given worst-case competitive ratio may generalize to other online problems, motivating further study of which problem classes and algorithm parameterizations admit such a framework.
}

\bibliographystyle{ACM-Reference-Format}
\bibliography{main}

@inproceedings{NEURIPS2024_11c6625b,
	title        = {{Overcoming Brittleness in Pareto-Optimal Learning Augmented Algorithms}},
	author       = {Elenter, Alex and Angelopoulos, Spyros and D\"{u}rr, Christoph and Lefki, Yanni},
	year         = 2024,
	booktitle    = {Advances in Neural Information Processing Systems},
	volume       = 37,
	pages        = {9329--9357},
	url          = {https://proceedings.neurips.cc/paper_files/paper/2024/file/11c6625b0481a7d5625831369f6b7c82-Paper-Conference.pdf},
}

@inproceedings{Manasse:88,
	title        = {{Competitive Algorithms for On-Line Problems}},
	author       = {Manasse, Mark and McGeoch, Lyle and Sleator, Daniel},
	year         = 1988,
	booktitle    = {Proceedings of the 20th Annual ACM Symposium on Theory of Computing},
	location     = {Chicago, Illinois, USA},
	publisher    = {Association for Computing Machinery},
	address      = {New York, NY, USA},
	series       = {STOC '88},
	pages        = {322–333},
	doi          = {10.1145/62212.62243},
	isbn         = {0897912640},
	numpages     = 12
}

@inproceedings{Daneshvaramoli:25,
	title        = {{Near-Optimal Consistency-Robustness Trade-Offs for Learning-Augmented Online Knapsack Problems}},
	author       = {Mohammadreza Daneshvaramoli and Helia Karisani and Adam Lechowicz and Bo Sun and Cameron N Musco and Mohammad Hajiesmaili},
	year         = 2025,
	booktitle    = {Forty-second International Conference on Machine Learning},
    eprint       = {2406.18752},
	archiveprefix = {arXiv},
}

@misc{GridStatus,
	title        = {{Grid Status}},
	author       = {{Grid Status}},
	year         = 2025,
	url          = {https://www.gridstatus.io},
	urldate      = {2025-10-11},
	organization = {Grid Status}
}

@misc{gurobi,
	title        = {{Gurobi Optimizer Reference Manual}},
	author       = {{Gurobi Optimization, LLC}},
	year         = 2024,
	url          = {https://www.gurobi.com}
}

@article{Jiang:14:LoadShaping,
	title        = {{Load Shaping Strategy Based on Energy Storage and Dynamic Pricing in Smart Grid}},
	author       = {Jiang,  Tao and Cao,  Yang and Yu,  Liang and Wang,  Zhiqiang},
	year         = 2014,
	month        = nov,
	journal      = {IEEE Transactions on Smart Grid},
	publisher    = {Institute of Electrical and Electronics Engineers (IEEE)},
	volume       = 5,
	number       = 6,
	pages        = {2868–2876},
	doi          = {10.1109/tsg.2014.2320261},
	issn         = {1949-3061},
	url          = {http://dx.doi.org/10.1109/TSG.2014.2320261}
}

@article{Aryandoust:17:LoadShaping,
	title        = {{The potential and usefulness of demand response to provide electricity system services}},
	author       = {Aryandoust,  Arsam and Lilliestam,  Johan},
	year         = 2017,
	month        = oct,
	journal      = {Applied Energy},
	publisher    = {Elsevier BV},
	volume       = 204,
	pages        = {749–766},
	doi          = {10.1016/j.apenergy.2017.07.034},
	issn         = {0306-2619},
	url          = {http://dx.doi.org/10.1016/j.apenergy.2017.07.034}
}

@misc{Benomar:25,
	title        = {{Pareto-Optimality, Smoothness, and Stochasticity in Learning-Augmented One-Max-Search}},
	author       = {Ziyad Benomar and Lorenzo Croissant and Vianney Perchet and Spyros Angelopoulos},
	year         = 2025,
	url          = {https://arxiv.org/abs/2502.05720},
	eprint       = {2502.05720},
	archiveprefix = {arXiv},
	primaryclass = {cs.DS}
}

@article{Borodin:92,
	title        = {{An Optimal On-Line Algorithm for Metrical Task System}},
	author       = {Borodin, Allan and Linial, Nathan and Saks, Michael E.},
	year         = 1992,
	month        = {Oct},
	journal      = {J. ACM},
	publisher    = {Association for Computing Machinery},
	address      = {New York, NY, USA},
	volume       = 39,
	number       = 4,
	pages        = {745–763},
	doi          = {10.1145/146585.146588},
	issn         = {0004-5411},
	url          = {https://doi.org/10.1145/146585.146588},
	issue_date   = {Oct. 1992},
	numpages     = 19
}

@inproceedings{Lykouris:18,
	title        = {{Competitive Caching with Machine Learned Advice}},
	author       = {Lykouris, Thodoris and Vassilvtiskii, Sergei},
	year         = 2018,
	month        = {10--15 Jul},
	booktitle    = {Proceedings of the 35th International Conference on Machine Learning},
	publisher    = {PMLR},
	series       = {Proceedings of Machine Learning Research},
	volume       = 80,
	pages        = {3296--3305},
	url          = {https://proceedings.mlr.press/v80/lykouris18a.html},
	editor       = {Dy, Jennifer and Krause, Andreas}
}

@inproceedings{Purohit:18,
	title        = {{Improving Online Algorithms via ML Predictions}},
	author       = {Kumar, Ravi and Purohit, Manish and Svitkina, Zoya},
	year         = 2018,
	booktitle    = {Advances in Neural Information Processing Systems},
	volume       = 31,
    eprint       = {2407.17712},
	archiveprefix = {arXiv},
}

@inproceedings{Lechowicz:24,
	title        = {{Online Conversion with Switching Costs: Robust and Learning-augmented Algorithms}},
	author       = {Adam Lechowicz and Nicolas Christianson and Bo Sun and Noman Bashir and Mohammad Hajiesmaili and Adam Wierman and Prashant Shenoy},
	year         = 2024,
	month        = {Jun},
	booktitle    = {Proceedings of the 2024 SIGMETRICS/Performance Joint International Conference on Measurement and Modeling of Computer Systems},
	location     = {Venice, Italy},
	publisher    = {Association for Computing Machinery},
	address      = {New York, NY, USA},
	series       = {SIGMETRICS / Performance '24}
}

@book{Mitrinovic:91,
	title        = {{Inequalities Involving Functions and Their Integrals and Derivatives}},
	author       = {Dragoslav S. Mitrinovic and Josip E. Pe{\v{c}}ari{\'c} and A. M. Fink},
	year         = 1991,
	publisher    = {Springer Science \& Business Media},
	volume       = 53
}

@article{Lorenz:08,
	title        = {{{Optimal Algorithms for }$k${-Search with Application in~Option Pricing}}},
	author       = {Julian Lorenz and Konstantinos Panagiotou and Angelika Steger},
	year         = 2008,
	month        = {Aug},
	journal      = {Algorithmica},
	publisher    = {Springer Science and Business Media {LLC}},
	volume       = 55,
	number       = 2,
	pages        = {311--328},
	doi          = {10.1007/s00453-008-9217-8},
	url          = {https://doi.org/10.1007/s00453-008-9217-8}
}

@incollection{Zhou:08,
	title        = {{Budget Constrained Bidding in Keyword Auctions and Online Knapsack Problems}},
	author       = {Yunhong Zhou and Deeparnab Chakrabarty and Rajan Lukose},
	year         = 2008,
	booktitle    = {Lecture Notes in Computer Science},
	publisher    = {Springer Berlin Heidelberg},
	pages        = {566--576}
}

@article{SunZeynali:20,
	title        = {{Competitive Algorithms for the Online Multiple Knapsack Problem with Application to Electric Vehicle Charging}},
	author       = {Sun, Bo and Zeynali, Ali and Li, Tongxin and Hajiesmaili, Mohammad and Wierman, Adam and Tsang, Danny H.K.},
	year         = 2021,
	month        = {June},
	journal      = {Proceedings of the ACM on Measurement and Analysis of Computing Systems},
	publisher    = {Association for Computing Machinery},
	address      = {New York, NY, USA},
	volume       = 4,
	number       = 3,
	issue_date   = {December 2020},
	articleno    = 51,
	numpages     = 32
}

@article{ElYaniv:01,
	title        = {{Optimal Search and One-Way Trading Online Algorithms}},
	author       = {Ran El-Yaniv and Amos Fiat and Richard M. Karp and Gordon Turpin},
	year         = 2001,
	month        = {May},
	journal      = {Algorithmica},
	publisher    = {Springer Science and Business Media {LLC}},
	volume       = 30,
	number       = 1,
	pages        = {101--139},
	doi          = {10.1007/s00453-001-0003-0},
	url          = {https://doi.org/10.1007/s00453-001-0003-0}
}

@article{radovanovic2022carbon,
	title        = {{Carbon-Aware Computing for Datacenters}},
	author       = {Radovanovic, Ana and Koningstein, Ross and Schneider, Ian and Chen, Bokan and Duarte, Alexandre and Roy, Binz and Xiao, Diyue and Haridasan, Maya and Hung, Patrick and Care, Nick and others},
	year         = 2022,
	journal      = {IEEE Transactions on Power Systems},
	publisher    = {IEEE}
}

@inproceedings{SunLee:21,
	title        = {{Pareto-Optimal Learning-Augmented Algorithms for Online Conversion Problems}},
	author       = {Sun, Bo and Lee, Russell and Hajiesmaili, Mohammad and Wierman, Adam and Tsang, Danny},
	year         = 2021,
	booktitle    = {Advances in Neural Information Processing Systems},
	volume       = 34,
	editor       = {M. Ranzato and A. Beygelzimer and Y. Dauphin and P.S. Liang and J. Wortman Vaughan}
}

@article{FriedmanLinial:93,
	title        = {{On convex body chasing}},
	author       = {Joel Friedman and Nathan Linial},
	year         = 1993,
	month        = {Mar},
	journal      = {Discrete \& Computational Geometry},
	publisher    = {Springer Science and Business Media {LLC}},
	volume       = 9,
	number       = 3,
	pages        = {293--321},
	doi          = {10.1007/bf02189324},
	url          = {https://doi.org/10.1007/bf02189324}
}

@inproceedings{Antoniadis:20MTS,
	title        = {{Online Metric Algorithms with Untrusted Predictions}},
	author       = {Antoniadis, Antonios and Coester, Christian and Elias, Marek and Polak, Adam and Simon, Bertrand},
	year         = 2020,
	month        = {Nov},
	booktitle    = {Proceedings of the 37th {{International Conference}} on {{Machine Learning}}},
	publisher    = {{PMLR}},
	pages        = {345--355},
	issn         = {2640-3498}
}

@article{Yang:20,
	title        = {{Online Linear Optimization with Inventory Management Constraints}},
	author       = {Yang, Lin and Hajiesmaili, Mohammad H. and Sitaraman, Ramesh and Wierman, Adam and Mallada, Enrique and Wong, Wing S.},
	year         = 2020,
	month        = {may},
	journal      = {Proceedings of the ACM on Measurement and Analysis of Computing Systems},
	publisher    = {Association for Computing Machinery},
	address      = {New York, NY, USA},
	volume       = 4,
	number       = 1,
	doi          = {10.1145/3379482},
	url          = {https://doi.org/10.1145/3379482},
	issue_date   = {March 2020},
	articleno    = 16,
	numpages     = 29
}

@article{CVXPY,
	title        = {{CVXPY: A Python-embedded modeling language for convex optimization}},
	author       = {Steven Diamond and Stephen Boyd},
	year         = 2016,
	journal      = {Journal of Machine Learning Research},
	volume       = 17,
	number       = 83,
	pages        = {1--5}
}

@article{Corless:96LambertW,
	title        = {{On the Lambert W function}},
	author       = {Corless, Robert M and Gonnet, Gaston H and Hare, David EG and Jeffrey, David J and Knuth, Donald E},
	year         = 1996,
	journal      = {{Advances in Computational Mathematics}},
	publisher    = {Springer},
	volume       = 5,
	pages        = {329--359}
}

@article{Kou:21:HVAC,
	title        = {{Model-based and data-driven HVAC control strategies for residential demand response}},
	author       = {Kou, Xiao and Du, Yan and Li, Fangxing and Pulgar-Painemal, Hector and Zandi, Helia and Dong, Jin and Olama, Mohammed M},
	year         = 2021,
	journal      = {IEEE Open Access Journal of Power and Energy},
	publisher    = {IEEE},
	volume       = 8,
	pages        = {186--197}
}

@article{Qu:24,
	title        = {{A Two-Stage Forecasting Approach for Day-Ahead Electricity Price Based on Improved Wavelet Neural Network With ELM Initialization}},
	author       = {Qu,  Ziyu and He,  Li and Ge,  Xinxin and Wang,  Fei and Xu,  Fei and Lu,  Jinling},
	year         = 2024,
	month        = may,
	journal      = {IEEE Transactions on Industry Applications},
	publisher    = {Institute of Electrical and Electronics Engineers (IEEE)},
	volume       = 60,
	number       = 3,
	pages        = {5061–5073},
	doi          = {10.1109/tia.2024.3365456},
	issn         = {1939-9367},
	url          = {http://dx.doi.org/10.1109/TIA.2024.3365456}
}

@article{Mitzenmacher:22:ALPS,
	title        = {{Algorithms with predictions}},
	author       = {Mitzenmacher, Michael and Vassilvitskii, Sergei},
	year         = 2022,
	month        = jun,
	journal      = {Commun. ACM},
	publisher    = {Association for Computing Machinery},
	address      = {New York, NY, USA},
	volume       = 65,
	number       = 7,
	pages        = {33–35},
	doi          = {10.1145/3528087},
	issn         = {0001-0782},
	url          = {https://doi.org/10.1145/3528087},
	issue_date   = {July 2022},
	numpages     = 3
}

@inproceedings{CVXPYlayers:19,
	title        = {{Differentiable Convex Optimization Layers}},
    author       = {Akshay Agrawal and Brandon Amos and Shane Barratt and Stephen Boyd and Steven Diamond and Zico Kolter},
	year         = 2019,
	booktitle    = {Advances in Neural Information Processing Systems}
}

@misc{Pinkus:2005:Density,
	title        = {{Density in Approximation Theory}},
	author       = {Allan Pinkus},
	year         = 2005,
	url          = {https://arxiv.org/abs/math/0501328},
	eprint       = {math/0501328},
	archiveprefix = {arXiv},
	primaryclass = {math.CA}
}

@book{Cheney:1998:Density,
	title        = {{Introduction to Approximation Theory}},
	author       = {Cheney, E.W.},
	year         = 1998,
	publisher    = {AMS Chelsea Pub.},
	series       = {AMS Chelsea Publishing Series},
	isbn         = 9780821813744,
	lccn         = 81067708
}

@article{Levin:22:Texas21,
	title        = {{Extreme weather and electricity markets: Key lessons from the February 2021 Texas crisis}},
	author       = {Levin, Todd and Botterud, Audun and Mann, W Neal and Kwon, Jonghwan and Zhou, Zhi},
	year         = 2022,
	journal      = {Joule},
	publisher    = {Elsevier},
	volume       = 6,
	number       = 1,
	pages        = {1--7}
}

@article{Yang:25:CarbonAwareHVAC,
	title        = {{Carbon-Aware Scheduling of Thermostatically Controlled Loads: A Bilevel DRCC Approach}},
	author       = {Yang, Yijie and Shi, Jian and Wang, Dan and Wu, Chenye and Han, Zhu},
	year         = 2025,
	journal      = {IEEE Transactions on Smart Grid},
	volume       = 16,
	number       = 2,
	pages        = {1233--1247}
}

@article{Mariano:21:BMS,
	title        = {{A review of strategies for building energy management system: Model predictive control,  demand side management,  optimization,  and fault detect \& diagnosis}},
	author       = {Mariano-Hernández,  D. and Hernández-Callejo,  L. and Zorita-Lamadrid,  A. and Duque-Pérez,  O. and Santos García,  F.},
	year         = 2021,
	month        = jan,
	journal      = {Journal of Building Engineering},
	publisher    = {Elsevier BV},
	volume       = 33,
	pages        = 101692,
	doi          = {10.1016/j.jobe.2020.101692},
	issn         = {2352-7102},
	url          = {http://dx.doi.org/10.1016/j.jobe.2020.101692}
}

@article{Stai:21:BatteryStorage,
	title        = {{Online Battery Storage Management via Lyapunov Optimization in Active Distribution Grids}},
	author       = {Stai,  Eleni and Wang,  Cong and Le Boudec,  Jean-Yves},
	year         = 2021,
	month        = mar,
	journal      = {IEEE Transactions on Control Systems Technology},
	publisher    = {Institute of Electrical and Electronics Engineers (IEEE)},
	volume       = 29,
	number       = 2,
	pages        = {672–690},
	doi          = {10.1109/tcst.2020.2975475},
	issn         = {2374-0159},
	url          = {http://dx.doi.org/10.1109/TCST.2020.2975475}
}

@misc{chehade:2024:MPC-bat-management,
	title        = {{Should we use model-free or model-based control? A case study of battery management systems}},
	author       = {Mohamad Fares El Hajj Chehade and Young-ho Cho and Sandeep Chinchali and Hao Zhu},
	year         = 2024,
	url          = {https://arxiv.org/abs/2407.15313},
	eprint       = {2407.15313},
	archiveprefix = {arXiv},
	primaryclass = {eess.SY}
}

@inproceedings{Blonsky:20:MPC-bat-management,
	title        = {{Time-of-use and Demand Charge Battery Controller Using Stochastic Model Predictive Control}},
	author       = {Blonsky, Michael and McKenna, Killian and Vincent, Tyrone and Nagarajan, Adarsh},
	year         = 2020,
	booktitle    = {2020 IEEE International Conference on Communications, Control, and Computing Technologies for Smart Grids},
	volume       = {},
	number       = {},
	pages        = {1--6},
	doi          = {10.1109/SmartGridComm47815.2020.9302943}
}

@article{MorenoCastro:23:microgrids,
	title        = {{Microgrid Management Strategies for Economic Dispatch of Electricity Using Model Predictive Control Techniques: A Review}},
	author       = {Moreno-Castro,  Juan and Ocaña Guevara,  Victor Samuel and León Viltre,  Lesyani Teresa and Gallego Landera,  Yandi and Cuaresma Zevallos,  Oscar and Aybar-Mejía,  Miguel},
	year         = 2023,
	month        = aug,
	journal      = {Energies},
	publisher    = {MDPI AG},
	volume       = 16,
	number       = 16,
	pages        = 5935,
	doi          = {10.3390/en16165935},
	issn         = {1996-1073},
	url          = {http://dx.doi.org/10.3390/en16165935}
}

@article{Lopez:21:MPC-hybridpp,
	title        = {{Day-Ahead MPC Energy Management System for an Island Wind/Storage Hybrid Power Plant}},
	author       = {López-Rodríguez,  Rubén and Aguilera-González,  Adriana and Vechiu,  Ionel and Bacha,  Seddik},
	year         = 2021,
	month        = feb,
	journal      = {Energies},
	publisher    = {MDPI AG},
	volume       = 14,
	number       = 4,
	pages        = 1066,
	doi          = {10.3390/en14041066},
	issn         = {1996-1073},
	url          = {http://dx.doi.org/10.3390/en14041066}
}

@article{Faraji:20:day-ahead-prosumer,
	title        = {{Day-ahead optimization of prosumer considering battery depreciation and weather prediction for renewable energy sources}},
	author       = {Faraji, Jamal and Abazari, Ahmadreza and Babaei, Masoud and Muyeen, SM and Benbouzid, Mohamed},
	year         = 2020,
	journal      = {Applied Sciences},
	publisher    = {MDPI},
	volume       = 10,
	number       = 8,
	pages        = 2774
}

@inproceedings{Zeng:23:day-ahead-AC,
	title        = {{Optimal Day-Ahead Dispatch of Air-Conditioning Load under Dynamic Carbon Emission Factors}},
	author       = {Zeng, Xianyi and Zhang, Yuanliang and Guo, Bin and Huang, Liyu and Li, Chuangzhi},
	year         = 2023,
	booktitle    = {2023 5th Asia Energy and Electrical Engineering Symposium (AEEES)},
	volume       = {},
	number       = {},
	pages        = {1--6},
	doi          = {10.1109/AEEES56888.2023.10114285}
}

@article{MohsenianRad:16:day-ahead-battery,
	title        = {{Optimal Bidding, Scheduling, and Deployment of Battery Systems in California Day-Ahead Energy Market}},
	author       = {Mohsenian-Rad, Hamed},
	year         = 2016,
	journal      = {IEEE Transactions on Power Systems},
	volume       = 31,
	number       = 1,
	pages        = {442--453},
	doi          = {10.1109/TPWRS.2015.2394355}
}

@article{Silva:20:day-ahead-microgrid,
	title        = {{Optimal Day-Ahead Scheduling of Microgrids with Battery Energy Storage System}},
	author       = {Silva,  Vanderlei Aparecido and Aoki,  Alexandre Rasi and Lambert-Torres,  Germano},
	year         = 2020,
	month        = oct,
	journal      = {Energies},
	publisher    = {MDPI AG},
	volume       = 13,
	number       = 19,
	pages        = 5188,
	doi          = {10.3390/en13195188},
	issn         = {1996-1073},
	url          = {http://dx.doi.org/10.3390/en13195188}
}

@article{Zhang:24:day-ahead-battery,
	title        = {{Optimal day-ahead large-scale battery dispatch model for multi-regulation participation considering full timescale uncertainties}},
	author       = {Zhang,  Mingze and Li,  Weidong and Yu,  Samson Shenglong and Wang,  Haixia and Ba,  Yu},
	year         = 2024,
	month        = jan,
	journal      = {Renewable and Sustainable Energy Reviews},
	publisher    = {Elsevier BV},
	volume       = 189,
	pages        = 113963,
	doi          = {10.1016/j.rser.2023.113963},
	issn         = {1364-0321},
	url          = {http://dx.doi.org/10.1016/j.rser.2023.113963}
}

@inproceedings{RigoMariani:22:ADMM,
	title        = {{An ADMM-based Coordination Strategy for the Control of Distributed Storage at the Household Level -Impact of the End-User Settings}},
	author       = {Rigo-Mariani, R{\'e}my and Debucsshere, Vincent},
	year         = 2022,
	month        = {May},
	booktitle    = {{ELECTRIMACS 2022}},
	address      = {Nancy, France},
	series       = {ELECTRIMACS 2022},
	url          = {https://hal.science/hal-03675572},
	hal_id       = {hal-03675572},
	hal_version  = {v1}
}

@article{Mamun:16:multi-objective-DR-battery,
	title        = {{Multi-objective optimization of demand response in a datacenter with lithium-ion battery storage}},
	author       = {Mamun,  A. and Narayanan,  I. and Wang,  D. and Sivasubramaniam,  A. and Fathy,  H.K.},
	year         = 2016,
	month        = aug,
	journal      = {Journal of Energy Storage},
	publisher    = {Elsevier BV},
	volume       = 7,
	pages        = {258–269},
	doi          = {10.1016/j.est.2016.08.002},
	issn         = {2352-152X},
	url          = {http://dx.doi.org/10.1016/j.est.2016.08.002}
}

@article{Azuatalam:20,
	title        = {{Reinforcement learning for whole-building HVAC control and demand response}},
	author       = {Azuatalam, Donald and Lee, Wee-Lih and De Nijs, Frits and Liebman, Ariel},
	year         = 2020,
	journal      = {Energy and AI},
	publisher    = {Elsevier},
	volume       = 2,
	pages        = 100020
}

@article{Sun:23:Lyapunov-battery-DC,
	title        = {{Battery-Assisted Online Operation of Distributed Data Centers With Uncertain Workload and Electricity Prices}},
	author       = {Sun,  Jun and Chen,  Shibo and You,  Pengcheng and Yang,  Qinmin and Yang,  Zaiyue},
	year         = 2023,
	month        = apr,
	journal      = {IEEE Transactions on Cloud Computing},
	publisher    = {Institute of Electrical and Electronics Engineers (IEEE)},
	volume       = 11,
	number       = 2,
	pages        = {1303–1314},
	doi          = {10.1109/tcc.2021.3132174},
	issn         = {2372-0018},
	url          = {http://dx.doi.org/10.1109/TCC.2021.3132174}
}

@inproceedings{Li:17:battery-DC,
	title        = {{Balancing the Use of Batteries and Opportunistic Scheduling Policies for Maximizing Renewable Energy Consumption in a Cloud Data Center}},
	author       = {Li,  Yunbo and Orgerie,  Anne-Cecile and Menaud,  Jean-Marc},
	year         = 2017,
	booktitle    = {2017 25th Euromicro International Conference on Parallel,  Distributed and Network-based Processing (PDP)},
	publisher    = {IEEE},
	pages        = {408–415},
	doi          = {10.1109/pdp.2017.24},
	url          = {http://dx.doi.org/10.1109/PDP.2017.24}
}

@inproceedings{Nasiriani:17:stochastic-peak-shaving,
	title        = {{Optimal Peak Shaving Using Batteries at Datacenters: Characterizing the Risks and Benefits}},
	author       = {Nasiriani,  Neda and Kesidis,  George and Wang,  Di},
	year         = 2017,
	month        = sep,
	booktitle    = {2017 IEEE 25th International Symposium on Modeling,  Analysis,  and Simulation of Computer and Telecommunication Systems (MASCOTS)},
	publisher    = {IEEE},
	pages        = {164–174},
	doi          = {10.1109/mascots.2017.27},
	url          = {http://dx.doi.org/10.1109/MASCOTS.2017.27}
}

@inproceedings{Sage:24:RL-battery-storage,
	title        = {{Enhancing Battery Storage Energy Arbitrage With Deep Reinforcement Learning and Time-Series Forecasting}},
	author       = {Sage, Manuel and Campbell, Joshua and Zhao, Yaoyao Fiona},
	year         = 2024,
	month        = jul,
	booktitle    = {ASME 2024 18th International Conference on Energy Sustainability},
	publisher    = {American Society of Mechanical Engineers},
	series       = {ES2024},
	doi          = {10.1115/es2024-130538},
	url          = {http://dx.doi.org/10.1115/ES2024-130538},
	collection   = {ES2024}
}

@inproceedings{Urgaonkar:11:lyapunov-power-cost-management,
	title        = {{Optimal power cost management using stored energy in data centers}},
	author       = {Urgaonkar, Rahul and Urgaonkar, Bhuvan and Neely, Michael J. and Sivasubramaniam, Anand},
	year         = 2011,
	booktitle    = {Proceedings of the ACM SIGMETRICS Joint International Conference on Measurement and Modeling of Computer Systems},
	location     = {San Jose, California, USA},
	publisher    = {Association for Computing Machinery},
	address      = {New York, NY, USA},
	series       = {SIGMETRICS '11},
	pages        = {221–232},
	doi          = {10.1145/1993744.1993766},
	isbn         = 9781450308144,
	url          = {https://doi.org/10.1145/1993744.1993766},
	numpages     = 12
}

@inproceedings{Tandukar:18:day-ahead-VPP,
	title        = {{Real-time Operation of a Data Center as Virtual Power Plant Considering Battery Lifetime}},
	author       = {Tandukar,  Prajina and Bajracharya,  Labi and Hansen,  Timothy M. and Fourney,  Robert and Tamrakar,  Ujjwol and Tonkoski,  Reinaldo},
	year         = 2018,
	month        = jun,
	booktitle    = {2018 International Symposium on Power Electronics,  Electrical Drives,  Automation and Motion (SPEEDAM)},
	publisher    = {IEEE},
	pages        = {81–86},
	doi          = {10.1109/speedam.2018.8445345},
	url          = {http://dx.doi.org/10.1109/SPEEDAM.2018.8445345}
}

@article{Guo:13:Lyapunov-battery-DC,
	title        = {{Electricity Cost Saving Strategy in Data Centers by Using Energy Storage}},
	author       = {Guo,  Yuanxiong and Fang,  Yuguang},
	year         = 2013,
	month        = jun,
	journal      = {IEEE Transactions on Parallel and Distributed Systems},
	publisher    = {Institute of Electrical and Electronics Engineers (IEEE)},
	volume       = 24,
	number       = 6,
	pages        = {1149–1160},
	doi          = {10.1109/tpds.2012.201},
	issn         = {1045-9219},
	url          = {http://dx.doi.org/10.1109/TPDS.2012.201}
}

@article{Yu:15:Lyapunov-battery-DC,
	title        = {{Joint Workload and Battery Scheduling with Heterogeneous Service Delay Guarantees for Data Center Energy Cost Minimization}},
	author       = {Yu,  Liang and Jiang,  Tao and Cao,  Yang and Qi,  Qi},
	year         = 2015,
	month        = jul,
	journal      = {IEEE Transactions on Parallel and Distributed Systems},
	publisher    = {Institute of Electrical and Electronics Engineers (IEEE)},
	volume       = 26,
	number       = 7,
	pages        = {1937–1947},
	doi          = {10.1109/tpds.2014.2329491},
	issn         = {1045-9219},
	url          = {http://dx.doi.org/10.1109/TPDS.2014.2329491}
}

@article{Morstyn:17:MPC-battery,
	title        = {{Model predictive control for distributed microgrid battery energy storage systems}},
	author       = {Morstyn, Thomas and Hredzak, Branislav and Aguilera, Ricardo P and Agelidis, Vassilios G},
	year         = 2017,
	journal      = {IEEE Transactions on Control Systems Technology},
	publisher    = {IEEE},
	volume       = 26,
	number       = 3,
	pages        = {1107--1114}
}

@article{Afram:15:MPC-HVAC,
	title        = {{Theory and applications of HVAC control systems--A review of model predictive control (MPC)}},
	author       = {Afram, Abdul and Janabi-Sharifi, Farrokh},
	year         = 2014,
	journal      = {{Building and Environment}},
	publisher    = {Elsevier},
	volume       = 72,
	pages        = {343--355}
}

@article{Serale:18:MPC-HVAC,
	title        = {{Model predictive control (MPC) for enhancing building and HVAC system energy efficiency: Problem formulation, applications and opportunities}},
	author       = {Serale, Gianluca and Fiorentini, Massimo and Capozzoli, Alfonso and Bernardini, Daniele and Bemporad, Alberto},
	year         = 2018,
	journal      = {Energies},
	publisher    = {MDPI},
	volume       = 11,
	number       = 3,
	pages        = 631
}

@article{Blum:22:MPC-HVAC,
	title        = {{Field demonstration and implementation analysis of model predictive control in an office HVAC system}},
	author       = {Blum, David and Wang, Zhe and Weyandt, Chris and Kim, Donghun and Wetter, Michael and Hong, Tianzhen and Piette, Mary Ann},
	year         = 2022,
	journal      = {Applied Energy},
	publisher    = {Elsevier},
	volume       = 318,
	number       = 119104,
	pages        = {10--1016}
}

@article{Lazic:18:MPC-BMS,
	title        = {{Data center cooling using model-predictive control}},
	author       = {Lazic, Nevena and Boutilier, Craig and Lu, Tyler and Wong, Eehern and Roy, Binz and Ryu, MK and Imwalle, Greg},
	year         = 2018,
	journal      = {Advances in Neural Information Processing Systems},
	volume       = 31
}

@article{Wang:23:hierarchical,
	title        = {{A hierarchical dispatch strategy of hybrid energy storage system in internet data center with model predictive control}},
	author       = {Wang, Kaifeng and Ye, Lin and Yang, Shihui and Deng, Zhanfeng and Song, Jieying and Li, Zhuo and Zhao, Yongning},
	year         = 2023,
	journal      = {Applied Energy},
	publisher    = {Elsevier},
	volume       = 331,
	pages        = 120414
}

@article{Mirhoseininejad:21:DC,
	title        = {{A data-driven, multi-setpoint model predictive thermal control system for data centers}},
	author       = {Mirhoseininejad, SeyedMorteza and Badawy, Ghada and Down, Douglas G},
	year         = 2021,
	journal      = {Journal of Network and Systems Management},
	publisher    = {Springer},
	volume       = 29,
	number       = 1,
	pages        = 7
}

@inproceedings{Kim:16,
	title        = {{Online Energy Storage Management: an Algorithmic Approach}},
	author       = {Kim,  Anthony and Liaghat,  Vahid and Qin,  Junjie and Saberi,  Amin},
	year         = 2016,
	publisher    = {Schloss Dagstuhl – Leibniz-Zentrum f\"{u}r Informatik},
	doi          = {10.4230/LIPICS.APPROX-RANDOM.2016.12},
	url          = {https://drops.dagstuhl.de/entities/document/10.4230/LIPIcs.APPROX-RANDOM.2016.12},
	copyright    = {Creative Commons Attribution 3.0 Unported license},
	language     = {en}
}

@inproceedings{Mo:21,
	title        = {{Optimal Online Algorithms for Peak-Demand Reduction Maximization with Energy Storage}},
	author       = {Mo,  Yanfang and Lin,  Qiulin and Chen,  Minghua and Qin,  Si-Zhao Joe},
	year         = 2021,
	month        = jun,
	booktitle    = {Proceedings of the Twelfth ACM International Conference on Future Energy Systems},
	publisher    = {ACM},
	series       = {e-Energy ’21},
	pages        = {73–83},
	doi          = {10.1145/3447555.3464857},
	url          = {http://dx.doi.org/10.1145/3447555.3464857},
	collection   = {e-Energy ’21}
}

@article{Lu:13,
	title        = {{Online energy generation scheduling for microgrids with intermittent energy sources and co-generation}},
	author       = {Lu, Lian and Tu, Jinlong and Chau, Chi-Kin and Chen, Minghua and Lin, Xiaojun},
	year         = 2013,
	month        = jun,
	journal      = {SIGMETRICS Perform. Eval. Rev.},
	publisher    = {Association for Computing Machinery},
	address      = {New York, NY, USA},
	volume       = 41,
	number       = 1,
	pages        = {53–66},
	doi          = {10.1145/2494232.2465551},
	issn         = {0163-5999},
	url          = {https://doi.org/10.1145/2494232.2465551},
	issue_date   = {June 2013},
	numpages     = 14
}

@inproceedings{Hajiesmaili:16,
	title        = {{Online microgrid energy generation scheduling revisited: the benefits of randomization and interval prediction}},
	author       = {Hajiesmaili, Mohammad H. and Chau, Chi-Kin and Chen, Minghua and Huang, Longbu},
	year         = 2016,
	booktitle    = {Proceedings of the Seventh International Conference on Future Energy Systems},
	location     = {Waterloo, Ontario, Canada},
	publisher    = {Association for Computing Machinery},
	address      = {New York, NY, USA},
	series       = {e-Energy '16},
	doi          = {10.1145/2934328.2934329},
	isbn         = 9781450343930,
	url          = {https://doi.org/10.1145/2934328.2934329},
	articleno    = 1,
	numpages     = 11
}

@article{Zhang:2018:microgrid,
	title        = {{Peak-Aware Online Economic Dispatching for Microgrids}},
	author       = {Zhang, Ying and Hajiesmaili, Mohammad H. and Cai, Sinan and Chen, Minghua and Zhu, Qi},
	year         = 2018,
	journal      = {IEEE Transactions on Smart Grid},
	volume       = 9,
	number       = 1,
	pages        = {323--335},
	doi          = {10.1109/TSG.2016.2551282}
}

@article{Bedi:18,
	title        = {{Online algorithms for storage utilization under real-time pricing in smart grid}},
	author       = {Bedi,  Amrit S. and Aditya P.,  P.V. and Waseem Ahmad,  Md. and Swapnil,  S. and Rajawat,  Ketan and Anand,  Sandeep},
	year         = 2018,
	month        = oct,
	journal      = {International Journal of Electrical Power \& Energy Systems},
	publisher    = {Elsevier BV},
	volume       = 101,
	pages        = {50–59},
	doi          = {10.1016/j.ijepes.2018.02.034},
	issn         = {0142-0615},
	url          = {http://dx.doi.org/10.1016/j.ijepes.2018.02.034}
}

@article{Larsen:10,
	title        = {{Competitive analysis of the online inventory problem}},
	author       = {Larsen,  Kim S. and Wøhlk,  Sanne},
	year         = 2010,
	month        = dec,
	journal      = {European Journal of Operational Research},
	publisher    = {Elsevier BV},
	volume       = 207,
	number       = 2,
	pages        = {685–696},
	doi          = {10.1016/j.ejor.2010.05.019},
	issn         = {0377-2217},
	url          = {http://dx.doi.org/10.1016/j.ejor.2010.05.019}
}

@article{Fung:21,
	title        = {{Online two-way trading: Randomization and advice}},
	author       = {Fung,  Stanley P.Y.},
	year         = 2021,
	month        = feb,
	journal      = {Theoretical Computer Science},
	publisher    = {Elsevier BV},
	volume       = 856,
	pages        = {41–50},
	doi          = {10.1016/j.tcs.2020.12.016},
	issn         = {0304-3975},
	url          = {http://dx.doi.org/10.1016/j.tcs.2020.12.016}
}

@article{Chen:22:InBev,
	title        = {{Approximation schemes for the joint inventory selection and online resource allocation problem}},
	author       = {Chen,  Xingxing and Feldman,  Jacob and Jung,  Seung Hwan and Kouvelis,  Panos},
	year         = 2022,
	month        = aug,
	journal      = {Production and Operations Management},
	publisher    = {SAGE Publications},
	volume       = 31,
	number       = 8,
	pages        = {3143–3159},
	doi          = {10.1111/poms.13742},
	issn         = {1937-5956},
	url          = {http://dx.doi.org/10.1111/poms.13742}
}

@article{Devanur:19,
	title        = {{Near Optimal Online Algorithms and Fast Approximation Algorithms for Resource Allocation Problems}},
	author       = {Devanur, Nikhil R. and Jain, Kamal and Sivan, Balasubramanian and Wilkens, Christopher A.},
	year         = 2019,
	month        = jan,
	journal      = {J. ACM},
	publisher    = {Association for Computing Machinery},
	address      = {New York, NY, USA},
	volume       = 66,
	number       = 1,
	doi          = {10.1145/3284177},
	issn         = {0004-5411},
	url          = {https://doi.org/10.1145/3284177},
	issue_date   = {February 2019},
	articleno    = 7,
	numpages     = 41
}

@inproceedings{Wang:17:ORA,
	title        = {{Online Resource Allocation for Arbitrary User Mobility in Distributed Edge Clouds}},
	author       = {Wang,  Lin and Jiao,  Lei and Li,  Jun and Muhlhauser,  Max},
	year         = 2017,
	month        = jun,
	booktitle    = {2017 IEEE 37th International Conference on Distributed Computing Systems (ICDCS)},
	publisher    = {IEEE},
	doi          = {10.1109/icdcs.2017.30},
	url          = {http://dx.doi.org/10.1109/ICDCS.2017.30}
}

@misc{jaillet2011online,
	title        = {{Online resource allocation problems}},
	author       = {Jaillet, Patrick and Lu, Xin},
	year         = 2011,
	url          = {https://web.mit.edu/jaillet/www/general/online-resource-pj-xl-1-2011.pdf}
}

@article{Jiao:17:ORA,
	title        = {{Smoothed Online Resource Allocation in Multi-Tier Distributed Cloud Networks}},
	author       = {Jiao, Lei and Tulino, Antonia Maria and Llorca, Jaime and Jin, Yue and Sala, Alessandra},
	year         = 2017,
	journal      = {IEEE/ACM Transactions on Networking},
	volume       = 25,
	number       = 4,
	pages        = {2556--2570},
	doi          = {10.1109/TNET.2017.2707142}
}

@inbook{Klimm:19,
	title        = {{The Online Best Reply Algorithm for Resource Allocation Problems}},
	author       = {Klimm,  Max and Schmand,  Daniel and T\"{o}nnis,  Andreas},
	year         = 2019,
	booktitle    = {Algorithmic Game Theory},
	publisher    = {Springer International Publishing},
	pages        = {200–215},
	doi          = {10.1007/978-3-030-30473-7_14},
	isbn         = 9783030304737,
	issn         = {1611-3349},
	url          = {http://dx.doi.org/10.1007/978-3-030-30473-7_14}
}

@article{Zhang:17:ORA,
	title        = {{Optimal Posted Prices for Online Cloud Resource Allocation}},
	author       = {Zhang, Zijun and Li, Zongpeng and Wu, Chuan},
	year         = 2017,
	month        = jun,
	journal      = {Proceedings of the ACM on Measurement and Analysis of Computing Systems},
	publisher    = {Association for Computing Machinery},
	address      = {New York, NY, USA},
	volume       = 1,
	number       = 1,
	issue_date   = {June 2017},
	articleno    = 23,
	numpages     = 26
}

@misc{Balseiro:23,
	title        = {{Online Resource Allocation under Horizon Uncertainty}},
	author       = {Santiago Balseiro and Christian Kroer and Rachitesh Kumar},
	year         = 2023,
	url          = {https://arxiv.org/abs/2206.13606},
	eprint       = {2206.13606},
	archiveprefix = {arXiv},
	primaryclass = {cs.DS}
}

@article{He:22,
	title        = {{An online algorithm for the inventory retrieval problem with an uncertain selling duration,  uncertain prices,  and price-dependent demands}},
	author       = {He,  Xiaozhou and Xiang,  Jie and Xiao,  Jin and Cheng,  T.C. E. and Tian,  Yuhang},
	year         = 2022,
	month        = dec,
	journal      = {Computers \& Operations Research},
	publisher    = {Elsevier BV},
	volume       = 148,
	pages        = 105991,
	doi          = {10.1016/j.cor.2022.105991},
	issn         = {0305-0548},
	url          = {http://dx.doi.org/10.1016/j.cor.2022.105991}
}

@inproceedings{HIHAT:23,
	title        = {{Online Inventory Problems: Beyond the i.i.d. Setting with Online Convex Optimization}},
	author       = {Hihat, Massil and Ga\"{\i}ffas, St\'{e}phane and Garrigos, Guillaume and Bussy, Simon},
	year         = 2023,
	booktitle    = {Advances in Neural Information Processing Systems},
	volume       = 36,
	pages        = {20421--20440},
	url          = {https://proceedings.neurips.cc/paper_files/paper/2023/file/41128e5b3a7622da5b17588757599077-Paper-Conference.pdf},
}

@article{Yang:24,
	title        = {{Online Allocation with Replenishable Budgets: Worst Case and Beyond}},
	author       = {Yang, Jianyi and Li, Pengfei and Islam, Mohammad Jaminur and Ren, Shaolei},
	year         = 2024,
	month        = feb,
	journal      = {Proceedings of the ACM on Measurement and Analysis of Computing Systems},
	publisher    = {Association for Computing Machinery},
	address      = {New York, NY, USA},
	volume       = 8,
	number       = 1,
	doi          = {10.1145/3639030},
	url          = {https://doi.org/10.1145/3639030},
	issue_date   = {March 2024},
	articleno    = 4,
	numpages     = 34
}

@misc{Cheung:21,
	title        = {{Inventory Balancing with Online Learning}},
	author       = {Wang Chi Cheung and Will Ma and David Simchi-Levi and Xinshang Wang},
	year         = 2021,
	url          = {https://arxiv.org/abs/1810.05640},
	eprint       = {1810.05640},
	archiveprefix = {arXiv},
	primaryclass = {cs.AI}
}

@article{Mehta:07,
	title        = {{Adwords and generalized online matching}},
	author       = {Mehta, Aranyak and Saberi, Amin and Vazirani, Umesh and Vazirani, Vijay},
	year         = 2007,
	journal      = {Journal of the ACM (JACM)},
	publisher    = {ACM New York, NY, USA},
	volume       = 54,
	number       = 5,
	pages        = {22--es}
}

@article{Huang:24,
	title        = {{Online Matching: A Brief Survey}},
	author       = {Huang, Zhiyi and Tang, Zhihao Gavin and Wajc, David},
	year         = 2024,
	month        = oct,
	journal      = {SIGecom Exch.},
	publisher    = {Association for Computing Machinery},
	address      = {New York, NY, USA},
	volume       = 22,
	number       = 1,
	pages        = {135–158},
	doi          = {10.1145/3699824.3699837},
	url          = {https://doi.org/10.1145/3699824.3699837},
	issue_date   = {June 2024},
	numpages     = 24
}

@misc{Alibaba:18,
	title        = {{Cluster data collected from production clusters in Alibaba for cluster management research}},
	author       = {Alibaba},
	year         = 2018,
	url          = {https://github.com/alibaba/clusterdata/tree/master/cluster-trace-v2018}
}

@article{Bloomberg:24,
	title        = {{AI Needs So Much Power, It's Making Yours Worse}},
	author       = {Leonardo Nicoletti and Naureen Malik and Andre Tartar},
	year         = 2024,
	month        = {dec},
	day          = 26,
	journal      = {Bloomberg},
	url          = {https://www.bloomberg.com/graphics/2024-ai-power-home-appliances/},
	urldate      = {2025-06-11}
}

@misc{NREL:ThermalStorage,
	title        = {{Thermal Energy Storage}},
	author       = {{National Renewable Energy Laboratory}},
	year         = 2024,
	url          = {https://www.nrel.gov/buildings/storage},
	urldate      = {2025-06-11},
	organization = {NREL}
}

@inbook{Breeze:19,
	title        = {{Power System Energy Storage Technologies}},
	author       = {Breeze,  Paul},
	year         = 2019,
	booktitle    = {Power Generation Technologies},
	publisher    = {Elsevier},
	pages        = {219–249},
	isbn         = 9780081026311,
	url          = {https://doi.org/10.1016/B978-0-08-102631-1.00010-9}
}

@article{Ruhnau:19,
	title        = {{Time series of heat demand and heat pump efficiency for energy system modeling}},
	author       = {Ruhnau,  Oliver and Hirth,  Lion and Praktiknjo,  Aaron},
	year         = 2019,
	month        = oct,
	journal      = {Scientific Data},
	publisher    = {Springer Science and Business Media LLC},
	volume       = 6,
	number       = 1,
	doi          = {10.1038/s41597-019-0199-y},
	issn         = {2052-4463},
	url          = {http://dx.doi.org/10.1038/s41597-019-0199-y}
}

@article{Li:22:ExpertCalibrated,
	title        = {{Expert-Calibrated Learning for Online Optimization with Switching Costs}},
	author       = {Li, Pengfei and Yang, Jianyi and Ren, Shaolei},
	year         = 2022,
	month        = jun,
	journal      = {Proceedings of the ACM on Measurement and Analysis of Computing Systems},
	publisher    = {Association for Computing Machinery},
	address      = {New York, NY, USA},
	volume       = 6,
	number       = 2,
	doi          = {10.1145/3530894},
	url          = {https://doi.org/10.1145/3530894},
	issue_date   = {June 2022},
	articleno    = 28,
	numpages     = 35
}

@article{FT:25:europe,
	title        = {{Europe's electricity system tested by heatwaves as air-conditioning demand soars}},
    author       = {Attracta Mooney and Alice Hancock and Amy Kazmin},
	year         = 2025,
	month        = 8,
	day          = 3,
	journal      = {Financial Times},
	url          = {https://www.ft.com/content/23b3dc59-b40f-48e2-ad93-e301de7ac5f2},
}

@article{Pons:2025,
	title        = {{Spain's grid operator warns of new voltage swings, urges measures to avoid blackout}},
	author       = {Pons, Corina and Khalip, Andrei},
	year         = 2025,
	month        = 10,
	day          = 8,
	journal      = {Reuters},
	url          = {https://www.reuters.com/business/energy/spains-grid-operator-warns-new-tension-swings-urges-measures-avoid-blackout-2025-10-08/},
	urldate      = {2025-10-14}
}

@misc{Kwon:25,
	title        = {{Operational Risks in Grid Integration of Large Data Center Loads: Characteristics, Stability Assessments, and Sensitivity Studies}},
	author       = {Kyung-Bin Kwon and Sayak Mukherjee and Veronica Adetola},
	year         = 2025,
	url          = {https://arxiv.org/abs/2510.05437},
	eprint       = {2510.05437},
	archiveprefix = {arXiv},
	primaryclass = {eess.SY}
}

@article{Zeynali:21:DataDriven,
	title        = {{Data-driven Competitive Algorithms for Online Knapsack and Set Cover}},
	author       = {Zeynali,  Ali and Sun,  Bo and Hajiesmaili,  Mohammad and Wierman,  Adam},
	year         = 2021,
	month        = may,
	journal      = {Proceedings of the AAAI Conference on Artificial Intelligence},
	publisher    = {Association for the Advancement of Artificial Intelligence (AAAI)},
	volume       = 35,
	number       = 12,
	pages        = {10833–10841},
	doi          = {10.1609/aaai.v35i12.17294},
	issn         = {2159-5399},
	url          = {http://dx.doi.org/10.1609/aaai.v35i12.17294}
}

@misc{Chen:21:L20,
	title        = {{Learning to Optimize: A Primer and A Benchmark}},
	author       = {Tianlong Chen and Xiaohan Chen and Wuyang Chen and Howard Heaton and Jialin Liu and Zhangyang Wang and Wotao Yin},
	year         = 2021,
	url          = {https://arxiv.org/abs/2103.12828},
	eprint       = {2103.12828},
	archiveprefix = {arXiv},
	primaryclass = {math.OC}
}

@misc{Paszke:19:PyTorch,
	title        = {{PyTorch: An Imperative Style, High-Performance Deep Learning Library}},
	author       = {Adam Paszke and Sam Gross and Francisco Massa and Adam Lerer and James Bradbury and Gregory Chanan and Trevor Killeen and Zeming Lin and Natalia Gimelshein and Luca Antiga and Alban Desmaison and Andreas Köpf and Edward Yang and Zach DeVito and Martin Raison and Alykhan Tejani and Sasank Chilamkurthy and Benoit Steiner and Lu Fang and Junjie Bai and Soumith Chintala},
	year         = 2019,
	url          = {https://arxiv.org/abs/1912.01703},
	eprint       = {1912.01703},
	archiveprefix = {arXiv},
	primaryclass = {cs.LG}
}

@article{Emerald:25,
  title = {AI data centres as grid-interactive assets},
  ISSN = {2058-7546},
  url = {http://dx.doi.org/10.1038/s41560-025-01927-1},
  DOI = {10.1038/s41560-025-01927-1},
  journal = {Nature Energy},
  publisher = {Springer Science and Business Media LLC},
  author = {Colangelo,  Philip and Coskun,  Ayse K. and Megrue,  Jack and Roberts,  Ciaran and Sengupta,  Shayan and Sivaram,  Varun and Tiao,  Ethan and Vijaykar,  Aroon and Williams,  Chris and Wilson,  Daniel C. and Records,  Brandon and MacFarland,  Zack and Dreiling,  Daniel and Morey,  Nathan and Ratnayake,  Anuja and Vairamohan,  Baskar},
  year = {2025},
  month = dec 
}

@InProceedings{Lobos:21,
  title = 	 {Joint Online Learning and Decision-making via Dual Mirror Descent},
  author =       {Lobos, Alfonso and Grigas, Paul and Wen, Zheng},
  booktitle = 	 {Proceedings of the 38th International Conference on Machine Learning},
  pages = 	 {7080--7089},
  year = 	 {2021},
  editor = 	 {Meila, Marina and Zhang, Tong},
  volume = 	 {139},
  series = 	 {Proceedings of Machine Learning Research},
  month = 	 {18--24 Jul},
  publisher =    {PMLR},
  url = 	 {https://proceedings.mlr.press/v139/lobos21a.html}
}

\appendix
\clearpage
\section*{Appendix}

\section{Deferred Experiment Setup and Results} \label{apx:exp}

In this section, we provide additional experimental results and details that were deferred from \autoref{sec:exp}.  We start by discussing the runtime of our algorithms in \autoref{apx:runtime}, before presenting additional experimental results in \autoref{apx:exp-results} and presenting details of the \PALD implementations in \autoref{apx:pald-implementation-details}.  Note that \autoref{fig:weighted-active-jobs} plots the ``\textit{weighted active jobs}'' metric used to model demand using the Alibaba trace~\cite{Alibaba:18}, described in \autoref{sec:exp}.  

\begin{figure*}[h]
    \hfill
    \minipage{0.75\textwidth}
    \begin{center}
    \includegraphics[width=0.5\linewidth]{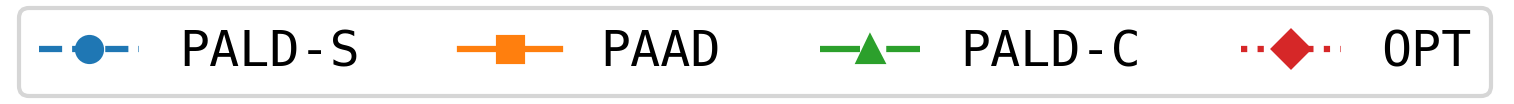}\vspace{0.1em}
    \end{center}
    \endminipage\hfill\\
    \minipage{0.24\textwidth}
    \includegraphics[width=\linewidth]{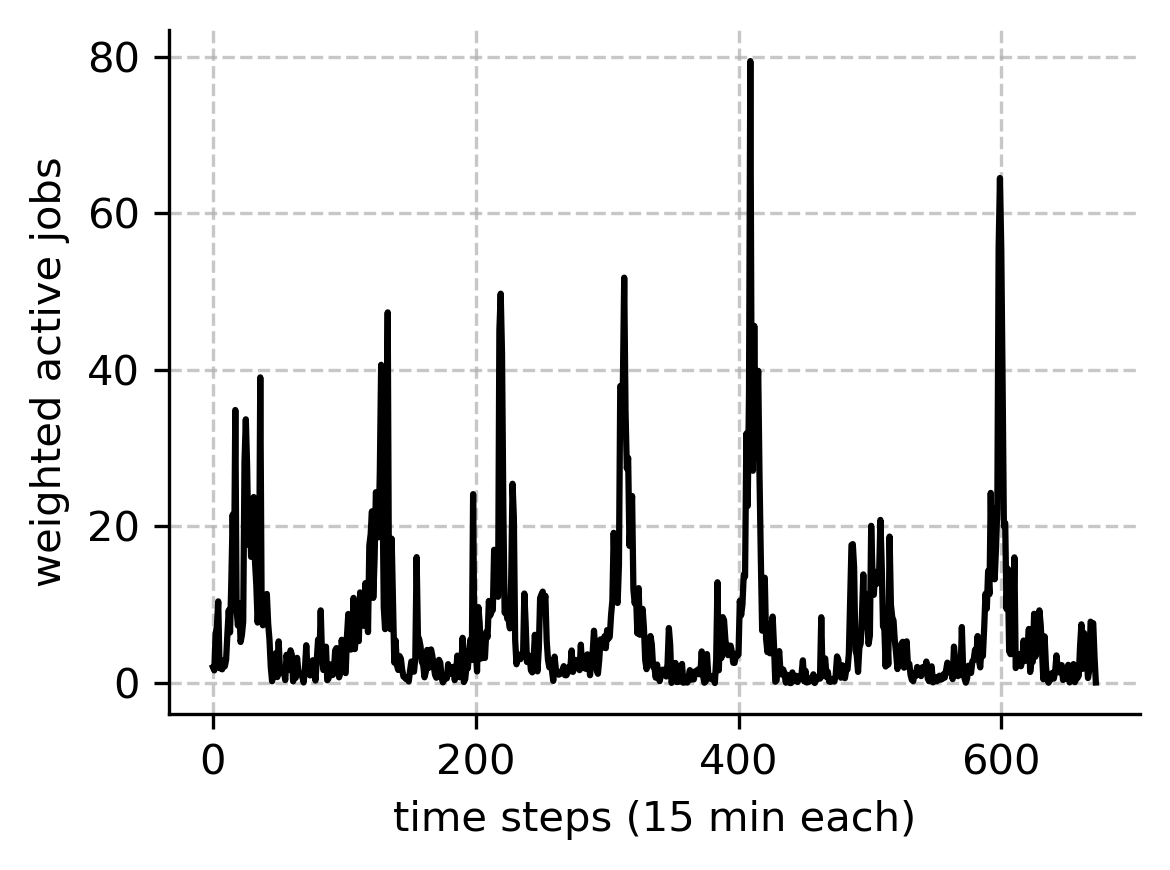}\vspace{-1em}
     \caption{ Weighted active jobs over time (bucketed in 15-min intervals) for the Alibaba trace~\cite{Alibaba:18} }\label{fig:weighted-active-jobs}
    \endminipage\hfill 
    \minipage{0.24\textwidth}
    \includegraphics[width=\linewidth]{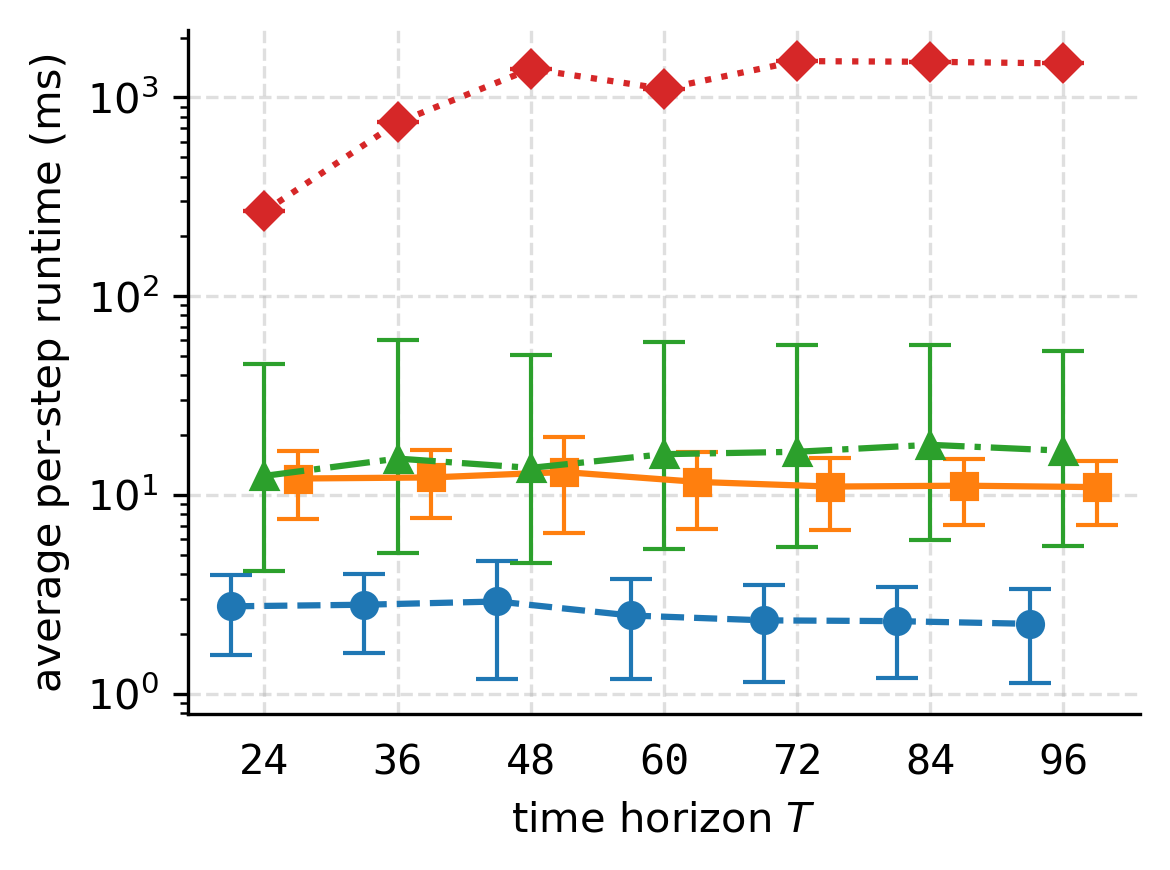}\vspace{-1em}
    \caption{ Average per-time-step runtime (in ms) for all algorithms with varying instance lengths $T$. }\label{fig:runtimes}
    \endminipage\hfill
    \minipage{0.24\textwidth}
    \includegraphics[width=\linewidth]{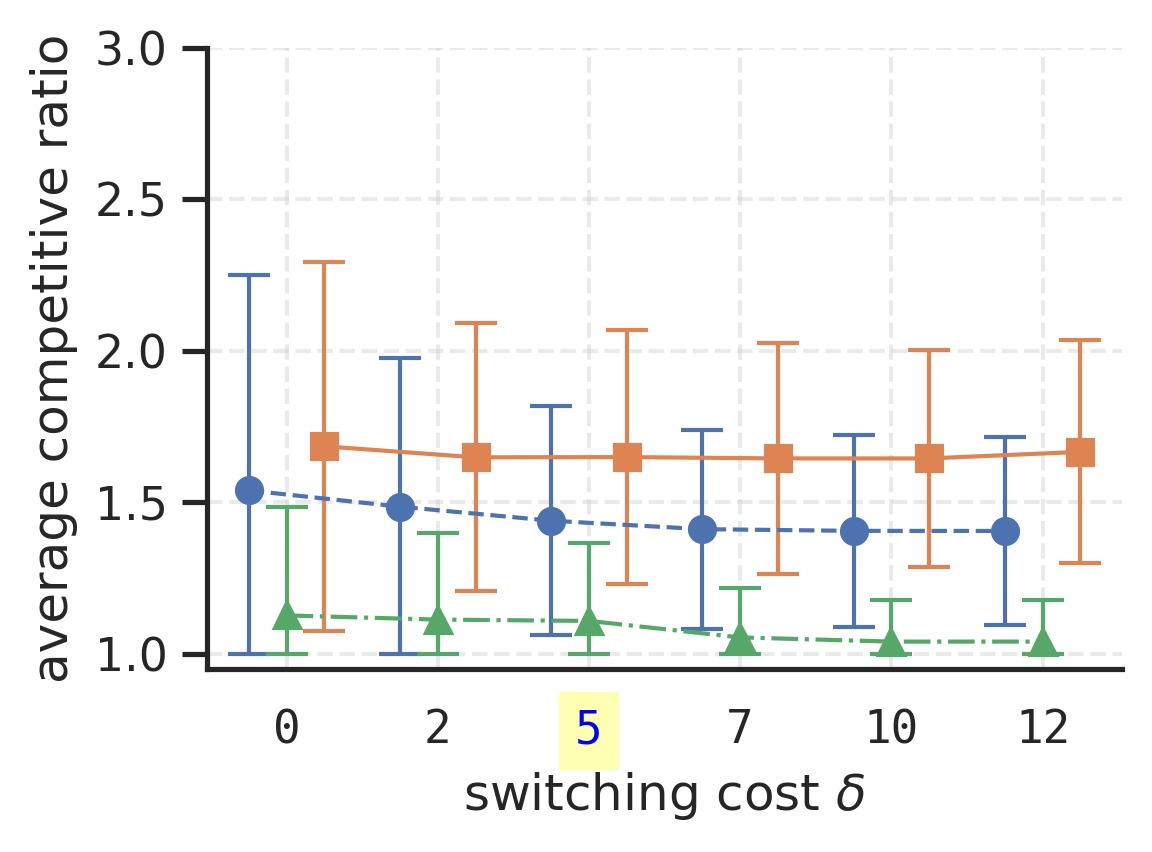}\vspace{-1em}
     \caption{ Average ECR for all algorithms in \texttt{CAISO} with varying delivery switching coefficients $\delta$. }\label{fig:deltas}
    \endminipage\hfill 
    \minipage{0.24\textwidth}
    \includegraphics[width=\linewidth]{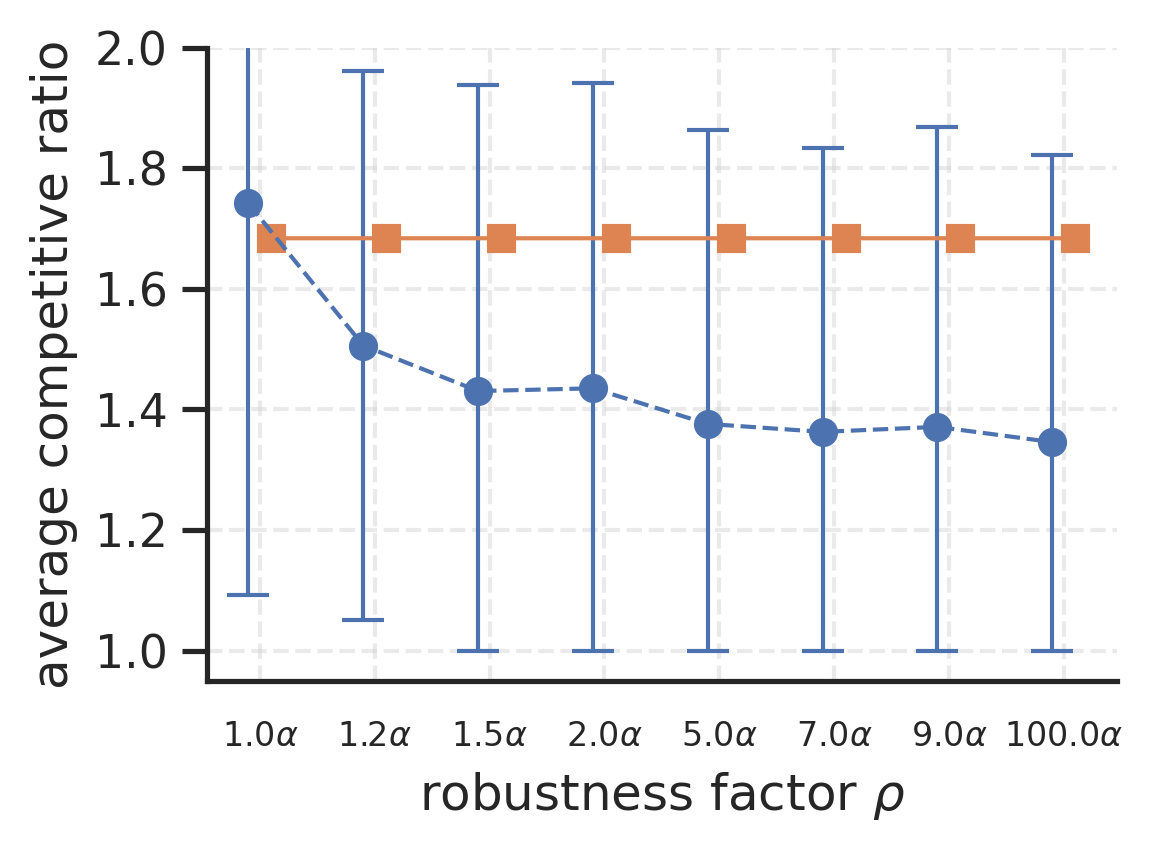}\vspace{-1em}
    \caption{ Average ECR for \PALDS in default \texttt{CAISO} experiment trained with varying $\rho$. }\label{fig:rhos}
    \endminipage\hfill  
\end{figure*}

\subsection{Runtime Measurements, Discussion, and Optimizations} \label{apx:runtime}

Building off the discussion in \autoref{sec:paad-algorithm}, we provide additional details about the empirical runtime of our algorithms and discuss some practical optimizations that can be used to improve the runtime of \PAAD and \PALD.

In \autoref{fig:runtimes}, we show the average per-time-step runtime (in milliseconds) for all algorithms with varying instance lengths $T$, averaging over 8,400 instances (1,200 for each setting of $T$).  
We also include the runtime of the offline solver (computed as a bilinear program using GurobiPy~\cite{gurobi}), normalized by the instance length $T$ to give a fair comparison with the online algorithms.
Each instance and each algorithm was run in a single thread on a MacBook Air with M2 processor and 24 GB of RAM.
Unsurprisingly, we find that the offline optimal solution is the most computationally expensive, exhibiting a ``per-time-step'' runtime an order of magnitude higher than the next closest online algorithm.  \PAAD and \PALDC have comparable runtimes, with averages of 62.5 ms and 65.7 ms per time step, respectively, which remains practical for real-time decision-making.

Interestingly, we find that \PALDS is the fastest online algorithm, with an average runtime of just 13.9 ms per time step.  We find that this is because of the re-parameterization of the psuedo-cost minimization problem in both variants of \PALD (see \autoref{apx:pald-implementation-details}).  Specifically, when solving the pseudo-cost minimization problem for each driver (see \autoref{alg:pcm-subroutine}), our \PAAD implementation solves the exact integral over the threshold function $\phi(\cdot)$, while our \PALD implementation uses a squared hinge re-parameterization that is faster for the CVXPY solver.  \PALDS is particularly fast because it uses a single set of thresholds for all time steps, allowing the CVXPY solver to cache and re-use computations across time steps.  \PALDC makes a forward pass through a neural network to predict the thresholds whenever a driver is created and project them into the feasible sets, which adds some overhead to the per-time-step runtime and brings it closer to \PAAD in runtime.

\PALDS's performance improvement demonstrates that several optimizations can be used to speed-up the core psuedo-cost minimization problem, e.g., by caching results, using approximations, or re-parameterizing the integral---we do not explore these optimizations in depth in this work, but they offer further evidence to suggest that the \PAAD and \PALD frameworks are practical for real-time decision-making in large-scale systems.

We also remark that additional driver-level optimizations can be employed to reduce the number of times that the psuedo-cost problem needs to be solved without affecting the theoretical results.  Perhaps the highest impact of these optimizations is \textit{driver consolidation}---if multiple drivers have the same type and relative state, they can be consolidated into a single driver with a larger size, allowing one solution to the pseudo-cost problem to essentially cover multiple drivers.  Formally, if for example, two drivers $i$ and $j$ have the same type (i.e., $d^{(i)} = d^{(j)}$) and the same relative state (i.e., $\nicefrac{w_b^{(i)}}{d^{(i)}} = \nicefrac{w_b^{(j)}}{d^{(j)}}$), then they can be consolidated into a single driver with size $d^{(i)} + d^{(j)}$.  We did not implement this optimization in our experiments, but it could be effective in practice to even further reduce the number of times that the pseudo-cost problem needs to be solved.

\subsection{Deferred Experimental Results} \label{apx:exp-results}

In this section, we present additional results and figures to complement those in the main body.  In subsequent figures (discussed below), we give experiment results for the effect of varying delivery switching coefficient $\delta$ and varying the robustness factor $\rho$ in our \PALD implementation.

\smallskip
\noindent \textit{$\triangleright$ Delivery Switching Cost. \ } \autoref{fig:deltas} shows the effect of varying the switching cost coefficient $\delta$ (with $\gamma$ fixed at $\gamma = 10$) between $0$ and $12.5$ in \texttt{CAISO}---other parameters are set to the experiment defaults (see \autoref{sec:exp}).  We find that the performance of all algorithms slightly improves as the switching cost coefficient $\delta$ increases.

\smallskip
\noindent \textit{$\triangleright$ Robustness Factor $\rho$. \ }
In \autoref{fig:rhos}, we show the effect of varying the robustness factor $\rho$ in our \PALDS implementation in \texttt{CAISO}.  
Recall that $\rho$ controls the degree of robustness guaranteed by the \PALD framework.  In our main experiments, we typically set $\rho = 5 \alpha$.  As $\rho \to \alpha$, the feasible sets for learning thresholds shrink to approximately those thresholds that would be valid for the robust \PAAD algorithm (see \autoref{sec:pald-learning}).  On the other hand, as $\rho \to \infty$, the feasible sets become unconstrained, allowing \PALDS to learn thresholds that may not be valid for any competitive online algorithm.  For \PALDS, we find that there is a steep drop in empirical competitive ratio as $\rho$ increases from $\alpha$ to $2\alpha$, after which the performance stabilizes. 
This suggests that there is a ``sweet spot'' for $\rho$ that balances robustness and flexibility, and that \PALDS can learn effective thresholds even when the feasible sets are relatively constrained. 
This also shows that in the case of \PALDS, where we just learn a single set of thresholds for each month in the LMP data, there is a practical limit to how well \PALDS can perform, as it cannot adapt to a specific instance and must ``average over'' many instances with different characteristics.

\subsection{\PALD Implementation Details} \label{apx:pald-implementation-details}

In this section, we provide additional details about our implementations of the \PALD framework in our case study (see \autoref{sec:exp})

We implement the \PALD algorithm using CVXPYLayers~\cite{CVXPYlayers:19,CVXPY} to enable optimization of the downstream task (i.e., minimizing the empirical competitive ratio on training instances) via backpropagation 
through the optimization problems that define \PALD's decisions.

As outlined in \autoref{sec:pald-learning}, \PALDS directly learns one set of piecewise-affine threshold functions to minimize overall empirical competitive ratio on training instances.  We implement projected gradient descent using PyTorch~\cite{Paszke:19:PyTorch} and CVXPY~\cite{CVXPY} to project the learned thresholds (specifically the parameter vectors) into the feasible sets.  To capture some degree of seasonality, we train a set of \PALDS thresholds for each month in a given region's LMP data (i.e., 12 sets of thresholds in total), using 100 random training instances for each month.  
In \autoref{fig:palds-diagram}, we give a high-level diagram of the \PALDS learning framework.

\begin{figure*}[h]
    \centering
    \vspace{-1em}
    \includegraphics[width=0.7\textwidth]{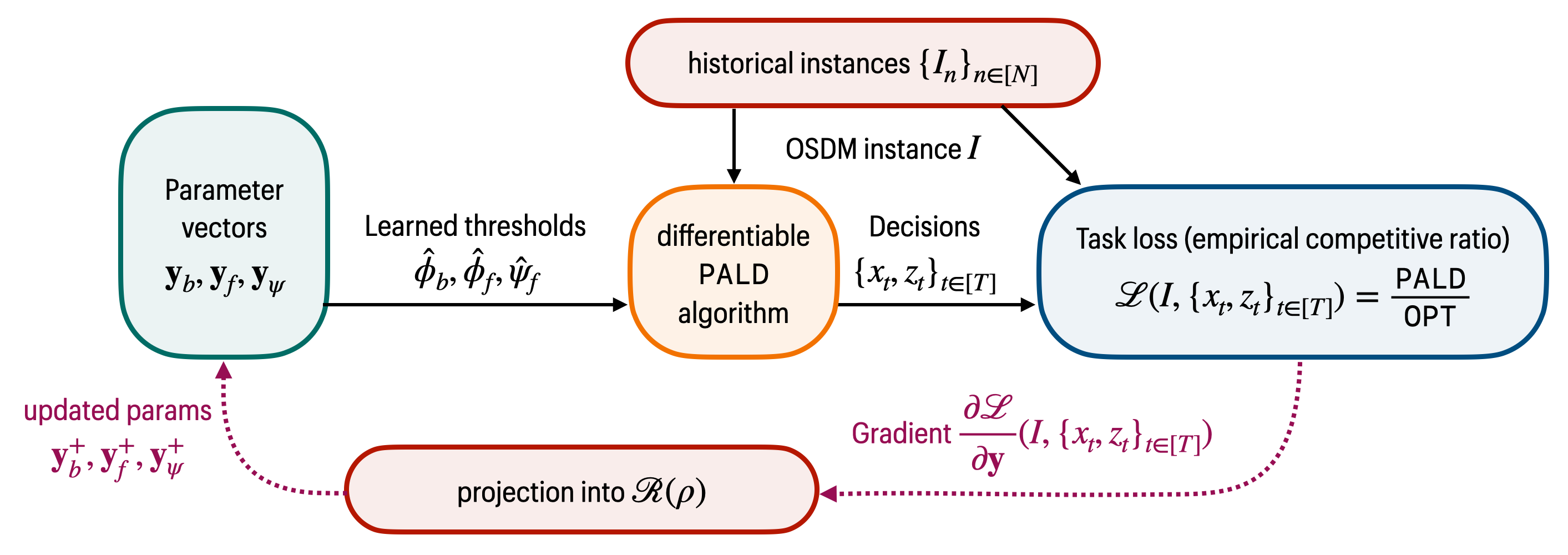} \vspace{-1em}
    \caption{\PALDS learns a single set of thresholds (via the parameter vectors $\mathbf{y}_b, \mathbf{y}_f, \mathbf{y}_\psi)$ to directly optimize a downstream task loss of the empirical competitive ratio.
    }
    \label{fig:palds-diagram}
    \vspace{-1em}
\end{figure*}

In contrast, \PALDC is a neural network-based approach (see \autoref{sec:pald-learning}) that uses contextual features to predict a set of piecewise-affine threshold functions for each driver.  In our implementation, we use a PyTorch neural network with two hidden layers (64 neurons each) and ReLU activations.  The output of this neural network is a monotone head that ensures the output parameter vectors (3 $K$-length vectors for the base and flexible driver threshold functions) are non-increasing by construction: we use a softplus activation to ensure non-negativity, and then take a cumulative sum in reverse order to ensure monotonicity.  The output of the monotone head is then passed through a CVXPYLayer that projects the outputs into the feasible sets.  We train \PALDC on 100 random training instances using all 12 months in a given region's LMP data, using the time, month, tracking target (if applicable), and statistics about the day-ahead LMP forecasts for future time steps (i.e., $\min$, $\max$, average, standard deviation), as the contextual features. In \autoref{fig:paldc-diagram}, we give a high-level diagram of the \PALDC learning framework.

\begin{figure*}[h]
    \centering
    \vspace{-1em}
    \includegraphics[width=0.8\textwidth]{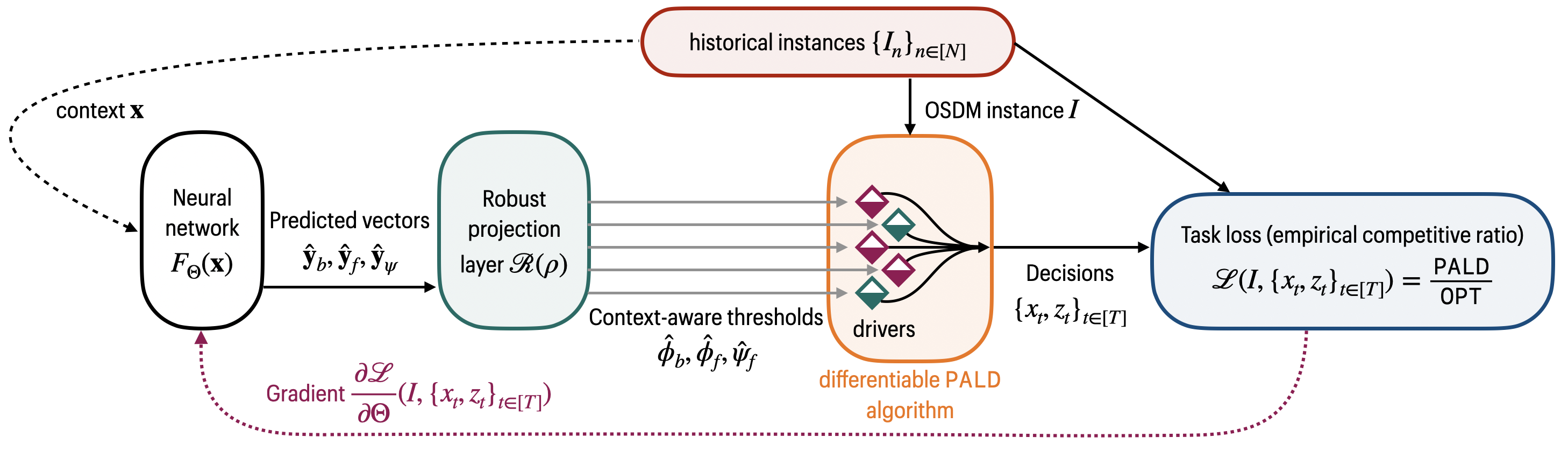} \vspace{-1em}
    \caption{\PALDC learns a mapping between context about each driver (see \autoref{alg:paad}) and the instance-specific threshold functions that minimize overall task loss.  Gradients flow through the differentiable \PALD implementation and the final projection layer of the neural network to improve this mapping.
    }
    \label{fig:paldc-diagram}
    \vspace{-1em}
\end{figure*}

In our implementation of both \PALDS and \PALDC, we use $K=10$ grid points to parameterize the piecewise-affine threshold functions.  We restrict the thresholds to lie in a concave region by enforcing that the second-order finite differences of the parameter vectors are non-positive.  The reason we do this is because we use a nonnegative sum of squared hinges to parameterize the computation of the integral over the threshold functions in the pseudo-cost minimization problem (i.e., $- \int_w^{w+x} \phi_b(u) du$, $- \int_w^{w+x} \phi_f(u) du$, and $- \int_v^{v+x} \psi_f(u) du$, see \autoref{alg:pcm-subroutine} for details).  This parameterization is convex and DPP (s disciplined parameterized program) in $x$ if the thresholds are linear decreasing or concave decreasing.  We implement this parameterization to employ CVXPYLayers, which require the optimization problem to be DPP in the optimization variables (i.e., $x$) and the parameters (i.e., the threshold function parameters).  However, we note that this concavity constraint is not required for the theoretical guarantees of \PALD and could likely be circumvented with a different implementation---the underlying integral is analytically convex in $x$ for any non-increasing threshold function.

\section{Deferred Discussion from \autoref{sec:problem}} \label{apx:examples}
In this section, we provide some deferred discussion about motivating examples for the \OSDM problem introduced in \autoref{sec:problem}.

\smallskip
\noindent \textbf{Grid-integrated Data Center with Storage. \ } Consider a grid-connected data center with local energy storage.  
At each discrete time step $t$ (e.g., 15-minute intervals), the data center must decide how much energy to purchase ($x_t$) from the grid at a price $p_t$, and how much energy to deliver (i.e., use now) to meet demand ($z_t$)---the difference between these quantities is stored in (resp. delivered from) the local energy storage with state of charge $s_t$.
In this setting, base demands $b_t$ model e.g., interactive workloads that must be satisfied immediately, while flexible demands $f_t$ model e.g., delay-tolerant batch jobs that can be deferred until a deadline $\Delta_t > t$.
The data center's goal is to minimize its electricity cost ($p_tx_t$) while smoothing its purchasing rate ($\mathcal{S}(\cdot)$), since large fluctuations in electricity consumption negatively impact the grid~\cite{Bloomberg:24}.  
The data center may also wish to place a small penalty on the cost of delivering energy ($\mathcal{D}(\cdot)$) from the battery to e.g., amortize the battery's degradation over time, as well as a penalty on the smoothness of the delivery rate (i.e., its own rate of energy consumption within the data center) to reduce wear-and-tear on internal components ($\delta |z_t - z_{t-1}|$).
Compared to prior work that has studied demand management in data centers~\cite{Yang:20}, \OSDM captures the practical considerations of \textit{both} flexible and inflexible demand, as well as the desire to make ``smooth'' decisions on the grid- and demand-sides of the problem.

\smallskip
\noindent \textbf{Thermal Energy Demand Management. \ } Consider a thermal energy system with a local energy storage tank~\cite{NREL:ThermalStorage}.  At each time step $t$, the operator must decide how much thermal energy to purchase ($x_t$) at a price $p_t$ (e.g., by converting grid electricity into heat), and how much thermal energy to deliver to meet demand ($z_t$)---the difference between these quantities is stored in (resp. delivered from) the local storage tank with state of charge $s_t$.  Base demands $b_t$ model e.g., heating or cooling loads that must be satisfied immediately, while flexible demands $f_t$ model e.g., delay-tolerant heating or cooling loads that can be deferred until a deadline $\Delta_t > t$.
The operator's goal is to minimize their energy cost ($p_tx_t$) while smoothing their purchasing rate ($\mathcal{S}(\cdot)$).
The delivery cost ($\mathcal{D}(\cdot)$) captures the extra input energy required to drive the delivery system (e.g., pumps, heat exchangers, etc.), which may depend on the price of energy ($p_t$).
The operator may also wish to place a small penalty on the smoothness of the delivery rate (i.e., its own rate of thermal energy consumption) to reduce wear-and-tear on distribution components such as valves ($\delta |z_t - z_{t-1}|$).

\smallskip
\noindent \textbf{Just-in-Time Manufacturing with Material Inventory. \ } 
Consider a factory that processes raw materials into a finished product, and has a warehouse to store raw materials.  At each time step $t$, the factory must decide how much raw material to purchase ($x_t$) at a price $p_t$, and how much finished product to deliver to meet demand ($z_t$)---the difference between these quantities is stored in (resp. delivered from) the warehouse with state $s_t$.
Orders for the product arrive over time, where base demands $b_t$ model e.g., orders that must be satisfied immediately, while flexible demands $f_t$ model scheduled orders that can be deferred until a deadline $\Delta_t > t$.
The factory's goal is to minimize their material cost ($p_tx_t$) while smoothing their purchasing rate ($\mathcal{S}(\cdot)$), which can reduce supply chain uncertainty and lead to discounts from suppliers.
The delivery cost ($\mathcal{D}(\cdot)$) captures the processing cost for converting raw materials into the finished product at time $t$.  The factory may also wish to place a small penalty on the smoothness of the delivery rate (i.e., its own rate of production) to encourage a more predictable production schedule ($\delta |z_t - z_{t-1}|$).

\smallskip
\noindent \textbf{Flow Battery Storage Management. \ }
Consider a large energy consumer that operates a grid-connected flow battery to store energy during off-peak hours.  Flow batteries are a type of rechargeable battery where energy is stored in the form of two (or more) liquids contained in external tanks, and energy is extracted by pumping the liquids through a cell stack to generate electricity~\cite{Breeze:19}.  At each time step $t$, the operator must decide how much electricity to purchase ($x_t$) from the grid at a price $p_t$, and how much electricity to deliver to meet demand ($z_t$)---the difference between these quantities is stored in (resp. delivered from) the flow battery with state of charge $s_t$.
The operator's goal is to minimize their electricity cost ($p_tx_t$) while smoothing their purchasing rate ($\mathcal{S}(\cdot)$).
The delivery cost ($\mathcal{D}(\cdot)$) captures the extra input energy required to drive the flow battery system (e.g., pumps, valves, etc.), which may depend on the price of energy ($p_t$).

\section{Details of the \OCS Problem and the \RORO Algorithm} \label{apx:roro-details}

In this section, we provide deferred details about the \OCS problem and the ``ramp-on, ramp-off'' (\RORO) algorithm proposed to solve it optimally~\cite{Lechowicz:24}, which we discuss in the warmup (see \autoref{sec:warmup}).
We start by formally defining the \OCS problem, which can be cast as a special case of \OSDM with a single unit of flexible demand and no storage.  We then describe the \RORO algorithm and its guarantee in \autoref{sec:roro-description}.

\subsection{Online Conversion with Switching Costs (\OCS) \cite{Lechowicz:24}}
\noindent \textbf{Problem Statement. \ } 
Consider an operator who must purchase an asset of unit size (without loss of generality) before a given deadline $T$ while minimizing their total cost.  At each time step $t \in [T]$, a price $p_t$ arrives online, and the decision-maker must choose the amount of the asset to purchase at the current time step, represented by $x_t \in [0,1]$.
Given a decision $x_t$, the operator's cost at time $t$ is $p_t x_t + \gamma \vert x_t - x_{t-1} \vert$, where the first term is the cost due to the time-varying price, and the second term penalizes the decision-maker for ``unsmooth jumps'' across time steps.  The coefficient $\gamma > 0$ is known a priori.  An offline formulation of \OCS is as follows:
\begin{align}
    [\OCS] \quad \min_{\mathbf{x} \coloneqq \{x_t\}_{t\in T}} & \underbrace{\sum_{t=1}^{T} p_t x_t}_{\text{total purchasing cost}} + \underbrace{\sum_{t=1}^{T+1} \gamma \vert x_t - x_{t-1} \vert,}_{\text{switching penalty}}\\
    \text{s.t.} & \underbrace{\sum_{t=1}^{T} x_t = 1,}_{\text{deadline constraint}} \quad  x_t \in [0, d_t] \ \forall t \in [T].
\end{align}
\citet{Lechowicz:24} focus on the online version of \OCS, where the decision-maker can only observe price signals up to time $t$ when selecting $x_t$, and each choice of $x_t$ is irrevocable (i.e., it cannot be revised at future time steps).  The goal is to design an online algorithm that achieves a small competitive ratio (see \sref{Def.}{dfn:comp-ratio})
They make the following assumptions on the problem:

\noindent \textbf{Assumptions. \ } 
Prices have bounded support, i.e., $p_t \in [p_{\min}, p_{\max}] \ \forall t \in [T]$, where $p_{\min}$ and $p_{\max}$ are known positive constants.  
They also assume that the switching coefficient $\gamma$ is ``not too large''. Formally, it is bounded within $\gamma \in \left[ 0, \frac{(p_{\max} - p_{\min})}{2} \right)$.  If $\gamma$ exceeds this range, its impact on the total cost exceeds that of the prices, and any competitive algorithm should simply minimize this term, making decision-making trivial~\cite{Lechowicz:24}.
They typically assume that the deadline $T$ is known in advance to facilitate a ``compulsory trade'' that ensures the entire asset is purchased before the end of the sequence.   If the operator has completed $w_T$ fraction of the workload at time $T$, they are forced to purchase the remaining $1 - w_T$ fraction at time $T$.

\subsection{The \RORO Algorithm \cite{Lechowicz:24}} \label{sec:roro-description}
In the classic competitive setting, \OCS is solved optimally by an online algorithm framework called ``ramp-on, ramp-off'' (\RORO), shown by \cite{Lechowicz:24}.  In what follows, we present this algorithm to contextualize our results in \autoref{sec:warmup}.

In \RORO~\cite{Lechowicz:24}, the online decision at each time step is made by solving a \emph{pseudo-cost minimization problem} to determine the amount to purchase at the current time step (i.e., $x_t \in [0, 1-w_{t-1}]$).  This minimization balances between the extreme options of buying ``too much'' early (thus incurring suboptimal costs if prices later drop) and waiting too long to purchase ( risking being forced to purchase a large portion at once, potentially at high cost).

Whenever the price is ``sufficiently attractive'', the pseudo-cost minimization finds the best decision that purchase just enough to maintain a certain competitive ratio.  To define this trade-off, \citet{Lechowicz:24} introduce a \emph{dynamic threshold function} $\phi(w) : [0,1] \rightarrow [p_{\min},p_{\max}]$, which is a function that maps the current amount that has been purchased $w$ to a price value.  It is defined as:
\begin{equation}
    \phi(w) = p_{\max} - \gamma + \left( \frac{p_{\max}}{\alpha_{\RORO}} - p_{\max} + 2 \gamma \right) \exp \left( \frac{w}{\alpha_{\RORO}} \right), 
\end{equation}
where $\alpha_{\RORO}$ is the optimal competitive ratio for \OCS, defined as the unique positive solution to $\exp \left( \frac{1}{\alpha_{\texttt{RORO}}}\right) = \frac{p_{\max} - p_{\min} - 2\gamma}{\nicefrac{p_{\max}}{\alpha_{\texttt{RORO}}} - p_{\max} - 2\gamma}$~\cite{Lechowicz:24}:
\begin{equation}
    \alpha_{\RORO} \coloneqq \left[ W \left( \left( \frac{2\gamma + p_{\min}}{p_{\max}} - 1\right) \exp \left( \frac{2\gamma}{p_{\max}} - 1 \right) \right) - \frac{2\gamma}{p_{\max}} + 1 \right]^{-1}. \label{eq:alpha_roro}
\end{equation}
\noindent In the above, $W(\cdot)$ is the Lambert $W$ function~\cite{Corless:96LambertW}.  Given this definition of $\alpha_{\RORO}$, note that $\phi(\cdot)$ is monotonically decreasing on the interval $w \in [0,1]$ (i.e., non-increasing in the amount of demand that has already been purchased), and captures the marginal trade-off of making progress towards satisfying the single unit of demand. 
We summarize the \RORO algorithm in \sref{Algorithm}{alg:roro}.
\begin{algorithm}[t]
\caption{Online Ramp-On, Ramp-Off Algorithm (\RORO)~\cite{Lechowicz:24}}\label{alg:roro}
{\small
    \begin{algorithmic}[1]
        \item \textbf{input: } threshold function $\phi(w) : [0,1] \rightarrow [p_{\min},p_{\max}]$
        \item \textbf{initialize: } initial decision $x_0 = 0$, current utilization $w_0 = 0$
        \While {price $p_t$ is revealed and $w_{t-1} < 1$}:
            \State{solve \textbf{pseudo-cost minimization problem} to obtain decision:}
            \begin{align*}
                x_t \gets \argmin_{x \in [0,1-w_{t-1}]} & p_t x_t + \gamma \vert x - x_{t-1} \vert - \int_{w_{t-1}}^{w_{t-1}+x} \phi(u) du
            \end{align*}
            \State{update utilization as $w^{(t)} = w_{t-1} + x_t$}
        \EndWhile
    \end{algorithmic}
}
\end{algorithm}

\section{Deferred Proofs from \autoref{sec:warmup}} \label{apx:roro-doubling-proof}
In this section, we give a full description and analysis of the warmup ``doubling strategy'' for \RORO.  We first give the pseudocode for this extension in \autoref{alg:roro-doubling}.  The notation $\instance_j(d)$ denotes the $j^\text{th}$ instance of \RORO that is created with a demand of size $\hat{d}$.  Note that each instance of \RORO operates until all of the flexible demands assigned to it are satisfied (or until their deadlines have passed), at which point they always return $0$ as their purchasing decision.

\begin{algorithm}[t]
\caption{Doubling Version of \RORO for \OSDM \textit{(only flexible demand \& no delivery cost)}}
\label{alg:roro-doubling}
{\small
\begin{algorithmic}[1]
\State \textbf{initialize} initial demand guess $\hat{d} \gets 1$,  index $j \gets 0$, set of active \RORO instances $\mathcal{R} \gets \emptyset$
\State \textbf{create} initial instance of \RORO $\mathcal{R} \gets \mathcal{R} \cup \instance_0(\hat{d})$ and associated demands $\mathcal{A}_0 \gets \emptyset$
\For{each time step $t$}
    \If{new non-zero flexible demand $f_t > 0$ arrives}
        \If{$\vert \mathcal{A}_0 \vert + f_t \leq \hat{d}$}
            \State Assign $f_t$ to current instance of \RORO: $\mathcal{A}_j \gets f_t$
        \Else
            \State $\hat{d} \gets 2^{j+1}, \quad j \gets j + 1$
            \State $\mathcal{R} \gets \mathcal{R} \cup \instance_j(\hat{d})$ and $\mathcal{A}_j \gets f_t$
            \State Assign $f_t$ to new instance of \RORO: $\mathcal{A}_j \gets f_t$
        \EndIf
    \EndIf
    \State \textbf{compute} initial purchase decision $x'_t \gets \sum_{\instance_j \in \mathcal{R}} x_t^{(j)}$ 
    \State \Comment $x_t^{(j)}$ is the purchase decision of $\instance_j$
    \State \textbf{compute} initial delivery decision $z'_t \gets \sum_{\instance_j \in \mathcal{R}} x_t^{(j)}$
    \If{demand is binding, i.e., $\sum_{\tau=1}^{t-1} z_\tau + z_t < \sum_{\tau: \Delta_\tau \le t} f_\tau$}
        \State \textbf{increment} $x_t \gets x'_t + \left( \sum_{\tau: \Delta_\tau \le t} f_\tau - \sum_{\tau=1}^{t-1} z_\tau - z'_t \right)$, \ $z_t \gets z'_t + \left( \sum_{\tau: \Delta_\tau \le t} f_\tau - \sum_{\tau=1}^{t-1} z_\tau - z'_t \right)$
    \Else
        \State $x_t \gets x'_t, \quad z_t \gets z'_t$
    \EndIf
\EndFor
\end{algorithmic}
}
\end{algorithm}

We now prove \autoref{thm:roro-doubling-upper-bound}, which states that under the assumptions in \autoref{sec:assumptions} and \autoref{sec:warmup}, the ``doubling extension'' of \RORO defined in \autoref{sec:warmup} is $\zeta$-competitive for \OSDM with only flexible demand and no delivery cost, where $\zeta$ is at least:
\[
    \zeta \geq \min \left\{ \frac{\alpha_{\texttt{RORO}}}{1-\sigma} + \frac{\alpha_{\texttt{RORO}}}{1-\sigma} \cdot \frac{\left( p_{\max} - 2\gamma - \frac{p_{\max}}{\alpha_{\texttt{RORO}}} \right) \exp\left( \frac{\sigma}{\alpha_{\texttt{RORO}}} \right) - \frac{\sigma p_{\max} (1-\sigma)}{\alpha_{\texttt{RORO}}}}{p_{\min}}, \frac{p_{\max} + 2\gamma}{p_{\min}} \right\}.
\]
In the above, $\alpha_{\texttt{RORO}}$ is the optimal competitive ratio for \OCS, which is defined in the context of \OSDM as the unique positive solution to $\exp \left( \frac{1}{\alpha_{\texttt{RORO}}}\right) = \frac{p_{\max} - p_{\min} - 2\gamma}{\nicefrac{p_{\max}}{\alpha_{\texttt{RORO}}} - p_{\max} - 2\gamma} $~\cite{Lechowicz:24}.

\begin{proof}[Proof of \autoref{thm:roro-doubling-upper-bound}]
    Consider the following \textit{worst-case} \OSDM instance $\mathcal{I}_K$ that satisfies the assumptions given in \autoref{sec:assumptions} and \autoref{sec:warmup}. Suppose that there are $K$ instances of \RORO created over the time horizon.  We denote the $j^\text{th}$ instance of \RORO as $\instance_j$, $j \in 0, 1, \dots K$.  Let $\mathcal{A}_j$ denote the set of flexible demands assigned to $\instance_j$, and let $D^{(j)} = \sum_{f_t \in \mathcal{A}_j} f_t$ denote the total size of these demands.

    We fix an arbitrary amount of total demand $D$, and a maximum size of each flexible demand $\sigma < 1$, and a small increment factor $\iota$.
    An adaptive adversary presents \RORO with a sequence of prices and flexible demands as follows.  Whenever a new instance of \RORO is created, the adversary presents a ``round'' of prices in iteratively descending order, starting from $p_{\max}$ and decreasing by $\iota$ at each time step.  Whenever \RORO purchases a non-zero amount of asset, the adversary's next presented price is $p_{\max}$, forcing \RORO to ``switch off'' and incur a high switching cost.

    Suppose the $j^\text{th}$ instance of \RORO has size $S^{(j)} = 2^j$.  From the start of the $j^\text{th}$ ``round'' of descending prices, the adversary begins presenting flexible demands with size smaller than or equal to $\sigma$ and a uniform deadline of $T$, until the total size of flexible demands presented so far is equal to $S^{(j)} - \sigma + \upsilon$, for small $\upsilon > 0$.  Note that at this point, the total size of flexible demands assigned to $\instance_j$ is exactly $S^{(j)} - \sigma + \upsilon$.  For sufficiently small $\iota$, the adversary reaches this point before the descending prices reach $p_{\min}$---for the rest of the ``round'', the adversary presents zero new flexible demand.

    By presenting prices in descending order (and for sufficiently small $\iota$) the adversary ensures that \RORO's purchasing is exactly captured by $\int_0^{S^{(j)} - \sigma + \upsilon} \phi^{(j)}(u) du$, where $\phi^{(j)}$ is the pseudo-cost function of the $j^\text{th}$ instance of \RORO.  Furthermore, due to the price fluctuations, \RORO will have incurred an extra switching cost (not accounted for by the integral) of at least $\gamma \cdot ( S^{(j)} - \sigma + \upsilon)$, because the adversary forces \RORO to switch off after each purchase.

    Once $\instance_j$ has purchased the total amount of demand assigned to it, the adversary continues to present descending prices until it reaches $p_{\min}$, presenting this best price up to $m$ times (for sufficiently large $m$).  Note that \RORO cannot purchase at these good prices, because it has already satisfied all of the flexible demand assigned to $\instance_j$ and is thus constrained from doing so.

    After presenting $m$ copies of $p_{\min}$, the adversary then presents a new flexible demand with size $\sigma$, which is large enough to create a new instance of \RORO.  At this point, a new ``round'' of descending prices and adversarial demand begins, repeating the above process for $\instance_{j+1}$, and so on.  This continues until $K$ instances of \RORO have been created.

    In what follows, we denote $D^{(j)} = S^{(j)} - \sigma + \upsilon$ as the total size of flexible demands assigned to $\instance_j$ (i.e., the set $\mathcal{A}_j$).  With a slight abuse of notation, let $\OPT^{(j)}$ denote the cost incurred by \OPT to satisfy the flexible demands assigned to $\instance_j$.  Since all flexible demands assigned to $\instance_j$ arrive at or after time $t_{-}^{(j)}$, and must be satisfied before or at time $T$, the cost incurred by \OPT to satisfy these demands is at least the best price that arrives during this time period (i.e., $p_{\min}$), times the total demand size $D^{(j)}$.  We assume switching costs are negligible for sufficiently large $m$ (see the instance description above). Thus, we have:
    \begin{equation}
        \OPT^{(j)} = D^{(j)} p_{\min}. 
    \end{equation}
    Furthermore, the total cost of \OPT on the \OSDM instance $\mathcal{I}_K$ is exactly the sum of the costs incurred to satisfy each $\mathcal{A}_j$, i.e., 
    $$\OPT(\mathcal{I}_K) = \sum_{j=0}^K \OPT^{(j)}.$$
    In the original \RORO competitive proof for \OCS, they use the following facts to relate the cost of \OPT and \RORO for a single instance of \RORO:
    \begin{lemma}[\protect{\cite[Lemma B.2]{Lechowicz:24}}] \label{lem:lb-roro-single-instance}
        For a single instance of \RORO $\mathbf{I}$ with unit-size demand, we have:
        \[ 
        \OPT(\mathbf{I}) \geq \phi(\bar{w}) - \gamma,
        \] 
        where $\bar{w}$ is the total amount of demand satisfied by \RORO (excluding demand purchased due to binding constraints).
    \end{lemma}

    \begin{lemma}[\protect{\cite[Lemma B.3]{Lechowicz:24}}] \label{lem:ub-roro-single-instance}
        For a single instance of \RORO $\mathbf{I}$ with unit-size demand, we have:
        \[
        \RORO(\mathbf{I}) \leq \int_0^{\bar{w}} \phi(u) du + \gamma \bar{w} + (1-\bar{w})p_{\max}.
        \]
    \end{lemma}

    To apply these bounds in our context, we partition the cost of doubling \RORO into the sum of the costs incurred by each $\instance_j$.  With a slight abuse of notation, let $\RORO^{(j)}$ denote the cost incurred by $\instance_j$ to satisfy the flexible demands captured by $\mathcal{A}_j$.  By the definition of the worst-case instance, we have:
    \[
    \RORO^{(j)} = \int_0^{D^{(j)}} \phi^{(j)}(u) du + \gamma D^{(j)}.
    \]
    \sref{Lemma}{lem:lb-roro-single-instance} suggests that we should be able to lower bound $\OPT^{(j)}$ by $D^{(j)} \left[ \phi^{(j)}(D^{(j)}) - \gamma \right] $, where $D^{(j)}$ is the total amount of demand satisfied by $\instance_j$.  However, this is \textit{not true} in \OSDM---we already established that $\OPT^{(j)} = D^{(j)} p_{\min}$, and the definition of the pseudo-cost function $\phi^{(j)}$ ensures that $\phi^{(j)}(S^{(j)}) - \gamma = p_{\min}$.  However, since $D^{(j)} < S^{(j)}$ (because $D^{(j)} = S^{(j)} - \sigma + \upsilon$) and $\phi^{(j)}$ is decreasing, we have $\phi^{(j)}(D^{(j)}) - \gamma > p_{\min}$, so the bound does not hold.  Thus, we must modify the lower bound on $\OPT^{(j)}$ to account for the difference between $D^{(j)}$ and $S^{(j)}$.

    Thus, we say that $\phi^{(j)} (\bar{w}^{(j)} + (S^{(j)} - D^{(j)})) - \gamma$ captures the best price that arrives during the lifetime of $\instance_j$ if it satisfies $\bar{w}^{(j)}$ amount of demand.  Since we have $\bar{w}^{(j)} = D^{(j)}$ in our worst-case instance, we can use the following lower bound on $\OPT^{(j)}$:
    \begin{equation}
        \OPT^{(j)} = D^{(j)} \left[ \phi^{(j)} (\bar{w}^{(j)} + (S^{(j)} - D^{(j)})) - \gamma \right]. \label{eq:opt-lb}
    \end{equation}
    Note that this lower bound is tight in our worst-case instance, because the best price that arrives during the lifetime of $\instance_j$ is exactly $p_{\min} = \phi^{(j)}(S^{(j)}) - \gamma$, and $\bar{w}^{(j)} + (S^{(j)} - D^{(j)}) = S^{(j)}$.

    In the case of unit-size demand, it is given that $\int_0^w \phi(u) du + \gamma w + (1-w)p_{\max} = \alpha_{\texttt{RORO}} (\phi(w) - \gamma)$ for all $w \in [0,1]$~\cite[Theorem 3.2]{Lechowicz:24}.  Extending this to the case of instance $j$ with size $S^{(j)}$, we have the following:
    \begin{align*}
    \int_0^{\bar{w}^{(j)}} \phi^{(j)}(u) du + \gamma \bar{w}^{(j)} + (S^{(j)} - \bar{w}^{(j)}) p_{\max} &= \alpha_{\texttt{RORO}} S^{(j)} (\phi^{(j)}(\bar{w}^{(j)}) - \gamma), \\
    \int_0^{\bar{w}^{(j)}} \phi^{(j)}(u) du + \gamma \bar{w}^{(j)} + (D^{(j)} - \bar{w}^{(j)}) p_{\max} &= \alpha_{\texttt{RORO}} S^{(j)} (\phi^{(j)}(\bar{w}^{(j)}) - \gamma) - (S^{(j)} - D^{(j)}) p_{\max}.
    \end{align*}
    Since we know that $\OPT^{(j)}$ is lower bounded as in \eqref{eq:opt-lb}, we can relate $S^{(j)} (\phi^{(j)}(\bar{w}^{(j)})$ to $\OPT^{(j)}$ as follows:
    \[
    S^{(j)} (\phi^{(j)}(\bar{w}^{(j)}) - \gamma) = \frac{S^{(j)}}{D^{(j)}} \cdot \left[ \OPT^{(j)} + D^{(j)} \left( p_{\max} - 2\gamma - \frac{p_{\max}}{\alpha_{\texttt{RORO}}}\right) \exp\left( \frac{S^{(j)} - D^{(j)}}{\alpha_{\texttt{RORO}}} \right)\right],
    \]
    where the trailing term comes from the additive $(S^{(j)} - D^{(j)})$ term inside the $\phi^{(j)}$ function in \eqref{eq:opt-lb}.      Combining the above two equations, we have the following characterization of $\RORO^{(j)}$:
    {\small
    \begin{equation}
    \RORO^{(j)} = \alpha_{\texttt{RORO}} \frac{S^{(j)}}{D^{(j)}} \OPT^{(j)} + \alpha_{\texttt{RORO}} S^{(j)} \left( p_{\max} - 2\gamma - \frac{p_{\max}}{\alpha_{\texttt{RORO}}}\right) \exp\left( \frac{S^{(j)} - D^{(j)}}{\alpha_{\texttt{RORO}}} \right) - (S^{(j)} - D^{(j)}) p_{\max}. \label{eq:roro-ub}
    \end{equation}
    }

    \noindent Furthermore, note that the total cost of the doubling \RORO instances on the \OSDM instance $\mathcal{I}_K$ is exactly the sum of the costs incurred to satisfy each $\mathcal{A}_j$, i.e.,
    $$\RORO(\mathcal{I}_K) = \sum_{j=0}^K \RORO^{(j)}.$$
    In what follows, we let $\zeta$ denote the competitive ratio of the doubling \RORO algorithm.  We have the following $\frac{\RORO(\mathcal{I}_K)}{\OPT(\mathcal{I}_K)}$ for the given adversarial instance $\mathcal{I}_K$, where $D = \sum_{j=0}^K D^{(j)}$ is the total demand in the instance:

    \begin{align*}
    \frac{\RORO(\mathcal{I}_K)}{\OPT(\mathcal{I}_K)} &= \frac{\sum_{j=0}^K \RORO^{(j)}}{\sum_{j=0}^K \OPT^{(j)}}, \\
    &= \sum_{j=0}^K \frac{D^{(j)}}{D} \frac{\RORO^{(j)}}{\OPT^{(j)}}
    \end{align*}

    \noindent We can use the fact that $D^{(j)} = S^{(j)} - \sigma$ to simplify the expression of $\nicefrac{\RORO^{(j)}}{\OPT^{(j)}}$ using the following facts: $S^{(j)} - D^{(j)} = \sigma$, $\frac{S^{(j)}}{D^{(j)}} = \frac{S^{(j)}}{S^{(j)}-\sigma}$, $\OPT^{(j)} \geq D^{(j)} p_{\min}$.  Thus, we have:
    \begin{align*}
    \frac{\RORO^{(j)}}{\OPT^{(j)}} &= \frac{S^{(j)} \alpha_{\texttt{RORO}}}{S^{(j)}-\sigma} + \frac{S^{(j)}\alpha_{\texttt{RORO}}}{S^{(j)}-\sigma} \cdot \frac{\left( p_{\max} - 2\gamma - \frac{p_{\max}}{\alpha_{\texttt{RORO}}} \right) \exp\left( \frac{\sigma}{\alpha_{\texttt{RORO}}} \right) - \frac{\sigma p_{\max} (S^{(j)}-\sigma)}{S^{(j)}\alpha_{\texttt{RORO}}}}{p_{\min}}.
    \end{align*}
    Note that the term $\frac{S^{(j)}}{S^{(j)}-\sigma}$ is maximized when $S^{(j)}$ is closest to $\sigma$, and thus the worst-case ratio occurs when $j = 0$ and $S^{(0)} = 1$.  This corresponds to the case where $K = 0$ (i.e., only the initial instance of \RORO is created).  In this instance, we have:
    \begin{align*}
    \frac{\RORO(\mathcal{I}_0)}{\OPT(\mathcal{I}_0)} &= \frac{D^{(0)}}{D} \frac{\RORO^{(0)}}{\OPT^{(0)}},
    \end{align*}
    where $D^{(0)} = D$ and $\frac{\RORO^{(0)}}{\OPT^{(0)}}$ is given by the following:
    \begin{align*}
    \frac{\RORO(\mathcal{I}_0)}{\OPT(\mathcal{I}_0)} &= \frac{\RORO^{(0)}}{\OPT^{(0)}} =  \left( \frac{\alpha_{\texttt{RORO}}}{1-\sigma} + \frac{\alpha_{\texttt{RORO}}}{1-\sigma} \cdot \frac{\left( p_{\max} - 2\gamma - \frac{p_{\max}}{\alpha_{\texttt{RORO}}} \right) \exp\left( \frac{\sigma}{\alpha_{\texttt{RORO}}} \right) - \frac{\sigma p_{\max} (1-\sigma)}{\alpha_{\texttt{RORO}}}}{p_{\min}} \right).
    \end{align*}
    Since the competitive ratio $\zeta$ is defined as the supremum over all valid \OSDM instances which includes $\mathcal{I}_0$, we have the following lower bound on $\zeta$:
    \[
    \zeta \geq \frac{\alpha_{\texttt{RORO}}}{1-\sigma} + \frac{\alpha_{\texttt{RORO}}}{1-\sigma} \cdot \frac{\left( p_{\max} - 2\gamma - \frac{p_{\max}}{\alpha_{\texttt{RORO}}} \right) \exp\left( \frac{\sigma}{\alpha_{\texttt{RORO}}} \right) - \frac{\sigma p_{\max} (1-\sigma)}{\alpha_{\texttt{RORO}}}}{p_{\min}}.
    \]

    We note that in a regime where $\sigma$ approaches $1$, the above lower bound on $\zeta$ is ill-defined.  In the \OCS problem, it is known that the worst-possible competitive ratio is $\frac{p_{\max} + 2\gamma}{p_{\min}}$~\cite{Lechowicz:24}, which is achieved by any feasible algorithm that purchases at any time step (i.e., possibly incurring the worst cost $p_{\max}$), and incurs the worst-possible switching cost to do so (i.e., $2\gamma$).  
    
    In the \OSDM problem with no delivery costs (i.e., $c = \varepsilon = 0$), no storage, and only flexible demands, we have a similar result---any algorithm is $\frac{p_{\max} + 2\gamma}{p_{\min}}$-competitive as long as it is feasible.  Thus, we have the following breakpoint for the lower bound on $\zeta$:
    \[
    \zeta \geq \min \left\{ \frac{\alpha_{\texttt{RORO}}}{1-\sigma} + \frac{\alpha_{\texttt{RORO}}}{1-\sigma} \cdot \frac{\left( p_{\max} - 2\gamma - \frac{p_{\max}}{\alpha_{\texttt{RORO}}} \right) \exp\left( \frac{\sigma}{\alpha_{\texttt{RORO}}} \right) - \frac{\sigma p_{\max} (1-\sigma)}{\alpha_{\texttt{RORO}}}}{p_{\min}}, \frac{p_{\max} + 2\gamma}{p_{\min}} \right\}.
    \]
    This completes the proof.

\end{proof}

\section{Deferred Proofs from \autoref{sec:alg} (Analysis of \PAAD Algorithm)} \label{apx:paad-proof}
In this section, we provide full proofs for the results in \autoref{sec:alg}, which describes and analyzes the \PAAD algorithm for the \OSDM problem.  We start by providing deferred plots of the competitive ratio $\alpha$ as a function of the problem's parameters to showcase the important dependencies in \autoref{fig:alpha-plots}.
We then prove \autoref{thm:osdm-upper-bound-1}, which states the competitive bound of \PAAD for \OSDMS, before proving \autoref{cor:osdmt-upper-bound-2}, which states the corresponding bound for \OSDMT.

\begin{figure*}[h]
    \centering
    \begin{subfigure}{0.24\textwidth}
        \centering
        \includegraphics[width=\textwidth]{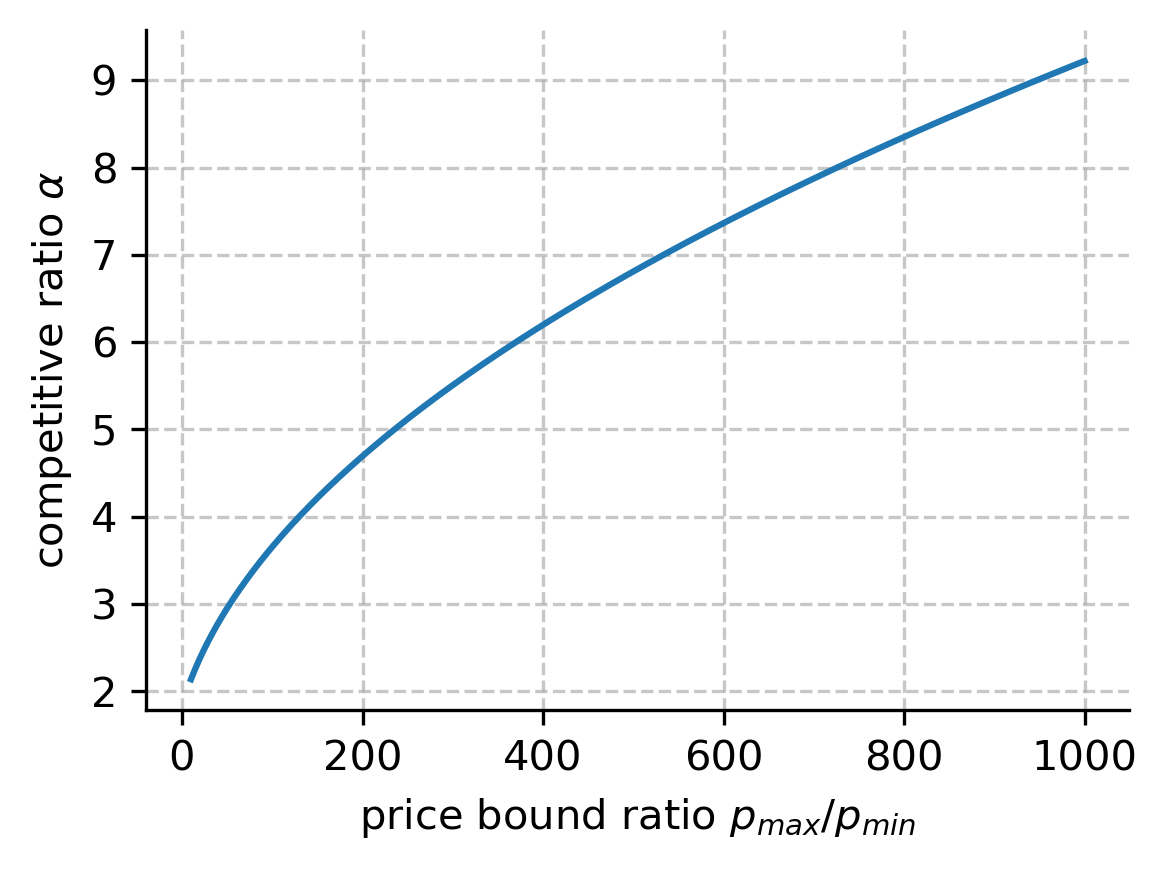}\vspace{-0.5em}
        \caption{Price bound ratio $\tfrac{p_{\min}}{p_{\max}}$}
    \end{subfigure}
    \hfill
    \begin{subfigure}{0.24\textwidth}
        \centering
        \includegraphics[width=\textwidth]{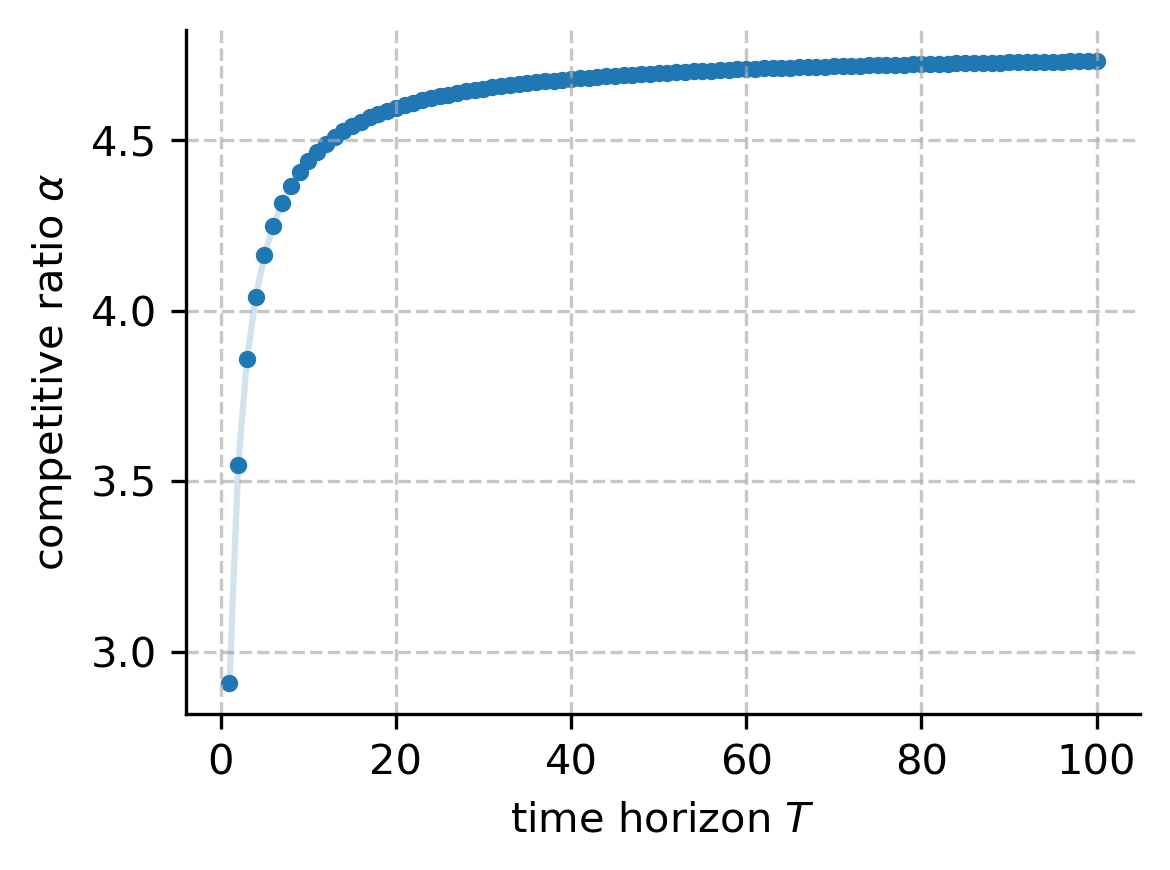}\vspace{-0.5em}
        \caption{Time horizon $T$}
    \end{subfigure}
    \hfill
    \begin{subfigure}{0.24\textwidth}
        \centering
        \includegraphics[width=\textwidth]{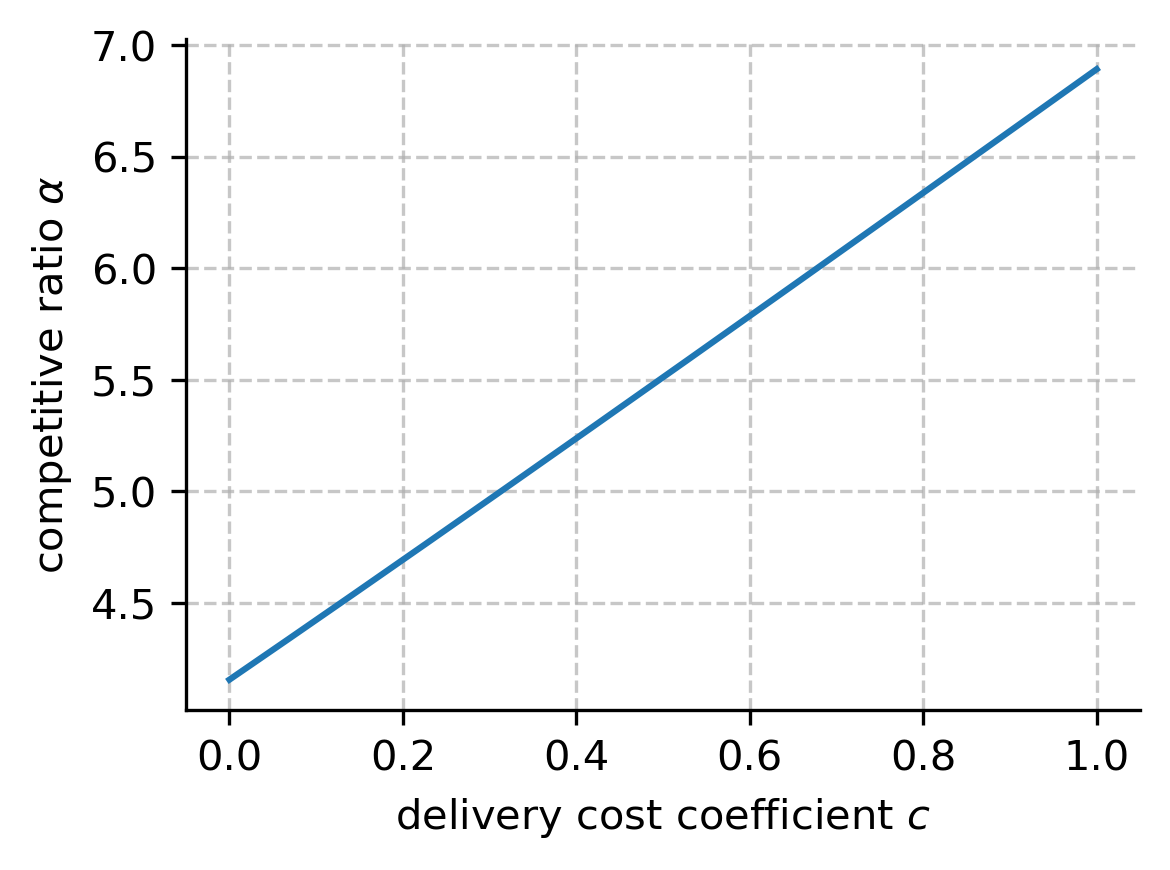}\vspace{-0.5em}
        \caption{Delivery coeff. $c$}
    \end{subfigure}
    \begin{subfigure}{.24\textwidth}
        \centering
        \includegraphics[width=\textwidth]{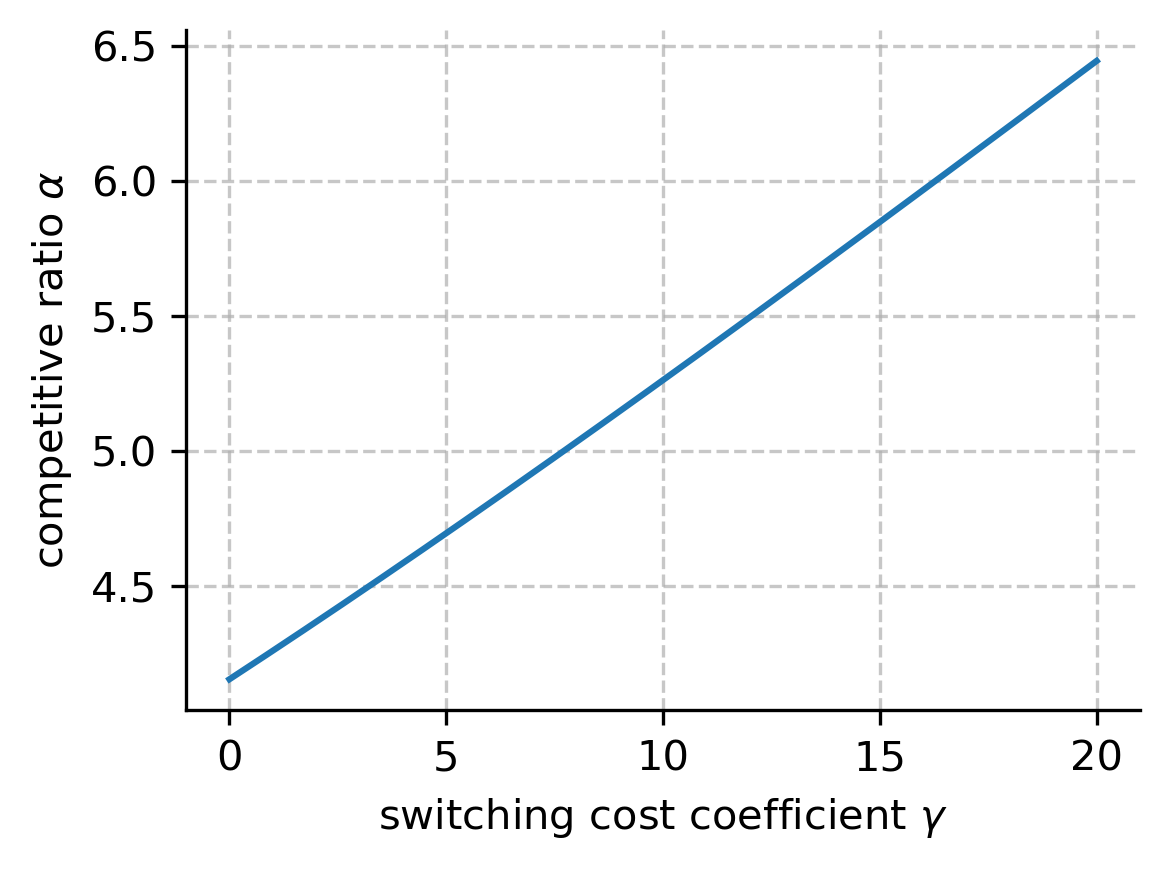}\vspace{-0.5em}
        \caption{Switching coeff. $\gamma$ }
    \end{subfigure}
    \vspace{-1em}
    \caption{Plotting the competitive ratio $\alpha$ (defined in \eqref{eq:alpha1}) as a function of the parameter dependencies described in each subfigure.  Unless a parameter is on an $x$-axis, we set $p_{\min} = 10$, $p_{\max} = 200$, $\gamma = 5$, $\delta = 1$, $c = 0.2$, $\varepsilon = 0.05$, and $T=48$.}
    \label{fig:alpha-plots}
    \vspace{-1em}
\end{figure*}

\subsection{Proof of \autoref{thm:osdm-upper-bound-1}}

In this section, we prove \autoref{thm:osdm-upper-bound-1}, which states that \PAAD achieves a competitive ratio of $\alpha$ (defined in \eqref{eq:alpha1}) for \OSDMS.  We start by proving feasibility of the solution produced by \PAAD, before proceeding to prove the competitive ratio.

\begin{proof}
Before the main competitive proof, we start our analysis by showing that \PAAD produces a feasible solution to \OSDM.  
\begin{lemma}
\label{lem:osdm-feasibility-1}
\PAAD produces a feasible solution to \OSDM, i.e., it satisfies all demand before their deadlines and never violates the storage capacity constraint.
\end{lemma}
\begin{proof}

We prove feasibility directly by the definition of the algorithm---there are two primary constraints to check.

First, we must ensure that all demand is covered.  In the case of base demand, this is immediate by line 18 of \autoref{alg:paad}, where the delivery decision $z_t$ always includes the base demand $b_t$ that arrives at time $t$.  In the case of flexible demand, this is ensured by line 19 and line 23 of \autoref{alg:paad}, which checks for binding constraints and projects the purchase and delivery decisions (if necessary) to ensure that all flexible demands are satisfied before their deadlines.

Next, we show that the storage capacity constraint is never violated.  
We first consider the base demand, as the base demand drivers have the ``role'' of replenishing the storage.  Note that by \sref{Definition}{dfn:active-inactive-periods}, the storage is empty $(s_t = 0)$ during any inactive period, which trivially satisfies the storage capacity constraint $0 \leq s_t \leq S$.  We thus then consider an arbitrary time step $t$ during an active period.  Again by \sref{Definition}{dfn:active-inactive-periods}, the total amount of demand purchased by base drivers in this $i^\text{th}$ active period is less than or equal to $\sum_{j=0}^{\nu_i} \hat{w}_b^{(i,j)}$---this itself is less than or equal to $\sum_{j=0}^{\nu_i} B_b^{(i,j)}$ by the definition of the base demand driver.  The cumulative base demand that arrives in the $i^\text{th}$ active period is equal to $\sum_{j=1}^{\nu_i} B_b^{(i,j)}$ (note the changed index to account for the initial storage manager driver $j=0$).  Thus, the amount stored is at most $\sum_{j=0}^{\nu_i} \hat{w}_b^{(i,j)} - \sum_{j=1}^{\nu_i} B_b^{(i,j)}$, which is at most $B_b^{(i,0)} = S$ by the definition of the storage manager driver. Thus, we have $s_t \leq S$ for all $t$ during active periods.

We now consider the flexible demand.  At a high level, the structure of \PAAD is such that flexible demand drivers do not rely on the storage to satisfy demand---rather, they rely on the temporal flexibility of flexible demand to reduce cost subject to constraints.  To see this, first note that when a flexible demand's deadline is binding, it is assumed in \autoref{sec:assumptions} to always be feasible to increment \textit{both} $z_t$ and $x_t$ by an amount $f_t$ to ensure that the demand is satisfied by its deadline.  In such a case, the storage dynamics leave the storage unchanged, i.e., $s_t = s_{t-1} + f_t - f_t = s_{t-1}$.  Thus, the storage state satisfies $s_t \leq S$ for all $t$ during flexible demand deadlines.  In other time steps, the effect of the flexible demand drivers on the storage is governed by their per-time-step delivery decisions---i.e., if purchasing exceeds delivery, the storage increases, and if delivery increases purchasing, the storage is depleted.  However, note that the delivery decisions of flexible demand drivers are always upper bounded by the amount of demand that has been purchased so far, i.e., $z_t^{(i)} \leq w_t^{(i)}$ for all $i \in \mathcal{F}$ and $t \in [T]$.  Furthermore, due to the analytical structure of the purchasing and delivery threshold functions in the context of a monotone price-dependent delivery cost,
the (cumulative) delivery decisions of a flexible demand driver at a given time step are always at least as large as its (cumulative) purchasing decisions, i.e., $v_t^{(i)} \geq w_t^{(i)}$ for all $i \in \mathcal{F}$ and $t \in [T]$.  Thus, the net effect of flexible demand drivers on the storage is 0, i.e., $\sum_{i \in \mathcal{F}} x_t^{(i)} = \sum_{i \in \mathcal{F}} z_t^{(i)}$ for all $t \in [T]$, and the storage state satisfies $s_t \leq S$ for all $t$ during non-binding flexible demand periods as well.

\end{proof}

We now proceed to prove the competitive ratio of \PAAD.  To do so, we introduce some notation and a definition to facilitate our analysis.  For notational brevity, the following considers \textit{any arbitrary} $\mathcal{I} \in \Omega$ (i.e., any arbitrary instance of \OSDM).

\begin{definition}[Active and inactive periods] \label{dfn:active-inactive-periods}
In our proof, we partition the time horizon $T$ into two types of \textbf{periods}.  An \textbf{active period} contains a contiguous interval where the storage is non-empty (i.e., $s_t > 0$).  An \textbf{inactive period} corresponds to any interval that lies in-between two adjacent active periods.
\end{definition}

We use the following additional notations.  Suppose that for an arbitrary instance, \PAAD's actions result in $n$ active periods.  During the $i^\text{th}$ active period, $i \in [n]$, we denote $\nu_i$ as the number of base demand drivers that are created.  Further, we let $\hat{w}_b^{(i,j)}$ denote the amount of the $j^\text{th}$ base demand driver that has been purchased by the end of the $i^\text{th}$ active period, $j \in [\nu_i]$, and let $B_b^{(i,j)}$ denote the total demand associated with the $j^\text{th}$ base demand driver in the $i^\text{th}$ active period.

We index flexible demand drivers according to the time at which they were created, $\tau \in [T]$.  In particular, we let $\hat{w}_f^{(\tau)}$ and $\hat{v}_f^{(\tau)}$ denote the purchasing and delivery amounts of the $\tau^\text{th}$ flexible demand driver before the deadline.  Note that $f_\tau \geq 0$ is the total demand associated with the $\tau^\text{th}$ flexible demand driver, and if $f_\tau = 0$, then $\hat{w}_f^{(\tau)} = \hat{v}_f^{(\tau)} = 0$.  

We let $D_b = \sum_{i=1}^n \sum_{j=1}^{\nu_i} B_b^{(i,j)}$ and $D_f = \sum_{\tau=1}^T f_\tau$ denote the total base and flexible demand, respectively, that arrive over the entire time horizon.  We let $D = D_b + D_f$ denote the total demand.  We let $\hat{s}$ denote the final status of the storage at the end of the time horizon, i.e., $\hat{s} = s_T$.  

Finally, we introduce the following notation to characterize the optimal solution: let $G_i(\beta)$ denote the minimum cost of purchasing $\beta$ units of asset during the $i^\text{th}$ active period, and let $H_\tau(f)$ denote the minimum cost of purchasing $f$ units of asset during the period $[\tau, \tau + \Delta_\tau]$ (i.e., during the lifetime of the $\tau^\text{th}$ flexible demand).  Finally, let $\tilde{p}$ denote the minimum price during idle periods, and let $\bar{p}$ denote the weighted average price during periods with non-zero base demand, i.e., $\bar{p} = \frac{\sum_{\tau \in [T]} p_t \cdot b_t}{D_b}$.

With these preliminaries, we are ready to proceed to the main competitive proof.  We begin by proving a lower bound on the cost of the offline optimal solution $\OPT$.
\begin{lemma}
\label{lem:osdm-optimal-lower-bound-1}
Given that \PAAD produces $n$ active periods, let $\beta_i$ denote the asset purchased towards the base demand by the offline optimal solution during the $i^\text{th}$ active period, $i \in [n]$, and let $\tilde{p}$ denote the minimum price during inactive periods.  Then $\OPT(\mathcal{I})$ is lower bounded as:
\begin{align*}
\OPT(\mathcal{I}) &\geq \sum_{i=1}^n G_i(\beta_i) + \left( D_b - \sum_{i=1}^n \beta_i \right) \tilde{p} + \varepsilon \bar{p} D_b + \sum_{\tau=1}^T \left(1 + \varepsilon \right) H_\tau (f_\tau) + \left( D_b + D_f \right) \frac{2\kappa}{T}
\end{align*}
\end{lemma}
\begin{proof}
The cost of an optimal offline solution, denoted as $\OPT(\mathcal{I})$, can be split into five components.  The first two components are the purchasing cost and delivery cost associated with the base demand, respectively.  The third and fourth components are the purchasing and delivery costs associated with the flexible demand, respectively.  The last component is the switching cost incurred by the optimal solution. 

Recall that we let $G_i(\beta)$ denote the minimum cost of purchasing $\beta$ units of asset during the $i^\text{th}$ active period, and let $\beta_i$ denote the amount of asset purchased towards the base demand during the $i^\text{th}$ active period by the optimal solution.  Then the total purchasing cost associated with the base demand is lower bounded by $\sum_{i=1}^n G_i(\beta_i) + \left( D_b - \sum_{i=1}^n \beta_i \right) \tilde{p}$, where $\tilde{p}$ is the \textit{minimum} price during inactive periods.  The delivery cost associated with the base demand is lower bounded by $\varepsilon \bar{p} D_b$, where $\bar{p}$ is the weighted average price during periods with non-zero base demand---note that we use $\bar{p}$ in this case because the base demand must be delivered at the time of its arrival.

Next, recall that we let $H_\tau(f)$ denote the minimum cost of purchasing $f$ units of asset during the period $[\tau, \tau + \Delta_\tau]$ (i.e., during the lifetime of the $\tau^\text{th}$ flexible demand).  Then the total purchasing cost associated with the flexible demand is lower bounded by $\sum_{\tau=1}^T H_\tau (f_\tau)$.  The delivery cost associated with the flexible demand is lower bounded by $\varepsilon \sum_{\tau=1}^T H_\tau (f_\tau)$, since the flexible demand can be delivered at any time before its deadline, and the smallest cost to do so is at the best (i.e., smallest) price during its lifetime.

Finally, note that the switching cost incurred by the optimal solution is at least $\left( D_b + D_f \right) \frac{2\kappa}{T}$, where $\kappa = \gamma + \delta$.  This follows because the total demand must be purchased and delivered over the entire time horizon, and the switching cost is minimized when the switching cost is evenly distributed over the entire time horizon.  This completes the proof.
\end{proof}

After \sref{Lemma}{lem:osdm-optimal-lower-bound-1}, we now proceed to prove an upper bound on the cost incurred by \PAAD.

\begin{lemma}
\label{lem:osdm-paad-upper-bound-1}
Given that \PAAD produces $n$ active periods, let $\hat{w}_b^{(i,j)}$ denote the amount of the $j^\text{th}$ base demand driver that has been purchased by the end of the $i^\text{th}$ active period, $j \in [\nu_i]$, and let $B_b^{(i,j)}$ denote the total demand associated with the $j^\text{th}$ base demand driver in the $i^\text{th}$ active period.  Further, let $\hat{w}_f^{(\tau)}$ and $\hat{v}_f^{(\tau)}$ denote the purchasing and delivery amounts of the $\tau^\text{th}$ flexible demand driver before the deadline, respectively.  Then the cost incurred by \PAAD is upper bounded as:
\begin{align}
\PAAD(\mathcal{I}) &\leq \sum_{i=1}^n \sum_{j=1}^{\nu_i} \int_0^{\hat{w}_b^{(i,j)}} \kern-1em \phi_b^{(i,j)}(u) du + \sum_{\tau=1}^T \left( \int_0^{\hat{w}_f^{(\tau)}} \kern-1em \phi_f^{(\tau)}(u) du + \int_0^{\hat{v}_f^{(\tau)}} \kern-1em \psi^{(\tau)}(z) dz \right) \label{eq:paad-upper-bound-1}\\
& \quad + \left( D_b + D_f - \sum_{i=1}^n \sum_{j=1}^{\nu_i} \hat{w}_b^{(i,j)} - \sum_{\tau=1}^{T} \hat{w}_f^{(\tau)}\right) (p_{\max} + 2 \gamma) \label{eq:paad-upper-bound-2}\\
& \quad + D_b \left( \bar{p} (c+\varepsilon) + 2\delta \right) + \left(D_f - \sum_{\tau=1}^{T} \hat{v}_f^{(\tau)}\right) (p_{\max} (c + \varepsilon) + 2\delta) \label{eq:paad-upper-bound-3}\\
& \quad - c p_{\min} \left( \sum_{i=1}^n \sum_{j=1}^{\nu_i} \hat{w}_b^{(i,j)} + \sum_{\tau=1}^{T} \hat{w}_f^{(\tau)}\right) + \hat{s} p_{\max}. \label{eq:paad-upper-bound-4}
\end{align}
\end{lemma}
\begin{proof}
This proof follows by explicitly characterizing the worst-case cost of \PAAD using the definition in \autoref{alg:paad}.  In the equation above, \eqref{eq:paad-upper-bound-1} corresponds to the worst-case cost that can be charged to each driver's purchasing and delivery threshold functions, \eqref{eq:paad-upper-bound-2} corresponds to the worst-case cost incurred by \PAAD for purchasing any remaining demand that has not been purchased by the drivers, \eqref{eq:paad-upper-bound-3} corresponds to the worst-case cost incurred by \PAAD for delivering demand (in the flexible demand case, this corresponds to any demand that has not been delivered thus far), and \eqref{eq:paad-upper-bound-4} corresponds to the worst-case \textit{improvement} in delivery costs due to increases in the storage state, plus the final storage level.  We unpack each below.

First, note that by the definition of each driver, \PAAD only purchases (or delivers) when the market price $p_t$ is sufficiently low to make the threshold minimization problem in lines 13 and 17 negative (see \autoref{alg:paad}).  For instance, we can say that the purchasing and switching cost of the $j^\text{th}$ base demand driver in the $i^\text{th}$ active period is upper bounded by $\int_0^{\hat{w}_b^{(i,j)}} \phi_b^{(i,j)}(u) du$.  Similar logic follows for the purchasing and delivery costs of the flexible demand drivers.  Aggregating over all drivers gives \eqref{eq:paad-upper-bound-1}.  We remark that the definition of the pseudo-cost minimization problem (Line 3, \autoref{alg:pcm-subroutine}) charges the current procurement cost ($p_tx$), the current switching cost $\gamma \vert x - x_{t-1}\vert$, and the \textit{future} (worst-case) switching cost $\gamma \vert x \vert$ to the threshold function.\footnote{The future switching cost captures the fact that the final global purchasing/delivery decisions $x_{T+1}$ and $z_{T+1}$ are both 0 by construction, and the pseudo-cost minimization captures this future cost in advance.}  This is feasible because the threshold function includes enough ``buffer'' to accommodate switching costs by definition, and captures any costs incurred before e.g., constraints become binding.

Any additional purchasing needed to satisfy all demands is exactly captured by the difference between $D_b + D_f$ (the total base and flexible demand) and the cumulative purchasing by the drivers $\sum_{i=1}^n \sum_{j=1}^{\nu_i} \hat{w}_b^{(i,j)} + \sum_{\tau=1}^{T} \hat{w}_f^{(\tau)}$.  In the worst case, this additional purchasing is done at the highest price $p_{\max}$, and incurs a switching cost of $2\gamma$ (i.e., a switch up and a switch down).  This gives \eqref{eq:paad-upper-bound-2}.

Next, the delivery cost is upper bounded as follows.  For the base demands that must be delivered immediately, the worst-case delivery cost is $\bar{p} (c+\varepsilon) + 2\delta$ per unit, where $\bar{p}$ is the weighted average price during periods with non-zero base demand.  Note that $c + \varepsilon$ is the worst-case for the delivery coefficient (i.e., the largest cost \PAAD can pay as a function of the storage state), and assumes that a maximal switching cost of $2\delta$ is incurred. 
For the flexible demands, the worst-case delivery cost is incurred for any remaining undelivered demand that is not already covered by the flexible demand driver---this is upper bounded by $p_{\max} (c + \varepsilon) + 2\delta$ per unit.  Note that the same worst-case assumptions are made as above, except that in the flexible demand case, this undelivered demand is delivered at the worst-case (highest) price $p_{\max}$.  This gives \eqref{eq:paad-upper-bound-3}.

Finally, \eqref{eq:paad-upper-bound-4} credits \PAAD for any improvements in delivery costs that occur due to increases in the storage state.  In particular, recall that in this setting of \OSDM, the delivery cost coefficient decreases as a function of the storage state, so any purchasing that feeds the storage will reduce future delivery costs.  In the worst-case however, this improvement is realized when it has the smallest possible impact on the delivery cost, i.e., when the overall price is $c_{\min}$.  Finally, the term $\hat{s} p_{\max}$ captures any excess left in the storage at the end of the time horizon, which was purchased at a price of at most $p_{\max}$.  This gives \eqref{eq:paad-upper-bound-4}, and completes the proof.

\end{proof}

In addition to the above lemmas, the following technical lemmas are necessary to prove a relation between the threshold functions and the optimal cost.  We restate \sref{Lemma}{lem:osdm-threshold-function-relation-1} and \sref{Lemma}{lem:osdm-threshold-function-relation-2} as follows:

By the definition of the threshold function $\phi_b^{(i,j)}(\cdot)$, the following relation always holds:
{\small 
\begin{align*}
\int_0^{w} \kern-1em \phi_b(u) du + (1-w) (p_{\max}+2\gamma) + p_{\max}(c+\varepsilon) + 2\delta - c w p_{\min} = \alpha\left[ \phi_b(w) - 2\gamma + \varepsilon p_{\max} + \frac{2\kappa}{T}\right] \ & \forall w \in [0,1].
\end{align*}
}

By the definitions of the threshold functions $\phi_f(\cdot)$ and $\psi_f(\cdot)$, the following relation always holds:
{\small
\begin{align*}
\int_0^{w} \kern-1em \phi_f(u) du + (1-w) (p_{\max}+2\gamma) - c w p_{\min} + \int_0^{v} \kern-1em \psi_f(z) dz + (1-v) (p_{\max}(c+\varepsilon) + 2\delta)  =\\
\alpha'\left[ \phi_f(w) + \psi_f(v) - 2\kappa + \frac{2\kappa (1+c+\varepsilon)}{T (1+\varepsilon) }\right] \ & \forall w \in [0,1], v \in [0,w].
\end{align*}
}

Using the previous lemmas, we now prove the competitive ratio.  First, in the simple case where $D_b + D_f = 0$, we have that the optimal cost is $0$ and the cost of \PAAD is at most $\hat{s} p_{\max} \leq S p_{\max}$, so \PAAD is $\alpha$-competitive under \sref{Definition}{dfn:comp-ratio} by setting $C = S p_{\max}$.  We thus focus on the more general case in which $D_b + D_f > 0$. 
Using the results in \sref{Lemmas}{lem:osdm-optimal-lower-bound-1}, \ref{lem:osdm-paad-upper-bound-1}, \ref{lem:osdm-threshold-function-relation-1}, and \ref{lem:osdm-threshold-function-relation-2}, we claim that the following holds:
\begin{align*}
\frac{\PAAD(\mathcal{I}) - p_{\max} \hat{s}}{\OPT(\mathcal{I})} &\leq \alpha.
\end{align*}
To show this result, we first substitute the bounds from \sref{Lemmas}{lem:osdm-optimal-lower-bound-1} and \ref{lem:osdm-paad-upper-bound-1} into the left-hand side of the above equation.  We define some shorthand notation to facilitate the presentation. 

\noindent Let $Q = \sum_{i=1}^n \sum_{j=1}^{\nu_i} \int_0^{\hat{w}_b^{(i,j)}} \kern-1em \phi_b^{(i,j)}(u) du + \sum_{\tau=1}^T \left( \int_0^{\hat{w}_f^{(\tau)}} \kern-1em \phi_f^{(\tau)}(u) du + \int_0^{\hat{v}_f^{(\tau)}} \kern-1em \psi^{(\tau)}(z) dz \right)$ denote the integrals over the thresholds.
Let $\hat{W}_b = \sum_{i=1}^n \sum_{j=1}^{\nu_i} \hat{w}_b^{(i,j)}$, let $\hat{W}_f = \sum_{\tau=1}^{T} \hat{w}_f^{(\tau)}$, and let $\hat{W} = \hat{W}_b + \hat{W}_f$ denote the total purchasing by all drivers.  

\noindent Further, let $\hat{V}_f = \sum_{i=1}^n \sum_{j=1}^{\nu_i} \hat{v}_f^{(\tau)}$ denote the total delivery by all flexible demand drivers.

\noindent Let ${\bm \beta} = \sum_{i=1}^n \beta_i$ denote the total amount of asset purchased towards the base demand by the optimal solution during all active periods, noting that $D_b - {\bm \beta} \geq 0$ by definition.   Finally, let $D = D_b + D_f$ denote the total demand.  Substituting the bounds from \sref{Lemmas}{lem:osdm-optimal-lower-bound-1} and \ref{lem:osdm-paad-upper-bound-1} into the left-hand side of the above equation, we have:

{\small
\begin{align*}
\frac{\PAAD(\mathcal{I}) - p_{\max} \hat{s}}{\OPT(\mathcal{I})} &\leq \frac{Q + (D - \hat{W}) (p_{\max} + 2 \gamma) + D_b \left( \bar{p} (c+\varepsilon) + 2\delta \right) + (D_f - \hat{V}_f) (p_{\max} (c + \varepsilon) + 2\delta) - c p_{\min} \hat{W}}{\sum_{i=1}^n G_i(\beta_i) + (D_b - {\bm \beta}) \tilde{p} + \varepsilon \bar{p} D_b + \sum_{\tau=1}^T (1 + \varepsilon) H_\tau (f_\tau) + D \frac{2\kappa}{T}}
\end{align*}
}
Using the fact that $D_b - {\bm \beta} \geq 0$, the following is equivalent:
{\small
\begin{align*}
&= \frac{Q\!+\!({\bm\beta}\!+\!D_f\!-\!\hat{W}) (p_{\max}\!+\!2 \gamma)\!+\!{\bm \beta} \left( \bar{p} (c\!+\!\varepsilon)\!+\!2\delta \right) + (D_f\!-\!\hat{V}_f) (p_{\max} (c\!+\!\varepsilon)\!+\!2\delta)\!-\!c p_{\min} \hat{W}\!+\!(D_b\!-\!{\bm \beta}) (p_{\max}\!+\!2\kappa\!+\!\bar{p}(c\!+\!\varepsilon))}{\sum_{i=1}^n G_i(\beta_i)\!+\!\varepsilon \bar{p} {\bm \beta}\!+\!\sum_{\tau=1}^T (1\!+\!\varepsilon) H_\tau (f_\tau)\!+\!(D_f\!+\!{\bm \beta}) \frac{2\kappa}{T}\!+\!(D_b\!-\!{\bm \beta})(\tilde{p}\!+\!\varepsilon \bar{p}\!+\!\frac{2\kappa}{T}) } \\
\end{align*}
}
Then, we have the following:
{\small
\begin{align*}
\leq \max \Bigg \{ \frac{Q + ({\bm\beta}+D_f - \hat{W}) (p_{\max} + 2 \gamma) + {\bm \beta} \left( \bar{p} (c+\varepsilon) + 2\delta \right) + (D_f - \hat{V}_f) (p_{\max} (c + \varepsilon) + 2\delta) - c p_{\min} \hat{W}}{\sum_{i=1}^n G_i(\beta_i) + \varepsilon \bar{p} {\bm \beta} + \sum_{\tau=1}^T (1 + \varepsilon) H_\tau (f_\tau) + (D_f + {\bm \beta}) \frac{2\kappa}{T} }, \ \ &\\
\quad \frac{(D_b - {\bm \beta}) (p_{\max} + 2\kappa + \bar{p}(c+\varepsilon))}{(D_b - {\bm \beta})(\tilde{p} + \varepsilon \bar{p} + \frac{2\kappa}{T}) } \Bigg \}, & \\
\end{align*}
}
\noindent where the definition of $\tilde{p}$ ensures that the second term in the $\max$ is at most $\alpha$.  We now focus on the first term.  For the sake of contradiction, suppose that
{\small
\begin{align}
    \frac{Q + ({\bm\beta}+D_f - \hat{W}) (p_{\max} + 2 \gamma) + {\bm \beta} \left( \bar{p} (c+\varepsilon) + 2\delta \right) + (D_f - \hat{V}_f) (p_{\max} (c + \varepsilon) + 2\delta) - c p_{\min} \hat{W}}{\sum_{i=1}^n G_i(\beta_i) + \varepsilon \bar{p} {\bm \beta} + \sum_{\tau=1}^T (1 + \varepsilon) H_\tau (f_\tau) + (D_f + {\bm \beta}) \frac{2\kappa}{T} } > \alpha. \label{eq:ratio-term-1}
\end{align}
}
\noindent Instead of working directly with the expression in terms of ${\bm \beta}$, we first reason about how the cost of \OPT and \PAAD relate to one another in terms of $\sum_{i=1}^n \sum_{j=1}^{\nu_i} B_b^{(i,j)}$, the total demand assigned to base drivers.
We introduce the following notation for the sake of brevity: let $D_b = \sum_{i=1}^n \sum_{j=1}^{\nu_i} B_b^{(i,j)}$ denote the total base demand.  Then, we have the following relation:
{\small
\begin{align*}
& \frac{Q + (D_b+D_f - \hat{W}) (p_{\max} + 2 \gamma) + D_b \left( \bar{p} (c+\varepsilon) + 2\delta \right) + (D_f - \hat{V}_f) (p_{\max} (c + \varepsilon) + 2\delta) - c p_{\min} \hat{W}}{\sum_{i=1}^n G_i(B_i) + \varepsilon \bar{p} D_b + \sum_{\tau=1}^T (1 + \varepsilon) H_\tau (f_\tau) + (D_f + D_b) \frac{2\kappa}{T} }, \\
& \geq \frac{Q\!+\!({\bm \beta }\!+\!D_f\!-\!\hat{W}) (p_{\max}\!+\!2 \gamma)\!+\!{\bm \beta } \left( \bar{p} (c\!+\!\varepsilon)\!+\!2\delta \right)\!+\!(D_f\!-\!\hat{V}_f) (p_{\max} (c\!+\!\varepsilon)\!+\!2\delta)\!-\!c p_{\min} \hat{W}\!+\!(D_b\!-\!{\bm \beta}) (p_{\max}\!+\!\bar{p}(c\!+\!\varepsilon)\!+\!2 \kappa)}{\sum_{i=1}^n G_i(\beta_i)\!+\!\varepsilon \bar{p} {\bm \beta }\!+\!\sum_{\tau=1}^T (1\!+\!\varepsilon) H_\tau (f_\tau)\!+\!(D_f\!+\!{\bm \beta }) \frac{2\kappa}{T}\!+\!(D_b\!-\!{\bm \beta}) (\tilde{p}\!+\!\varepsilon \bar{p}\!+\!\frac{2\kappa}{T})}, \\
& > \alpha.
\end{align*}
}
where in the second inequality, we have used the fact that there is a worst-case input instance such that for all $0 < \beta < \beta' < \sum_{j=1}^{\nu_i} B_b^{(i,j)} - \hat{b}^{(i)}$, we have $G_i(\beta') - G_i(\beta) \leq \tilde{p} (\beta' - \beta)$, where $\hat{b}^{(i)}$ is the initial state of the storage at the start of the $i^\text{th}$ active period in \OPT's solution.  
See \sref{Lemma}{lem:upper-bound-B-beta-tilde-p} for a formal proof of this.

\begin{lemma} \label{lem:upper-bound-B-beta-tilde-p}
    Defining $\hat{b}^{(i)}$ as the initial state of the storage at the start of the $i^\text{th}$ active period in \OPT's solution, there exists a worst-case input instance such that for all $0 < \beta < \beta' < \sum_{j=1}^{\nu_i} B_b^{(i,j)} - \hat{b}^{(i)}$, we have:
    \begin{enumerate}
    \item $G_i(0) = 0$,
    \item $G_i(\beta') > G_i(\beta)$,
    \item $G_i(\beta') - G_i(\beta) \leq \tilde{p} (\beta' - \beta)$.
    \end{enumerate}
\end{lemma}
\begin{proof}
Statements 1 and 2 are immediate from the definition of $G_i(\cdot)$.  
Recall that $\tilde{p}$ is defined as the minimum price during inactive periods, which corresponds with the worst (highest) price that \PAAD's base drivers are willing to pay to fill the storage.

To prove statement 3, consider any input instance $\mathcal{I} = [(p_t, b_t, f_t, \Delta_t)]_{t \in [T]}$. An adversary can construct a new input instance $\mathcal{I}' = [(p_t', b_t', f_t', \Delta_t')]_{t \in [T']}$ as follows:
\[
\mathcal{I}' \coloneqq \left[ \left( \tilde{p}, 0, 0, \cdot \right), \left(p_1, b_1, f_1, \Delta_1\right), \left( \tilde{p}, 0, 0, \cdot \right), \left(p_2, b_2, f_2, \Delta_2\right), \ldots, \left(p_T, b_T, f_T, \Delta_T\right) \right].
\]
The purpose of constructing this new input instance is to guarantee that the purchasing cost of \PAAD does not change, and the purchasing cost of \OPT will not increase (for this worst-case logic, we only consider the purchasing cost, ignoring additional switching costs that would increase \OPT's cost).
Note that this is possible because the adversary can set the length of the time horizon.
Let $G_i(\beta)$ be the minimum cost of purchasing $\beta$ units of asset during the $i^\text{th}$ active period.
When \OPT purchases another $\beta' - \beta, \beta' \leq \sum_{j=1}^{\nu_i} B_b^{(i,j)} - \hat{b}^{(i)}$ units of asset during the $i^\text{th}$ active period, the cost will not be larger than $\tilde{p} (\beta' - \beta)$, since \OPT can always purchase at any newly added time slots.  Thus, there always exists a worst-case input instance such that $G_i(\beta') - G_i(\beta) \leq \tilde{p} (\beta' - \beta)$.

\end{proof}

Given the result in \sref{Lemma}{lem:upper-bound-B-beta-tilde-p}, we proceed to work with the expression in terms of $D_b$.
 During the $i^\text{th}$ active period and the lifetime of the $j^\text{th}$ base demand driver, the minimum marginal purchasing price observed is given by $\phi_b^{(i,j)}(\hat{w}_b^{(i,j)}) - 2\gamma$ by the definition of the threshold function.  Similarly, during the lifetime of the $\tau^\text{th}$ flexible demand driver, the minimum marginal purchasing and delivery cost observed is given by $\frac{1+\varepsilon}{1+c+\varepsilon} \left( \phi_f^{(\tau)}(\hat{w}_f^{(\tau)}) + \psi_f^{(\tau)}(\hat{v}_f^{(\tau)}) - 2 \kappa \right)$.  

\noindent This gives the following lower bounds on the terms that depend on $G_i$ and $H_\tau$, respectively:
{\small
\begin{align*}
\sum_{i=1}^n G_i\left( \sum_{j=1}^{\nu_i} B_b^{(i,j)} \right) &\geq \sum_{i=1}^n \sum_{j=1}^{\nu_i} \left( \phi_b^{(i,j)}(\hat{w}_b^{(i,j)}) - 2\gamma \right) \times B_b^{(i,j)} \\
\sum_{\tau=1}^T (1 + \varepsilon) H_\tau (f_\tau) &\geq \sum_{\tau=1}^T \frac{1+\varepsilon}{1+c+\varepsilon} \left( \phi_f^{(\tau)}(\hat{w}_f^{(\tau)}) + \psi_f^{(\tau)}(\hat{v}_f^{(\tau)}) - 2 \kappa \right) \times f_\tau.
\end{align*}
}

\noindent Substituting these bounds into the previous expression, we have that the left-hand-side of \eqref{eq:ratio-term-1} is less than or equal to:
{\small
\begin{align*}
&\leq \frac{Q + (D_b+D_f - \hat{W}) (p_{\max} + 2 \gamma) + D_b \left( \bar{p} (c+\varepsilon) + 2\delta \right) + (D_f - \hat{V}_f) (p_{\max} (c + \varepsilon) + 2\delta) - c p_{\min} \hat{W}}{\sum_{i=1}^n \sum_{j=1}^{\nu_i} \left( \phi_b^{(i,j)}(\hat{w}_b^{(i,j)}) - 2\gamma \right) B_b^{(i,j)} + \varepsilon \bar{p} D_b + \sum_{\tau=1}^T \frac{1+\varepsilon}{1+c+\varepsilon} \left( \phi_f^{(\tau)}(\hat{w}_f^{(\tau)}) + \psi_f^{(\tau)}(\hat{v}_f^{(\tau)}) - 2 \kappa \right) f_\tau + (D_f + D_b) \frac{2\kappa}{T} }.
\end{align*}
}

\noindent By rearranging the terms in the above and substituting for $Q$, we obtain the following:
{\small
\begin{align}
&= \frac{\sum_{i=1}^n \sum_{j=1}^{\nu_i} \left[ X_b^{(i,j)} \right] + \sum_{\tau=1}^T \left[ X_f^{(\tau)} \right]}{\sum_{i=1}^n \sum_{j=1}^{\nu_i} \left[ \left( \phi_b^{(i,j)}(\hat{w}_b^{(i,j)})\!-\!2\gamma\!+\!\varepsilon \bar{p}\!+\!\frac{2\kappa}{T}\right) B_b^{(i,j)}  \right]\!+\!\sum_{\tau=1}^T \left[ \frac{1+\varepsilon}{1+c+\varepsilon} \left( \phi_f^{(\tau)}(\hat{w}_f^{(\tau)})\!+\!\psi_f^{(\tau)}(\hat{v}_f^{(\tau)})\!-\!2 \kappa \right) f_\tau\!+\!f_\tau \frac{2\kappa}{T} \right]}, \label{eq:partitioned-sums-upper-bound-1}
\end{align}
}

\noindent where $X_b^{(i,j)}$ and $X_f^{(\tau)}$ are defined as follows (for each driver):
{\small
\begin{align*}
X_b^{(i,j)} &= \int_0^{\hat{w}_b^{(i,j)}} \kern-1em \phi_b^{(i,j)}(u) du + (B_b^{(i,j)} -\hat{w}_b^{(i,j)}) (p_{\max}+2\gamma) + \left( p_{\max}(c+\varepsilon) + 2\delta \right) B_b^{(i,j)} - c \hat{w}_b^{(i,j)} p_{\min}, \\
X_f^{(\tau)} &= \int_0^{\hat{w}_f^{(\tau)}} \kern-1em \phi_f^{(\tau)}(u) du + (f_\tau -\hat{w}_f^{(\tau)}) (p_{\max}+2\gamma) + \int_0^{\hat{v}_f^{(\tau)}} \kern-1em \psi_f^{(\tau)}(u) du + \left( p_{\max} (c+\varepsilon) + 2\delta \right)(f_\tau - \hat{v}_f^{(\tau)}) - c \hat{w}_f^{(\tau)} p_{\min}.
\end{align*}
}

\noindent Since \eqref{eq:partitioned-sums-upper-bound-1} is $> \alpha$ by assumption, one of the following cases must be true:

\smallskip
\noindent\textbf{Case I:} There exists some $i \in [n]$ and $j \in [\nu_i]$ such that
{\small
\begin{align*}
X_b^{(i,j)} &> \alpha B_b^{(i,j)} \left( \phi_b^{(i,j)}(\hat{w}_b^{(i,j)}) - 2\gamma + \varepsilon \bar{p} + \frac{2\kappa}{T}\right).
\end{align*}
}

\smallskip
\noindent\textbf{Case II:} There exists some $\tau \in [T]$ such that
{\small
\begin{align*}
X_f^{(\tau)} &> \alpha' f_\tau \left[ \frac{1+\varepsilon}{1+c+\varepsilon} \left( \phi_f^{(\tau)}(\hat{w}_f^{(\tau)}) + \psi_f^{(\tau)}(\hat{v}_f^{(\tau)}) - 2 \kappa \right) + \frac{2\kappa}{T} \right],
\end{align*}
}

\noindent But, if either of these two cases are true (i.e., for any unit of base or flexible demand), they contradict \sref{Lemmas}{lem:osdm-threshold-function-relation-1} and \ref{lem:osdm-threshold-function-relation-2}, respectively.  Thus, we have a contradiction, and the original assumption that \eqref{eq:partitioned-sums-upper-bound-1} is $> \alpha$ must be false.  This completes the proof that $\frac{\PAAD(\mathcal{I}) - p_{\max} \hat{s}}{\OPT(\mathcal{I})} \leq \alpha$. 

\noindent Using this result, we have:
\begin{align*}
\PAAD(\mathcal{I}) &\leq \alpha \OPT(\mathcal{I}) + p_{\max} \hat{s} \leq \alpha \OPT(\mathcal{I}) + p_{\max} S,
\end{align*}
where $p_{\max} S$ is a constant.  This shows that \PAAD is $\alpha$-competitive under \sref{Definition}{dfn:comp-ratio}, which completes the proof.

\end{proof}

\subsection{Proof of \sref{Corollary}{cor:osdmt-upper-bound-2}} \label{apx:osdmt-upper-bound-2}

In this section, we prove \sref{Corollary}{cor:osdmt-upper-bound-2}, which states that under the assumptions in \autoref{sec:assumptions}, \PAAD is $\alpha_{\texttt{T}}$-competitive for \OSDMT, where $\alpha_{\texttt{T}}$ is given by:
{\color{blue}
\begin{align}
\alpha_{\texttt{T}} \geq \frac{\omega}{\left[ W\left(-\frac{\left( (1+c+\varepsilon) p_{\max} - (1+\varepsilon) p_{\min} \right) e^{\frac{\nicefrac{\omega 2\delta}{T} - p_{\min} c - (1+c+\varepsilon) p_{\max}}{(1+c+\varepsilon)p_{\max} + 2(\eta+\delta)}}}{(1+c+\varepsilon)p_{\max} + 2(\eta + \delta)}\right)+\frac{(1+c+\varepsilon) p_{\max} + p_{\min} c - \nicefrac{\omega 2\delta}{T}}{(1+c+\varepsilon)p_{\max} + 2(\eta+\delta)} \right] }, 
\end{align}}
where $\omega = \frac{(1+c+\varepsilon)}{(1+\varepsilon)}$.  
We start with a (re-)definition of \PAAD's threshold functions and pseudo-cost minimization problem for the \OSDMT setting.
\begin{definition}[Threshold functions and pseudo-cost minimization for \PAAD in \OSDMT]
{\it
\label{dfn:new-thresholds-and-pseudo-cost-osdmt}
Given a base driver with size denoted by $d$, the purchasing threshold function $\phi_b$ for \OSDMT is defined as:
{\small
\begin{align*}
\phi_b(w) &= p_{\max} + 2 \eta + p_{\min} c + \left( \frac{p_{\max}(1+c+\varepsilon) + 2(\eta+\delta)}{\alpha_{\texttt{T}}} - \left( p_{\max}(1+\varepsilon) + p_{\min}c + \frac{2\delta}{T} \right) \right) \exp\left( \frac{w}{\alpha_{\texttt{T}} d} \right) : w \in [0,d].
\end{align*}
}
Furthermore, given a flexible driver with size denoted by $d$, the purchasing threshold function $\phi_f$ and delivery threshold function $\psi_f$ for \OSDMT are defined as:
{\small
\begin{align*}
\phi_f(w) &= p_{\max} + p_{\min} c + 2 \eta + \left( \frac{p_{\max}+ 2\eta}{\alpha'_{\texttt{T}}} - \left( p_{\max} + p_{\min}c \right) \right) \exp\left( \frac{w}{\alpha'_{\texttt{T}} d} \right) : w \in [0,d], \\
\psi(v) &= p_{\max} (c + \varepsilon) + 2\delta + \left( \frac{p_{\max}(c+\varepsilon) + 2\delta}{\alpha'_{\texttt{T}}} - \left( p_{\max}(c+\varepsilon) + \frac{2\delta}{T} \omega \right) \right) \exp\left( \frac{v}{\alpha'_{\texttt{T}} d} \right) : v \in [0,d],
\end{align*}
}
where $\omega = \frac{(1+c+\varepsilon)}{(1+\varepsilon)}$ and $\alpha'_{\texttt{T}} = \nicefrac{\alpha_{\texttt{T}}}{\omega}$.  %

Since \OSDMT includes a tracking cost on the purchasing side (i.e., rather than a switching cost), we modify the pseudo-cost minimization problem solved by the base demand driver and the flexible demand driver (just on the purchasing side) accordingly.  In \autoref{alg:pcm-subroutine}, we replace the input of a switching coeff. $\gamma$ with inputs of a tracking coeff. $\eta$ and target $a_t$ in line 1. 
We replace line 2 with:
\begin{align*}
\text{\textbf{define }} \textit{pseudo-target } \hat{a}^{(i)}_{t} \gets \frac{ a_t \cdot d^{(i)} }{\sum_{i\in \mathcal{B} \cup \mathcal{F}} d^{(i)} }
\end{align*}
and replace line 3 with:
\begin{align*}
\text{\textbf{solve }} x^{(i)}_t \gets \text{argmin}_{x \in [0, d^{(i)} - w_{t-1}^{(i)}]} \ p_t x + \eta \vert x - \hat{a}^{(i)}_{t} \vert - \Phi^{(i)}(w^{(i)}_{t-1}, w^{(i)}_{t-1}+x)
\end{align*}
}
\end{definition}

\begin{proof}

The result largely follows by the proof of \autoref{thm:osdm-upper-bound-1}, with modifications to account for the tracking cost instead of the switching cost on the grid side.
Thus, to show the result, we state several lemmas that replace the corresponding logic in the proof of \autoref{thm:osdm-upper-bound-1}.
We start by stating a lower bound on the cost of the offline optimal solution $\OPT$ for \OSDMT, analogous to \sref{Lemma}{lem:osdm-optimal-lower-bound-1}.

\begin{lemma}
\label{lem:osdmt-optimal-lower-bound-2}
Given that \PAAD produces $n$ active periods, let $\beta_i$ denote the asset purchased towards the base demand by the offline optimal solution during the $i^\text{th}$ active period, $i \in [n]$, and let $\tilde{p}$ denote the minimum price during inactive periods.  Then $\OPT(\mathcal{I})$ is lower bounded as:
\begin{align*}
\OPT(\mathcal{I}) &\geq \sum_{i=1}^n G_i(\beta_i) + \left( D_b - \sum_{i=1}^n \beta_i \right) \tilde{p} + \varepsilon \bar{p} D_b + \sum_{\tau=1}^T \left(1 + \varepsilon \right) H_\tau (f_\tau) + \left( D_b + D_f \right) \frac{2\delta}{T}
\end{align*}
\end{lemma}
\begin{proof}
The proof follows by the same logic as \sref{Lemma}{lem:osdm-optimal-lower-bound-1}, except instead of a switching cost on the grid side (i.e., with coefficient $\gamma$), there is a tracking cost with coefficient $\eta$.  The switching cost on the demand side remains the same (with coefficient $\delta$).  The worst case for \PAAD is when \OPT pays zero tracking cost, yielding that the total tracking \& switching cost incurred by the optimal solution is at least $\left( D_b + D_f \right) \frac{2\delta}{T}$.
\end{proof}

\begin{lemma}
\label{lem:osdmt-paad-upper-bound-2}
Given that \PAAD produces $n$ active periods, let $\hat{w}_b^{(i,j)}$ denote the amount of the $j^\text{th}$ base demand driver that has been purchased by the end of the $i^\text{th}$ active period, $j \in [\nu_i]$, and let $B_b^{(i,j)}$ denote the total demand associated with the $j^\text{th}$ base demand driver in the $i^\text{th}$ active period.  Further, let $\hat{w}_f^{(\tau)}$ and $\hat{v}_f^{(\tau)}$ denote the purchasing and delivery amounts of the $\tau^\text{th}$ flexible demand driver before the deadline, respectively.  Then the cost incurred by \PAAD is upper bounded as:
\begin{align}
\PAAD(\mathcal{I}) &\leq \sum_{i=1}^n \sum_{j=1}^{\nu_i} \int_0^{\hat{w}_b^{(i,j)}} \kern-1em \phi_b^{(i,j)}(u) du + \sum_{\tau=1}^T \left( \int_0^{\hat{w}_f^{(\tau)}} \kern-1em \phi_f^{(\tau)}(u) du + \int_0^{\hat{v}_f^{(\tau)}} \kern-1em \psi^{(\tau)}(z) dz \right) \label{eq:paad-upper-bound-1-tweaked} \\
& \quad + \left( D_b + D_f - \sum_{i=1}^n \sum_{j=1}^{\nu_i} \hat{w}_b^{(i,j)} - \sum_{\tau=1}^{T} \hat{w}_f^{(\tau)}\right) (p_{\max} + 2 \eta) \label{eq:paad-upper-bound-2-tweaked}\\
& \quad + D_b \left( \bar{p} (c+\varepsilon) + 2\delta \right) + \left(D_f - \sum_{\tau=1}^{T} \hat{v}_f^{(\tau)}\right) (p_{\max} (c + \varepsilon) + 2\delta) \\
& \quad - c p_{\min} \left( \sum_{i=1}^n \sum_{j=1}^{\nu_i} \hat{w}_b^{(i,j)} + \sum_{\tau=1}^{T} \hat{w}_f^{(\tau)}\right) + \hat{s} p_{\max}.
\end{align}
\end{lemma}
\begin{proof}
The proof follows by the same logic as \sref{Lemma}{lem:osdm-paad-upper-bound-1}---the main change is with respect to \eqref{eq:paad-upper-bound-2-tweaked}, to capture the worst-case purchasing cost incurred by \PAAD for purchasing the remaining demand that has not been purchased by the drivers.

First, note that when the threshold functions $\phi_b$ and $\phi_f$ are defined according to \sref{Corollary}{cor:osdmt-upper-bound-2}, the pseudo-cost terms in \eqref{eq:paad-upper-bound-1-tweaked} capture the worst-case \textit{tracking} cost incurred by \PAAD due to deviations from the signal wherever demand \textit{has} been purchased by the drivers.

Finally, any additional purchasing needed to satisfy demand is exactly captured by the difference between $D_b + D_f$ (the total base and flexible demand) and the cumulative purchasing by the drivers $\sum_{i=1}^n \sum_{j=1}^{\nu_i} \hat{w}_b^{(i,j)} + \sum_{\tau=1}^{T} \hat{w}_f^{(\tau)}$.  In the worst case, this additional purchasing is done at the highest price $p_{\max}$, and incurs a tracking cost of $2\eta$ (capturing the case where purchasing is simultaneously not done when the tracking target is non-zero and is done when the tracking target is zero). This gives \eqref{eq:paad-upper-bound-2-tweaked}, and completes the proof.
\end{proof}

Due to the redefinition of the threshold functions $\phi_b$, $\phi_f$ and $\psi_f$ in \sref{Corollary}{cor:osdmt-upper-bound-2}, we have the following technical lemmas to relate the threshold functions and the optimal cost:
\begin{lemma}
\label{lem:osdmt-threshold-function-relation-1}
By the (re-)definition of the threshold function $\phi_b^{(i,j)}(\cdot)$ in \sref{Corollary}{cor:osdmt-upper-bound-2}, the following relation always holds:
{\small 
\begin{align*}
\int_0^{w} \kern-1em \phi_b(u) du + (1-w) (p_{\max}+2\eta) + p_{\max}(c+\varepsilon) + 2\delta - c w p_{\min} = \alpha_{\texttt{T}}\left[ \phi_b(w) - 2\eta + \varepsilon p_{\max} + \frac{2\delta}{T}\right] \ & \forall w \in [0,1].
\end{align*}
}
\end{lemma}

\begin{lemma}
\label{lem:osdmt-threshold-function-relation-2}
By the (re-)definitions of the threshold functions $\phi_f^{(\tau)}(\cdot)$ and $\psi^{(\tau)}(\cdot)$ in \sref{Corollary}{cor:osdmt-upper-bound-2}, the following relation always holds:
{\small
\begin{align*}
\int_0^{w} \kern-1em \phi_f(u) du + (1-w) (p_{\max}+2\eta) - c w p_{\min} + \int_0^{v} \kern-1em \psi_f(z) dz + (1-v) (p_{\max}(c+\varepsilon) + 2\delta)  =\\
\alpha'_{\texttt{T}}\left[ \phi_f(w) + \psi_f(v) - 2(\eta + \delta) + \frac{2\delta (1+c+\varepsilon)}{T (1+\varepsilon) }\right] \ & \forall w \in [0,1], v \in [0,w],
\end{align*}
}
\end{lemma}

\noindent Using the results in \sref{Lemmas}{lem:osdmt-optimal-lower-bound-2}, \ref{lem:osdmt-paad-upper-bound-2}, \ref{lem:osdmt-threshold-function-relation-1}, and \ref{lem:osdmt-threshold-function-relation-2}, we claim that the following holds:
\begin{align*}
\frac{\PAAD(\mathcal{I}) - p_{\max} \hat{s}}{\OPT(\mathcal{I})} &\leq \alpha_{\texttt{T}}.
\end{align*}

To show this result, we first substitute the bounds from \sref{Lemmas}{lem:osdmt-optimal-lower-bound-2} and \ref{lem:osdmt-paad-upper-bound-2} into the left-hand side of the above equation.  As in the proof of \autoref{thm:osdm-upper-bound-1}, we define some shorthand notation to facilitate the presentation.

\noindent Let $Q = \sum_{i=1}^n \sum_{j=1}^{\nu_i} \int_0^{\hat{w}_b^{(i,j)}} \kern-1em \phi_b^{(i,j)}(u) du + \sum_{\tau=1}^T \left( \int_0^{\hat{w}_f^{(\tau)}} \kern-1em \phi_f^{(\tau)}(u) du + \int_0^{\hat{v}_f^{(\tau)}} \kern-1em \psi^{(\tau)}(z) dz \right)$ denote the integrals over the thresholds.
Let $\hat{W}_b = \sum_{i=1}^n \sum_{j=1}^{\nu_i} \hat{w}_b^{(i,j)}$, let $\hat{W}_f = \sum_{\tau=1}^{T} \hat{w}_f^{(\tau)}$, and let $\hat{W} = \hat{W}_b + \hat{W}_f$ denote the total purchasing by all drivers.  

\noindent Further, let $\hat{V}_b = \sum_{i=1}^n \sum_{j=1}^{\nu_i} \hat{v}_b^{(i,j)}$ denote the total delivery by all flexible demand drivers.

\noindent Let ${\bm \beta} = \sum_{i=1}^n \beta_i$ denote the total amount of asset purchased towards the base demand by the optimal solution during all active periods, noting that $D_b - {\bm \beta} \geq 0$ by definition.   Finally, let $D = D_b + D_f$ denote the total demand.  Substituting the bounds from \sref{Lemmas}{lem:osdmt-optimal-lower-bound-2} and \ref{lem:osdmt-paad-upper-bound-2} into the left-hand side of the above inequality, we have:

{\small
\begin{align*}
\frac{Q\!+\!({\bm\beta}\!+\!D_f\!-\!\hat{W}) (p_{\max}\!+\!2 \eta)\!+\!{\bm \beta} \left( \bar{p} (c\!+\!\varepsilon)\!+\!2\delta \right) + (D_f\!-\!\hat{V}_f) (p_{\max} (c\!+\!\varepsilon)\!+\!2\delta)\!-\!c p_{\min} \hat{W}\!+\!(D_b\!-\!{\bm \beta}) (p_{\max}\!+\!2(\eta\!+\!\delta)\!+\!\bar{p}(c\!+\!\varepsilon))}{\sum_{i=1}^n G_i(\beta_i)\!+\!\varepsilon \bar{p} {\bm \beta}\!+\!\sum_{\tau=1}^T (1\!+\!\varepsilon) H_\tau (f_\tau)\!+\!(D_f\!+\!{\bm \beta}) \frac{2\delta}{T}\!+\!(D_b\!-\!{\bm \beta})(\tilde{p}\!+\!\varepsilon \bar{p}\!+\!\frac{2\delta}{T}) } \\
\end{align*}
}
Then, we have the following:
{\small
\begin{align*}
\leq \max \Bigg \{ \frac{Q + ({\bm\beta}+D_f - \hat{W}) (p_{\max} + 2 \eta) + {\bm\beta} \left( \bar{p} (c+\varepsilon) + 2\delta \right) + (D_f - \hat{V}_f) (p_{\max} (c + \varepsilon) + 2\delta) - c p_{\min} \hat{W}}{\sum_{i=1}^n G_i(\beta_i) + \varepsilon \bar{p} {\bm\beta} + \sum_{\tau=1}^T (1 + \varepsilon) H_\tau (f_\tau) + (D_f + {\bm\beta}) \frac{2\delta}{T} }, \ \ &\\
\quad \frac{(D_b\!-\!{\bm\beta}) (p_{\max}\!+\!2(\eta\!+\!\delta)\!+\!\bar{p}(c+\varepsilon))}{(D_b - {\bm\beta})(\tilde{p} + \varepsilon \bar{p} + \frac{2\delta}{T}) } \Bigg \}, & \\
\end{align*}
}
where the definition of $\tilde{p}$ ensures that the second term in the $\max$ is at most $\alpha_{\texttt{T}}$.  We now focus on the first term.  For the sake of contradiction, suppose that
{\small
\begin{align}
    \frac{Q + ({\bm\beta}+D_f - \hat{W}) (p_{\max} + 2 \eta) + {\bm \beta} \left( \bar{p} (c+\varepsilon) + 2\delta \right) + (D_f - \hat{V}_f) (p_{\max} (c + \varepsilon) + 2\delta) - c p_{\min} \hat{W}}{\sum_{i=1}^n G_i(\beta_i) + \varepsilon \bar{p} {\bm \beta} + \sum_{\tau=1}^T (1 + \varepsilon) H_\tau (f_\tau) + (D_f + {\bm \beta}) \frac{2\delta}{T} } > \alpha_{\texttt{T}}. \label{eq:ratio-term-2}
\end{align}
}
Instead of working directly with the expression in terms of ${\bm \beta}$, we first reason about how the cost of \OPT and \PAAD relate to one another in terms of $\sum_{i=1}^n \sum_{j=1}^{\nu_i} B_b^{(i,j)}$, the total demand assigned to base drivers.
We introduce the following notation for the sake of brevity: let ${\bm B} = \sum_{i=1}^n \sum_{j=1}^{\nu_i} B_b^{(i,j)}$ denote the total base demand.  Then, we have the following relation:
{\small
\begin{align*}
& \frac{Q + ({\bm B}+D_f - \hat{W}) (p_{\max} + 2 \eta) + {\bm B} \left( \bar{p} (c+\varepsilon) + 2\delta \right) + (D_f - \hat{V}_f) (p_{\max} (c + \varepsilon) + 2\delta) - c p_{\min} \hat{W}}{\sum_{i=1}^n G_i(B_i) + \varepsilon \bar{p} {\bm B} + \sum_{\tau=1}^T (1 + \varepsilon) H_\tau (f_\tau) + (D_f + {\bm B}) \frac{2\delta}{T} }, \\
& \geq \frac{Q\!+\!({\bm \beta }\!+\!D_f\!-\!\hat{W}) (p_{\max}\!+\!2 \eta)\!+\!{\bm \beta } \left( \bar{p} (c\!+\!\varepsilon)\!+\!2\delta \right)\!+\!(D_f\!-\!\hat{V}_f) (p_{\max} (c\!+\!\varepsilon)\!+\!2\delta)\!-\!c p_{\min} \hat{W}\!+\!({\bm B}\!-\!{\bm \beta}) (p_{\max}\!+\!\bar{p}(c\!+\!\varepsilon)\!+\!2 \kappa)}{\sum_{i=1}^n G_i(\beta_i)\!+\!\varepsilon \bar{p} {\bm \beta }\!+\!\sum_{\tau=1}^T (1\!+\!\varepsilon) H_\tau (f_\tau)\!+\!(D_f\!+\!{\bm \beta }) \frac{2\delta}{T}\!+\!({\bm B}\!-\!{\bm \beta}) (\tilde{p}\!+\!\varepsilon \bar{p}\!+\!\frac{2\delta}{T})}, \\
& > \alpha_{\texttt{T}}.
\end{align*}
}
where in the second inequality, we have used \sref{Lemma}{lem:upper-bound-B-beta-tilde-p}. We proceed to work with the expression in terms of ${\bm B}$.
During the $i^\text{th}$ active period and the lifetime of the $j^\text{th}$ base demand driver, the minimum marginal purchasing price observed is given by $\phi_b^{(i,j)}(\hat{w}_b^{(i,j)}) - 2\eta$ by the definition of the threshold function.  Similarly, during the lifetime of the $\tau^\text{th}$ flexible demand driver, the minimum marginal purchasing and delivery cost observed is given by $\frac{1+\varepsilon}{1+c+\varepsilon} \left( \phi_f^{(\tau)}(\hat{w}_f^{(\tau)}) + \psi_f^{(\tau)}(\hat{v}_f^{(\tau)}) - 2 (\eta+\delta) \right)$.  

\noindent This gives the following lower bounds on the terms that depend on $G_i$ and $H_\tau$, respectively:
{\small
\begin{align*}
\sum_{i=1}^n G_i\left( \sum_{j=1}^{\nu_i} B_b^{(i,j)} \right) &\geq \sum_{i=1}^n \sum_{j=1}^{\nu_i} \left( \phi_b^{(i,j)}(\hat{w}_b^{(i,j)}) - 2\eta \right) \times B_b^{(i,j)} \\
\sum_{\tau=1}^T (1 + \varepsilon) H_\tau (f_\tau) &\geq \sum_{\tau=1}^T \frac{1+\varepsilon}{1+c+\varepsilon} \left( \phi_f^{(\tau)}(\hat{w}_f^{(\tau)}) + \psi_f^{(\tau)}(\hat{v}_f^{(\tau)}) - 2 (\eta+\delta) \right) \times f_\tau.
\end{align*}
}

\noindent Substituting these bounds into the previous expression, we have that the left-hand-side of \eqref{eq:ratio-term-2} is less than or equal to:
{\small
\begin{align*}
&\leq \frac{Q + ({\bm B}+D_f - \hat{W}) (p_{\max} + 2 \eta) + {\bm B} \left( \bar{p} (c+\varepsilon) + 2\delta \right) + (D_f - \hat{V}_f) (p_{\max} (c + \varepsilon) + 2\delta) - c p_{\min} \hat{W}}{\sum_{i=1}^n \sum_{j=1}^{\nu_i} \left( \phi_b^{(i,j)}(\hat{w}_b^{(i,j)})\!-\!2\eta \right) B_b^{(i,j)}\!+\! \varepsilon \bar{p} {\bm B}\!+\!\sum_{\tau=1}^T \frac{1+\varepsilon}{1+c+\varepsilon} \left( \phi_f^{(\tau)}(\hat{w}_f^{(\tau)})\!+\!\psi_f^{(\tau)}(\hat{v}_f^{(\tau)})\!-\!2 (\eta\!+\!\delta) \right) f_\tau\!+\!(D_f\!+\!{\bm B}) \frac{2\delta}{T} }.
\end{align*}
}

\noindent By rearranging the terms in the above and substituting for $Q$, we obtain the following:
{\small
\begin{align}
&= \frac{\sum_{i=1}^n \sum_{j=1}^{\nu_i} \left[ X_b^{(i,j)} \right] + \sum_{\tau=1}^T \left[ X_f^{(\tau)} \right]}{\sum_{i=1}^n \sum_{j=1}^{\nu_i} \left[ \left( \phi_b^{(i,j)}(\hat{w}_b^{(i,j)})\!-\!2\eta\!+\!\varepsilon \bar{p}\!+\!\frac{2\delta}{T}\right) B_b^{(i,j)}  \right]\!+\!\sum_{\tau=1}^T \left[ \frac{1+\varepsilon}{1+c+\varepsilon} \left( \phi_f^{(\tau)}(\hat{w}_f^{(\tau)})\!+\!\psi_f^{(\tau)}(\hat{v}_f^{(\tau)})\!-\!2 (\eta\!+\!\delta) \right) f_\tau\!+\!f_\tau \frac{2\delta}{T} \right]}, \label{eq:partitioned-sums-upper-bound-2}
\end{align}
}

\noindent where $X_b^{(i,j)}$ and $X_f^{(\tau)}$ are defined as follows (for each driver):
{\small
\begin{align*}
X_b^{(i,j)} &= \int_0^{\hat{w}_b^{(i,j)}} \kern-1em \phi_b^{(i,j)}(u) du + (B_b^{(i,j)} -\hat{w}_b^{(i,j)}) (p_{\max}+2\eta) + \left( p_{\max}(c+\varepsilon) + 2\delta \right) B_b^{(i,j)} - c \hat{w}_b^{(i,j)} p_{\min}, \\
X_f^{(\tau)} &= \int_0^{\hat{w}_f^{(\tau)}} \kern-1em \phi_f^{(\tau)}(u) du + (f_\tau -\hat{w}_f^{(\tau)}) (p_{\max}+2\eta) + \int_0^{\hat{v}_f^{(\tau)}} \kern-1em \psi_f^{(\tau)}(u) du + \left( p_{\max} (c+\varepsilon) + 2\delta \right)(f_\tau - \hat{v}_f^{(\tau)}) - c \hat{w}_f^{(\tau)} p_{\min}.
\end{align*}
}

\noindent Since \eqref{eq:partitioned-sums-upper-bound-2} is $> \alpha_{\texttt{T}}$ by assumption, one of the following cases must be true:

\smallskip
\noindent\textbf{Case I:} There exists some $i \in [n]$ and $j \in [\nu_i]$ such that
{\small
\begin{align*}
X_b^{(i,j)} &> \alpha_{\texttt{T}}  B_b^{(i,j)} \left( \phi_b^{(i,j)}(\hat{w}_b^{(i,j)}) - 2\eta + \varepsilon \bar{p} + \frac{2\delta}{T}\right).
\end{align*}
}

\smallskip
\noindent\textbf{Case II:} There exists some $\tau \in [T]$ such that
{\small
\begin{align*}
X_f^{(\tau)} &> \alpha'_{\texttt{T}} f_\tau \left[ \frac{1+\varepsilon}{1+c+\varepsilon} \left( \phi_f^{(\tau)}(\hat{w}_f^{(\tau)}) + \psi_f^{(\tau)}(\hat{v}_f^{(\tau)}) - 2 (\eta+\delta) \right) + \frac{2\delta}{T} \right],
\end{align*}
}

\noindent But, if either of these two cases are true (i.e., for any unit of base or flexible demand), they contradict \sref{Lemmas}{lem:osdmt-threshold-function-relation-1} and \ref{lem:osdmt-threshold-function-relation-2}, respectively.  Thus, we have a contradiction, and the original assumption that \eqref{eq:partitioned-sums-upper-bound-2} is $> \alpha_{\texttt{T}}$ must be false.  This completes the proof that $\frac{\PAAD(\mathcal{I}) - p_{\max} \hat{s}}{\OPT(\mathcal{I})} \leq \alpha_{\texttt{T}}$. 

\noindent Using this result, we have:
\begin{align*}
\PAAD(\mathcal{I}) &\leq \alpha_{\texttt{T}} \OPT(\mathcal{I}) + p_{\max} \hat{s} \leq \alpha_{\texttt{T}} \OPT(\mathcal{I}) + p_{\max} S,
\end{align*}
where $p_{\max} S$ is a constant.  This shows that \PAAD is $\alpha_{\texttt{T}}$-competitive under \sref{Definition}{dfn:comp-ratio}, which completes the proof.

\end{proof}

\section{Deferred Proofs from \autoref{sec:fundamental-limits}} \label{apx:lowerbounds}
In this section, we provide full proofs for the results in \autoref{sec:fundamental-limits}, which establishes lower bounds on the competitive ratio achievable by deterministic online algorithms for several instantiations of \OSDM.
We start by proving \autoref{thm:osdm-lower-bound-1}, which shows a lower bound on the best competitive ratio achievable by any deterministic online algorithm for \OSDMS.  Then, we prove \sref{Corollary}{cor:osdm-lower-bound-base-demand}, which shows a lower bound on the best competitive ratio achievable for \OSDMS with only base demand, before proving \autoref{thm:osdm-decreasing-worse-increasing}, which shows that the decreasing delivery cost case is strictly harder than the increasing delivery cost case within the class of monotone delivery costs defined by \sref{Def.}{def:osdm-delivery-pd-monotone}.  Finally, we prove \autoref{thm:osdmt-lower-bound-2}, which shows a lower bound on the best competitive ratio achievable by any deterministic online algorithm for \OSDMT.

\subsection{Proof of \autoref{thm:osdm-lower-bound-1}} \label{apx:osdm-lower-bound-1}

In this section, we prove \autoref{thm:osdm-lower-bound-1}, which states that no deterministic online algorithm can achieve a competitive ratio better than $\alpha$ (defined in \eqref{eq:alpha1}) for \OSDMS.  We start by defining a family of difficult instances for \OSDMS, before proceeding to prove the lower bound on the competitive ratio achievable by any deterministic online algorithm.

\begin{proof}
To show this result, we first define a family of instances, and then show that the competitive ratio of any deterministic algorithm is lower bounded under these instances.  
Consider the following set of instances $\{\mathcal{I}_{x}\}_{x\in [p_{\min}, p_{\max}]}$, where $\mathcal{I}_{x}$ is called an \emph{$x$-decreasing instance}.

\begin{definition}[$x$-decreasing instance for \OSDMS]
\label{def:x-decreasing-instance-switching}
Let $n, m \in \mathbb{N}$ be sufficiently large, and denote $\varrho \coloneqq \nicefrac{p_{\max} - p_{\min}}{n}$. 

For $x \in [p_{\min}, p_{\max}]$, we define a decreasing market price sequence as follows: 
The sequence is partitioned into $n_x := 2 \cdot \lceil \nicefrac{(x - p_{\min})}{\varrho} \rceil + 1$ alternating batches of prices. The $i^{\text{th}}$ batch ($i\in [n_x-2]$) contains $m$ prices $p_{\max}$ if $i$ is odd, and $1$ cost function with coefficient $p_{\max}-(\lceil i / 2 \rceil)\varrho$ if $i$ is even. The last two batches consist of $m$ cost functions with coefficient $x + \iota$, followed by a final batch of $1$ cost function with coefficient $p_{\max}$.

The demand sequence is defined as follows:  At $t=1$, a batch of flexible demand $f_1 = 1$ arrives with deadline $\Delta_1 = T$.  All other $b_t$ and $f_t$ are zero.  

\end{definition}

Let $h(x)$ and $z(x)$ denote \emph{conversion functions} that both map $[p_{\min}, p_{\max}] \to [0,1]$, respectively.  
Suppose these arbitrary functions fully describe the actions of a deterministic \ALG for \OSDM on an instance $\mathcal{I}_{x}$.  Specifically, suppose that $h(x)$ describes \ALG's purchasing decisions, and $z(x)$ describes \ALG's delivery decisions \emph{before} the arrival of the last batch of prices ($p_{\max}$) in the instance.  Note that by definition, $z(x) \leq h(x)$, since the delivery cannot exceed what has been purchased thus far.
Note that for large $n$, processing an instance $\mathcal{I}_{x-\varrho}$ is equivalent to first processing $\mathcal{I}_{x}$ (besides the last two batches), and then processing batches with prices $x - \varrho$ and $p_{\max}$.

Since \ALG is deterministic and both conversions (i.e., purchasing and delivery) are both unidirectional (irrevocable), we must have that $h(x-\varrho) \geq h(x)$ and $z(x-\varrho) \geq z(x)$, i.e., $h(x)$ and $z(x)$ are both non-increasing in $[p_{\min}, p_{\max}]$.  Intuitively, the entire demand should be purchased before the end of the sequence if the lowest price appears, i.e., $h(p_{\min}) = 1$.  Furthermore, since the lowest price appears at least twice, 

\ALG pays switching cost proportional to $h(x)$ and $z(x)$, because each good price is ``interrupted'' by bad prices in the $\mathcal{I}_x$ instance, captured by $2\gamma h(x)$ and $2\delta z(x)$, respectively.  However, if \ALG is forced to purchase at the deadline $T$, it also pays a switching cost of $2\gamma (1-h(x))$ and $2\delta (1-z(x))$, respectively.  Thus, the total worst-case switching cost paid by \ALG is $2\gamma + 2\delta$.

Recall that in the price-dependent decreasing delivery cost case, the worst case delivery coefficient is $c + \varepsilon \leq 1$, when the storage is empty.  We bound the total \emph{effective} delivery cost coefficient (i.e., across all time steps) as follows:
\[
\sum_{t \in [T]} (c - cs_t + \varepsilon) = c + \varepsilon - \sum_{t \in [T]} c s_t \leq c + \varepsilon - c h(x),
\]
where we have used the fact that for any \ALG, $\sum_{t\in[T]} s_t \leq h(x)$.  Since the delivery cost is price-dependent, the worst-case for \ALG occurs when the \emph{savings} in the delivery cost are realized at the lowest price, i.e., $p_{\min} c h(x)$.
Then the total cost of \ALG on instance $\mathcal{I}_{x}$ is described by:
\begin{align*}
\ALG(\mathcal{I}_{x}) &= h(\ell) \ell - \int_{\ell}^x udh(u) + (1 - h(x))p_{\max} + 2 \gamma - p_{\min} c h(x) + \\
&( c + \varepsilon ) \ell z(\ell) - (c+\varepsilon) \int_{\ell}^x udz(u) + (1 - z(x))(c+\varepsilon)p_{\max} + 2\delta.
\end{align*}

\noindent The optimal cost on instance $\mathcal{I}_{x}$ (if all prices are known in advance) is given by:
\begin{align*}
\OPT(\mathcal{I}_{x}) &= x + \frac{2\gamma}{T} + \varepsilon x + \frac{2\delta}{T}.
\end{align*}

\noindent If $\ALG$ is $\alpha^\star$-competitive, then a necessary condition is that for all $x \in [p_{\min}, p_{\max}]$, $\ALG(\mathcal{I}_x) \leq \alpha^\star \OPT(\mathcal{I}_x)$. 
This imposes a necessary condition on $h(x)$ and $z(x)$, which can be expressed as the following differential inequality:
\begin{align*}
h(\ell) \ell - \int_{\ell}^x udh(u) + (1 - h(x))p_{\max} + 2 \gamma - p_{\min} c h(x) +& \\
( c + \varepsilon ) \ell z(\ell) - (c+\varepsilon) \int_{\ell}^x udz(u) + (1 - z(x))(c+\varepsilon)p_{\max} + 2\delta & \\
\leq \alpha^\star \left( x + \frac{2\gamma}{T} + \varepsilon x + \frac{2\delta}{T} \right)&.
\end{align*}

\noindent Recall that $z(x) \leq h(x)$.  On the instances described above, under the assumption that $\kappa = \gamma + \delta \leq \frac{p_{\max} - p_{\min}}{2}$ (see \autoref{sec:assumptions}), note that for sufficiently good $x$ (i.e., $x < \ell$) it is strictly better for $\ALG$ to deliver as much as possible, i.e., $z(x) = h(x)$, reducing its last-minute delivery cost.  Thus, we proceed by assuming that $z(x) = h(x)$.  Using integration by parts and letting $\kappa = \gamma + \delta$, this gives the following condition on $h(x)$:
\begin{align*}
h(x) x - \int_{\ell}^x h(u) du + (1 - h(x))p_{\max} + 2 \kappa - p_{\min} c h(x) +& \\
( c + \varepsilon ) x h(x) - (c+\varepsilon) \int_{\ell}^x h(u) du + (1 - h(x))(c+\varepsilon)p_{\max} & \\
\leq \alpha^\star \left( x + \frac{2\kappa}{T} + \varepsilon x \right)&.
\end{align*}

\noindent To solve for $\alpha^\star$, we can use Grönwall's inequality~\cite[Theorem 1, p. 356]{Mitrinovic:91}. First, observe that $h(x)$ must satisfy the following:
\begin{align*}
h(x) \geq \frac{p_{\max} + \frac{2\kappa}{(1+c+\varepsilon)} - \alpha^\star \frac{x + \frac{2\kappa}{T} + \varepsilon x }{1+c+\varepsilon}}{p_{\max} -x + \frac{p_{\min} c}{1+c+\varepsilon}} - \frac{1}{p_{\max} -x  + \frac{p_{\min} c}{1+c+\varepsilon}} \int_\ell^x h(u) du.
\end{align*}

\noindent By Grönwall's inequality~\cite[Theorem 1, p. 356]{Mitrinovic:91}, it follows that:
\begin{align*}
h(x) &\geq \frac{p_{\max} + \frac{2\kappa}{(1+c+\varepsilon)} - \alpha^\star \frac{x + \frac{2\kappa}{T} + \varepsilon x }{1+c+\varepsilon}}{p_{\max} -x + \frac{p_{\min} c}{1+c+\varepsilon}} - \int_\ell^x \frac{p_{\max} + \frac{2\kappa}{(1+c+\varepsilon)} - \alpha^\star \frac{u + \frac{2\kappa}{T} + \varepsilon u }{1+c+\varepsilon}}{\left( p_{\max} -u  + \frac{p_{\min} c}{1+c+\varepsilon} \right)^2} du,\\ 
h(x) &\geq \alpha^\star \frac{1+\varepsilon}{1+c+\varepsilon} \ln \left[ \frac{(1+c+\varepsilon)(p_{\max}-x) + p_{\min} c}{(1+c+\varepsilon)(p_{\max}-\ell)  + p_{\min} c} \right].
\end{align*}

\noindent Recall that we have a boundary condition that $h(p_{\min}) = 1$---we can combine this with the above to obtain:
\begin{align*}
1 &\geq \alpha^\star \frac{1+\varepsilon}{1+c+\varepsilon} \ln \left[ \frac{(1+c+\varepsilon)(p_{\max}-p_{\min}) + p_{\min} c}{(1+c+\varepsilon)(p_{\max}-\ell) + p_{\min} c} \right].
\end{align*}

\noindent Letting $\ell = \frac{(1+c+\varepsilon) p_{\max} + 2\kappa}{\alpha^\star(1+\varepsilon)} - \frac{2\kappa}{T (1+\varepsilon)}$, the optimal $\alpha^\star$ is obtained when the above inequality is binding, so we have:
\begin{align*}
1 &= \alpha^\star \frac{1+\varepsilon}{1+c+\varepsilon} \ln \left[ \frac{(1+c+\varepsilon)p_{\max} - (1+\varepsilon) p_{\min}}{(1+c+\varepsilon)p_{\max} - \frac{(1+c+\varepsilon)\left[ (1+c+\varepsilon) p_{\max} + 2 \kappa \right]}{\alpha^\star(1+\varepsilon)} + \frac{(1+c+\varepsilon) 2\kappa}{T (1+\varepsilon)} + p_{\min} c} \right].
\end{align*}
Solving the above yields that the optimal $\alpha^\star$ for any \ALG solving \OSDM is lower bounded by:
{\color{blue}
\begin{align*}
\alpha^\star \geq \frac{\omega}{\left[ W\left(-\frac{\left( (1+c+\varepsilon) p_{\max} - (1+\varepsilon) p_{\min} \right) e^{\frac{-\nicefrac{\omega 2\kappa}{T} - p_{\min} c - (1+c+\varepsilon) p_{\max}}{(1+c+\varepsilon)p_{\max} + 2\kappa}}}{(1+c+\varepsilon)p_{\max} + 2\kappa}\right)+\frac{(1+c+\varepsilon) p_{\max} + p_{\min} c - \nicefrac{\omega 2\kappa}{T}}{(1+c+\varepsilon)p_{\max} + 2\kappa} \right] },
\end{align*}}
where $\kappa = \gamma + \delta$ and $\omega = \frac{(1+c+\varepsilon)}{(1+\varepsilon)}$.  This completes the proof.

\end{proof}

\subsection{Proof of \sref{Corollary}{cor:osdm-lower-bound-base-demand}} \label{apx:osdm-lower-bound-base-demand}

In this section, we prove \sref{Corollary}{cor:osdm-lower-bound-base-demand}, which states that no deterministic online algorithm can achieve a competitive ratio better than $\alpha_{\texttt{B}}$ (defined in \eqref{eq:alphaB}) for \OSDMS instances that only contain base demand.  
To do so, we define a tweaked set of difficult instances for \OSDMS with only base demand, and use a similar conversion function argument to bound the best competitive ratio.

\begin{proof}
To show the result, we start by defining a set of instances similar to those defined in \sref{Def.}{def:x-decreasing-instance-switching}, where the instance contains only base demand.

\begin{definition}[$x$-decreasing instance for \OSDMS with just base demand]
\label{def:x-decreasing-instance-switching-just-base}
Recall the price sequence defined in \sref{Def.}{def:x-decreasing-instance-switching}.  We define a similar set of instances $\{\mathcal{I}'_{x}\}_{x\in [p_{\min}, p_{\max}]}$, where $\mathcal{I}'_{x}$ is called an \emph{$x$-decreasing instance with just base demand}.

The demand sequence is defined as follows:  At $t=T$, a batch of base demand $b_T = 1$ arrives.  All other $b_t$ and $f_t$ are zero.  

\end{definition}

We now show that the competitive ratio of any deterministic online algorithm \ALG on the above instances is lower bounded by $\alpha_{\texttt{B}}$ defined in \eqref{eq:alphaB}.  We assume without loss of generality that $\ALG$ knows the value of $b_T = 1$ in advance, since this can only help \ALG.  

Let $h(x)$ denote a \emph{conversion function} that maps $[p_{\min}, p_{\max}] \to [0,1]$---this describes \ALG's purchasing decisions before the final batch of prices ($p_{\max}$) in the instance.  Since \ALG is deterministic and purchasing is unidirectional (irrevocable), we must have that $h(x-\varrho) \geq h(x)$, i.e., $h(x)$ is non-increasing in $[p_{\min}, p_{\max}]$.  Intuitively, the entire demand should be purchased before the end of the sequence if the lowest price appears, i.e., $h(p_{\min}) = 1$.

\ALG pays switching cost proportional to $h(x)$, because each good price is ``interrupted'' by bad prices in the $\mathcal{I}_x$ instance, captured by $2\gamma h(x)$.  If \ALG is forced to purchase when the demand arrives, it also pays a switching cost of $2\gamma (1-h(x))$.  Thus, the total worst-case switching cost on the purchasing side paid by \ALG is $2\gamma$.

The switching cost on the delivery side for both \ALG and \OPT is $2\delta$, since the entire base demand arrives at the end of the sequence, and it must be delivered at that time.

Recall that in the price-dependent decreasing delivery cost case, the worst case delivery coefficient is $c + \varepsilon \leq 1$, when the storage is empty.  We bound the total \emph{effective} delivery cost coefficient (i.e., across all time steps) as follows:
\[
\sum_{t \in [T]} (c - cs_t + \varepsilon) = c + \varepsilon - \sum_{t \in [T]} c s_t \leq c + \varepsilon - c h(x),
\]
where we have used the fact that for any \ALG, $\sum_{t\in[T]} s_t \leq h(x)$.  Since the delivery cost is price-dependent, the \emph{savings} in the delivery cost are realized at the price at the final time step, i.e., $p_{\max} c h(x)$.
Then the total cost of \ALG on instance $\mathcal{I}'_{x}$ is described by:
\begin{align*}
\ALG(\mathcal{I}'_{x}) &= h(\ell) \ell - \int_{\ell}^x udh(u) + (1 - h(x))p_{\max} + 2 \gamma - p_{\max} c h(x) + (c+\varepsilon)p_{\max} + 2\delta.
\end{align*}

\noindent The optimal cost on instance $\mathcal{I}'_{x}$ (if all prices are known in advance) is given by:
\begin{align*}
\OPT(\mathcal{I}'_{x}) &\geq x + \frac{2\gamma}{T} + \varepsilon p_{\max} + 2\delta.
\end{align*}

\noindent If $\ALG$ is $\alpha^\star$-competitive, then a necessary condition is that for all $x \in [p_{\min}, p_{\max}]$, $\ALG(\mathcal{I}_x) \leq \alpha^\star \OPT(\mathcal{I}_x)$. 
This imposes a necessary condition on $h(x)$, which can be expressed as the following differential inequality:
\begin{align*}
h(\ell) \ell - \int_{\ell}^x udh(u) + (1 - h(x))p_{\max} + 2 \gamma - p_{\max} c h(x) +(c+\varepsilon)p_{\max} + 2\delta & \\
\leq \alpha^\star \left( x + \frac{2\gamma}{T} + \varepsilon p_{\max} + 2\delta \right)&.
\end{align*}

\noindent Using integration by parts and letting $\kappa = \gamma + \delta$, this gives the following condition on $h(x)$:
\begin{align*}
h(x) \geq& \frac{p_{\max}(1+c+\varepsilon)+2\kappa - \alpha^\star \left( x + \frac{2\gamma}{T} + \varepsilon p_{\max} +2\delta \right)}{p_{\max} + p_{\max}c -x} - \frac{1}{p_{\max} + p_{\max}c - x} \int_\ell^x h(u) du.
\end{align*}

\noindent By Grönwall's inequality~\cite[Theorem 1, p. 356]{Mitrinovic:91}, it follows that:
{\small
\begin{align*}
h(x) &\geq \frac{p_{\max}(1\!+\!c\!+\!\varepsilon)+2\kappa - \alpha^\star \left( x + \frac{2\gamma}{T} + \varepsilon p_{\max} +2\delta \right)}{p_{\max} + p_{\max}c - x} - \int_\ell^x \frac{p_{\max}(1\!+\!c\!+\!\varepsilon)+2\kappa - \alpha^\star \left( u + \frac{2\gamma}{T} + \varepsilon p_{\max} +2\delta \right)}{ \left( p_{\max} + p_{\max}c - u \right)^2 } du,\\ 
h(x) &\geq \alpha^\star \ln \left[ \frac{p_{\max} + p_{\max}c -x}{p_{\max} + p_{\max}c -\ell} \right],
\end{align*}
where $\ell = \frac{p_{\max} (1+c+\varepsilon) + 2\kappa}{\alpha^\star} - \frac{2\gamma}{T} - p_{\max} \varepsilon - 2\delta$.
}

\noindent Recall that we have a boundary condition that $h(p_{\min}) = 1$---we can combine this with the above to obtain:
\begin{align*}
1 &\geq \alpha^\star \ln \left[ \frac{p_{\max} + p_{\max}c -p_{\min}}{p_{\max} + p_{\max}c -\ell} \right].
\end{align*}

\noindent Letting $\ell = \frac{(1+c+\varepsilon) p_{\max} + 2\kappa}{\alpha^\star} - \frac{2\kappa}{T} - p_{\max} \varepsilon - 2\delta$, the optimal $\alpha^\star$ is obtained when the above inequality is binding, so we have:
\begin{align*}
1 &= \alpha^\star \ln \left[ \frac{(1+c)p_{\max} - p_{\min}}{(1+c)p_{\max} - \frac{(1+c+\varepsilon) p_{\max} + 2 \kappa}{\alpha^\star} + \frac{2\gamma}{T} + p_{\max} \varepsilon + 2\delta} \right].
\end{align*}
Solving the above yields that the optimal $\alpha^\star$ for any \ALG solving \OSDMS with just base demand is lower bounded by:
\begin{align*}
\alpha^\star \geq \left[ W\left(-\frac{\left( (1+c)p_{\max} - p_{\min} \right) e^{ - \frac{(1+c+\varepsilon)p_{\max} +2\delta + \frac{2\gamma}{T}}{(1+c+\varepsilon)p_{\max} + 2\kappa} }}{(1+c+\varepsilon)p_{\max} + 2\kappa}\right)+\frac{(1+c+\varepsilon)p_{\max} +2\delta + \frac{2\gamma}{T}}{(1+c+\varepsilon)p_{\max} + 2\kappa} \right]^{-1},
\end{align*}
which completes the proof.
\end{proof}

\subsection{Proof of \autoref{thm:osdm-decreasing-worse-increasing}} \label{apx:osdm-decreasing-worse-increasing}

In this section, we prove \autoref{thm:osdm-decreasing-worse-increasing}, which states that amongst the class of monotone delivery costs defined in \sref{Def.}{def:osdm-delivery-pd-monotone}, the best achievable competitive ratios for \OSDMS with an \textit{increasing delivery cost} are strictly better than those for \OSDMS with a decreasing delivery cost (i.e., the cases considered in \autoref{thm:osdm-lower-bound-1} and \sref{Corollary}{cor:osdm-lower-bound-base-demand})

\begin{proof}

To show this result, we prove two lower bounds (one in the case of just base demand, and one in the case of just flexible demand) to compare against the lower bounds in \autoref{thm:osdm-lower-bound-1} and \sref{Corollary}{cor:osdm-lower-bound-base-demand}, respectively.
We start with the case of flexible demand and the following lemma:
\begin{lemma}\label{lem:lb-flexible-increasing}
There exists a set of \OSDMS instances with only flexible demand and a monotone increasing delivery cost such that no deterministic online algorithm \ALG can achieve a competitive ratio better than $\alpha_{\texttt{IF}}$, given by: 
\begin{align}
\alpha_{\texttt{IF}} \geq \left[ W\left( -\frac{\left( p_{\max} - p_{\min} \right) e^{-1}}{p_{\max} + \frac{2\kappa}{T (1+\varepsilon)}} \right) + 1 \right]^{-1}. \label{eq:alphaIF}
\end{align}
\end{lemma}
\begin{proof}
We recall \sref{Def.}{def:x-decreasing-instance-switching} and consider the same set of instances $\{\mathcal{I}_{x}\}_{x\in [p_{\min}, p_{\max}]}$, where $\mathcal{I}_{x}$ is called an \emph{$x$-decreasing instance}.

Let $h(x)$ and $z(x)$ denote \emph{conversion functions} that both map $[p_{\min}, p_{\max}] \to [0,1]$, respectively.  
Suppose these arbitrary functions fully describe the actions of a deterministic \ALG for \OSDM with a monotone increasing delivery cost on an instance $\mathcal{I}_{x}$.  
Specifically, $h(x)$ describes \ALG's purchasing decisions and $z(x)$ describes \ALG's delivery decisions \emph{before} the arrival of the last batch of prices ($p_{\max}$) in the instance.  Note that by definition, $z(x) \leq h(x)$, since the delivery cannot exceed what has been purchased thus far.

Since \ALG is deterministic and both conversions (i.e., purchasing and delivery) are both unidirectional (irrevocable), we must have that $h(x-\varrho) \geq h(x)$ and $z(x-\varrho) \geq z(x)$, i.e., $h(x)$ and $z(x)$ are both non-increasing in $[p_{\min}, p_{\max}]$.  Intuitively, the entire demand should be purchased before the end of the sequence if the lowest price appears, i.e., $h(p_{\min}) = 1$.  Furthermore, since the lowest price appears at least twice, \ALG pays switching cost proportional to $h(x)$ and $z(x)$, because each good price is ``interrupted'' by bad prices in the $\mathcal{I}_x$ instance, captured by $2\gamma h(x)$ and $2\delta z(x)$, respectively.  However, if \ALG is forced to purchase at the deadline $T$, it also pays a switching cost of $2\gamma (1-h(x))$ and $2\delta (1-z(x))$, respectively.  Thus, the total worst-case switching cost paid by \ALG is $2\gamma + 2\delta$.

Recall that in the price-dependent increasing delivery cost case, the worst case delivery coefficient is $c + \varepsilon \leq 1$, when the storage is fully charged.  Thus, in this case, \ALG's optimal decision is to deliver the demand as early as possible (i.e., $h(x) = z(x) \forall x \in [p_{\min}, p_{\max}]$), reducing its delivery cost.  Under this assumption, \ALG's effective delivery coefficient is $\varepsilon$.
Then the total cost of \ALG on instance $\mathcal{I}_{x}$ is described by:
\begin{align*}
\ALG(\mathcal{I}_{x}) &= h(\ell) \ell - \int_{\ell}^x udh(u) + (1 - h(x))p_{\max} + 2 \gamma + \varepsilon \ell z(\ell) - \varepsilon \int_{\ell}^x udz(u) + (1 - z(x))\varepsilon p_{\max} + 2\delta.
\end{align*}

\noindent And the optimal cost on instance $\mathcal{I}_{x}$ (if all prices are known in advance) is given by:
\begin{align*}
\OPT(\mathcal{I}_{x}) &= x + \frac{2\gamma}{T} + \varepsilon x + \frac{2\delta}{T}.
\end{align*}

\noindent If $\ALG$ is $\alpha^\star$-competitive, then a necessary condition is that for all $x \in [p_{\min}, p_{\max}]$, $\ALG(\mathcal{I}_x) \leq \alpha^\star \OPT(\mathcal{I}_x)$. 
This imposes a necessary condition on $h(x)$ and $z(x)$, which can be expressed as the following differential inequality:
\begin{align*}
h(\ell) \ell - \int_{\ell}^x udh(u) + (1 - h(x))p_{\max} + 2 \gamma + \varepsilon \ell z(\ell) - \varepsilon \int_{\ell}^x udz(u) + (1 - z(x))\varepsilon p_{\max} + 2\delta & \\
\leq \alpha^\star \left( x + \frac{2\gamma}{T} + \varepsilon x + \frac{2\delta}{T} \right)&.
\end{align*}

\noindent Note that since $z(x) \leq h(x)$, we can proceed by assuming that $z(x) = h(x)$, since this minimizes the effective delivery cost that \ALG pays.  Using integration by parts and letting $\kappa = \gamma + \delta$, this gives the following condition on $h(x)$:
\begin{align*}
h(x) &\geq \frac{p_{\max} + p_{\max} \varepsilon + 2\kappa - \alpha^\star \left( x + \frac{2\kappa}{T} + \varepsilon x \right)}{p_{\max} + p_{\max} \varepsilon -x - \varepsilon x} - \frac{1 + \varepsilon}{p_{\max} + p_{\max} \varepsilon -x - \varepsilon x} \int_\ell^x h(u) du.
\end{align*}
By Grönwall's inequality~\cite[Theorem 1, p. 356]{Mitrinovic:91}, it follows that:
{\small
\begin{align*}
h(x) &\geq \frac{p_{\max} + p_{\max} \varepsilon + 2\kappa - \alpha^\star \left( x + \frac{2\kappa}{T} + \varepsilon x \right)}{p_{\max} + p_{\max} \varepsilon -x - \varepsilon x} - (1+\varepsilon) \int_\ell^x \frac{p_{\max} + p_{\max} \varepsilon + 2\kappa - \alpha^\star \left( u + \frac{2\kappa}{T} + \varepsilon u \right)}{\left(p_{\max} + p_{\max} \varepsilon -u - \varepsilon u\right)^2} du,\\
h(x) &\geq \alpha^\star \ln \left[ \frac{p_{\max} - x}{p_{\max} - \ell} \right],
\end{align*}
}
where $\ell = \frac{(1+\varepsilon) p_{\max} + 2\kappa}{\alpha^\star (1+\varepsilon)} - \frac{2\kappa}{T (1+\varepsilon)}$.

\noindent Recall that we have a boundary condition that $h(p_{\min}) = 1$---we can combine this with the above to obtain:
\begin{align*}
1 &\geq \alpha^\star \ln \left[ \frac{p_{\max} - p_{\min}}{p_{\max} - \ell} \right].
\end{align*}

\noindent Letting $\ell = \frac{(1+\varepsilon) p_{\max} + 2\kappa}{\alpha^\star (1+\varepsilon)} - \frac{2\kappa}{T (1+\varepsilon)}$, the optimal $\alpha^\star$ is obtained when the above inequality is binding, so we have:
\begin{align*}
1 &= \alpha^\star \ln \left[ \frac{p_{\max} - p_{\min}}{p_{\max} - \frac{(1+\varepsilon) p_{\max} + 2\kappa}{\alpha^\star (1+\varepsilon)} + \frac{2\kappa}{T (1+\varepsilon)}} \right].
\end{align*}
Solving the above yields that the optimal $\alpha^\star$ for any \ALG solving \OSDMS with just flexible demand and increasing delivery cost is lower bounded by:
\begin{align*}
\alpha^\star \geq \left[ W\left( -\frac{\left( p_{\max} - p_{\min} \right) e^{-1}}{p_{\max} + \frac{2\kappa}{T (1+\varepsilon)}} \right) + 1 \right]^{-1},
\end{align*}
where $\kappa = \gamma + \delta$, completing the proof.
\end{proof}

\noindent Next, we prove a similar lower bound in the case of base demand:
\begin{lemma}\label{lem:lb-base-increasing}
There exists a set of \OSDMS instances with only base demand and a monotone increasing delivery cost such that no deterministic online algorithm \ALG can achieve a competitive ratio better than $\alpha_{\texttt{IB}}$, given by: 
\begin{align}
\alpha_{\texttt{IB}} \geq \left[ W\left(-\frac{\left( (1-c)p_{\max} - p_{\min} \right) e^{ - \frac{(1+\varepsilon)p_{\max} +2\delta + \frac{2\gamma}{T}}{(1+\varepsilon)p_{\max} + 2\kappa} }}{(1+\varepsilon)p_{\max} + 2\kappa}\right)+\frac{(1+\varepsilon)p_{\max} +2\delta + \frac{2\gamma}{T}}{(1+\varepsilon)p_{\max} + 2\kappa} \right]^{-1}. \label{eq:alphaIB}
\end{align}
\end{lemma}
\begin{proof}
We recall \sref{Def.}{def:x-decreasing-instance-switching-just-base} and consider the same set of instances $\{\mathcal{I}'_{x}\}_{x\in [p_{\min}, p_{\max}]}$, where $\mathcal{I}'_{x}$ is called an \emph{$x$-decreasing instance}.

Let $h(x)$ denote a \emph{conversion function} that maps $[p_{\min}, p_{\max}] \to [0,1]$---this describes \ALG's purchasing decisions before the final batch of prices ($p_{\max}$) in the instance.  Since \ALG is deterministic and purchasing is unidirectional (irrevocable), we must have that $h(x-\varrho) \geq h(x)$, i.e., $h(x)$ is non-increasing in $[p_{\min}, p_{\max}]$.  Intuitively, the entire demand should be purchased before the end of the sequence if the lowest price appears, i.e., $h(p_{\min}) = 1$.

\ALG pays switching cost proportional to $h(x)$, because each good price is ``interrupted'' by bad prices in the $\mathcal{I}'_x$ instance, captured by $2\gamma h(x)$.  If \ALG is forced to purchase when the demand arrives, it also pays a switching cost of $2\gamma (1-h(x))$.  Thus, the total worst-case switching cost on the purchasing side paid by \ALG is $2\gamma$.

The switching cost on the delivery side for both \ALG and \OPT is $2\delta$, since the entire base demand arrives at the end of the sequence, and it must be delivered at that time.

Recall that in the price-dependent increasing delivery cost case, the worst case delivery coefficient is $c + \varepsilon \leq 1$, when the storage is full.  We bound the total \emph{effective} delivery cost coefficient (i.e., across all time steps) as follows:
\[
\sum_{t \in [T]} (cs_t + \varepsilon) = \varepsilon + \sum_{t \in [T]} c s_t \leq \varepsilon + c h(x),
\]
where we have used the fact that for any \ALG, $\sum_{t\in[T]} s_t \leq h(x)$.  Since the delivery cost is price-dependent, the extra delivery cost is realized at the price at the final time step, i.e., $p_{\max} c h(x)$.

\noindent Then the total cost of \ALG on instance $\mathcal{I}'_{x}$ is described by:
\begin{align*}
\ALG(\mathcal{I}'_{x}) &= h(\ell) \ell - \int_{\ell}^x udh(u) + (1 - h(x))p_{\max} + 2 \gamma + p_{\max} \varepsilon + p_{\max} c h(x) + 2\delta.
\end{align*}

\noindent The optimal cost on instance $\mathcal{I}'_{x}$ (if all prices are known in advance) is given by:
\begin{align*}
\OPT(\mathcal{I}'_{x}) &\geq \min \left\{ x + \frac{2\gamma}{T} + (c+\varepsilon) p_{\max} + 2\delta, p_{\max} + \frac{2\gamma}{T} + 2\delta + p_{\max} \varepsilon \right\}.
\end{align*}

\noindent If $\ALG$ is $\alpha^\star$-competitive, then a necessary condition is that for all $x \in [p_{\min}, p_{\max}]$, $\ALG(\mathcal{I}'_{x}) \leq \alpha^\star \OPT(\mathcal{I}'_{x})$. 
This imposes a necessary condition on $h(x)$, which can be expressed as the following differential inequality:
\begin{align*}
h(\ell) \ell - \int_{\ell}^x udh(u) + (1 - h(x))p_{\max} + 2 \gamma + p_{\max} \varepsilon + p_{\max} c h(x) + 2\delta & \\
\leq \alpha^\star \left( x + \frac{2\gamma}{T} + (c+\varepsilon) p_{\max} + 2\delta \right)&.
\end{align*}

\noindent Using integration by parts and letting $\kappa = \gamma + \delta$, this gives the following condition on $h(x)$:
\begin{align*}
h(x) \geq& \frac{p_{\max}(1+\varepsilon)+2\kappa - \alpha^\star \left( x + \frac{2\gamma}{T} + (c + \varepsilon) p_{\max} +2\delta \right)}{p_{\max} - p_{\max}c -x} - \frac{1}{p_{\max} - p_{\max}c - x} \int_\ell^x h(u) du.
\end{align*}

\noindent By Grönwall's inequality~\cite[Theorem 1, p. 356]{Mitrinovic:91}, it follows that:
{\small
\begin{align*}
h(x) &\geq \frac{p_{\max}(1+\varepsilon)+2\kappa - \alpha^\star \left( x + \frac{2\gamma}{T} + (c + \varepsilon) p_{\max} +2\delta \right)}{p_{\max} - p_{\max}c -x} - \int_\ell^x \frac{p_{\max}(1\!+\!\varepsilon)+2\kappa - \alpha^\star \left( u \!+ \!\frac{2\gamma}{T} \!+ \!(c \!+\! \varepsilon) p_{\max} \!+\!2\delta \right)}{ \left(p_{\max} \!-\! p_{\max}c \!-\!u\right)^2} du,\\ 
h(x) &\geq \alpha^\star \ln \left[ \frac{p_{\max} - p_{\max} c - x}{p_{\max} - p_{\max} c - \ell} \right],
\end{align*}
where $\ell = \frac{p_{\max} (1+\varepsilon) + 2\kappa}{\alpha^\star} - \frac{2\gamma}{T} - p_{\max} (\varepsilon + c) - 2\delta$.
}

\noindent Recall that we have a boundary condition that $h(p_{\min}) = 1$---we can combine this with the above to obtain:
\begin{align*}
1 &\geq \alpha^\star \ln \left[ \frac{p_{\max} - p_{\max}c -p_{\min}}{p_{\max} - p_{\max}c -\ell} \right].
\end{align*}

\noindent Letting $\ell = \frac{p_{\max} (1+\varepsilon) + 2\kappa}{\alpha^\star} - \frac{2\gamma}{T} - p_{\max} (\varepsilon + c) - 2\delta$, the optimal $\alpha^\star$ is obtained when the above inequality is binding, so we have:
\begin{align*}
1 &= \alpha^\star \ln \left[ \frac{(1-c)p_{\max} - p_{\min}}{p_{\max} - \frac{p_{\max} (1+\varepsilon) + 2\kappa}{\alpha^\star} + \frac{2\gamma}{T} + p_{\max} \varepsilon + 2\delta} \right].
\end{align*}
Solving the above yields that the optimal $\alpha^\star$ for any \ALG solving \OSDMS with just base demand with increasing delivery cost is lower bounded by:
\begin{align*}
\alpha^\star \geq \left[ W\left(-\frac{\left( (1-c)p_{\max} - p_{\min} \right) e^{ - \frac{(1+\varepsilon)p_{\max} +2\delta + \frac{2\gamma}{T}}{(1+\varepsilon)p_{\max} + 2\kappa} }}{(1+\varepsilon)p_{\max} + 2\kappa}\right)+\frac{(1+\varepsilon)p_{\max} +2\delta + \frac{2\gamma}{T}}{(1+\varepsilon)p_{\max} + 2\kappa} \right]^{-1},
\end{align*}
completing the proof.
\end{proof}

By \sref{Lemmas}{lem:lb-flexible-increasing} and \ref{lem:lb-base-increasing}, we have that the best competitive ratio achievable by any deterministic online algorithm for \OSDMS with a monotone increasing delivery cost is $\alpha_{\texttt{IF}}$ in a general case and $\alpha_{\texttt{IB}}$ in the case of just base demand.  Given parameters of $p_{\min}, p_{\max}, c, \varepsilon, \gamma, \delta, T$, it can be verified that $\alpha_{\texttt{IF}} < \alpha$ and $\alpha_{\texttt{IB}} < \alpha_{\texttt{B}}$, where $\alpha$ and $\alpha_{\texttt{B}}$ are defined in \eqref{eq:alpha1} and \eqref{eq:alphaB}, respectively, and represent the corresponding optimal competitive bounds for the case of \textit{monotone decreasing delivery cost}.  This completes the proof.
\end{proof}

\subsection{Proof of \autoref{thm:osdmt-lower-bound-2}} \label{apx:osdmt-lower-bound-2}

In this section, we prove \autoref{thm:osdmt-lower-bound-2}, which states that there exists a set of \OSDMT instances (i.e., with a tracking cost) such that no deterministic online algorithm \ALG can achieve a competitive ratio better than $\alpha_{\texttt{T}}$ (for $\alpha_{\texttt{T}}$ defined in \eqref{eq:alphaT}).

\begin{proof}
To show this result, we first define a family of instances, and then show that the competitive ratio of any deterministic algorithm is lower bounded under these instances.  
Consider the following set of instances $\{\mathcal{I}_{x}\}_{x\in [p_{\min}, p_{\max}]}$, where $\mathcal{I}_{x}$ is called an \emph{$x$-decreasing instance}.

\begin{definition}[$x$-decreasing instance for \OSDMT]
\label{def:x-decreasing-instance-tracking}
Let $n, m \in \mathbb{N}$ be sufficiently large, and denote $\varrho \coloneqq \nicefrac{p_{\max} - p_{\min}}{n}$. 

For $x \in [p_{\min}, p_{\max}]$, we define a decreasing market price sequence as follows: 
The sequence is partitioned into $n_x := 2 \cdot \lceil \nicefrac{(x - p_{\min})}{\varrho} \rceil + 1$ alternating batches of prices. The $i^{\text{th}}$ batch ($i\in [n_x-2]$) contains $m$ prices $p_{\max}$ if $i$ is odd, and $2$ cost functions with coefficient $p_{\max}-(\lceil i / 2 \rceil)\varrho$ if $i$ is even. The last two batches consist of $m$ cost functions with coefficient $x + \iota$, followed by a final batch of $m$ cost functions with coefficient $p_{\max}$.

The tracking target is defined as follows: $a_t = 0$ for all batches of prices except the second-to-last batch, where $a_t = \nicefrac{1}{m}$ for all $m$ prices.  Note that $\sum_{t=1}^T a_t = 1$.

The demand sequence is defined as follows:  At $t=1$, a batch of flexible demand $f_1 = 1$ arrives with deadline $\Delta_1 = T$.  All other $b_t$ and $f_t$ are zero.  

\end{definition}

Let $h(x)$ and $z(x)$ denote \emph{conversion functions} that both map $[p_{\min}, p_{\max}] \to [0,1]$, respectively.  
Suppose these arbitrary functions fully describe the actions of a deterministic \ALG for \OSDM on an instance $\mathcal{I}_{x}$.  Specifically, suppose that $h(x)$ describes \ALG's purchasing decisions, and $z(x)$ describes \ALG's delivery decisions \emph{before} the arrival of the last batch of prices ($p_{\max}$) in the instance.  Note that by definition, $z(x) \leq h(x)$, since the delivery cannot exceed what has been purchased thus far.
Note that for large $n$, processing an instance $\mathcal{I}_{x-\varrho}$ is equivalent to first processing $\mathcal{I}_{x}$ (besides the last two batches), and then processing batches with prices $x - \varrho$ and $p_{\max}$.

Since \ALG is deterministic and both conversions (i.e., purchasing and delivery) are both unidirectional (irrevocable), we must have that $h(x-\varrho) \geq h(x)$ and $z(x-\varrho) \geq z(x)$, i.e., $h(x)$ and $z(x)$ are both non-increasing in $[p_{\min}, p_{\max}]$.  Intuitively, the entire demand should be purchased before the end of the sequence if the lowest price appears, i.e., $h(p_{\min}) = 1$.  Furthermore, since the lowest price appears at least twice, 

Due to the tracking target, any procurement before the second-to-last batch incurs a tracking penalty of $\eta h(x)$, since the target is zero during this time.  Furthermore, during the last two batches, \ALG incurs a tracking penalty of $\eta + \eta (1 - h(x))$, since the target is $\nicefrac{1}{m}$ during the first $m$ prices in the batch, and zero during the second $m$ prices in the batch---note that \ALG must purchase $(1-h(x))$ at the last price to satisfy the deadline.  Thus, the total tracking penalty incurred by \ALG is given by $2 \eta$.

Recall that in the price-dependent decreasing delivery cost case, the worst case delivery coefficient is $c + \varepsilon \leq 1$, when the storage is empty.  We bound the total \emph{effective} delivery cost coefficient (i.e., across all time steps) as follows:
\[
\sum_{t \in [T]} (c - cs_t + \varepsilon) = c + \varepsilon - \sum_{t \in [T]} c s_t \leq c + \varepsilon - c h(x),
\]
where we have used the fact that for any \ALG, $\sum_{t\in[T]} s_t \leq h(x)$.  Since the delivery cost is price-dependent, the worst-case for \ALG occurs when the \emph{savings} in the delivery cost are realized at the lowest price, i.e., $p_{\min} c h(x)$.
Then the total cost of \ALG on instance $\mathcal{I}_{x}$ is described by:
\begin{align*}
\ALG(\mathcal{I}_{x}) &= h(\ell) \ell - \int_{\ell}^x udh(u) + (1 - h(x))p_{\max} - p_{\min} c h(x) + 2 \eta + \\
&( c + \varepsilon ) \ell z(\ell) - (c+\varepsilon) \int_{\ell}^x udz(u) + (1 - z(x))(c+\varepsilon)p_{\max} + 2\delta.
\end{align*}

\noindent The optimal cost on instance $\mathcal{I}_{x}$ (if all prices are known in advance) is given by:
\begin{align*}
\OPT(\mathcal{I}_{x}) &= x + \varepsilon x + \frac{2\delta}{T},
\end{align*}
where note that \OPT incurs no tracking penalty by perfectly matching the target.
If $\ALG$ is $\alpha^\star$-competitive, then a necessary condition is that for all $x \in [p_{\min}, p_{\max}]$, $\ALG(\mathcal{I}_x) \leq \alpha^\star \OPT(\mathcal{I}_x)$. 
This imposes a necessary condition on $h(x)$ and $z(x)$, which can be expressed as the following differential inequality:
\begin{align*}
h(\ell) \ell - \int_{\ell}^x udh(u) + (1 - h(x))p_{\max} + 2 \eta - p_{\min} c h(x) +& \\
( c + \varepsilon ) \ell z(\ell) - (c+\varepsilon) \int_{\ell}^x udz(u) + (1 - z(x))(c+\varepsilon)p_{\max} + 2\delta & \\
\leq \alpha^\star \left( x + \varepsilon x + \frac{2\delta}{T} \right)&.
\end{align*}

\noindent Recall that $z(x) \leq h(x)$.  On the instances described above, under the assumption that $\kappa = \gamma + \delta \leq \frac{p_{\max} - p_{\min}}{2}$ (see \autoref{sec:assumptions}), note that for sufficiently good $x$ (i.e., $x < \ell$) it is strictly better for $\ALG$ to deliver as much as possible, i.e., $z(x) = h(x)$, reducing its last-minute delivery cost.  Thus, we proceed by assuming that $z(x) = h(x)$.  Using integration by parts, this gives the following condition on $h(x)$:
\begin{align*}
h(x) x - \int_{\ell}^x h(u) du + (1 - h(x))p_{\max} + 2\eta + 2 \delta - p_{\min} c h(x) +& \\
( c + \varepsilon ) x h(x) - (c+\varepsilon) \int_{\ell}^x h(u) du + (1 - h(x))(c+\varepsilon)p_{\max} & \\
\leq \alpha^\star \left( x + \frac{2\delta}{T} + \varepsilon x \right)&.
\end{align*}

\noindent To solve for $\alpha^\star$, we can use Grönwall's inequality~\cite[Theorem 1, p. 356]{Mitrinovic:91}. First, observe that $h(x)$ must satisfy the following:
\begin{align*}
h(x) \geq \frac{p_{\max} + \frac{2\eta + 2\delta}{(1+c+\varepsilon)} - \alpha^\star \frac{x + \frac{2\delta}{T} + \varepsilon x }{1+c+\varepsilon}}{p_{\max} -x + \frac{p_{\min} c}{1+c+\varepsilon}} - \frac{1}{p_{\max} -x  + \frac{p_{\min} c}{1+c+\varepsilon}} \int_\ell^x h(u) du.
\end{align*}

\noindent By Grönwall's inequality, it follows that:
\begin{align*}
h(x) &\geq \frac{p_{\max} + \frac{2\eta + 2\delta}{(1+c+\varepsilon)} - \alpha^\star \frac{x + \frac{2\delta}{T} + \varepsilon x }{1+c+\varepsilon}}{p_{\max} -x + \frac{p_{\min} c}{1+c+\varepsilon}} - \int_\ell^x \frac{p_{\max} + \frac{2\eta + 2\delta}{(1+c+\varepsilon)} - \alpha^\star \frac{u + \frac{2\delta}{T} + \varepsilon u }{1+c+\varepsilon}}{\left( p_{\max} -u  + \frac{p_{\min} c}{1+c+\varepsilon} \right)^2} du,\\ 
h(x) &\geq \alpha^\star \frac{1+\varepsilon}{1+c+\varepsilon} \ln \left[ \frac{(1+c+\varepsilon)(p_{\max}-x) - 2\delta + \frac{2\delta}{m} + p_{\min} c}{(1+c+\varepsilon)(p_{\max}-\ell) - 2\delta + \frac{2\delta}{m} + p_{\min} c} \right].
\end{align*}

\noindent Recall that we have a boundary condition that $h(p_{\min}) = 1$---we can combine this with the above to obtain:
\begin{align*}
1 &\geq \alpha^\star \frac{1+\varepsilon}{1+c+\varepsilon} \ln \left[ \frac{(1+c+\varepsilon)(p_{\max}-p_{\min}) + p_{\min} c}{(1+c+\varepsilon)(p_{\max}-\ell) + p_{\min} c} \right].
\end{align*}

\noindent Letting $\ell = \frac{(1+c+\varepsilon) p_{\max} + 2(\eta+\delta)}{\alpha^\star(1+\varepsilon)} - \frac{2\delta}{T (1+\varepsilon)}$, the optimal $\alpha^\star$ is obtained when the above inequality is binding, so we have:
\begin{align*}
1 &= \alpha^\star \frac{1+\varepsilon}{1+c+\varepsilon} \ln \left[ \frac{(1+c+\varepsilon)p_{\max} - (1+\varepsilon) p_{\min}}{(1+c+\varepsilon)p_{\max} - \frac{(1+c+\varepsilon)\left[ (1+c+\varepsilon) p_{\max} + 2 (\eta+\delta) \right]}{\alpha^\star(1+\varepsilon)} + \frac{(1+c+\varepsilon) 2\delta}{T (1+\varepsilon)} + p_{\min} c} \right].
\end{align*}
Solving the above yields that the optimal $\alpha^\star$ for any \ALG solving \OSDMT is lower bounded by:
\begin{align*}
\alpha^\star \geq \frac{\omega}{\left[ W\left(-\frac{\left( (1+c+\varepsilon) p_{\max} - (1+\varepsilon) p_{\min} \right) e^{\frac{-\nicefrac{\omega 2\delta}{T} - p_{\min} c - (1+c+\varepsilon) p_{\max}}{(1+c+\varepsilon)p_{\max} + 2(\eta+\delta)}}}{(1+c+\varepsilon)p_{\max} + 2(\eta + \delta)}\right)+\frac{(1+c+\varepsilon) p_{\max} + p_{\min} c - \nicefrac{\omega 2\delta}{T}}{(1+c+\varepsilon)p_{\max} + 2(\eta+\delta)} \right] },
\end{align*}
where $\omega = \frac{(1+c+\varepsilon)}{(1+\varepsilon)}$.  This completes the proof.

\end{proof}

\section{Deferred Proofs from \autoref{sec:pald} (Analysis of \PALD Framework)} \label{apx:pald-proof}
In this section, we provide full proofs for the results in \autoref{sec:pald}, which describes and analyzes the \PALD learning-augmented framework for the \OSDM problem.  We start by proving \autoref{thm:pald-robustness}, which provides a robustness certificate that guarantees \PALD's competitive ratio when given learned threshold functions that lie within feasible sets.  We also state and prove \sref{Corollary}{cor:pald-robustness-osdmt}, which provides an analogous robustness certificate for the case of \OSDMT, before proving \sref{Lemma}{lem:convex-robust-set}, which shows that the feasible sets are convex (and thus efficient to enforce) when the threshold functions are parameterized as piecewise-affine functions.

\subsection{Proof of \autoref{thm:pald-robustness}} \label{apx:pald-robustness}

In this section, we prove \autoref{thm:pald-robustness}, which states that given learned threshold functions that lie in the feasible sets $\hat{\phi}_b \in \mathcal{R}_b(\rho)$ and $(\hat{\phi}_f, \hat{\psi}_f) \in \mathcal{R}_f(\rho)$ for some $\rho > \alpha$ ($\alpha$ defined in \eqref{eq:alpha1}), the \PALD algorithm is $\rho$-robust for \OSDMS.

\begin{proof}

Before the main robustness proof, we show that \PALD produces a feasible solution to \OSDM.

\begin{corollary}\label{cor:pald-feasibility-1}
\PALD produces a feasible solution to \OSDM, i.e., it satisfies all demand before their deadlines and never violates the storage capacity constraint.
\end{corollary}
\begin{proof}
This corollary follows by the same logic as \sref{Lemma}{lem:osdm-feasibility-1}, since \PALD is identical to \PALD except for the threshold functions used, and the feasibility proof does not depend on the specific form of the threshold functions.
\end{proof}

We now proceed to prove the worst-case competitive ratio of \PALD---for notational brevity, the following considers \textit{any arbitrary} $\mathcal{I} \in \Omega$ (i.e., any arbitrary instance of \OSDM).  We recall the definition of active and inactive periods from \sref{Definition}{dfn:active-inactive-periods}, and the additional notations relating to driver indexing used in the proof of \autoref{thm:osdm-upper-bound-1} in \autoref{apx:paad-proof}.

We let $D_b = \sum_{i=1}^n \sum_{j=1}^{\nu_i} B_b^{(i,j)}$ and $D_f = \sum_{i=1}^n \sum_{j=1}^{\hat{\phi}_i} B_f^{(i,j)}$ denote the total base and flexible demand, respectively, that arrive over the entire time horizon.  We let $D = D_b + D_f$ denote the total demand.  We let $\hat{s}$ denote the final status of the storage at the end of the time horizon, i.e., $\hat{s} = s_T$.  

Finally, we introduce the following notation to characterize the optimal solution: let $G_i(\beta)$ denote the minimum cost of purchasing $\beta$ units of asset during the $i^\text{th}$ active period, and let $H_\tau(f)$ denote the minimum cost of purchasing $f$ units of asset during the period $[\tau, \tau + \Delta_\tau]$ (i.e., during the lifetime of the $\tau^\text{th}$ flexible demand).  Finally, let $\tilde{p}$ denote the minimum price during idle periods, and let $\bar{p}$ denote the weighted average price during periods with non-zero base demand, i.e., $\bar{p} = \frac{\sum_{\tau \in [T]} p_t \cdot b_t}{D_b}$.

We begin the proof by stating a lower bound on the cost of the offline optimal solution $\OPT$, leveraging the exact same logic as \sref{Lemma}{lem:osdm-optimal-lower-bound-1}.

\begin{corollary}
\label{cor:pald-optimal-lower-bound-1}
Given that \PALD produces $n$ active periods, let $\beta_i$ denote the asset purchased towards the base demand by the offline optimal solution during the $i^\text{th}$ active period, $i \in [n]$, and let $\tilde{p}$ denote the minimum price during inactive periods.  Then $\OPT(\mathcal{I})$ is lower bounded as:
\begin{align*}
\OPT(\mathcal{I}) &\geq \sum_{i=1}^n G_i(\beta_i) + \left( D_b - \sum_{i=1}^n \beta_i \right) \tilde{p} + \varepsilon \bar{p} D_b + \sum_{\tau=1}^T \left(1 + \varepsilon \right) H_\tau (f_\tau) + \left( D_b + D_f \right) \frac{2\kappa}{T}
\end{align*}
\end{corollary}

Next, we state an upper bound on the cost incurred by \PALD, leveraging the same logic as \sref{Lemma}{lem:osdm-paad-upper-bound-1}, but with the threshold functions $\hat{\phi}_b$ and $(\hat{\phi}_f, \hat{\psi}_f)$ used by \PALD.

\begin{lemma}
\label{lem:pald-upper-bound-1}
Given that \PALD produces $n$ active periods, let $\hat{w}_b^{(i,j)}$ denote the amount of the $j^\text{th}$ base demand driver that has been purchased by the end of the $i^\text{th}$ active period, $j \in [\nu_i]$, and let $B_b^{(i,j)}$ denote the total demand associated with the $j^\text{th}$ base demand driver in the $i^\text{th}$ active period.  Further, let $\hat{w}_f^{(\tau)}$ and $\hat{v}_f^{(\tau)}$ denote the purchasing and delivery amounts of the $\tau^\text{th}$ flexible demand driver before the deadline, respectively.  Then the cost incurred by \PALD is upper bounded as:
\begin{align}
\PALD(\mathcal{I}) &\leq \sum_{i=1}^n \sum_{j=1}^{\nu_i} \int_0^{\hat{w}_b^{(i,j)}} \kern-1em \hat{\phi}_b^{(i,j)}(u) du + \sum_{\tau=1}^T \left( \int_0^{\hat{w}_f^{(\tau)}} \kern-1em \hat{\phi}_f^{(\tau)}(u) du + \int_0^{\hat{v}_f^{(\tau)}} \kern-1em \hat{\psi}^{(\tau)}(z) dz \right) \label{eq:pald-upper-bound-1-tweaked} \\
& \quad + \left( D_b + D_f - \sum_{i=1}^n \sum_{j=1}^{\nu_i} \hat{w}_b^{(i,j)} - \sum_{\tau=1}^{T} \hat{w}_f^{(\tau)}\right) (p_{\max} + 2 \gamma)\\
& \quad + D_b \left( \bar{p} (c+\varepsilon) + 2\delta \right) + \left(D_f - \sum_{\tau=1}^{T} \hat{v}_f^{(\tau)}\right) (p_{\max} (c + \varepsilon) + 2\delta) \\
& \quad - c p_{\min} \left( \sum_{i=1}^n \sum_{j=1}^{\nu_i} \hat{w}_b^{(i,j)} + \sum_{\tau=1}^{T} \hat{w}_f^{(\tau)}\right) + \hat{s} p_{\max}.
\end{align}
\end{lemma}
\begin{proof}
The proof follows by the same logic as \sref{Lemma}{lem:osdm-paad-upper-bound-1}---the main change is with respect to \eqref{eq:pald-upper-bound-1-tweaked}, to capture the fact that the threshold functions used by \PALD are the learned variants, denoted by $\hat{\phi}_b$ and $(\hat{\phi}_f, \hat{\psi}_f)$, instead of the analytical variants $\phi_b$ and $(\phi_f, \psi_f)$ used by \PAAD.

Any additional purchasing needed to satisfy demand (i.e., that isn't captured by the integrals over the threshold functions) is exactly captured by the existing logic in the proof of \sref{Lemma}{lem:osdm-paad-upper-bound-1}, completing the proof.
\end{proof}

Leveraging the definition of the feasible robustness sets $\mathcal{R}_b(\rho)$ and $\mathcal{R}_f(\rho)$ in \sref{Definitions}{dfn:robust-threshold-base} and \ref{dfn:robust-threshold-flexible}, respectively, we have the following technical lemmas to relate the threshold functions and the optimal cost:
\begin{lemma}
\label{lem:pald-threshold-function-relation-1}
Given a learned threshold function $\hat{\phi}_b(\cdot)$ that lies within the feasible set $\mathcal{R}_b(\rho)$ for some $\rho > \alpha$, the following relation always holds:
{\small 
\begin{align*}
\int_0^{w} \kern-1em \hat{\phi}_b(u) du + (1-w) (p_{\max}+2\gamma) + p_{\max}(c+\varepsilon) + 2\delta - c w p_{\min} = \rho \left[ \hat{\phi}_b(w) - 2\gamma + \varepsilon p_{\max} + \frac{2\kappa}{T}\right] \ & \forall w \in [0,1].
\end{align*}
}
\end{lemma}

\begin{lemma}
\label{lem:pald-threshold-function-relation-2}
Given learned threshold functions $\hat{\phi}_f(\cdot)$ and $\hat{\psi}(\cdot)$ that lie within the feasible set $\mathcal{R}_f(\rho)$ for some $\rho > \alpha$, the following relation always holds:
{\small
\begin{align*}
\int_0^{w} \kern-1em \hat{\phi}_f(u) du + (1-w) (p_{\max}+2\gamma) - c w p_{\min} + \int_0^{v} \kern-1em \hat{\psi}_f(z) dz + (1-v) (p_{\max}(c+\varepsilon) + 2\delta)  =\\
\rho \left[ \frac{1+\varepsilon}{1+c+\varepsilon} \left( \hat{\phi}_f(w) + \hat{\psi}_f(v) - 2\kappa \right) + \frac{2\kappa}{T}\right] \ & \forall w \in [0,1], v \in [0,w],
\end{align*}
}
\end{lemma}

\noindent The proofs of both lemmas follow directly from the definitions of the feasible sets $\mathcal{R}_b(\rho)$ and $\mathcal{R}_f(\rho)$.  Using the results in \sref{Corollary}{cor:pald-optimal-lower-bound-1} and \sref{Lemmas}{lem:pald-upper-bound-1}, \ref{lem:pald-threshold-function-relation-1}, and \ref{lem:pald-threshold-function-relation-2}, we claim that the following holds:
\begin{align*}
\frac{\PALD(\mathcal{I}) - p_{\max} \hat{s}}{\OPT(\mathcal{I})} &\leq \rho.
\end{align*}

To show this result, we first substitute the bounds from \sref{Corollary}{cor:pald-optimal-lower-bound-1} and \sref{Lemma}{lem:pald-upper-bound-1} into the left-hand side of the above equation.  As in the proof of \autoref{thm:osdm-upper-bound-1}, we define some shorthand notation to facilitate the presentation.

\noindent Let $Q = \sum_{i=1}^n \sum_{j=1}^{\nu_i} \int_0^{\hat{w}_b^{(i,j)}} \kern-1em \hat{\phi}_b^{(i,j)}(u) du + \sum_{\tau=1}^T \left( \int_0^{\hat{w}_f^{(\tau)}} \kern-1em \hat{\phi}_f^{(\tau)}(u) du + \int_0^{\hat{v}_f^{(\tau)}} \kern-1em \hat{\psi}^{(\tau)}(z) dz \right)$ denote the integrals over the thresholds.
Let $\hat{W}_b = \sum_{i=1}^n \sum_{j=1}^{\nu_i} \hat{w}_b^{(i,j)}$, let $\hat{W}_f = \sum_{\tau=1}^{T} \hat{w}_f^{(\tau)}$, and let $\hat{W} = \hat{W}_b + \hat{W}_f$ denote the total purchasing by all drivers.  

\noindent Further, let $\hat{V}_b = \sum_{i=1}^n \sum_{j=1}^{\nu_i} \hat{v}_b^{(i,j)}$ denote the total delivery by all flexible demand drivers.

\noindent Let ${\bm \beta} = \sum_{i=1}^n \beta_i$ denote the total amount of asset purchased towards the base demand by the optimal solution during all active periods, noting that $D_b - {\bm \beta} \geq 0$ by definition.   Finally, let $D = D_b + D_f$ denote the total demand.  Substituting the bounds from \sref{Corollary}{cor:pald-optimal-lower-bound-1} and \sref{Lemma}{lem:pald-upper-bound-1} into the left-hand side of the above inequality, we have:

{\small
\begin{align*}
\frac{Q\!+\!({\bm\beta}\!+\!D_f\!-\!\hat{W}) (p_{\max}\!+\!2 \gamma)\!+\!{\bm \beta} \left( \bar{p} (c\!+\!\varepsilon)\!+\!2\delta \right) + (D_f\!-\!\hat{V}_f) (p_{\max} (c\!+\!\varepsilon)\!+\!2\delta)\!-\!c p_{\min} \hat{W}\!+\!(D_b\!-\!{\bm \beta}) (p_{\max}\!+\!2\kappa\!+\!\bar{p}(c\!+\!\varepsilon))}{\sum_{i=1}^n G_i(\beta_i)\!+\!\varepsilon \bar{p} {\bm \beta}\!+\!\sum_{\tau=1}^T (1\!+\!\varepsilon) H_\tau (f_\tau)\!+\!(D_f\!+\!{\bm \beta}) \frac{2\kappa}{T}\!+\!(D_b\!-\!{\bm \beta})(\tilde{p}\!+\!\varepsilon \bar{p}\!+\!\frac{2\kappa}{T}) } \\
\end{align*}
}
Then, we have the following:
{\small
\begin{align*}
\leq \max \Bigg \{ \frac{Q + ({\bm\beta}+D_f - \hat{W}) (p_{\max} + 2 \gamma) + {\bm\beta} \left( \bar{p} (c+\varepsilon) + 2\delta \right) + (D_f - \hat{V}_f) (p_{\max} (c + \varepsilon) + 2\delta) - c p_{\min} \hat{W}}{\sum_{i=1}^n G_i(\beta_i) + \varepsilon \bar{p} {\bm\beta} + \sum_{\tau=1}^T (1 + \varepsilon) H_\tau (f_\tau) + (D_f + {\bm\beta}) \frac{2\kappa}{T} }, \ \ &\\
\quad \frac{(D_b\!-\!{\bm\beta}) (p_{\max}\!+\!2\kappa\!+\!\bar{p}(c+\varepsilon))}{(D_b - {\bm\beta})(\tilde{p} + \varepsilon \bar{p} + \frac{2\kappa}{T}) } \Bigg \}, & \\
\end{align*}
}
where the definition of $\tilde{p}$ ensures that the second term in the $\max$ is at most $\rho$.  Note that this follows because $\tilde{p}$ is the best price during inactive periods (i.e., a price that is not accepted by any base demand driver), and by the definition of the feasible set $\mathcal{R}_b(\rho)$, any $\rho$-robust threshold function $\hat{\phi}_b$ must satisfy $p_{\max}(1+c+\varepsilon) + 2\kappa = \rho \left[ \hat{\phi}_b(0) - 2\gamma + \varepsilon p_{\max} + \frac{2\kappa}{T} \right]$ for $w = 0$.

\noindent We now focus on the first term.  For the sake of contradiction, suppose that
{\small
\begin{align}
    \frac{Q + ({\bm\beta}+D_f - \hat{W}) (p_{\max} + 2 \gamma) + {\bm \beta} \left( \bar{p} (c+\varepsilon) + 2\delta \right) + (D_f - \hat{V}_f) (p_{\max} (c + \varepsilon) + 2\delta) - c p_{\min} \hat{W}}{\sum_{i=1}^n G_i(\beta_i) + \varepsilon \bar{p} {\bm \beta} + \sum_{\tau=1}^T (1 + \varepsilon) H_\tau (f_\tau) + (D_f + {\bm \beta}) \frac{2\kappa}{T} } > \alpha_{\texttt{T}}. \label{eq:ratio-term-pald-1}
\end{align}
}
Instead of working directly with the expression in terms of ${\bm \beta}$, we first reason about how the cost of \OPT and \PALD relate to one another in terms of $\sum_{i=1}^n \sum_{j=1}^{\nu_i} B_b^{(i,j)}$, the total demand assigned to base drivers.
We introduce the following notation for the sake of brevity: let ${\bm B} = \sum_{i=1}^n \sum_{j=1}^{\nu_i} B_b^{(i,j)}$ denote the total base demand.  Then, we have the following relation:
{\small
\begin{align*}
& \frac{Q + ({\bm B}+D_f - \hat{W}) (p_{\max} + 2 \gamma) + {\bm B} \left( \bar{p} (c+\varepsilon) + 2\delta \right) + (D_f - \hat{V}_f) (p_{\max} (c + \varepsilon) + 2\delta) - c p_{\min} \hat{W}}{\sum_{i=1}^n G_i(B_i) + \varepsilon \bar{p} {\bm B} + \sum_{\tau=1}^T (1 + \varepsilon) H_\tau (f_\tau) + (D_f + {\bm B}) \frac{2\kappa}{T} }, \\
& \geq \frac{Q\!+\!({\bm \beta }\!+\!D_f\!-\!\hat{W}) (p_{\max}\!+\!2 \gamma)\!+\!{\bm \beta } \left( \bar{p} (c\!+\!\varepsilon)\!+\!2\delta \right)\!+\!(D_f\!-\!\hat{V}_f) (p_{\max} (c\!+\!\varepsilon)\!+\!2\delta)\!-\!c p_{\min} \hat{W}\!+\!({\bm B}\!-\!{\bm \beta}) (p_{\max}\!+\!\bar{p}(c\!+\!\varepsilon)\!+\!2 \kappa)}{\sum_{i=1}^n G_i(\beta_i)\!+\!\varepsilon \bar{p} {\bm \beta }\!+\!\sum_{\tau=1}^T (1\!+\!\varepsilon) H_\tau (f_\tau)\!+\!(D_f\!+\!{\bm \beta }) \frac{2\kappa}{T}\!+\!({\bm B}\!-\!{\bm \beta}) (\tilde{p}\!+\!\varepsilon \bar{p}\!+\!\frac{2\kappa}{T})}, \\
& > \rho.
\end{align*}
}
where in the second inequality, we have used \sref{Lemma}{lem:upper-bound-B-beta-tilde-p}. We proceed to work with the expression in terms of ${\bm B}$.

During the $i^\text{th}$ active period and the lifetime of the $j^\text{th}$ base demand driver, the minimum marginal purchasing price observed is given by $\hat{\phi}_b^{(i,j)}(\hat{w}_b^{(i,j)}) - 2\gamma$ by the definition of the threshold function.  Note that this follows because the feasible set $\mathcal{R}_b(\rho)$ ensures that $\hat{\phi}_b^{(i,j)}(w)$ is monotone non-increasing for $w \in [0,1]$, and it is given that $\hat{\phi}_b^{(i,j)}(1) \leq p_{\min} + 2\gamma$, guaranteeing that the true minimum price is captured by the threshold function.  

Using the same logic and by definition of the feasible set $\mathcal{R}_f(\rho)$, we have that during the lifetime of the $\tau^\text{th}$ flexible demand driver, the minimum marginal purchasing and delivery cost observed is given by $\frac{1+\varepsilon}{1+c+\varepsilon} \left( \hat{\phi}_f^{(\tau)}(\hat{w}_f^{(\tau)}) + \hat{\psi}_f^{(\tau)}(\hat{v}_f^{(\tau)}) - 2 \kappa \right)$.  
This gives the following lower bounds on the terms that depend on $G_i$ and $H_\tau$, respectively:
{\small
\begin{align*}
\sum_{i=1}^n G_i\left( \sum_{j=1}^{\nu_i} B_b^{(i,j)} \right) &\geq \sum_{i=1}^n \sum_{j=1}^{\nu_i} \left( \hat{\phi}_b^{(i,j)}(\hat{w}_b^{(i,j)}) - 2\gamma \right) \times B_b^{(i,j)} \\
\sum_{\tau=1}^T (1 + \varepsilon) H_\tau (f_\tau) &\geq \sum_{\tau=1}^T \frac{1+\varepsilon}{1+c+\varepsilon} \left( \hat{\phi}_f^{(\tau)}(\hat{w}_f^{(\tau)}) + \hat{\psi}_f^{(\tau)}(\hat{v}_f^{(\tau)}) - 2 (\gamma+\delta) \right) \times f_\tau.
\end{align*}
}

Substituting these bounds into the previous expression, we have that the left-hand-side of \eqref{eq:ratio-term-pald-1} is less than or equal to:
{\small
\begin{align*}
&\leq \frac{Q + ({\bm B}+D_f - \hat{W}) (p_{\max} + 2 \gamma) + {\bm B} \left( \bar{p} (c+\varepsilon) + 2\delta \right) + (D_f - \hat{V}_f) (p_{\max} (c + \varepsilon) + 2\delta) - c p_{\min} \hat{W}}{\sum_{i=1}^n \sum_{j=1}^{\nu_i} \left( \hat{\phi}_b^{(i,j)}(\hat{w}_b^{(i,j)})\!-\!2\gamma \right) B_b^{(i,j)}\!+\! \varepsilon \bar{p} {\bm B}\!+\!\sum_{\tau=1}^T \frac{1+\varepsilon}{1+c+\varepsilon} \left( \hat{\phi}_f^{(\tau)}(\hat{w}_f^{(\tau)})\!+\!\hat{\psi}_f^{(\tau)}(\hat{v}_f^{(\tau)})\!-\!2 \kappa \right) f_\tau\!+\!(D_f\!+\!{\bm B}) \frac{2\kappa}{T} }.
\end{align*}
}

By rearranging the terms in the above and substituting for $Q$, we obtain the following:
{\small
\begin{align}
&= \frac{\sum_{i=1}^n \sum_{j=1}^{\nu_i} \left[ X_b^{(i,j)} \right] + \sum_{\tau=1}^T \left[ X_f^{(\tau)} \right]}{\sum_{i=1}^n \sum_{j=1}^{\nu_i} \left[ \left( \hat{\phi}_b^{(i,j)}(\hat{w}_b^{(i,j)})\!-\!2\gamma\!+\!\varepsilon \bar{p}\!+\!\frac{2\kappa}{T}\right) B_b^{(i,j)}  \right]\!+\!\sum_{\tau=1}^T \left[ \frac{1+\varepsilon}{1+c+\varepsilon} \left( \hat{\phi}_f^{(\tau)}(\hat{w}_f^{(\tau)})\!+\!\hat{\psi}_f^{(\tau)}(\hat{v}_f^{(\tau)})\!-\!2 \kappa \right) f_\tau\!+\!f_\tau \frac{2\kappa}{T} \right]}, \label{eq:partitioned-sums-upper-bound-pald-1}
\end{align}
}

where $X_b^{(i,j)}$ and $X_f^{(\tau)}$ are defined as follows (for each driver):
{\small
\begin{align*}
X_b^{(i,j)} &= \int_0^{\hat{w}_b^{(i,j)}} \kern-1em \hat{\phi}_b^{(i,j)}(u) du + (B_b^{(i,j)} -\hat{w}_b^{(i,j)}) (p_{\max}+2\gamma) + \left( p_{\max}(c+\varepsilon) + 2\delta \right) B_b^{(i,j)} - c \hat{w}_b^{(i,j)} p_{\min}, \\
X_f^{(\tau)} &= \int_0^{\hat{w}_f^{(\tau)}} \kern-1em \hat{\phi}_f^{(\tau)}(u) du + (f_\tau -\hat{w}_f^{(\tau)}) (p_{\max}+2\gamma) + \int_0^{\hat{v}_f^{(\tau)}} \kern-1em \hat{\psi}_f^{(\tau)}(u) du + \left( p_{\max} (c+\varepsilon) + 2\delta \right)(f_\tau - \hat{v}_f^{(\tau)}) - c \hat{w}_f^{(\tau)} p_{\min}.
\end{align*}
}

Since \eqref{eq:partitioned-sums-upper-bound-pald-1} is $> \rho$ by assumption, one of the following cases must be true:

\smallskip
\noindent\textbf{Case I:} There exists some $i \in [n]$ and $j \in [\nu_i]$ such that
{\small
\begin{align*}
X_b^{(i,j)} &> \rho  B_b^{(i,j)} \left( \hat{\phi}_b^{(i,j)}(\hat{w}_b^{(i,j)}) - 2\gamma + \varepsilon \bar{p} + \frac{2\kappa}{T}\right).
\end{align*}
}

\smallskip
\noindent\textbf{Case II:} There exists some $\tau \in [T]$ such that
{\small
\begin{align*}
X_f^{(\tau)} &> \rho f_\tau \left[ \frac{1+\varepsilon}{1+c+\varepsilon} \left( \hat{\phi}_f^{(\tau)}(\hat{w}_f^{(\tau)}) + \hat{\psi}_f^{(\tau)}(\hat{v}_f^{(\tau)}) - 2 (\gamma+\delta) \right) + \frac{2\kappa}{T} \right],
\end{align*}
}

\noindent But, if either of these two cases are true (i.e., for any unit of base or flexible demand), they contradict \sref{Lemmas}{lem:pald-threshold-function-relation-1} and \ref{lem:pald-threshold-function-relation-2}, respectively.  Thus, we have a contradiction, and the original assumption that \eqref{eq:partitioned-sums-upper-bound-pald-1} is $> \rho$ must be false.  This completes the proof that $\frac{\PALD(\mathcal{I}) - p_{\max} \hat{s}}{\OPT(\mathcal{I})} \leq \rho$. 

\noindent Using this result, we have:
\begin{align*}
\PALD(\mathcal{I}) &\leq \rho \OPT(\mathcal{I}) + p_{\max} \hat{s} \leq \rho \OPT(\mathcal{I}) + p_{\max} S,
\end{align*}
where $p_{\max} S$ is a constant.  This shows that \PALD is $\rho$-competitive under \sref{Definition}{dfn:comp-ratio}, which completes the proof.

\end{proof}

\subsection{\PALD Robustness Certificate for \OSDMT} \label{apx:pald-robustness-osdmt}

In this section, we define the feasible robust sets for \PALD in the context of the \OSDMT problem (i.e., in the tracking case), and state and prove the corresponding robustness certificate.  These (re-)definitions account for the tracking cost on the purchasing side of the problem.

\begin{definition}[Robust threshold set for base drivers in \OSDMT]
\label{dfn:robust-threshold-base-osdmt}
{\it
Given a target robustness $\rho > \alpha_{\texttt{T}}$ (where $\alpha_{\texttt{T}}$ is defined in \eqref{eq:alphaT}), a learned threshold function $\hat{\phi}_b$ for base demand drivers must lie in the following feasible set $\mathcal{R}'_b(\rho)$:
}
{\small
\begin{align*}
\mathcal{R}'_b(\rho) = \Bigg\{ & \phi_b: [0,1] \to [p_{\min}, p_{\max}] \ \Bigg| \ \phi_b \text{ is monotone non-increasing},  \phi_b(1) \leq p_{\min} + 2\eta, \text{ and } \forall w \in [0,1] : \\
& \ \ \ \Phi_b(0,w) + (1-w)(p_{\max} + 2\eta) + p_{\max}(c+\varepsilon) + 2\delta - c w p_{\min} \leq \rho \left[ \phi_b(w) - 2\eta + \varepsilon p_{\max} + \nicefrac{2\delta}{T}\right] \Bigg\}.
\end{align*}
}
\end{definition}

\begin{definition}[Robust threshold set for flexible drivers in \OSDMT]
\label{dfn:robust-threshold-flexible-osdmt}
{\it
Given a target robustness $\rho > \alpha_{\texttt{T}}$ (where $\alpha_{\texttt{T}}$ is defined in \eqref{eq:alphaT}), the learned threshold functions $\hat{\phi}_f$ and $\hat{\psi}_f$ for flexible demand drivers must lie in the following joint feasible set $\mathcal{R}'_f(\rho)$, where $\omega = \frac{1+c+\varepsilon}{1+\varepsilon}$:
}
{\small
\begin{align*}
\mathcal{R}'_f(\rho) = \Bigg\{ & (\phi_f, \psi_f): [0,1]^2 \to [p_{\min}, p_{\max}]^2 \ \Bigg| \ \phi_f, \psi_f \text{ are monotone non-increasing}, \\
& \phi_f(1) \leq p_{\min} + 2\eta, \psi_f(1) \leq p_{\min}(c+\varepsilon) + 2\delta,  \text{ and } \forall w \in [0,1], v \in [0,w]: \\
& \Phi_f(0,w) + (1-w)(p_{\max} + 2\eta) - c w p_{\min} + \Psi_f(0,v) + (1-v)(p_{\max}(c+\varepsilon) + 2\delta) \\
& \hspace{18em} \leq \rho \left[ \nicefrac{1}{\omega} \left( \phi_f(w) + \psi_f(v) - 2(\eta+\delta) \right) + \nicefrac{2\delta}{T}\right] \Bigg\}.
\end{align*}
}
\end{definition}

\noindent Under \sref{Def.}{dfn:robust-threshold-base-osdmt} and \sref{Def.}{dfn:robust-threshold-flexible-osdmt}, we have the following robustness certificate for \PALD when given any learned thresholds $\hat{\phi}_b, \hat{\phi}_f$ and $\hat{\psi}_f$ that lie in the respective feasible sets:

\begin{corollary}
\label{cor:pald-robustness-osdmt}
Given learned threshold functions that lie in the feasible sets $\hat{\phi}_b \in \mathcal{R}'_b(\rho)$ and $(\hat{\phi}_f, \hat{\psi}_f) \in \mathcal{R}'_f(\rho)$ for some $\rho > \alpha_{\texttt{T}}$ ($\alpha_{\texttt{T}}$ defined in \eqref{eq:alphaT}), the \PALD algorithm is $\rho$-robust for \OSDMT.
\end{corollary}
\begin{proof}
The feasibility of \PALD's solution follows from \sref{Corollary}{cor:pald-feasibility-1}.  

We now proceed to prove the worst-case competitive ratio of \PALD---for notational brevity, the following considers \textit{any arbitrary} $\mathcal{I} \in \Omega$ (i.e., any arbitrary instance of \OSDM). 
We begin the proof by stating a lower bound on the cost of the offline optimal solution $\OPT$, leveraging the same logic as \sref{Lemma}{lem:osdmt-optimal-lower-bound-2}.
\begin{corollary}
\label{cor:pald-optimal-lower-bound-2}
Given that \PALD produces $n$ active periods for \OSDM, let $\beta_i$ denote the asset purchased towards the base demand by the offline optimal solution during the $i^\text{th}$ active period, $i \in [n]$, and let $\tilde{p}$ denote the minimum price during inactive periods.  Then $\OPT(\mathcal{I})$ is lower bounded as:
\begin{align*}
\OPT(\mathcal{I}) &\geq \sum_{i=1}^n G_i(\beta_i) + \left( D_b - \sum_{i=1}^n \beta_i \right) \tilde{p} + \varepsilon \bar{p} D_b + \sum_{\tau=1}^T \left(1 + \varepsilon \right) H_\tau (f_\tau) + \left( D_b + D_f \right) \frac{2\delta}{T}
\end{align*}
\end{corollary}

\noindent Next, we state an upper bound on the cost incurred by \PALD for \OSDMT, leveraging the same logic as \sref{Lemma}{lem:osdmt-paad-upper-bound-2}, but with the threshold functions $\hat{\phi}_b$ and $(\hat{\phi}_f, \hat{\psi}_f)$ used by \PALD.

\begin{lemma}
\label{lem:pald-upper-bound-2}
Given that \PALD produces $n$ active periods, let $\hat{w}_b^{(i,j)}$ denote the amount of the $j^\text{th}$ base demand driver that has been purchased by the end of the $i^\text{th}$ active period, $j \in [\nu_i]$, and let $B_b^{(i,j)}$ denote the total demand associated with the $j^\text{th}$ base demand driver in the $i^\text{th}$ active period.  Further, let $\hat{w}_f^{(\tau)}$ and $\hat{v}_f^{(\tau)}$ denote the purchasing and delivery amounts of the $\tau^\text{th}$ flexible demand driver before the deadline, respectively.  Then the cost incurred by \PALD is upper bounded as:
\begin{align}
\PALD(\mathcal{I}) &\leq \sum_{i=1}^n \sum_{j=1}^{\nu_i} \int_0^{\hat{w}_b^{(i,j)}} \kern-1em \hat{\phi}_b^{(i,j)}(u) du + \sum_{\tau=1}^T \left( \int_0^{\hat{w}_f^{(\tau)}} \kern-1em \hat{\phi}_f^{(\tau)}(u) du + \int_0^{\hat{v}_f^{(\tau)}} \kern-1em \hat{\psi}^{(\tau)}(z) dz \right)  \\
& \quad + \left( D_b + D_f - \sum_{i=1}^n \sum_{j=1}^{\nu_i} \hat{w}_b^{(i,j)} - \sum_{\tau=1}^{T} \hat{w}_f^{(\tau)}\right) (p_{\max} + 2 \eta) \label{eq:pald-upper-bound-2-tweaked} \\
& \quad + D_b \left( \bar{p} (c+\varepsilon) + 2\delta \right) + \left(D_f - \sum_{\tau=1}^{T} \hat{v}_f^{(\tau)}\right) (p_{\max} (c + \varepsilon) + 2\delta) \\
& \quad - c p_{\min} \left( \sum_{i=1}^n \sum_{j=1}^{\nu_i} \hat{w}_b^{(i,j)} + \sum_{\tau=1}^{T} \hat{w}_f^{(\tau)}\right) + \hat{s} p_{\max}.
\end{align}
\end{lemma}
\begin{proof}
The proof follows by the same logic as \sref{Lemma}{lem:pald-upper-bound-1}---the main change is with respect to \eqref{eq:pald-upper-bound-2-tweaked}, to capture the tracking cost (with parameter $\eta$ on the purchasing side of the problem).
\end{proof}

Leveraging the definition of the feasible robustness sets $\mathcal{R}'_b(\rho)$ and $\mathcal{R}'_f(\rho)$ in \sref{Definitions}{dfn:robust-threshold-base-osdmt} and \ref{dfn:robust-threshold-flexible-osdmt}, respectively, we have the following lemmas to relate the threshold functions and the optimal cost:

\begin{lemma}
\label{lem:pald-threshold-function-relation-1-osdmt}
Given a learned threshold function $\hat{\phi}_b(\cdot)$ that lies within the feasible set $\mathcal{R}'_b(\rho)$ for some $\rho > \alpha_{\texttt{T}}$, the following relation always holds:
{\small 
\begin{align*}
\int_0^{w} \kern-1em \hat{\phi}_b(u) du + (1-w) (p_{\max}+2\eta) + p_{\max}(c+\varepsilon) + 2\delta - c w p_{\min} = \rho \left[ \hat{\phi}_b(w) - 2\eta + \varepsilon p_{\max} + \frac{2\delta}{T}\right] \ & \forall w \in [0,1].
\end{align*}
}
\end{lemma}

\begin{lemma}
\label{lem:pald-threshold-function-relation-2-osdmt}
Given learned threshold functions $\hat{\phi}_f(\cdot)$ and $\hat{\psi}(\cdot)$ that lie within the feasible set $\mathcal{R}'_f(\rho)$ for some $\rho > \alpha_{\texttt{T}}$, the following relation always holds:
{\small
\begin{align*}
\int_0^{w} \kern-1em \hat{\phi}_f(u) du + (1-w) (p_{\max}+2\eta) - c w p_{\min} + \int_0^{v} \kern-1em \hat{\psi}_f(z) dz + (1-v) (p_{\max}(c+\varepsilon) + 2\delta)  =\\
\rho \left[ \frac{1+\varepsilon}{1+c+\varepsilon} \left( \hat{\phi}_f(w) + \hat{\psi}_f(v) - 2(\eta+\delta) \right) + \frac{2\delta}{T}\right] \ & \forall w \in [0,1], v \in [0,w],
\end{align*}
}
\end{lemma}

The proofs of both lemmas follow directly from the definitions of the feasible sets $\mathcal{R}'_b(\rho)$ and $\mathcal{R}'_f(\rho)$.  Using \sref{Corollary}{cor:pald-optimal-lower-bound-2} and \sref{Lemmas}{lem:pald-upper-bound-2}, \ref{lem:pald-threshold-function-relation-1-osdmt}, and \ref{lem:pald-threshold-function-relation-2-osdmt}, we claim that the following holds:
\begin{align*}
\frac{\PALD(\mathcal{I}) - p_{\max} \hat{s}}{\OPT(\mathcal{I})} &\leq \rho.
\end{align*}

To show this result, we first substitute the bounds from \sref{Corollary}{cor:pald-optimal-lower-bound-1} and \sref{Lemma}{lem:pald-upper-bound-1} into the left-hand side of the above equation.  We use the same shorthand notation as in the proof of \autoref{thm:pald-robustness} to facilitate the presentation.

{\small
\begin{align*}
\frac{Q\!+\!({\bm\beta}\!+\!D_f\!-\!\hat{W}) (p_{\max}\!+\!2 \eta)\!+\!{\bm \beta} \left( \bar{p} (c\!+\!\varepsilon)\!+\!2\delta \right) + (D_f\!-\!\hat{V}_f) (p_{\max} (c\!+\!\varepsilon)\!+\!2\delta)\!-\!c p_{\min} \hat{W}\!+\!(D_b\!-\!{\bm \beta}) (p_{\max}\!+\!2(\eta\!+\!\delta)\!+\!\bar{p}(c\!+\!\varepsilon))}{\sum_{i=1}^n G_i(\beta_i)\!+\!\varepsilon \bar{p} {\bm \beta}\!+\!\sum_{\tau=1}^T (1\!+\!\varepsilon) H_\tau (f_\tau)\!+\!(D_f\!+\!{\bm \beta}) \frac{2\delta}{T}\!+\!(D_b\!-\!{\bm \beta})(\tilde{p}\!+\!\varepsilon \bar{p}\!+\!\frac{2\delta}{T}) } \\
\end{align*}
}
Then, we have the following:
{\small
\begin{align*}
\leq \max \Bigg \{ \frac{Q + ({\bm\beta}+D_f - \hat{W}) (p_{\max} + 2 \eta) + {\bm\beta} \left( \bar{p} (c+\varepsilon) + 2\delta \right) + (D_f - \hat{V}_f) (p_{\max} (c + \varepsilon) + 2\delta) - c p_{\min} \hat{W}}{\sum_{i=1}^n G_i(\beta_i) + \varepsilon \bar{p} {\bm\beta} + \sum_{\tau=1}^T (1 + \varepsilon) H_\tau (f_\tau) + (D_f + {\bm\beta}) \frac{2\delta}{T} }, \ \ &\\
\quad \frac{(D_b\!-\!{\bm\beta}) (p_{\max}\!+\!2(\eta\!+\!\delta)\!+\!\bar{p}(c+\varepsilon))}{(D_b - {\bm\beta})(\tilde{p} + \varepsilon \bar{p} + \frac{2\delta}{T}) } \Bigg \}, & \\
\end{align*}
}
where the definition of $\tilde{p}$ ensures that the second term in the $\max$ is at most $\rho$.

\noindent We now focus on the first term.  For the sake of contradiction, suppose that
{\small
\begin{align}
    \frac{Q + ({\bm\beta}+D_f - \hat{W}) (p_{\max} + 2 \eta) + {\bm \beta} \left( \bar{p} (c+\varepsilon) + 2\delta \right) + (D_f - \hat{V}_f) (p_{\max} (c + \varepsilon) + 2\delta) - c p_{\min} \hat{W}}{\sum_{i=1}^n G_i(\beta_i) + \varepsilon \bar{p} {\bm \beta} + \sum_{\tau=1}^T (1 + \varepsilon) H_\tau (f_\tau) + (D_f + {\bm \beta}) \frac{2\delta}{T} } > \alpha_{\texttt{T}}. \label{eq:ratio-term-pald-2}
\end{align}
}
Instead of working directly with the expression in terms of ${\bm \beta}$, we first reason about how the cost of \OPT and \PALD relate to one another in terms of $\sum_{i=1}^n \sum_{j=1}^{\nu_i} B_b^{(i,j)}$, the total demand assigned to base drivers.
We introduce the following notation for the sake of brevity: let ${\bm B} = \sum_{i=1}^n \sum_{j=1}^{\nu_i} B_b^{(i,j)}$ denote the total base demand.  Then, we have the following relation:
{\small
\begin{align*}
& \frac{Q + ({\bm B}+D_f - \hat{W}) (p_{\max} + 2 \eta) + {\bm B} \left( \bar{p} (c+\varepsilon) + 2\delta \right) + (D_f - \hat{V}_f) (p_{\max} (c + \varepsilon) + 2\delta) - c p_{\min} \hat{W}}{\sum_{i=1}^n G_i(B_i) + \varepsilon \bar{p} {\bm B} + \sum_{\tau=1}^T (1 + \varepsilon) H_\tau (f_\tau) + (D_f + {\bm B}) \frac{2\delta}{T} }, \\
& \geq \frac{Q\!+\!({\bm \beta }\!+\!D_f\!-\!\hat{W}) (p_{\max}\!+\!2 \eta)\!+\!{\bm \beta } \left( \bar{p} (c\!+\!\varepsilon)\!+\!2\delta \right)\!+\!(D_f\!-\!\hat{V}_f) (p_{\max} (c\!+\!\varepsilon)\!+\!2\delta)\!-\!c p_{\min} \hat{W}\!+\!({\bm B}\!-\!{\bm \beta}) (p_{\max}\!+\!\bar{p}(c\!+\!\varepsilon)\!+\!2 \kappa)}{\sum_{i=1}^n G_i(\beta_i)\!+\!\varepsilon \bar{p} {\bm \beta }\!+\!\sum_{\tau=1}^T (1\!+\!\varepsilon) H_\tau (f_\tau)\!+\!(D_f\!+\!{\bm \beta }) \frac{2\delta}{T}\!+\!({\bm B}\!-\!{\bm \beta}) (\tilde{p}\!+\!\varepsilon \bar{p}\!+\!\frac{2\delta}{T})}, \\
& > \rho.
\end{align*}
}
where in the second inequality, we have used \sref{Lemma}{lem:upper-bound-B-beta-tilde-p}. We proceed to work with the expression in terms of ${\bm B}$.
During the $i^\text{th}$ active period and the lifetime of the $j^\text{th}$ base demand driver, the minimum marginal purchasing price observed is given by $\hat{\phi}_b^{(i,j)}(\hat{w}_b^{(i,j)}) - 2\eta$ by the definition of the threshold function.  Note that this follows because the feasible set $\mathcal{R}'_b(\rho)$ ensures that $\hat{\phi}_b^{(i,j)}(w)$ is monotone non-increasing for $w \in [0,1]$, and it is given that $\hat{\phi}_b^{(i,j)}(1) \leq p_{\min} + 2\eta$, guaranteeing that the true minimum price is captured by the threshold function.  

Using the same logic and by definition of the feasible set $\mathcal{R}'_f(\rho)$, we have that during the lifetime of the $\tau^\text{th}$ flexible demand driver, the minimum marginal purchasing and delivery cost observed is given by $\frac{1+\varepsilon}{1+c+\varepsilon} \left( \hat{\phi}_f^{(\tau)}(\hat{w}_f^{(\tau)}) + \hat{\psi}_f^{(\tau)}(\hat{v}_f^{(\tau)}) - 2 (\eta+\delta) \right)$.  
This gives the following lower bounds on the terms that depend on $G_i$ and $H_\tau$, respectively:
{\small
\begin{align*}
\sum_{i=1}^n G_i\left( \sum_{j=1}^{\nu_i} B_b^{(i,j)} \right) &\geq \sum_{i=1}^n \sum_{j=1}^{\nu_i} \left( \hat{\phi}_b^{(i,j)}(\hat{w}_b^{(i,j)}) - 2\eta \right) \times B_b^{(i,j)} \\
\sum_{\tau=1}^T (1 + \varepsilon) H_\tau (f_\tau) &\geq \sum_{\tau=1}^T \frac{1+\varepsilon}{1+c+\varepsilon} \left( \hat{\phi}_f^{(\tau)}(\hat{w}_f^{(\tau)}) + \hat{\psi}_f^{(\tau)}(\hat{v}_f^{(\tau)}) - 2 (\eta+\delta) \right) \times f_\tau.
\end{align*}
}

Substituting these bounds into the previous expression, we have that the left-hand-side of \eqref{eq:ratio-term-pald-2} is less than or equal to:
{\small
\begin{align*}
&\leq \frac{Q + ({\bm B}+D_f - \hat{W}) (p_{\max} + 2 \eta) + {\bm B} \left( \bar{p} (c+\varepsilon) + 2\delta \right) + (D_f - \hat{V}_f) (p_{\max} (c + \varepsilon) + 2\delta) - c p_{\min} \hat{W}}{\sum_{i=1}^n \sum_{j=1}^{\nu_i} \left( \hat{\phi}_b^{(i,j)}(\hat{w}_b^{(i,j)})\!-\!2\eta \right) B_b^{(i,j)}\!+\! \varepsilon \bar{p} {\bm B}\!+\!\sum_{\tau=1}^T \frac{1+\varepsilon}{1+c+\varepsilon} \left( \hat{\phi}_f^{(\tau)}(\hat{w}_f^{(\tau)})\!+\!\hat{\psi}_f^{(\tau)}(\hat{v}_f^{(\tau)})\!-\!2 (\eta+\delta) \right) f_\tau\!+\!(D_f\!+\!{\bm B}) \frac{2\delta}{T} }.
\end{align*}
}

By rearranging the terms in the above and substituting for $Q$, we obtain the following:
{\small
\begin{align}
&= \frac{\sum_{i=1}^n \sum_{j=1}^{\nu_i} \left[ X_b^{(i,j)} \right] + \sum_{\tau=1}^T \left[ X_f^{(\tau)} \right]}{\sum_{i=1}^n \sum_{j=1}^{\nu_i} \left[ \left( \hat{\phi}_b^{(i,j)}(\hat{w}_b^{(i,j)})\!-\!2\eta\!+\!\varepsilon \bar{p}\!+\!\frac{2\delta}{T}\right) B_b^{(i,j)}  \right]\!+\!\sum_{\tau=1}^T \left[ \frac{1+\varepsilon}{1+c+\varepsilon} \left( \hat{\phi}_f^{(\tau)}(\hat{w}_f^{(\tau)})\!+\!\hat{\psi}_f^{(\tau)}(\hat{v}_f^{(\tau)})\!-\!2 (\eta+\delta) \right) f_\tau\!+\!f_\tau \frac{2\delta}{T} \right]}, \label{eq:partitioned-sums-upper-bound-pald-2}
\end{align}
}

where $X_b^{(i,j)}$ and $X_f^{(\tau)}$ are defined as follows (for each driver):
{\small
\begin{align*}
X_b^{(i,j)} &= \int_0^{\hat{w}_b^{(i,j)}} \kern-1em \hat{\phi}_b^{(i,j)}(u) du + (B_b^{(i,j)} -\hat{w}_b^{(i,j)}) (p_{\max}+2\eta) + \left( p_{\max}(c+\varepsilon) + 2\delta \right) B_b^{(i,j)} - c \hat{w}_b^{(i,j)} p_{\min}, \\
X_f^{(\tau)} &= \int_0^{\hat{w}_f^{(\tau)}} \kern-1em \hat{\phi}_f^{(\tau)}(u) du + (f_\tau -\hat{w}_f^{(\tau)}) (p_{\max}+2\eta) + \int_0^{\hat{v}_f^{(\tau)}} \kern-1em \hat{\psi}_f^{(\tau)}(u) du + \left( p_{\max} (c+\varepsilon) + 2\delta \right)(f_\tau - \hat{v}_f^{(\tau)}) - c \hat{w}_f^{(\tau)} p_{\min}.
\end{align*}
}

Since \eqref{eq:partitioned-sums-upper-bound-pald-2} is $> \rho$ by assumption, one of the following cases must be true:

\smallskip
\noindent\textbf{Case I:} There exists some $i \in [n]$ and $j \in [\nu_i]$ such that
{\small
\begin{align*}
X_b^{(i,j)} &> \rho  B_b^{(i,j)} \left( \hat{\phi}_b^{(i,j)}(\hat{w}_b^{(i,j)}) - 2\eta + \varepsilon \bar{p} + \frac{2\delta}{T}\right).
\end{align*}
}

\smallskip
\noindent\textbf{Case II:} There exists some $\tau \in [T]$ such that
{\small
\begin{align*}
X_f^{(\tau)} &> \rho f_\tau \left[ \frac{1+\varepsilon}{1+c+\varepsilon} \left( \hat{\phi}_f^{(\tau)}(\hat{w}_f^{(\tau)}) + \hat{\psi}_f^{(\tau)}(\hat{v}_f^{(\tau)}) - 2 (\eta+\delta) \right) + \frac{2\delta}{T} \right],
\end{align*}
}

But, if either of these two cases are true (i.e., for any unit of base or flexible demand), they contradict \sref{Lemmas}{lem:pald-threshold-function-relation-1-osdmt} and \ref{lem:pald-threshold-function-relation-2-osdmt}, respectively.  Thus, we have a contradiction, and the original assumption that \eqref{eq:partitioned-sums-upper-bound-pald-2} is $> \rho$ must be false.  This completes the proof that $\frac{\PALD(\mathcal{I}) - p_{\max} \hat{s}}{\OPT(\mathcal{I})} \leq \rho$. 

Using this result, we have:
\begin{align*}
\PALD(\mathcal{I}) &\leq \rho \OPT(\mathcal{I}) + p_{\max} \hat{s} \leq \rho \OPT(\mathcal{I}) + p_{\max} S,
\end{align*}
where $p_{\max} S$ is a constant.  This shows that \PALD is $\rho$-competitive for \OSDMT under \sref{Definition}{dfn:comp-ratio}, which completes the proof.

\end{proof}

\subsection{Proof of \sref{Lemma}{lem:convex-robust-set}} \label{apx:convex-robust-set}

In this section, we prove \sref{Lemma}{lem:convex-robust-set}, which states that if $\phi_b, \phi_f, \psi_f$ are piecewise-affine functions parameterized by $\mathbf{y}_b, \mathbf{y}_f, \mathbf{y}_\psi \in \mathbb{R}^K$ for $K > 1$, the feasible sets $\mathcal{R}_b(\rho)$ and $\mathcal{R}_f(\rho)$ defined in \sref{Def.}{dfn:robust-threshold-base} and \sref{Def.}{dfn:robust-threshold-flexible} are convex sets in $\mathbf{y}_b$ and $(\mathbf{y}_f, \mathbf{y}_\psi)$, respectively.

\begin{proof}[Proof of \sref{Lemma}{lem:convex-robust-set}]

Denote the grid discretization points by $0 = w_1 < w_2 < \cdots < w_K = 1$, with $w_j = \frac{j-1}{K-1}$.  Each threshold function $\phi: [0,1] \to \mathbb{R}$ or $\psi: [0,1] \to \mathbb{R}$ is parameterized by its knot values $\mathbf{y} = (y_1, \ldots, y_K)$ and linear interpolation on each interval $[w_j, w_{j+1}]$. 

Then we want to show that $\mathcal{R}_b(\rho)$ is a convex set in $\mathbf{y}_b$ and $\mathcal{R}_f(\rho)$ is a convex set in $(\mathbf{y}_f, \mathbf{y}_\psi)$.

First, we introduce the hat-basis representation of piecewise-affine functions, which will be useful in the proof.  Then we prove the convexity of $\mathcal{R}_b(\rho)$ and $\mathcal{R}_f(\rho)$ in their respective parameter spaces.

Let $\{ h_j(\cdot) \}_{j=1}^K$ be the standard PWA ``tent'' basis: each $h_j$ is piecewise-linear, $h_j(w_i)=\mathbf{1}\{i=j\}$, supported on $[w_{j-1},w_{j+1}]$ (with boundary conventions for $j=1,K$), and $\sum_j h_j(u)=1$ for all $u\in[0,1]$. Under linear interpolation, we have:
\[
\phi(u)=\sum_{j=1}^K y_j h_j(u),\qquad
\Phi(0,w)=\int_0^w \phi(u)\,du=\sum_{j=1}^K y_j H_j(w),
\]
so for fixed $u,w$, both $\phi(u)$ and $\Phi(0,w)$ are affine in the parameter vector $\mathbf y$.

\smallskip
\noindent\textbf{Base Drivers. \ }
All defining constraints of $\mathcal R_b(\rho)$ are affine in $\mathbf y_b$:
(i) range and endpoint: $p_{\min}\le y_{b,j}\le p_{\max}$ and $y_{b,K}\le p_{\min}+2\gamma$ (box and halfspace);
(ii) monotone non-increasing: $y_{b,j+1}-y_{b,j}\le 0$ for $j=1,\dots,K-1$ (segment slopes);
(iii) robustness inequality for fixed $w\in[0,1]$:
\[
\sum_{j=1}^K\!\big(H_j(w)-\rho\,h_j(w)\big)\,y_{b,j}
\;\le\;
\rho\!\left(\varepsilon p_{\max}+\tfrac{2\kappa}{T}-2\gamma\right)
-(1-w)(p_{\max}+2\gamma)-p_{\max}(c+\varepsilon)-2\delta+cwp_{\min},
\]
which is a single affine halfspace in $\mathbf y_b$. The condition “for all $w$” is the intersection (possibly uncountable) of affine halfspaces, hence convex. Intersecting with the box and linear inequalities preserves convexity. Thus $\mathcal R_b(\rho)$ is convex.

\smallskip
\noindent\textbf{Flexible Drivers. \ }
Analogously, in the product space $(\mathbf y_f,\mathbf y_\psi)$:
(i) range and endpoints give affine constraints on $y_{f,j}$ and $y_{\psi,j}$;
(ii) monotonicity gives $y_{f,j+1}-y_{f,j}\le 0$ and $y_{\psi,j+1}-y_{\psi,j}\le 0$;
(iii) for fixed $(w,v)$ with $0\le v\le w\le 1$, the robustness inequality

{\small
\begin{align*}
& \sum_{j=1}^K H_j(w)\,y_{f,j} + \sum_{j=1}^K H_j(v)\,y_{\psi,j} \le \\
& \ \ \rho \left[\tfrac{1}{\omega} \left(\sum_{j=1}^K h_j(w)\,y_{f,j}+\sum_{j=1}^K h_j(v)\,y_{\psi,j}-2\kappa\right)
+\tfrac{2\kappa}{T} \right]
-(1-w)(p_{\max}+2\gamma)- (1-v)(p_{\max}(c+\varepsilon)+2\delta)+cwp_{\min}
\end{align*}
}

is an affine halfspace in $(\mathbf y_f,\mathbf y_\psi)$. Intersecting these halfspaces over all $(w,v)$ and with the linear constraints yields a convex set. Hence $\mathcal R_f(\rho)$ is convex.

Therefore, both feasible sets are convex in their respective parameter spaces.
\end{proof}

\end{document}